\PassOptionsToPackage{noabbrev,capitalize}{cleveref}

\documentclass[conference]{IEEEtran}
\usepackage{caption}
\def\BibTeX{{\rm B\kern-.05em{\sc i\kern-.025em b}\kern-.08em
    T\kern-.1667em\lower.7ex\hbox{E}\kern-.125emX}}

\usepackage{amsthm}
\usepackage[all]{xy}
\usepackage{amssymb}
\usepackage[inference]{semantic}
\usepackage{CTLC}
\usepackage{bitml}
\usepackage{paralist}
\usepackage{nicefrac}
\usepackage{siunitx}
\usepackage{array,framed}
\usepackage{booktabs}
\usepackage{
  color,
  float,
  epsfig,
  wrapfig,
  graphics,
  subcaption
}
\usepackage{textcomp}
\usepackage{setspace}
\usepackage{latexsym,fancyhdr,url}
\usepackage{enumerate}
\usepackage{algorithm2e}
\usepackage{algpseudocode}
\usepackage{graphics}
\usepackage{xparse} 
\usepackage{xspace}
\usepackage{multirow}
\usepackage{csvsimple}
\usepackage{balance}
\usepackage{
  tikz,
  pgfplots,
  pgfplotstable
}
\usepackage[hidelinks]{hyperref}
\usepackage{cryptocode}
\usepackage{amsmath}
\usepackage{cleveref}
\usepackage{array}
\usepackage{url}

\usetikzlibrary{
  shapes.geometric,
  arrows,
  external,
  pgfplots.groupplots,
  matrix
}

\pgfplotsset{compat=1.9}
\usepackage{mathtools}

\DeclareMathAlphabet{\mathcal}{OMS}{cmsy}{m}{n}
\DeclareGraphicsExtensions{%
    .png,.PNG,%
    .pdf,.PDF,%
    .jpg,.mps,.jpeg,.jbig2,.jb2,.JPG,.JPEG,.JBIG2,.JB2}

\ifCLASSOPTIONcompsoc
  \usepackage[nocompress]{cite}
\else
  \usepackage{cite}
\fi

\usepackage{listings, xcolor}

\definecolor{verylightgray}{rgb}{.97,.97,.97}

\lstdefinelanguage{Solidity}{
	keywords=[1]{anonymous, assembly, assert, balance, break, call, callcode, case, catch, class, constant, continue, constructor, contract, debugger, default, delegatecall, delete, do, else, emit, event, experimental, export, external, false, finally, for, function, gas, if, implements, import, in, indexed, instanceof, interface, internal, is, length, library, log0, log1, log2, log3, log4, memory, modifier, new, payable, pragma, private, protected, public, pure, push, require, return, returns, revert, selfdestruct, send, solidity, storage, struct, suicide, super, switch, then, this, throw, transfer, true, try, typeof, using, value, view, while, with, addmod, ecrecover, keccak256, mulmod, ripemd160, sha256, sha3}, 
	keywordstyle=[1]\color{blue}\bfseries,
	keywords=[2]{address, bool, byte, bytes, bytes1, bytes2, bytes3, bytes4, bytes5, bytes6, bytes7, bytes8, bytes9, bytes10, bytes11, bytes12, bytes13, bytes14, bytes15, bytes16, bytes17, bytes18, bytes19, bytes20, bytes21, bytes22, bytes23, bytes24, bytes25, bytes26, bytes27, bytes28, bytes29, bytes30, bytes31, bytes32, enum, int, int8, int16, int24, int32, int40, int48, int56, int64, int72, int80, int88, int96, int104, int112, int120, int128, int136, int144, int152, int160, int168, int176, int184, int192, int200, int208, int216, int224, int232, int240, int248, int256, mapping, string, uint, uint8, uint16, uint24, uint32, uint40, uint48, uint56, uint64, uint72, uint80, uint88, uint96, uint104, uint112, uint120, uint128, uint136, uint144, uint152, uint160, uint168, uint176, uint184, uint192, uint200, uint208, uint216, uint224, uint232, uint240, uint248, uint256, var, void, ether, finney, szabo, wei, days, hours, minutes, seconds, weeks, years},	
	keywordstyle=[2]\color{teal}\bfseries,
	keywords=[3]{block, blockhash, coinbase, difficulty, gaslimit, number, timestamp, msg, data, gas, sender, sig, value, now, tx, gasprice, origin},	
	keywordstyle=[3]\color{violet}\bfseries,
	identifierstyle=\color{black},
	sensitive=true,
	comment=[l]{//},
	morecomment=[s]{/*}{*/},
	commentstyle=\color{gray}\ttfamily,
	stringstyle=\color{red}\ttfamily,
	morestring=[b]',
	morestring=[b]"
}

\usepackage{xparse}
\newcommand{\bnm}{\begin{newmath}}
\newcommand{\enm}{\end{newmath}}

\newcommand{\bea}{\begin{eqnarray*}}%
\newcommand{\eea}{\end{eqnarray*}}%

\newcommand{\bne}{\begin{newequation}}
\newcommand{\ene}{\end{newequation}}

\newcommand{\bal}{\begin{newalign}}
\newcommand{\eal}{\end{newalign}}

\newenvironment{newalign}{\begin{align}%
\setlength{\abovedisplayskip}{4pt}%
\setlength{\belowdisplayskip}{4pt}%
\setlength{\abovedisplayshortskip}{6pt}%
\setlength{\belowdisplayshortskip}{6pt} }{\end{align}}

\newenvironment{newmath}{\begin{displaymath}%
\setlength{\abovedisplayskip}{4pt}%
\setlength{\belowdisplayskip}{4pt}%
\setlength{\abovedisplayshortskip}{6pt}%
\setlength{\belowdisplayshortskip}{6pt} }{\end{displaymath}}

\newenvironment{newequation}{\begin{equation}%
\setlength{\abovedisplayskip}{4pt}%
\setlength{\belowdisplayskip}{4pt}%
\setlength{\abovedisplayshortskip}{6pt}%
\setlength{\belowdisplayshortskip}{6pt} }{\end{equation}}

\newcounter{ctr}

\newcounter{mytable}
\def\mytable{\begin{centering}\refstepcounter{mytable}}
\def\endmytable{\end{centering}}

\newcounter{myfig}
\def\myfig{\begin{centering}\refstepcounter{myfig}}
\def\endmyfig{\end{centering}}

\newlength{\saveparindent}
\newlength{\saveparskip}
\newcommand{\E}{{\rm I\kern-.3em E}}

\renewcommand{\eqref}[1]{\mbox{Equation~(\ref{#1})}}

\newcommand{\R}{\mathbb{R}}
\newcommand{\N}{\mathbb{N}}
\newcommand{\Z}{\mathbb{Z}}

\def \part {part}

\renewcommand{\paragraph}[1]{\smallskip \noindent\textbf{#1.}\;}

\def \blackslug{\hbox{\hskip 1pt \vrule width 4pt height 8pt
    depth 1.5pt \hskip 1pt}}
\def \qed{\quad\blackslug\lower 8.5pt\null\par}

\newcounter{mynote}[section]

\newcommand\ignore[1]{}

\newcounter{rcnote}[section]

\newcounter{mrnote}[section]

\newcounter{fknote}[section]

\newcounter{anote}[section]

\DeclareMathSymbol{\mlq}{\mathord}{operators}{``}
\DeclareMathSymbol{\mrq}{\mathord}{operators}{`'}

\newcommand{\rhf}[2]{R_{f, \gamma}}

\DeclareDocumentCommand{\edist}{o o}{
  \ensuremath{
    \IfNoValueTF{#1}{{d}}{{\sf d}(#1,#2)}
  }
}

\newcommand{\olrk}[1]{\ifx\nursymbol#1\else\!\!\mskip4.5mu plus 0.5mu\left(\mskip0.5mu plus0.5mu #1\mskip1.5mu plus0.5mu \right)\fi}

\NewDocumentCommand{\indseq}{ O{1} O{r} }{{#1}\ldots {#2}}

\newcommand{\run}{R}

\newcommand{\fulldigraph}{\mathcal{D} = (\mathcal{N}, \mathcal{A})}
\newcommand{\nodes}{\mathcal{N}}

\newcommand{\atg}{atomic transfer graph\xspace}
\newcommand{\atgs}{atomic transfer graphs\xspace}

\newcommand{\ATGs}{Atomic Transfer Graphs\xspace}
\newcommand{\atgacronym}{ATG\xspace}
\newcommand{\atgsacronym}{ATGs\xspace}

\newcommand{\hbe}{heterogenous blockchain ecosystem\xspace}

\newcommand{\hbeacronym}{HBE\xspace}

\newcommand{\swapgraphs}{swap graphs\xspace}

\newcommand{\gtree}{xtree\xspace}
\newcommand{\Gtree}{Xtree\xspace}
\newcommand{\GTree}{Xtree\xspace}
\newcommand{\gtrees}{xtrees\xspace}
\newcommand{\Gtrees}{Xtrees\xspace}
\newcommand{\GTrees}{Xtrees\xspace}

\newcommand{\tam}{TAM\xspace}

\newcommand{\tams}{TAMs\xspace}
\newcommand{\Tams}{TAMs\xspace}
\newcommand{\tamfull}{transfer agreement mechanism\xspace}

\newcommand{\tamsfull}{transfer agreement mechanisms\xspace}

\newcommand{\tammath}{\textit{tam}}

\newcommand{\chenvs}{\tam environments\xspace}

\newcommand{\twopartyaswap}{two-party atomic swap\xspace}

\newcommand{\actlc}{ethCTLC\xspace}

\newcommand{\nodesymbol}{\ensuremath{\mathcal{N}}\xspace}
\newcommand{\arcsymbol}{\ensuremath{\mathcal{A}}\xspace}
\newcommand{\graphsymbol}{\ensuremath{\mathcal{D}}\xspace}
\newcommand{\treesymbol}{\ensuremath{\mathcal{T}}\xspace}

\newcommand{\graphtext}{graph\xspace}
\newcommand{\graphstext}{graphs\xspace}

\newcommand{\arctext}{arc\xspace}
\newcommand{\arcstext}{arcs\xspace}
\newcommand{\nodetext}{node\xspace}
\newcommand{\nodestext}{nodes\xspace}

\newcommand{\edgetext}{edge\xspace}

\newcommand{\edgestext}{edges\xspace}
\newcommand{\Edgestext}{Edges\xspace}

\newcommand{\htlc}{HTLC\xspace}
\newcommand{\htlcs}{HTLCs\xspace}
\newcommand{\ctlc}{CTLC\xspace}
\newcommand{\ctlcs}{CTLCs\xspace}
\newcommand{\fundtext}{fund\xspace}
\newcommand{\fundstext}{funds\xspace}

\newcommand{\CTLClong}{Conditional Timelock Contract}
\newcommand{\CTLCslong}{Conditional Timelock Contracts}
\newcommand{\htlclong}{hashed timelock contract\xspace}

\newcommand{\ctlcprotocol}{\ctlc-based protocol\xspace}

\newcommand{\arcpaths}{E}
\newcommand{\levels}{\mathcal{L}}
\newcommand{\levelsecrets}[1]{\psi_{#1}}
\newcommand{\funconditions}{\textit{cond}}
\newcommand{\edgesecretfull}[2]{s_{#1}^{#2}}

\definecolor{darkred}{HTML}{E6194B}
\definecolor{lightblue}{HTML}{42D4F4}
\definecolor{middlegreen}{HTML}{3CB44B}
\definecolor{middleorange}{HTML}{F58231}
\definecolor{darkblue}{HTML}{000075}
\definecolor{kingblue}{HTML}{4363D8}

\newcommand{\honestuser}{B}

\newcommand{\ctlcargs}[3]{\textit{CTLC}(#1, #2, #3)}
\newcommand{\subctlcargs}[2]{\textit{sCTLC}(#1, #2)}
\newcommand{\htlcargs}[4]{\textit{HTLC}(#1, #2, #3, #4)}
\newcommand{\edgevar}{e}
\newcommand{\arc}{a}

\newcommand{\chainA}{\mathbb{A}}
\newcommand{\chainB}{\mathbb{B}}
\newcommand{\sender}{\textsf{S}}
\newcommand{\receiver}{\textsf{R}}
\newcommand{\contractAB}{{\color{darkred}c_{\textit{AB}}}}
\newcommand{\contractBA}{{\color{lightblue}{c_\textit{BA}}}}
\newcommand{\outcomes}[1]{\mathcal{O}^{\untree}_{#1}}

\newcommand{\level}{\ell}
\newcommand{\edgesecret}[2]{s^{\scriptscriptstyle{#2}}}
\newcommand{\ctlcvar}{c}
\newcommand{\subctlcvar}{\textit{sc}}
\newcommand{\ctlcsubcontracts}{\textit{SCs}}

\newcommand{\attackerstrategy}{\Sigma_{\mathcal{A}}}
\newcommand{\userstrategy}[1]{\Sigma_{#1}}
\newcommand{\sstep}[1]{\xrightarrow{#1}}
\newcommand{\conforms}{\vdash}

\newcommand*\colorcircled[2]{\tikz[baseline=(char.base)]{
            \node[shape=circle,draw,inner sep=1pt, color=#2, text=black] (char) {#1};}}

\newcommand*\colorrectboxed[2]{\tikz[baseline=(char.base)]{
                \node[shape=rectangle,draw,inner sep=1pt, color=#2, text=black] (char) {#1};}}

\newcommand{\edgeBAone}{\colorcircled{1}{darkred}}
\newcommand{\edgeCAone}{\colorcircled{6}{middlegreen}}
\newcommand{\edgeABtwo}{\colorcircled{2}{lightblue}}

\newcommand{\edgeBCtwo}{\colorcircled{8}{darkblue}}
\newcommand{\edgeABthree}{\colorcircled{9}{lightblue}}
\newcommand{\edgeCBthree}{\colorcircled{10}{kingblue}}

\newcommand{\subctlcABtwo}{\colorrectboxed{$\subctlcvar^2$}{lightblue}}
\newcommand{\subctlcABthree}{\colorrectboxed{$\subctlcvar^3$}{lightblue}}

\newif\iffullversion
\fullversiontrue

\begin{document}

\title{Atomic Transfer Graphs: Secure-by-design Protocols for Heterogeneous Blockchain Ecosystems}

        \newtheorem{definition}{Definition}[section]
        \newtheorem{theorem}[definition]{Theorem}
        \newtheorem{corollary}[definition]{Corollary}
        \newtheorem{lemma}[definition]{Lemma}
        \newtheorem{prop}[definition]{Proposition}
        \newtheorem{remark}[definition]{Remark}
        \newtheorem{example}[definition]{Example}

\author{
    \IEEEauthorblockN{Stephan Dübler}
    \IEEEauthorblockA{\textit{MPI-SP} \\
    stephan.duebler@mpi-sp.org}
    \and
    \IEEEauthorblockN{Federico Badaloni}
    \IEEEauthorblockA{\textit{MPI-SP} \\
    federico.badaloni@mpi-sp.org}%
    \and
    \and
    \IEEEauthorblockN{Pedro Moreno-Sanchez}
    \IEEEauthorblockA{\textit{IMDEA Software Institute} \\
    \textit{VISA Research}\\
    \textit{MPI-SP}\\
    pedro.moreno@imdea.org}
    \and
    \IEEEauthorblockN{Clara Schneidewind}
    \IEEEauthorblockA{\textit{MPI-SP} \\
    clara.schneidewind@mpi-sp.org}
}

\maketitle

\begin{abstract}

The heterogeneity of the blockchain landscape has motivated the design of blockchain protocols tailored to specific blockchains and applications that, hence, require custom security proofs. 
We observe that many blockchain protocols share common security and functionality goals, which can be captured by an \emph{\atg} (\atgacronym) describing the structure of desired transfers. 
Based on this observation, we contribute a framework for generating secure-by-design protocols that realize these goals.
The resulting protocols build upon \emph{\CTLCslong} (\ctlcs), a novel minimal smart contract functionality that can be implemented in a large variety of cryptocurrencies with a restricted scripting language (e.g., Bitcoin), and payment channels. 
We show how \atgsacronym, in addition to enabling novel applications, capture the security and functionality goals of existing applications, including many examples from payment channel networks and complex multi-party cross-currency swaps among Ethereum-style cryptocurrencies.  
Our framework is the first to provide generic and provably secure protocols for all these use cases while matching or improving the performance of existing use-case-specific protocols. 

\end{abstract}

\IEEEpeerreviewmaketitle

\section{Introduction}
\label{sec:intro}

Many existing \emph{blockchain protocols}, at their core, rely on the atomic execution of a set of (financial) transfers. A canonical example of such protocols are multi-hop payment protocols in payment channel networks (PCNs)~\cite{multi-hop-locks}: PCNs are a widely used solution to address the scalability issues of cryptocurrencies~\cite{GudgeonMRMG20}. In PCNs, pairs of users lock funds into so-called payment channels (PCs), where they can renegotiate the ownership distribution of those funds. A renegotiation of the PC funds effectively realizes an \emph{off-chain} transfer between the PC parties, which is not recorded on the blockchain.
Multiple bilateral off-chain transfers can be securely chained into a multi-hop payment with the help of a cryptographic protocol~\cite{state-channel-networks,lightning-network,fulgor,multi-hop-locks,teechain}. 
Such a multi-hop payment protocol enables an off-chain payment between users who do not share a PC by routing the payment through a path of PCs that connects the sender and receiver in the PCN. 

For a multi-hop payment protocol to be secure, it needs to be guaranteed that the users acting as intermediaries on the payment paths (so those who forward the payment on their respective PCs) never lose funds: Whenever an honest intermediary transfers funds to their successor on the payment path, they also need to receive the transfer from their predecessor. In other words, from the perspective of honest intermediate users, the transfers should be executed atomically. 

\begin{figure}
\includegraphics[width=\columnwidth]{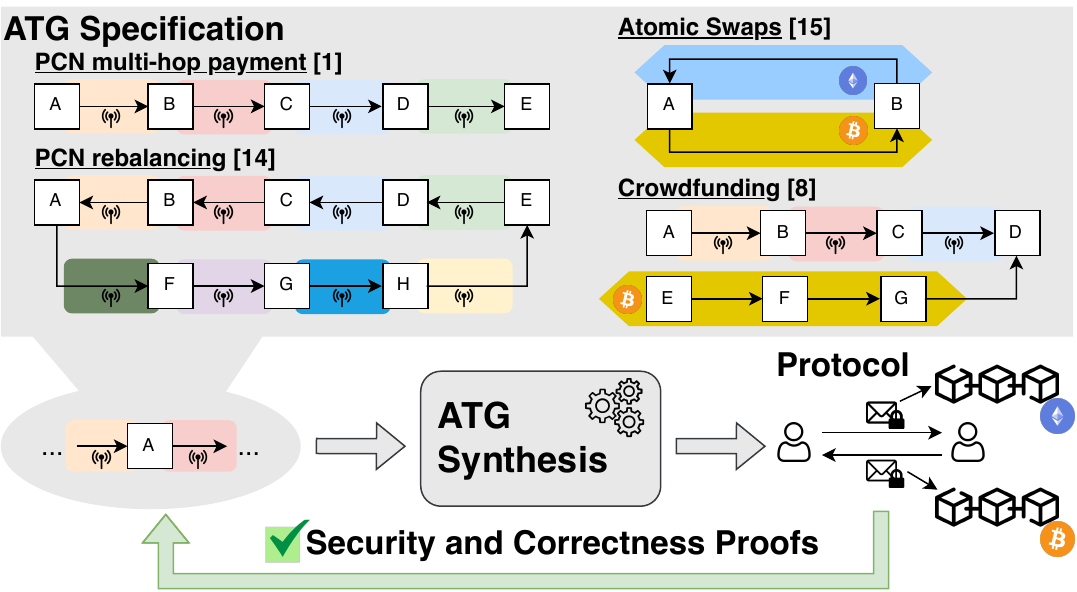}
\caption{Applications covered by our work can be specified as an ATG (gray box on top) and then, using our framework, synthesized into a protocol involving different transfer agreement mechanisms (e.g., blockchains or PCs). 
In ATGs, square boxes represent users, and arcs represent transfers. Round, colored boxes indicate PCs, and hexagons indicate a blockchain.}
\label{fig:intro-applications}
\end{figure}

A multi-hop payment, hence, can be specified as a linear graph of transfers, with nodes representing the users on the payment path and the arcs denoting the transfers on PCs between the users of the arc.
A graph for a multi-hop payment from user $A$ to $E$ is illustrated in~\cref{fig:intro-applications} (top left). 

Such a graph, which we will hereby call \emph{atomic transfer graph} (in short ATG), can be seen as a specification of the protocol in terms of functionality and security: The goal of the protocol execution is to execute all transfers of the graph (functionality). This goal must be achieved when the protocol users are honest. But even if all but one user is malicious, the protocol should still ensure that honest users obtain an outcome that is at least as good as the one resulting from an atomic execution of the transfers in the graph (security). 

In this work, we are interested in ATGs that feature diverse topologies (e.g., include cycles) and whose arcs may indicate transfers not only in payment channels but also in (arbitrary) cryptocurrencies or a combination thereof. 
To reflect that transfers may be realized using different agreement mechanisms (e.g., a blockchain or a payment channel built on-top of a blockchain), we say that an ATG specifies atomic transfers in a \hbe (short \hbeacronym). 
Such general ATGs serve as a specification format for applications beyond multi-hop payments, including crowdfunding~\cite{pathshuffle,amcu-crowdfunding}, where several users atomically fund a certain receiver; rebalancing in PCNs~\cite{revive,hide-and-seek,cycle,shaduf,amcu,zeta-rebalancing}, where a cycle payment is used to redistribute balances among the involved PCs; or atomic swaps~\cite{atomic-swaps,p2dex,bitcoin-wiki-htlc,htlc-bip,thyagarajan2022universal,BentovJ0BDJ19,imoto2023atomic}, where users intent to atomically exchange several assets of their interest held at different cryptocurrencies; and beyond. 

\paragraph{Problem} Despite their structural similarity, currently, all of the mentioned use cases are solved by custom protocols whose corresponding functionality and security notions are restated and adapted for each use case and then proven from scratch (requiring involved proofs in complex cryptographic proof frameworks~\cite{TairiMS23}). This is not only cumbersome and error-prone, as demonstrated by the security flaws found so far in blockchain protocols proposed by both academia~\cite{foundations-coin-mixing,foundations-adaptor-signatures} and industry~\cite{multi-hop-locks,flood-loot}, but also hinders the design of new protocols that aim to achieve similar functionality and security goals.

To help this, in this work, we want to answer the following research question: \textit{Can we generate secure-by-design blockchain protocols from an ATG specification?}

\paragraph{Challenges} So far, the only attempts towards generating blockchain protocols from a high-level specification format have been made in the area of atomic cross-chain swaps. The goal of an atomic cross-chain swap is to allow users with funds on different cryptocurrencies to securely exchange these funds (without involving a trusted party). While many concrete protocols exist for the two-party case, Herlihy, in \cite{atomic-swaps}, develops a generic protocol for more complex swap scenarios among multiple users, which can be expressed as strongly connected graphs (also called swap graphs).

These swap graphs are a special form of ATGs where the arcs represent on-chain transactions in different cryptocurrencies. However, the protocol proposed by Herlihy does not apply to our scenario for two fundamental reasons:

1) The protocol is inherently limited to cryptocurrencies that
allow for locking funds into complex smart contracts. 
Smart contracts are programs that govern funds in a cryptocurrency. 
Cryptocurrencies like Ethereum support rich, stateful smart contract languages, while others, e.g., Bitcoin, only allow for locking funds under simple payment conditions, such as passing a certain time (timelock) or providing the preimage of hash value (hashlock).
Herlihy explicitly states that it is \emph{an open research problem to develop a generic protocol that would enable cross-chain swaps based on simple payment conditions.}
Overcoming this limitation boosts protocol performance and, most importantly, enables swaps of funds in cryptocurrencies like Bitcoin. 
Supporting cryptocurrencies with simple payment conditions is a prerequisite for integrating off-chain solutions of these currencies (e.g., PCs on Bitcoin) and, hence, for realizing any of the aforementioned applications.

2) The protocol (and its security and correctness proofs) only covers strongly connected graphs. However, many interesting use cases (e.g., multi-hop payments) are characterized by ATGs that are not strongly connected.

\paragraph{Our Approach} We overcome the aforementioned challenges by providing the first framework for generating secure-by-design blockchain protocols from a large class of ATG specifications (c.f.~\cref{fig:intro-applications}). 
To capture a broad range of existing application scenarios, we abstract from the concrete transfer mechanism (e.g., a PC or a concrete cryptocurrency) and build our protocol on an abstraction layer that we call \emph{transfer agreement mechanism} (or \tam in short). 
As depicted in~\cref{fig:intro-applications} (bottom), our framework takes as input an ATG specification and it (1) generates the corresponding blockchain protocol that can be executed on any \tam that supports locking funds under simple payment conditions; (2) provides generic security and correctness notions (based on the ATG specification); (3) provides a generic security and correctness proof. 
As a by-product, our approach generates general cross-chain swap protocols that are more efficient than those presented in \cite{atomic-swaps}.

\paragraph{Our Contributions} We make the following contributions:

\begin{asparaitem}
\item 
For defining provably secure protocols for arbitrary \atgsacronym, we introduce the concept of a \emph{transfer tree} (short \gtree, c.f.~\cref{cha:trees}) as a general intermediate representation for blockchain protocols across different \tams. 
\Gtrees capture the flow of such blockchain protocols in terms of a multi-stage, multi-player \fundtext redistribution game and, as such, constitute a contribution of independent interest.
We show how to synthesize a \gtree for a given \atgacronym and prove the security and correctness of the resulting \gtree w.r.t. the \atgacronym. 

\item 
We provide a generic protocol to execute a given \gtree on \tams that support a novel building block called \emph{\CTLClong \, (\ctlc)} (c.f.~\cref{sec:CTLCs}). 
A \CTLC\ can be realized on any \tam that supports simple payment conditions (i.e., timelocks and transfer authorization based on digital signatures),  
so we solve the open research problem stated by Herlihy in~\cite{atomic-swaps}. 
We formally prove the security and correctness of our \ctlc-based protocol w.r.t. the original \gtree. 
We carry out this security analysis in a symbolic model of users interacting with multiple \tams in an adversarial environment, where the attacker controls the order of \tam interactions.

\item 
To demonstrate the practical applicability of our framework we
(1) show how to capture the functionality and security requirements of many existing cross-chain and off-chain applications in terms of \atgsacronym (c.f.~\cref{sec:applications}); 
(2) validate \ctlc-supporting \tams as an abstraction layer by demonstrating how to realize \ctlcs on virtually all existing cryptocurrencies and PCNs (c.f.~\cref{sec:ctlc-impl}).

\end{asparaitem}

\section{Key Ideas}
\label{sec:keyideas}

In this section, we overview our approach to systematically construct protocols that realize the functionality and security goals specified by an \atg (\atgacronym). 
We use atomic swaps as the running example application as they illustrate well the challenges of constructing such protocols and the design rationale for our approach. 

\paragraph{A Simple Atomic Swap Protocol}
We first revisit simple \twopartyaswap protocols~\cite{thyagarajan2022universal}, which are among the many supported applications of our framework.

\begin{figure}[b]
    \centering
    \includegraphics[width= 0.7\columnwidth]{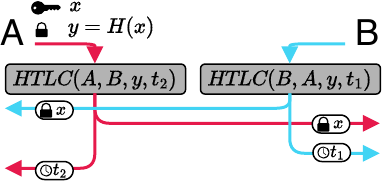}
    \caption{Two-party atomic swap protocol.}
    \label{fig:two-party-swap}
\end{figure}

In a \twopartyaswap protocol, two users, Alice (A) and Bob (B), want to exchange \fundstext that they own on different blockchains (called chain $\chainA$ and chain $\chainB$ here).
To this end, as illustrated in~\cref{fig:two-party-swap}, they first lock the corresponding \fundstext in both $\chainA$ and $\chainB$ by transferring them into a simple \emph{\htlclong} (\htlc). 
A \htlclong $\htlcargs{\sender}{\receiver}{y}{t}$ is a smart contract whose funds can only be withdrawn in two ways: 
1) 
The \fundstext are transferred to the receiver $\receiver$ given that a secret value $x$ for condition $y$ is provided such that $H(x) = y$ (for some fixed hash function $H$); we say the \htlc got \emph{claimed} in this case.
2) 
The \fundstext are transferred (back) to the sender $\sender$ after the time reaches timelock $t$; we say the \htlc got \emph{refunded} in this case.
To initiate an atomic swap, $A$ chooses a secret $x$ and transfers their funds into a contract $\contractAB = \htlcargs{A}{B}{y}{t_2}$ (with $y = H(x)$) on $\chainA$. 
Based on this, $B$ transfers their \fundstext into a contract $\contractBA = \htlcargs{B}{A}{y}{t_1}$ on $\chainB$ such that $t_1 < t_2$.
To initiate the swap, $A$ withdraws $B$'s \fundstext from $\contractBA$ before $t_1$, publishing secret $x$.
Bob, learning $x$, can claim the \fundstext from $\contractAB$ before $t_2$.
The timelocks $t_1$ and $t_2$ ensure that $A$'s and $B$'s \fundstext will not be indefinitely locked in the respective contracts.
It is crucial that there is time between $t_1$ and $t_2$ for Bob to safely refund their \fundstext before Alice can do so. 
Otherwise, a malicious $A$ could potentially withdraw the \fundstext from both $\contractAB$ and $\contractBA$ at time $t_2$.

\paragraph{Fund Redistribution Games} 
While the \twopartyaswap protocol above is general in that it relies on the \htlc primitive, which is known to be realizable in many cryptocurrencies (including Ethereum and Bitcoin), it cannot be easily generalized to more complicated scenarios involving additional users or transfers. 
To overcome this limitation, we take a more systematic view by observing that the protocol relies on two key elements: 
First, claiming $\contractAB$ is dependent on the claiming of $\contractBA$, ensuring that they can only be claimed in a predefined order (first $\contractBA$, then $\contractAB$).
Second, both \htlcs can be refunded in the very same order.

\begin{figure}
    \centering
    \includegraphics[width=0.75\columnwidth]{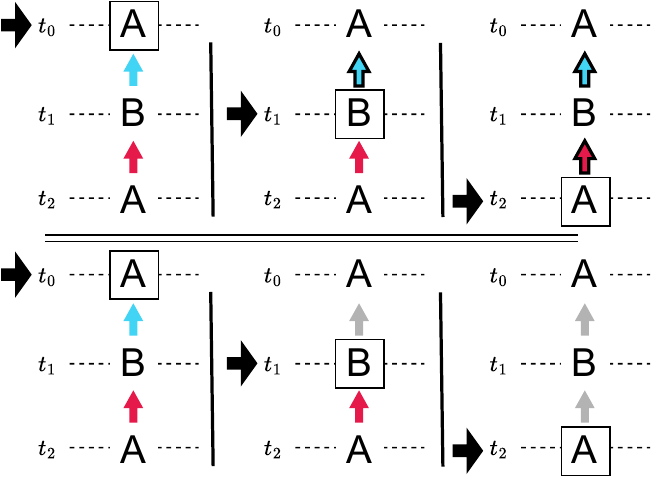}
    \caption{A \twopartyaswap as a round-based game. The black arrow indicates the current round. Players eligible to make a move are indicated by a square.
    Arrows with a black border indicate pulled edges; arrows without borders indicate available edges, and grayed-out arrows indicate disabled edges.
    The top picture depicts the game during an honest execution, and the lower picture depicts the game during an execution involving malicious user A.}
    \label{fig:two-party-game}
\end{figure}

Based on this observation, we can see the atomic swap as a round-based multi-player game represented as a simple linear \emph{transfer tree} (short \gtree) shown in~\cref{fig:two-party-game}.

Intuitively, the game starts on the top level of the \gtree and proceeds in rounds corresponding to the \gtree levels. In every round, the users located on the \gtree level have the option to pull their ingoing \edgestext in case their outgoing \edgetext has been pulled in the previous round. If this is not the case, they can \emph{disable} the outgoing \edgetext, preventing the user on top from pulling it later on.  
For the first round, there are no outgoing edges, so the initial ingoing \edgestext can be pulled immediately. 

Alice is located at the root, and hence, it is Alice's turn to make a move, where the only available move is pulling the ingoing \edgetext (indicating the claiming of $\contractBA$). 
In the second round, it is Bob's turn. 
Given that Bob's outgoing \edgetext got pulled before, Bob can now choose to pull its ingoing \edgetext (the claiming of $\contractAB$). 
Alternatively, if Alice chose not to pull their ingoing \edgetext in the first round, the move of pulling its ingoing \edgetext is not available to Bob in the second round. Instead, Bob can disable this \edgetext (corresponding to refunding $\contractBA$). Then, Alice could disable their outgoing \edgetext in round three (by refunding $\contractAB$). 

\paragraph {Blockchain Protocols as Games}
\label{sec:fromgraphstotreessubsection}
Introducing the concept of \gtrees, we can study how the goals of an \atgacronym specification can be met by users playing fund redistribution games and only later consider how the mechanics of these games can be realized with cryptographic protocols. 
More precisely, we will show how to directly transform an \atgacronym specification into such a game (represented as a \gtree), which satisfies the \atgacronym's functionality and security goals.
We hereby model an \atgacronym as a directed \graphtext $\graphsymbol$, 
consisting of a non-empty set $\nodesymbol$ of \nodestext (representing users) and a finite set $\arcsymbol$ of \arcstext (representing transfers between users).
To illustrate the transformation procedure, we use the example of a three-party swap \graphtext (c.f.~\cref{fig:three-party-tree}). 

\begin{figure}
    \centering
    \includegraphics[width=0.75\columnwidth]{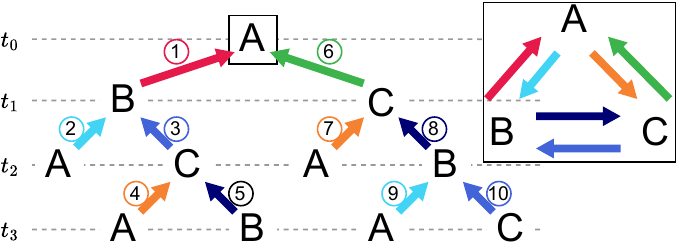 }
    \caption{Tree unfolding of the three-party graph in the top right. Edge numbers are given in pre-order and just serve to identify the edges. 
    Duplicated edges are depicted in the same color. 
    Player A is enclosed in a square to indicate that they will be the first player to make a move in the game.}
    \label{fig:three-party-tree}
\end{figure}

As in the two-player case, we choose an arbitrary user as the tree's root (hereby called \emph{leader}). 
The leader will be in the position to initiate the swap by pulling all their ingoing \arcstext from the \graphtext. 
In the further construction of the \gtree, it needs to be ensured that whenever the outgoing \arcstext of a user are pulled, they can also pull all ingoing \arcstext. 
Correspondingly, the \gtree is created: 
Starting from the leader, all ingoing \arcstext from the \graphtext are added as \edgestext to the tree. 
On the next level, for all users present at this level, all ingoing \arcstext will be appended as \edgestext, constituting the next tree level. 
This \graphtext unfolding stops whenever a user on a path is encountered for the second time (since otherwise, a user could claim an \arctext twice). 

We show the unfolding of the example \graphtext in~\cref{fig:three-party-tree}. 
Note that due to the original \graphtext's structure, several \arcstext appear multiple times in the \gtree, for example, \edgestext \edgeABtwo{} and \edgeABthree{} correspond to the same \arctext. 
This is as in the \graphtext in~\Cref{fig:three-party-tree}, there are multiple paths from C to the leader A. 
Such duplicate \edgestext (as indicated by the same color) should be considered representatives of the same \arctext.
Consequently, only one of these duplicate \edgestext can be executed. 
In the \gtree execution, this is reflected by duplicate \edgestext getting disabled (so becoming unavailable) once another representative has been executed. 
Intuitively, this is because the \fundstext to be transferred in this case have already been claimed otherwise (when executing the duplicate \edgetext).

\begin{figure}[tb]
    \centering
    \includegraphics[width=0.75\columnwidth]{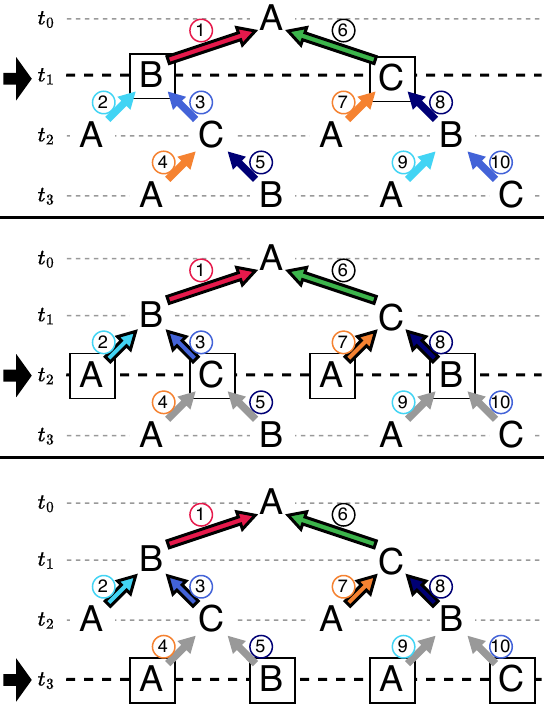}
    \caption{Honest execution of the game tree in~\Cref{fig:three-party-tree}. Execution states are depicted with the symbols in~\Cref{fig:two-party-game}.}
    \label{fig:three-party-game-honest}
\end{figure}

An honest execution of the \gtree in~\cref{fig:three-party-tree} is depicted in~\cref{fig:three-party-game-honest}: 
Suppose every user pulls all their ingoing \edgestext in every round. Then, the pulled \edgestext of the \gtree execution cover the \arcstext of the original \graphtext already after three levels, and all \edgestext on other levels (due to the existence of duplicate \edgestext) got disabled (indicated by grayed-out \edgestext).

The appearance of these duplicate \edgestext is crucial for the security of the game since it is not ensured that malicious users will always pull all their ingoing \edgestext.
For example, consider an execution of the \gtree (depicted in~\cref{fig:three-party-game-dishonest}) where A only pulls \edgetext \edgeCAone{} and C pulls \edgetext \edgeBCtwo{}. 
Without the ingoing \edgestext \edgeABthree{} and \edgeCBthree{} appearing in the right subtree, B would lose \fundstext now as they could not claim their \fundstext from A and C.  

\paragraph{Game Security}
We will formally prove in~\cref{cha:trees} that \gtrees constructed from an \atgacronym specification meet the \atgacronym goals. 
To this end, we characterize the possible outcomes that can result from an honest user \honestuser{} playing their honest strategy on an \gtree $\treesymbol$. 
The honest strategy of \honestuser{} consists of the user eagerly pulling all possible ingoing \edgestext (only possible when corresponding outgoing \edgestext have been pulled before) and disabling all possible outgoing \edgestext. 
An outcome of a \gtree execution will be represented by all \edgestext that got pulled during the \gtree execution. 
We will denote with $\outcomes{\honestuser}$ the set of outcomes that can result from $\honestuser$ following the honest strategy (while other users may behave arbitrarily).
Based on this notion, we will show that honest users, when playing the game, always enjoy the following local atomicity guarantee: 
\begin{theorem}[\Gtree Security (informal)]
Let $\treesymbol$ be an \gtree resulting from unfolding \graphtext $\graphsymbol = (\nodesymbol, \arcsymbol)$. 
All outcomes $\outcome \in \outcomes{\honestuser}$ of $\treesymbol$ for honest user $\honestuser \in \nodesymbol$ satisfy that if they contain an \arctext $(\honestuser, X) \in \arcsymbol$ (corresponding to an outgoing \arctext of $\honestuser$ in $\graphsymbol$), then also all ingoing \arcstext $(Y, \honestuser) \in \arcsymbol$ are contained in $\outcome$.
\end{theorem}
This result ensures an honest user never loses \fundstext during the \gtree execution: If \fundstext are pulled from them, they can also claim all \fundstext that they should get according to $\graphsymbol$.\footnote{Note that an honest user may end up better off than specified by $\graphsymbol$:
A user can receive all \fundstext corresponding to their ingoing \arcstext in $\graphsymbol$ without spending all their \fundstext corresponding to their outgoing \arcstext. 
E.g., in the execution in~\cref{fig:three-party-game-dishonest}, \edgetext \edgeBAone{} is not pulled, indicating that user B does not need to spend their \fundstext intended for A.}

\begin{figure}[tb]
    \centering
    \includegraphics[width=0.75\columnwidth]{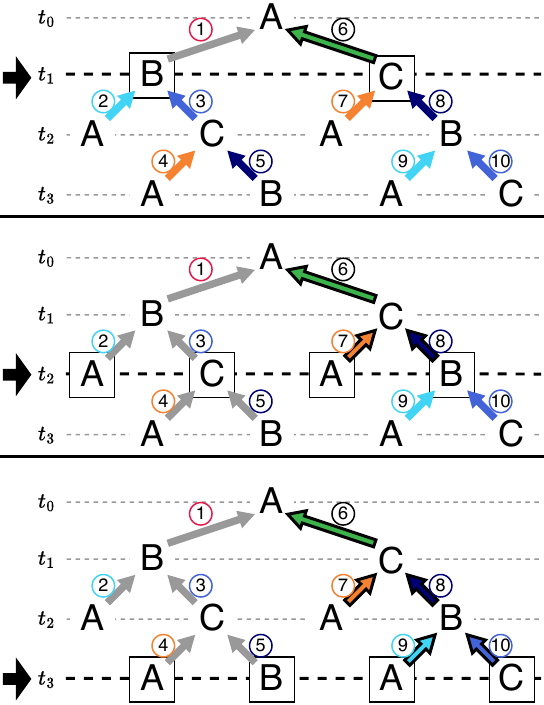}
    \caption{Dishonest execution of the game tree in~\Cref{fig:three-party-tree}. Execution states are depicted with the symbols in~\Cref{fig:two-party-game}.}
    \label{fig:three-party-game-dishonest}
\end{figure}

The described procedure of transforming graphs into \gtrees not only applies to strongly connected \graphstext as the one given in the example but to all \graphstext that are \emph{\insemi}. 
A graph $\graphsymbol = (\nodesymbol, \arcsymbol)$ is \insemi if it contains a node $n \in \nodesymbol$ that can be reached from every other node $n' \in \nodesymbol$. 
Each such node $n \in \nodesymbol$ is a possible leader in the \gtree construction. 
As we show in~\cref{sec:applications}, supporting \insemi graphs paves the way to cover many applications beyond atomic swaps. 

\paragraph{Protocols for \GTrees}
\label{From Trees to Protocols}
\Gtrees constitute a powerful intermediate representation to capture the essence of complex protocols realizing \atgsacronym.
We now show how to enforce the mechanics of \gtrees with the help of cryptographic protocols operating in a \hbe (\hbeacronym) that consists of different transfer agreement mechanisms (\tams). 
The challenge here lies in finding a building block that is powerful enough to support arbitrary \gtrees and at the same time sufficiently simple so that it can be realized on a large variety of \tams, including such with limited capabilities (e.g., cryptocurrencies like Bitcoin, which supports only checking simple payment conditions). 
To meet these requirements, we introduce \emph{\CTLClong~($\ctlc$)}, a generalization of the \htlc primitive used in the \twopartyaswap protocol (c.f.~\cref{fig:two-party-swap}).
Note that the \htlc there serves two purposes:
1) The condition of the \htlc establishes a dependency between the two \gtree \edgestext, ensuring that the second \edgetext can be pulled if and only if the first one was pulled, and, thus, enabling round-based pulling of \edgestext.
2) The consecutive timeouts enable a round-based disabling of \edgestext, ensuring that a user can disable an outgoing \edgetext before their ingoing \edgestext can get disabled. 

However, \htlcs do not suffice to enable the same properties for general \gtrees, since in \gtrees
1) the execution of an \edgetext depends on the prior execution of a whole path of \edgestext (and not only a single \edgetext) and
2) \gtrees contain duplicate edges, indicating alternative forms of spending the same \fundstext to the receiver in different phases of the protocol. 

\ctlcs address these limitations. 
Firstly, they support flexible conditions that can be composed of several secrets.
This is to reflect that each \edgetext of the \gtree $\treesymbol$ representing an \arctext $(X,Y)$ corresponds to a walk $\walk$ from $(X,Y)$ to the leader (the root) of $\treesymbol$ and hence should only be pulled once all other \edgestext along $\walk$ have been pulled.
To realize this, we identify each \edgetext $\e$ of $\treesymbol$ with this walk $\walk$ (so $\e = (X,Y)_{\walk}$) and assign it a unique secret $\edgesecret{(X,Y)}{\walk}$ owned by the receiver $Y$ of the edge (similar to how A owns $x$ in~\Cref{fig:two-party-swap}). 
The condition for pulling edge $(X,Y)_{\walk}$ will then require the knowledge of all secrets for the edges on the path $\walk$ to the root in $\treesymbol$. 
Second, \ctlc{s} generalize \htlc{s} in that they allow for nesting multiple \ctlc{s}, meaning that refunding a \ctlc{} can result in the funds being transferred to a follow-up \ctlc{} instead of returning them to the receiver.
More precisely, we will consider a \ctlc{} denoted by $\ctlcargs{\sender}{\receiver}{\ctlcsubcontracts}$ to contain a non-empty list $\ctlcsubcontracts$ of subcontracts of the form $\subctlcargs{\secretset}{t}$ where $\secretset$ represents a set of composed spending conditions and $t$ represents the timelock of the subcontract. 
For the sake of generality, we will not require the individual conditions $Y_i$ within a composed condition $\vec{Y} \in \secretset$ to be fixed to hash values for a specific hash function $H$ but will simply consider them to be witnesses for some hard relation $R$ (meaning that given $Y$, it is computationally hard to find $x$ such that $(x,Y) \in R$).
Further, we will require that the timelocks of all subcontracts in $\ctlcsubcontracts$ are strictly increasing.
Then $\ctlcargs{\sender}{\receiver}{\ctlcsubcontracts}$ can evolve in the following ways: 
If $\ctlcsubcontracts = [(\subctlcargs{\secretset_1}{t_1}), \dots]$ then the contract funds can be claimed (transferred to $\receiver$) when providing $\vec{x}$ such that there is a $\vec{Y} \in \secretset_1$ and for all $Y^i \in \vec{Y}$ it holds $(x^i, Y^i) \in R$.
Alternatively, after time $t_1$ the contract funds can be refunded to $\sender$ (if $\subctlcargs{\secretset_1}{t_1}$ was the last element in $\ctlcsubcontracts$), or spent to $\ctlcargs{\sender}{\receiver}{[(\subctlcargs{\secretset_2}{t_2}), \dots]}$. 

Duplicate edges can be realized through a \ctlc{} 
that contains subcontracts for all \edgestext representing the same \arctext in the graph in ascending order of their appearance in the \gtree and increasing timeouts according to their \gtree level\footnote{
Duplicate edges on the same level are modeled by different conditions $\vec{Y} \in \secretset$.
}.
For example, the duplicate edges \edgeABtwo{} and \edgeABthree{} 
from~\Cref{fig:three-party-tree} would be realized by the following \ctlc{} $\ctlcvar_{(A,B)}$:
\vspace{-2pt}
\begin{equation}
\hspace{-6pt}\scalebox{0.85}{$
    \ctlcargs{A}{B}{[ 
        \overbrace{\subctlcargs{\{ \textstyle (\edgesecret{(B,A)}{\text{\edgeBAone}}, \edgesecret{(A,B)}{\text{\edgeABtwo}})\}}{t_2}}^{\text{\subctlcABtwo}}, 
        \overbrace{\subctlcargs{\{ ( \edgesecret{(C,A)}{\text{\edgeCAone}}, \edgesecret{(B,C)}{\text{\edgeBCtwo}}, \edgesecret{(A, B)}{\text{\edgeABthree}})\}}{t_3}}^{\text{\subctlcABthree}}]}
        $} \hspace{-1pt} \label{eq:CTLC-running-example}
\end{equation}
Before time $t_2$, $B$ can claim the funds by providing secrets 
$\edgesecret{(B,A)}{\text{\edgeBAone}}$ and $\edgesecret{(A,B)}{\text{\edgeABtwo}}$. 
While $\edgesecret{(A,B)}{\text{\tiny\edgeABtwo}}$ is chosen by $B$, 
$B$ would obtain secret $\edgesecret{(B,A)}{\text{\edgeBAone}}$ if $A$ claims the funds from $\ctlcvar_{(B,A)}$ in edge \edgeBAone{}.
Consequently, providing these secrets to claim the funds of $\ctlcvar_{(A,B)}$ using subcontract \subctlcABtwo{} would correspond to $B$ pulling the edge \edgeABtwo{} given that edge \edgeBAone{} was pulled.

If edge \edgeBAone{} was not pulled at time $t_2$, $B$ already disabled \edgeBAone{} and, hence, also \subctlcABtwo{} (representing edge \edgeABtwo) can be safely discarded so that only \subctlcABthree{}  representing edge \edgeABthree{} is left.
Subcontract \subctlcABthree{} enables $B$ to claim the funds when providing the secrets $\edgesecret{(C,A)}{\text{\edgeCAone}}$, $\edgesecret{(B,C)}{\text{\edgeBCtwo}}$, and $\edgesecret{(A, B)}{\text{\edgeABthree}}$.
Again $\edgesecret{(C,A)}{\text{\edgeCAone}}$ and $\edgesecret{(B,C)}{\text{\edgeBCtwo}}$ will be learned from $C$ pulling edge \edgeBCtwo{}, while $\edgesecret{(A, B)}{\text{\edgeABthree}}$ was chosen by $B$.

Following these ideas, an \gtree $\untree$ can be translated into a set of \ctlc{s} representing its edges. 
Pulling edges in $\untree$ corresponds to claiming the \ctlc{} modeling the edge, and disabling an edge is reflected by removing a subcontract for this edge from the corresponding \ctlc{}. 
Based on this correspondence, we can characterize a general protocol $\paramstrategy{\honestuser}{\untree}$ that implements the \gtree strategy of honest user $\honestuser$ on the \gtree $\untree$ for the \ctlcs{} resulting from the translation of $\untree$. 

\paragraph{Protocol Security}
To formally prove that the mechanics of an \gtree $\untree$ are faithfully captured by the protocol $\paramstrategy{\honestuser}{\untree}$, we provide a formal symbolic model for the execution of \ctlc{}-based protocols:
We define how users can interact with an \hbeacronym consisting of different \tams that support \ctlcs. 
More precisely, we characterize all possible protocol runs $\run$ that can incur when executing a specific honest user protocol $\userstrategy{\honestuser}$ in such an ecosystem and write $\userstrategy{\honestuser} \conforms \run$ if $\run$ results from such an execution.
The symbolic execution model thereby takes security-relevant blockchain-specific characteristics into account, e.g., that the interactions of honest users with the \ctlcs on the different \tams get known to the attacker before execution and may be maliciously delayed or reordered.

We obtain the following security result:

\begin{theorem}[Protocol Security (Informal)] 
Let $\treesymbol$ be an \gtree and $\run$ be a run stemming from honest user $\honestuser$ executing $\paramstrategy{\honestuser}{\untree}$ ($\paramstrategy{\honestuser}{\untree} \conforms \run$) such that the timelocks of all subcontracts $\subctlcvar_\edgevar$ for edges in $\treesymbol$ have passed.
Then there exists an outcome $\outcome \in \outcomes{\honestuser}$ such that the subcontracts $\subctlcvar_\edgevar$ involving $\honestuser$ claimed in $\run$ correspond to the edges $\edgevar$ of $\honestuser$ in $\outcome$.
\end{theorem}
This result ensures that whenever an execution of the protocol $\paramstrategy{\honestuser}{\untree}$ advanced enough (namely reached the timelocks of all subcontracts constructed for the \gtree $\treesymbol$), then the executed subcontracts correspond to the \edgestext of a valid outcome for the honest user $\honestuser$ in $\treesymbol$. 
Hence, the outcomes of $\honestuser$ in $\treesymbol$ as given by $\outcomes{\honestuser}$ soundly reflect how \fundstext of $\honestuser$ will be claimed in the protocol.

Intuitively, this result gives us end-to-end security 
guarantees:
A (sufficiently advanced) execution of the honest user protocol $\paramstrategy{\honestuser}{\untree}$ corresponds to an outcome $\outcome \in \outcomes{\honestuser}$ and such an outcome $\outcome$ for a tree $\treesymbol$ resulting from unfolding a graph $\graphsymbol$ was shown to ensure that $\honestuser$ does not lose funds (w.r.t. the \atgacronym specification $\graphsymbol$).

In summary, we have shown how to synthesize an \atgacronym specification $\graphsymbol$ into a game represented as an \gtree $\treesymbol$, and how to transform this \gtree $\treesymbol$ into a \ctlcprotocol.
From our security results for both of these transformations, we can show that for an honest user $\honestuser$, the resulting \ctlcprotocol is guaranteed to execute transfers that correspond to a beneficial trade of $\honestuser$ w.r.t. $\graphsymbol$.

\section{\GTree Unfolding}
\label{cha:trees}

In the following, we will give a more formal description of how to represent \atgs and how to transform them into \gtrees with the intended behavior.

\paragraph{\ATGs}
An \atg (\atgacronym) is a directed \graphtext (short digraph) $\graphsymbol = (\nodesymbol, \arcsymbol)$
where $\nodesymbol$ denotes the set of \nodestext (representing users) and $\arcsymbol$ denotes the set of \arcstext (representing transfers between users).
 Arcs $\arc$ are given as tuples $\arc = (A, B)$ with $A, B \in \nodesymbol$, and we call $A$ the sender (written $\funsender(\arc)$) and $B$ the receiver of $\arc$ (written $\funreceiver(\arc)$). 
 We call \emph{walk} a sequence of \arcstext $\walk = [\arc_{\size{\walk} - 1}, \arc_{\size{\walk} - 2}, ..., \arc_0] \in \arcsymbol^{\size{\walk}}$ where the receiver of each \arctext coincides with the sender of its predecessor.
 We use $\concatvec{\walk_1}{\walk_2}$ to denote the concatenation of two walks $\walk_1$ and $\walk_2$. 
Further, we use $\walk_2 \suffix \walk_1$ to denote that $\walk_1$ is a suffix of $\walk_2$ (so that $\exists \walk:~ \walk_2 = \concatvec{\walk}{\walk_1}$) and 
$\walk_2 \directsucc \walk_1$ to say that $\walk_2$ extends $\walk_1$ by one arc (so $\exists \arc:~\walk_2 = \concatvec{\unitvec{\arc}}{\walk_1}$). 

In this work, we are interested in the class of digraphs that we call \emph{\insemi}.
A \graphtext is \insemi if there is a \nodetext $A \in \nodesymbol$ (which we call \emph{leader}), which can be reached (with a walk) from every other \nodetext in $\nodesymbol$. 
Every strongly connected \graphtext is also \insemi while the contrary does not hold. A formal proof and a more detailed discussion of the underlying graph theory can be found in
\iffullversion
Appendix~\ref{sec:Graph Theory}.
\else 
\cite{extended-version}. 
\fi  

\paragraph{\GTrees Unfolding}
We next define how to transform an \atgacronym given as a digraph $\graphsymbol$ into an \gtree. 
To this end, we will represent \gtrees as sets of \emph{\edgestext}, where an \edgetext $\edgevar = (A, B)_{\walk}$ is an \arctext indexed with its walk $\walk$ to the root of the \gtree.

\begin{definition}[Tree unfolding]
    Let $\graphsymbol = (\nodesymbol, \arcsymbol)$ be a digraph that is in-semiconnected in $A \in \nodesymbol$. 
    The tree unfolding $\fununfold(\graphsymbol, A)$ of $\graphsymbol$ with leader $A$ is defined as follows:
\begin{align*}
    \small
        &\fununfold(\graphsymbol, A) :=  \\
        &~\{ 
            \arc_{\walk^{\textit{s}}} ~|~ \exists \walk: \walk \suffix \walk^{\scriptscriptstyle\textit{s}} ~\land~ \arc \in \arcsymbol ~\land~ a = \arc^{\scriptscriptstyle\textit{s}}_{\size{\walk^{\textit{s}}} -1} \\
             &~\land~
             \exists B \in \nodesymbol : \textit{noDupWalk}(\graphsymbol, \walk, B, A)\\
         &~\land~  (\exists j \leq \size{\walk} -2 : B = \funreceiver(a_j) \\
         &\hspace{21pt} ~\lor~ \nexists \arc' \in \arcsymbol: B = \funreceiver(\arc')  )
             \}
\end{align*}
where $\textit{noDupWalk}(\graphsymbol, \walk, B, A)$ denotes that $\walk$ is a walk from $B$ to $A$ in $\graphsymbol$ that does not contain the same \arctext twice.
\end{definition}

The definition states that the unfolding contains all \edgestext $\arc_{\walk^{\textit{s}}}$ on walks $\walk$ from \nodestext $B \in \nodesymbol$ to the leader $A$, which satisfy the following properties
(1) $\walk$ does not contain repeated \arcstext and either
 (2a) $\walk$ already contains an \arctext with $B$ in the position of a receiver or 
 (2b) $\walk$ could not be extended beyond $B$ because there is no $\arc' \in \arcsymbol$ with $B$ as receiver. 
This unfolding ensures for every user $B$ that if it occurs as a sender in a path of the resulting \gtree, then it also occurs as a receiver (if it has receiving \arcstext in $\graphsymbol$). 
We defer to 
\iffullversion
Appendix~\ref{sec:TtreeSize}
\else 
\cite{extended-version}. 
\fi  
an analysis of the scalability of the \gtrees unfolding mechanism.

\paragraph{Outcome Sets}
We provide a form of game semantics for \gtrees by introducing the outcome sets (written $\outcomeset{\untree}$) of a user $\honestuser$ when interacting with an \gtree $\untree$. 
Intuitively, an outcome $\outcome \in \outcomeset{\untree}$ corresponds to a partial execution of $\untree$ that may result from such an interaction. 
In particular, the outcomes reflect that the user $\honestuser$ can enforce certain minimal guarantees on the \gtree execution.
More precisely:
\begin{enumerate}
    \item If $\untree$ contains duplicate \edgestext $(X,Y)_{\walk}$,  $(X,Y)_{\walk'}$ involving $\honestuser$ (so $\honestuser \in \{ X, Y\} $), at most one of them may be executed.
    \item If $\honestuser$ is the root user of $\untree$ then all ingoing edges $(Y,B)_{\walk}$ (with $\size{\walk} =1$) are executed. 
    \item If an outgoing \edgetext $(B, X)_{\walk}$ of $\honestuser$ is executed then also all ingoing \edgestext $(Y,B)_{\walk'}$ with $\walk' \directsucc \walk$ are executed.
\end{enumerate}

\begin{definition}[Outcome Set]
    \label{def:outcomeset}
    Let $\untree$ be a \gtree. 
    Then the outcome set $\outcomeset{\untree}$ of user $B$ in $\untree$ is given as 
    \begin{align*}
        \outcomeset{\untree} := \{ 
            &\outcome \in \partialtreeoutcomes (\untree) ~|~
             \predNoDup{\untree}{B}{\outcome}  \\
             &\land~\predHonestRoot{\untree}{B}{\outcome}  
             ~\land~\predEagerPull{\untree}{B}{\outcome} \} 
    \end{align*}  
    where $\partialtreeoutcomes(\untree)$ denotes the set of all \gtrees whose paths (from the leaves to the root) are suffixes of paths in $\untree$ and
    \begin{align*}
        &\predNoDup{\untree}{B}{\outcome} : \Leftrightarrow
        \\
        &\edge, (X,Y)_{\walk'} \in \outcome \land B \in \{ X, Y\} \Rightarrow \walk = \walk' ,
        \\
        &\predHonestRoot{\untree}{B}{\outcome} :\Leftrightarrow \\
         &(X, B)_{\walk} \in \untree \land \size{\walk} = 1 \Rightarrow (X,B)_{\walk} \in \outcome ,
        \\
        &\predEagerPull{\untree}{B}{\outcome} :\Leftrightarrow
        \\
        &(X, B)_{\walk_1} \in \untree \land (B,Y)_{\walk_2} \in \outcome \land \walk_1 \directsucc \walk_2 , \\
            & \qquad \Rightarrow \exists \walk_3: (X, B)_{\walk_3} \in \outcome \land \size{\walk_3} \leq \size{\walk_1} .
    \end{align*}
\end{definition}

The predicates $\predNoDupP$, $\predHonestRootP$,  $\predEagerPullP$ capture exactly the three requirements on the partial \gtree executions.

\paragraph{Security and Correctness of \GTree Unfolding}
To show the security of the \gtree unfolding, we show that each outcome $\outcome \in \outcomeset{\untree}$ of a \gtree $\untree  =  \fununfold(\graphsymbol, A)$ is a safe outcome for user \hspace{-1pt}$B$\hspace{-1pt} w.r.t. the \atgacronym \hspace{-1pt}$\graphsymbol$\hspace{-1pt}. 
An outcome is considered safe if $B$ does not end up \emph{underwater}, meaning that some outgoing \arcstext of \hspace{-1pt}$\honestuser$\hspace{-1pt} in \hspace{-1pt}$\graphsymbol$\hspace{-1pt} are triggered but not all their ingoing \arcstext in $\graphsymbol$.
\begin{theorem}[Security of Tree Unfolding]
    \label{thm:unfolding-security}
    Let $\graphsymbol = (\nodesymbol, \arcsymbol)$ be a digraph that is in-semiconnected in $A \in \nodesymbol$ and 
     $\untree = \fununfold(\graphsymbol, A)$ and $B \in \nodesymbol$.
    Then it holds that $\forall \outcome \in \outcomeset{\untree} :$
    \begin{align*}
        &\forall \walk, \walk' : \bigl( \edge \in \outcome \land (X,Y)_{\walk'} \in \outcome \land B \in \{ X, Y\} \notag \\
        & \qquad \quad \Rightarrow \walk = \walk' \land (X,Y) \in \arcsymbol \bigr) \label{eq:underwater1}\\
        & \hspace{3pt} \land \forall (B,Y)_{\walk} \in \outcome : \, \bigl( (X,B) \in \arcsymbol \Rightarrow \exists \walk': (X,B)_{\walk'} \in \outcome \bigr)
    \end{align*}
\end{theorem}
Intuitively, the statement says that all outcomes $\outcome$ (so partial \gtree executions) that a user $\honestuser$ may incur during interaction with the \gtree $\untree$ resulting from unfolding the \graphtext $\graphsymbol$ have the following properties:
(1) \Edgestext involving user $\honestuser$ correspond to a unique \arctext in the \graphtext, and 
(2) if $(B,Y)_{\walk} \in \outcome$ (corresponding to an outgoing \arctext $(B,Y)$ in $\graphsymbol$ being executed) then for each ingoing \arctext $(X,B)$ in $\graphsymbol$, $\outcome$ also contains a corresponding \edgetext $(X,B)_{\walk'}$ for this \arctext. 
The formal proof of this theorem is given in 
\iffullversion
Appendix~\ref{sec:GraphtoTree}.
\else 
\cite{extended-version}. 
\fi  

In addition to security for an honest party, we show that if all parties are honest, all arcs from the original \graphtext get executed, so the intersection of the outcome sets of all users only contains outcomes that cover the whole \graphtext. 
The proof of~\cref{theorem:correctness-graph-tree} can be found in
\iffullversion
Appendix~\ref{sec:GraphtoTree}.
\else 
\cite{extended-version}. 
\fi  

\begin{theorem}[Correctness of Tree Unfolding]\label{theorem:correctness-graph-tree}
    Let $\graphsymbol = (\nodesymbol, \arcsymbol)$ be a digraph that is in-semiconnected in $A \in \nodesymbol$ and 
    $\untree = \fununfold(\graphsymbol, A)$ the tree unfolding of that graph. It holds:
    \begin{align*}
        \forall \outcome^* \in \bigcap_{B_i \in \nodesymbol} \outcomesetBi{\untree}: 
        (X,Y) \in \arcsymbol \Leftrightarrow \exists \walk: (X,Y)_{\walk} \in \outcome^*
    \end{align*}
\end{theorem}

\section{Protocols for \GTrees}
\label{sec:CTLCs}
\Gtrees provide a general abstraction layer for describing interactive protocols that involve the orchestrated execution of transfers in different \tams.
In this section, we show how \gtrees can be securely realized by cryptographic protocols in a \hbeacronym encompassing multiple \tams. 
This will allow us to prove end-to-end security and correctness for general blockchain protocols specified as an \atgacronym.

\subsection{\ctlcs}
Our protocols rely on \CTLCslong~(\ctlcs), the core building block that needs to be provided from the \tam (e.g., the underlying blockchain). 
We provide here a formal model of the execution of \ctlcs and describe in~\cref{sec:ctlc-impl} how \ctlcs can be realized in practice.

We represent a \CTLC{} contract by a list $\CTLCcontract := [ \CTLCsubcontract_1, ..., \CTLCsubcontract_s ]$
of subcontracts
$\CTLCsubcontract_i := (X,Y,\fundfull,\timelock_i, \secretset_i)$
where $X$ denotes the contract's sender, 
$Y$ the contract's receiver,
$f^{\zeta}$ the contract fund (with identifier $\fundid$),
$\timelock_i$ the \emph{timelock} of the subcontract, and 
$\secretset_i$ the subcontract's \emph{condition}. 
While the sender, receiver, and fund need to match for all $\CTLCsubcontract_i \in \CTLCcontract$ (we hence also write $\funsender(\CTLCcontract)$, $\funreceiver(\CTLCcontract)$, and $\funfund(\CTLCcontract)$), 
timelock $\timelock_i$ and condition $\secretset_i$ are specific to $\CTLCsubcontract_i$. 

The timelock $\timelock_i$ denotes the time starting from which subcontract $\CTLCsubcontract_i \in \CTLCcontract$ may be removed from $\CTLCcontract$  given that it is the first element of $\CTLCcontract$ (we say that $\CTLCsubcontract_i$ gets \emph{timed out} in this case). 
Since subcontracts can only be timed out in order, we require that the timelocks of the subcontracts must be strictly increasing. 
The condition $\secretset_i$ is a set of sets of secrets describing different options for claiming $\CTLCsubcontract_i$:
A member $\singlecondition \in \secretset_i$ describes a set of secrets whose knowledge is sufficient for claiming the contract. 

\paragraph{Semantics}
We formally describe the execution of \ctlcs using a small-step semantics, so a relation
$\environmentvec \sstep{\action} \environmentvecprime$ characterizing how a \hbeacronym of \ctlc-supporting \tams evolves from one state (denoted by $\environmentvec$) to another ($\environmentvecprime$) when executing an action $\action$. 
The state $\environmentvec$  thereby is given as a vector of individual \chenvs $\envindex{\tammath}$ per \tamfull $\tammath$, keeping track of the individual stages of the \ctlc execution. Each $\envindex{\tammath}$ corresponds to the state of one of the \tams from the \hbeacronym.
E.g., the components $\envindex{\tammath}.\avFunds$ and $\envindex{\tammath}.\resFunds$ track the funds currently available in $\tammath$ (e.g., owned by a user), or reserved by a \ctlc, respectively.  
The component $\envindex{\tammath}.\enCTLCs$ contains \ctlcs that have been set up (enabled) for execution in the specified \tam.
More concretely, the semantics covers the next stages of \ctlc execution on this \tam: 

(1) Since \ctlcs are used as parts of protocols involving multiple \ctlc instances, which reside in different \chenvs, users will agree on the execution of \ctlcs in batches. 
For initiating protocol execution, users will broadcast their intention to execute the protocol and its specification consisting of a set $\batch$ of \ctlcs across the different \tams. 
(2) Based on the announcement of a \ctlc batch $\batch$, the protocol participants will commit to the secrets used in the \ctlcs $\CTLCcontract \in \batch$. 
(3) Once all participants committed their secrets, the users in the different \chenvs can initiate the pair-wise setup of the individual \ctlcs $\CTLCcontract$ on the respective \tams. 
Importantly, starting from this point, \ctlcs $\CTLCcontract$ are considered local objects residing in a single $\environmenttam$ such that only users of $\tammath$ may interact with them.
(4) Contracts $\CTLCcontract$ previously advertised can be authorized by $\funsender(\CTLCcontract)$ and $\funreceiver(\CTLCcontract)$, effectively marking the $\CTLCcontract$ as authorized in the corresponding \chenvs .
(5) Once authorized by both parties a \ctlc $\CTLCcontract$ can be enabled, effectively marking the \fundstext $\funfund(\CTLCcontract)$ as reserved. 
(6) An enabled contract $\CTLCcontract$ can either be claimed by $\funreceiver(\CTLCcontract)$ (by claiming its top-level subcontract), 
or  $\CTLCcontract$'s subcontracts can be successively timed out 
until $\CTLCcontract$ can finally be refunded to $\funsender(\CTLCcontract)$. 
In both cases, the \fundstext $\funfund(\CTLCcontract)$ are unmarked as reserved and assigned to either $\funreceiver(\CTLCcontract)$ or $\funsender(\CTLCcontract)$.
\begin{figure}
    \begin{equation*}
\scalebox{0.92}{$
        \inference{ \advCTLCcontract \in \environmenttamCTLCAdvertised, \CTLCsubcontract \in \CTLCcontract \in \environmenttamCTLCEnabled, 
        \\ \singlecondition \in \funconditions(\CTLCsubcontract),
         \singlecondition \subseteq \environmenttamRevealedSecrets, 
        \\ \nexists \, \dotCTLCsubcontract \in \advCTLCcontract : \funposition(\dotCTLCsubcontract) < \funposition(\CTLCsubcontract),
        \\ \advCTLCs' := \environmenttamCTLCAdvertised \backslash \{ \CTLCcontract \},
        \enCTLCs' := \environmenttamCTLCEnabled \backslash \{ \CTLCcontract \}, \\
        \avFunds' := \environmenttam.\avFunds \cup \{ (\funreceiver(\CTLCcontract): \funfund(\CTLCcontract)) \}, \\ 
        \resFunds' := \environmenttam.\avFunds \backslash \{ \funfund(\CTLCcontract) \}, \\
         \\ \environmenttam' := \environmenttam[{\advCTLCs} \rightarrow \advCTLCs', \enCTLCs \rightarrow \enCTLCs', \avFunds \rightarrow \avFunds', \resFunds \rightarrow \resFunds'] }
        {\environmentvec \overset{ \DecideCo (\CTLCcontract ,\CTLCsubcontract,  \singlecondition)}{\longrightarrow} \Update{\environmentvec}{\environmenttam}{\environmenttam'}} 
        $}
        \end{equation*}
        \caption{Inference rule for claiming a \ctlc subcontract (slightly simplified).}
        \label{fig:ctlc-claim}
\end{figure}

The small-step relation $\sstep{}$ that formally captures these execution steps, is defined by a set of inference rules. 
An example of such a rule for the \emph{claim} case is provided in~\Cref{fig:ctlc-claim}: 
The rule first checks whether all preconditions for claiming a subcontract $\CTLCsubcontract$ of \ctlc $\CTLCcontract$ with conditions $\singlecondition$ are met, namely that 
(1) $\CTLCsubcontract$ was enabled in $\environmenttam$ ($\CTLCsubcontract \in \CTLCcontract \in \environmenttamCTLCEnabled$)
(2) all secrets as specified by condition $\singlecondition \in \funconditions(\CTLCsubcontract)$ have been revealed in $\environmenttam$ ($\singlecondition \subseteq \environmenttamRevealedSecrets$) and
(3) $\CTLCsubcontract$ is the top-level contract of $\CTLCcontract$ (ensured by checking that there is no other contract $\dotCTLCsubcontract$ according to the original \ctlc advertisement $\advCTLCcontract \in \environmenttamCTLCAdvertised$ occurring in a position before $\CTLCsubcontract$, so $\funposition(\dotCTLCsubcontract) < \funposition(\CTLCsubcontract)$).
Note that the last condition accesses the original advertisement $\advCTLCcontract \in \environmenttamCTLCAdvertised$ since this contains subcontracts $\dotCTLCsubcontract$ of $\CTLCcontract$ that might not have been enabled (and hence not in $\environmenttamCTLCEnabled$) but still need to be timed out before a lower-level subcontract can be claimed.

If all conditions are met, $\environmenttam$ is updated to remove the claimed $\CTLCcontract$ from both $\environmenttamCTLCEnabled$ and $\environmenttamCTLCAdvertised$ (reflecting that those contracts have been resolved) and to assign the contract funds to $\funreceiver(\CTLCcontract)$ (indicated by moving $\funfund(\CTLCcontract)$ from the set of reserved funds $\resFunds$ to the set of available funds $\avFunds$, annotating the new ownership $(\funreceiver(\CTLCcontract): \funfund(\CTLCcontract)$).
The full specification of the semantics is in 
\iffullversion
Appendix \ref{sec:inferenceRules}.
\else 
\cite{extended-version}. 
\fi  

\subsection{Blockchain Execution Model}
The small-step semantics describes the (concurrent) execution of \ctlcs in a \hbeacronym.
However, when defining cryptographic protocols that leverage \ctlcs, the peculiar execution environment of the \tams executing \ctlcs needs to be considered. 
In particular, in such \tams, user actions (e.g., the execution of a transaction in a blockchain) do not happen instantaneously but are subject to interference with a (potentially malicious) scheduler. 
Such a scheduler (e.g., a block builder in a blockchain) learns about the intended actions of honest users \hspace{-1pt}(e.g., \hspace{-1pt}when users submit transactions for inclusion in the blockchain) and can, based on this knowledge, decide on the execution order of actions or insert their own. 
However, \tams (such as blockchains) ensure that a malicious scheduler cannot defer honest user actions indefinitely but provide an eventual inclusion guarantee for these actions.

To reflect this in our formal model, we adopt the approach taken in~\cite{BitML} and model protocols of honest users (given by a set $\honestusers$) as symbolic \emph{strategies}. 
A strategy $\strategy$ is a function operating on sequences $\run = \environmentvec_0 \sstep{\action_0} \environmentvec_1 \sstep{\action_1} \dots \sstep{\action_{n-1}} \environmentvec_n$ of valid transitions according to the small-step semantics, which represent the execution history of the \hbeacronym. 
We will call such sequences \emph{runs}. 
On input of a run $\run$, a strategy $\strategy$ outputs a set of actions $\{ \action'_1, \action'_2, \dots, \action'_m\}$ that the user aims to execute and hence should be appended to $\run$. 

In addition to the strategies $\strategy_{\honestuser_i}$ of honest users $\honestuser_i \in \honestusers$ (modeling the behavior of honest protocol participants), we assume the existence of an adversarial strategy $\advstrategy$ that models the behavior of malicious protocol participants and the malicious scheduler. 
Such a strategy $\advstrategy$, in addition to the run $\run$, gets as input the set $\mempool = \bigcup_{\honestuser_i \in \honestusers}{\strategy_{\honestuser}(\run)}$ of all outputs of honest user strategies and based on that outputs the next action $\alpha$ to append to $\run$. 
$\advstrategy$ is limited to only output actions $\alpha = \advstrategy(\run, \mempool)$ that are valid extensions of $\run$ according to the small-step semantics.
Further, $\advstrategy$ may not schedule any actions $\alpha$ for which honest users have the privilege unless those are included in $\mempool$ (e.g., actions revealing the secrets of honest users).
Finally, $\advstrategy$ may only advance the time $\chtime$ of the \hbeacronym by an offset $\delta$ (indicated by the execution of a dedicated action $\elapse \, \delta$) if all honest users agreed to this by scheduling actions $\elapse \, \delta_i \in \strategy_{\honestuser_i}(\run)$ with $\delta_i \geq \delta$, see 
\iffullversion
Appendix \ref{sec:TimeProgression} 
\else 
\cite{extended-version} 
\fi  
for details. 
This requirement reflects the inclusion guarantees of the \tams that enable honest users to meet deadlines\footnote{It may at first seem like a restriction that honest user actions will be included at the same time as scheduled. However, the attacker can still schedule arbitrarily many actions before the inclusion of an honest user action. This, in particular, models the effects of a malicious miner appending several blocks before including the honest user transaction in a blockchain.}.

We say that a run $\run = \environmentvec_0 \sstep{\action_0} \environmentvec_1 \sstep{\action_1} \dots \sstep{\action_{n-1}} \environmentvec_n$
conforms to $(\strategy_\honestusers, \advstrategy)$ (written $(\strategy_\honestusers, \advstrategy) \conforms \run$)
if $\run$ results from the interactions of the honest user strategies $\strategy_{\honestuser_j} \in \strategy_\honestusers$ with the adversarial strategy $\advstrategy$, meaning that $\action_i = \advstrategy(\run_i, {\mempool_i})$ for $\run_i = \environmentvec_0 \sstep{\action_0} \dots \sstep{\action_{i-1}} \environmentvec_i$ and $\mempool_i =  \bigcup_{\honestuser_j \in \honestusers} \strategy_{\honestuser_j}(\run_i)$ for all $i = 0, \dots, n-1$.
For proving (security) properties for a protocol specified by a set of $\strategy_\honestusers$ of honest user strategies, one needs to consider all runs $\run$ such that $(\strategy_\honestusers, \advstrategy) \conforms \run$ for any adversarial strategy $\advstrategy$.
\vspace{-2pt}

\subsection{\GTree Protocols}
We formally define the strategy $\paramstrategy{\honestuser}{\untree}$ of an honest user $\honestuser$ that sets up and executes an \gtree $\untree$.
To this end, we first specify how to translate an  \gtree $\untree$ to a batch $\batch$ of \ctlcs whose subcontracts represent the \edgestext $\e \in \untree$. 
$\paramstrategy{\honestuser}{\untree}$ will then try to advertise and enable $\batch$, and, if successful, interact with the \ctlcs $\CTLCcontract \in \batch$ according to the way that user $\honestuser$ would interact with the \gtree $\untree$. 

\paragraph{From \GTrees to \ctlcs}
Defining a \gtree protocol requires additional information on how a given \gtree $\untree$ should integrate with the \hbeacronym, e.g., which \edgestext should use which \fundstext of which \tam.
For this reason, we consider tuples of the form $\fulltreeobj$ where $\treeid$ denotes a unique \gtree identifier, 
$\untree$ the \gtree to be executed, $\specalone$ a function mapping \edgestext $\e \in \untree$ to pairs $(\tammath, \fundfull)$ of a \tam $\tammath$ and \fundstext $\fundfull$, and $\inittime$ the time when the \gtree execution should start. 
Note that for $\specalone$ to be valid, we need to assign the same $(\tammath, \fundfull)$ to all \edgestext $\e$, $\e'$ with the same sender and receiver (since those constitute duplicate \edgestext that should be represented by the same \ctlc).

We define the batch $\batch$ of \ctlcs for \gtree $\untree$ as follows:
\begin{definition}[\Gtree to \ctlc conversion]
    Let $\untree$ be an \gtree, $\treeid$ an identifier, $\inittime \in \mathbb{R}$, and $\specalone$ a valid specification. 
\begin{align*}
& \treetoCTLCBadge\fulltreeobj := \\
\{& \CTLCcontract_{(X,Y)} ~|~   \CTLCcontract_{(X,Y)} = [\CTLCsubcontract_{\level_1}, \dots, \CTLCsubcontract_{\level_n}]  ~\land~ \level_1  < \dots < \level_n \\
&\qquad \hspace{-23pt} \land~ \arcpaths_{(X,Y)} = \{ (X, Y)_{\walk} \in \untree \} \\
&\qquad \hspace{-23pt} \land~ \levels_{(X,Y)} = \{ \size{\walk} \hspace{2pt} | \hspace{2pt} (X, Y)_{\walk} \in \arcpaths_{(X,Y)} \} = \{ \level_1, \dots, \level_n \} \\
&\qquad \hspace{-23pt} \land~ x  = (\treeid,  X, Y) ~\land~ \level \in \levels_{(X,Y)} ~\land \e^* \in \arcpaths_{(X,Y)} \\
&\qquad  \hspace{-23pt} \land~ \specalone(\e^*) = (\tammath, \fundfull) \land \CTLCsubcontract_{\level} = (X,Y,\fundfull, \inittime + \level \Delta , \secretset_{\level}) \\
&\qquad \hspace{-23pt} \land~  \secretset_{\level} = \{  \levelsecrets{\e}  ~|~ \e \in \arcpaths_{(X,Y)} ~\land~ \e = (X,Y)_{\walk} \\
&\qquad \hspace{-23pt} \qquad \hspace{-10pt} \land~ \size{\walk} = \level \land \levelsecrets{\e} = \{ \edgesecretfull{\e'}{\treeid} \hspace{2pt} | \hspace{2pt} \e' \in \funonPathtoRoot(\untree, \e) \}  
\} 
\}
\end{align*}
\end{definition}
This definition states that there is a \ctlc $[\CTLCsubcontract_{\level_1}, \dots, \CTLCsubcontract_{\level_n}]$ in $\batch$ 
for each unique sender-receiver pair $(X,Y)$ for which there is an \edgetext $\e^* = (X,Y)_{\walk}$ in $\untree$.
The subcontracts $\CTLCsubcontract_{\level}$ of this \ctlc then correspond to the levels $\level \in \levels_{(X,Y)}$ where \edgestext $(X,Y)_{\walk'}$ occur in $\untree$ and are ordered correspondingly in increasing order (placing the subcontract corresponding to the lowest \gtree level first). 
The subcontract $\CTLCsubcontract_{\level}$ for level $\level$ has timeout $\inittime + \level \Delta$, where $\Delta$ is a sufficient amount of time to execute an action on the specified \tam. The condition $\secretset_{\level}$ contains an element $\levelsecrets{\e}$ for each \edgetext $\e = (X,Y)_{\walk'}$ at level $\level$.
One such element $\levelsecrets{\e}$ contains secrets $\edgesecretfull{\e'}{\treeid}$ for all \edgestext $\e'$ on the path from $\e$ to the root of $\untree$ (denoted by $\funonPathtoRoot(\untree, \e)$).
According to this construction, executing an \edgetext $\e = (X,Y)_{\walk} \in \untree$ corresponds to claiming a subcontract $\CTLCsubcontract_{\size{\walk}}$ with secrets $\levelsecrets{\e}$ (so all secrets from $\e$ to the root). 
This ensures that the subcontract can only be claimed (by $Y$) if (1) all subcontracts corresponding to \edgestext $\e' = (X,Y)_{\walk'}$ on a higher level ($\size{\walk'} < \size{\walk}$) have been timed out before and (2) the secrets for all \edgestext $\e'$ on the path to the root have been revealed (indicating that these \edgestext have been claimed). This realizes the intended game semantics described in~\cref{sec:fromgraphstotreessubsection}.
In particular, this construction gives us a mapping from \edgestext $(X,Y)_{\walk}$ in a \gtree $\untree$ specified by $\fulltreeobj = \treespec$ to a corresponding contract $\CTLCcontract_{(X,Y)}$, subcontract $\CTLCsubcontract_{\size{\walk}} \in \CTLCcontract_{(X,Y)}$ and claiming condition $\levelsecrets{(X,Y)_{\walk}} \in \CTLCsubcontract_{\size{\walk}}.\secretset$ modeling that edge. 
We formally capture this correspondence by a function $\funedgemap(\treespec, (X,Y)_{\walk}) =  (\CTLCcontract_{(X,Y)}, \CTLCsubcontract_{\size{\walk}}, \levelsecrets{(X,Y)_{\walk}})$.

\paragraph{Honest User Protocol}
The honest user protocol (given by a strategy $\paramstrategy{\honestuser}{\untree}$) for executing an \gtree as given by $\fulltreeobj$ specifies how the user behaves in the different phases of the setup and execution of a \ctlc batch $\batch = \treetoCTLCBadge\fulltreeobj$. 

The strategy, in the setup phase, advertises the contract batch and eagerly tries to set up the \ctlcs corresponding to edges $\e \in \untree$. 
This means it advertises, authorizes, and enables those \ctlcs of $\batch$ containing subcontracts representing $\e$. 
Here, it is taken into account that a subcontract $\CTLCsubcontract$ representing $\e$ may only be enabled once all subcontracts representing its ingoing \edgestext have been enabled before (to avoid the loss of funds). 
After the setup phase, $\paramstrategy{\honestuser}{\untree}$ times out and refunds subcontracts representing edges $\e \in \untree$ of $\honestuser$ as soon as the corresponding timeouts are reached.
Further, $\paramstrategy{\honestuser}{\untree}$ claims subcontracts $\CTLCsubcontract$ representing ingoing edges $\e \in \untree$ if their outgoing edge has been pulled before (or there is no such edge). 
To this end, $\paramstrategy{\honestuser}{\untree}$ shares secrets that were revealed in other \tams (for executing outgoing edges there) and reveals the remaining secret of $\honestuser$ to claim $\CTLCsubcontract$, and, once this was successful, claims $\CTLCsubcontract$.
Finally, when all possible actions have been executed, $\paramstrategy{\honestuser}{\untree}$ schedules an $\elapse \, \delta$ to proceed to the next protocol round.

Note that we can easily lift user strategies to operate on sets $\treeobj$ of trees of the form $\fulltreeobj$, given that their identifiers are unique and their funds are disjoint. This enables the secure concurrent execution of multiple \gtree protocols.
A full specification of such generalized strategies $\treestrategy{\honestuser}$ can be found in  
\iffullversion
Appendix \ref{sec:honestStrategy}. 
\else 
\cite{extended-version}.
\fi  

\subsection{Security and Correctness}
We prove that the protocol given by $\treestrategy{\honestuser}$ is secure and correct. 
Intuitively, security in this context means that 
\fundtext transfers observable in the protocol execution can be mapped to a valid \gtree behavior. 
More formally, this is captured by the following theorem: 

\begin{theorem}[Protocol Security] \label{th:protocolsecurity-body}
    Let $B$ be an honest user, $\treeobj$ be a set of tuples of the form $\fulltreeobj$, which is well-formed, and $\Bstrategy$ the honest user strategy for $B$ executing $\treeobj$. Let $\Astrategy$ be an arbitrary adversarial strategy. Then for all final runs $\run$ with $(\Bstrategy , \Astrategy) \conforms \run$, starting from an initial environment, and for all $\treespec = \fulltreeobj \in \treeobj$ there exists $\widetilde{\outcome}_{\treeid} \in \outcomeset{\untree} \cup \{ \emptyset \}$ such that 
        \begin{align*}
              &\forall \treeid, \e \in \widetilde{\outcome}_{\treeid}: B \in \funusers(\e) \\
              & \qquad \Rightarrow \DecideCo(\funedgemap(\treespec, \e)) \in \actions{\run} ~\text{and} \\
              & \forall \DecideCo (\CTLCcontract ,\CTLCsubcontract, \secretsetelem) \in \actions{\run}: B \in \funusers(\CTLCcontract) \\
              & \qquad \Rightarrow \exists \treeid, \e \in  \widetilde{\outcome}_{\treeid}: \funedgemap(\treespec, \e) = (\CTLCcontract ,\CTLCsubcontract, \secretsetelem).
        \end{align*}
    where a run $\run$ is considered \emph{final} if it passed time \linebreak $\textit{max} \{t_0 + \fundepth(\untree) \Delta \vert \treespec \in \treeobj \}$, an \emph{initial} environment is a vector $\environmentvec^0$ where the components of all elements but $\envindex{\tammath}^0.\avFunds$ are empty and $\actions{\run}$ denote the actions appearing in $\run$.
    \end{theorem}

The theorem states that for every protocol run $\run$ of $\Bstrategy$ passing time $\textit{max} \{t_0 + \fundepth(\untree) \Delta \vert \treespec \in \treeobj \}$ and interacting with an arbitrary adversary, there is an outcome $\widetilde{\outcome}_\treeid \in \outcomeset{\untree} \cup \{ \emptyset \}$ for every \gtree $\untree$ in $\treeobj$ such that these $\widetilde{\outcome}_\treeid$ correspond exactly to all \ctlcs that have been claimed during $\run$. 
In particular, this excludes that there are funds claimed by the attacker in a way that is not covered by $\outcomeset{\untree}$ for the \gtrees $\untree$ in $\treeobj$. 
The case $\widetilde{\outcome}_{\treeid} = \emptyset$ covers the case where $\honestuser$ may be the root node of $\untree$ but still their ingoing \edgestext are not executed because the adversary caused the \gtree set up to fail. 
However, in this case, the security statement ensures that no \fundstext were transferred between users, and hence $\honestuser$ did not lose money. 
In particular, together with~\cref{thm:unfolding-security}, we immediately obtain end-to-end security for \gtrees $\untree = \fununfold(\graphsymbol, A)$ resulting from the unfolding of an \insemi{} \atgacronym $\graphsymbol$ as we know that a user $\honestuser$, cannot be underwater in any of its outcomes $\outcome \in \outcomeset{\untree}$.
The formal proof of~\cref{th:protocolsecurity-body} and the one of an end-to-end security statement can be found in 
\iffullversion
Appendix \ref{sec:protocolSecurity}. 
\else 
\cite{extended-version}.
\fi  

Similar to security, we can show that the protocol $\treestrategy{\honestuser_i}$ faithfully executes the tree given that all users $\honestuser_i$ involved in $\treeobj$ follow the protocol. 
Formally, this is captured by the following theorem: 

\begin{theorem}[Protocol Correctness]
    \label{th:protocolcorrectness-body}
    Let $\honestusers =  \{ \honestuser_1, \dots, \honestuser_k\}$ be a set of honest users, $\treeobj$\hspace{-2pt} a well-formed set of tuples of the form $\fulltreeobj$ and $\funusers(\treeobj) \hspace{-1pt} \subseteq \hspace{-1pt} \honestusers$, $\honestuserstrategies^{\treeobj} \hspace{-2pt} = \hspace{-2pt} \{ \treestrategy{\honestuser_1}, \dots, \treestrategy{\honestuser{_k}} \}$ a set of honest user strategies for $\treeobj$, $\Astrategy$ an arbitrary adversarial strategy (for $\honestusers$). 
    Let \hspace{-2pt}$\run$\hspace{-2pt} with \hspace{-3pt}$(\honestuserstrategies^{\treeobj}, \Astrategy) \hspace{-4pt} \results \hspace{-4pt} \run$\hspace{-2pt} be a final run starting from an initial configuration $\environmentvec_0$ that is liquid w.r.t. $\treeobj$. 
    Further, let $\environmentvec_0.t < \textit{min} \{ t_0 - \fundepth(\untree)\Delta \vert \treespec \in \treeobj \}$ then for all $\treespec \in \treeobj$ there exists 
    $\outcome^*_{\treeid} \in \bigcap_{\honestuser \in \honestusers} \outcomeset{\untree}$ such that 
    \begin{align*}
     &  \forall \treeid, \e \in \outcome^*_{\treeid} \Rightarrow \DecideCo(\funedgemap(\treespec, \e))  \in \actions{\run} ~\text{and} \\
     & \forall \DecideCo (\CTLCcontract ,\CTLCsubcontract, \secretsetelem) \in \actions{\run}: \funusers(\CTLCcontract) \cap \honestusers \neq \emptyset \\
     & \qquad \Rightarrow \exists \treeid, \e \in \outcome^*_{\treeid}: \funedgemap(\treespec, \e) = (\CTLCcontract ,\CTLCsubcontract, \secretsetelem).
    \end{align*}
    Where $\environmentvec_0$ is liquid w.r.t. $\treeobj$ requires that all funds as specified in $\treeobj$ are present in $\environmentvec_0$.
\end{theorem}
The theorem states that for every final protocol run $\run$ with users behaving honestly (following their corresponding $\treestrategy{\honestuser_i}$ strategies) and that started sufficiently early to complete the setup (so before $\textit{min} \{ t_0 - \fundepth(\untree)\Delta \vert \treespec \in \treeobj \}$) and in a state were all funds to be used by the protocol are available,
there is an outcome $\outcome^*_{\treeid} \in \bigcap_{\honestuser \in \honestusers} \outcomeset{\untree}$ for every tree $\untree$ in $\treeobj$ such that these $\outcome^*_{\treeid}$ correspond exactly to all \ctlcs that have been claimed during $\run$. 
Intuitively, such outcomes $\outcome^*_{\treeid}$ represent the protocol outcomes for tree $\untree$ with identifier $\treeid$ that all honest users agree upon.
In the case of $\untree = \fununfold(\graphsymbol, A)$, $\outcome^*_{\treeid}$ corresponds to the execution of the whole graph $\graphsymbol$ (as proven in~\Cref{theorem:correctness-graph-tree}), which gives us the end-to-end correctness guarantee that if all preconditions of~\Cref{th:protocolcorrectness-body} are met, the protocol will execute all transfers described by $\graphsymbol$.
A formal proof of this statement is given in 
\iffullversion
Appendix \ref{sec:protocolSecurity}. 
\else 
\cite{extended-version}.
\fi  

\paragraph{Discussion} The chosen model for \ctlc-supporting \tams abstracts from several real-world aspects: 
(1) it features a symbolic model of cryptography and
(2) it assumes synchronized time across the different \tams.
We opted for those abstractions since (1) the underlying cryptographic building blocks are simple, and the main complexity arises from the interactions with the reordering scheduler. Further, we closely follow the symbolic model presented in~\cite{BitML}, which has been proven sound w.r.t. to a computational execution model of Bitcoin.
Concerning (2), in practice, different \tams may feature different inclusion time guarantees, which may require adjusting timeouts to be in line with the respective guarantees. 
This, however, could be realized with simple \tam-specific time conversions (when assuming synchronized user clocks, a standard assumption in the analysis of blockchain protocols).

The conducted simplifications allow us to focus on the main aspects of the proof, namely on showing how an honest user can achieve reliable security guarantees in a highly distributed adversarial \hbeacronym where the user has only access to limited information. 
Here, a particular challenge lies in accounting for the impact of a malicious scheduler on the \CTLC{} execution. 
E.g., it needs to be considered that the execution order as dictated by the \gtree cannot always be upheld while running the protocol: 
While intuitively, at least those \edgestext involving $\honestuser$ should only be executed in order (an edge $\e_1$ on the path from edge $\e_2$ to the root in $\untree$ should be executed before $\e_2$), even this invariant can be violated during protocol execution: 
After submitting a transaction pulling an ingoing edge $\e_1 = (A_1, \honestuser)_{\walk_1}$ (corresponding to $\treestrategy{\honestuser}$ outputting a $\DecideCo$ action $\action_1$) a malicious scheduler controlling users $A_1, \dots, A_n$ on the path between $\e_2 = (\honestuser, A_n)_{\walk_2}$ and $e_1$ could publish the transaction pulling $\e_2$ (corresponding to $\advstrategy$ outputting a corresponding $\DecideCo$ action $\action_2$) before $\action_1$.
It hence needs to be proven that a secure state (including both $\action_1$ and $\action_2$) will eventually be reestablished in such a case.

\section{\CTLC{} Implementations in Existing \tams}
\label{sec:ctlc-impl}
\label{sec:ctlc-impl-app}

Our \ctlc{} protocol builds on \tams as an abstraction layer of the concrete payment channel or blockchain. In this section, we describe how to instantiate \ctlc{}s over concrete \tams.

\subsection{\CTLC{} Implementation in Blockchains}
A blockchain with support for complex smart contracts (e.g., Ethereum) permits implementing \CTLC{}s directly as such smart contracts. 
However, \CTLC{}s do not require the full expressiveness of a Turing complete scripting language but can be constructed from basic primitives for 
(i) transaction authorization based on digital signatures
and
(ii) timelock checks.
Consequently, \CTLC{}s can also be implemented in many existing blockchains, including Bitcoin, which do not support expressive smart contracts, but only offer (i) and (ii).
In the following, we present concrete implementations of the \CTLC{} for a blockchain with complex smart contract support (e.g., Ethereum) and for a blockchain with simple payment conditions (e.g., Bitcoin), using adaptor signatures~\cite{AumayrEEFHMMR21}.  

\paragraph{\CTLC{} Implementation in Ethereum}
We provide in~\cref{fig:ctlc-sc} an excerpt of \actlc, the Ethereum-based implementation of the \CTLC{} between two users called party and counterparty in this case. 
When a \CTLC{} is created, the first subcontract is enabled (i.e., by having the variable \textit{height} pointing to it). 
The contract permits the counterparty to claim the current subcontract by providing the secrets for the corresponding conditions (i.e., claim function). 
Alternatively, the party can disable the current subcontract when the corresponding timeout expires, thereby enabling the next subcontract (i.e., disableSubcontract function). 
Finally, the party can get refunded if the last subcontract is enabled and its corresponding timeout is expired (i.e., refund function). The source code is made available at~\cite{CTLC-implementation}.

As opposed to Ethereum, Bitcoin only supports a simplistic (non-Turing-complete) scripting language, Bitcoin Script, to formulate payment conditions. Bitcoin Script is still sufficient to encode \CTLC{}s when combing multiple transactions with custom payment conditions, as we demonstrate in 
\iffullversion
Appendix \ref{sec:CTLC-in-BitML}. 
\else 
\cite{extended-version}.
\fi  
The Bitcoin-Script-based encoding of \CTLC{}s is structurally similar to the adaptor-signature-based one that we discuss in the following paragraph.

\begin{figure}[t]
    \begin{lstlisting}[language=Solidity,numbers=none]
    contract ethCTLC {
        
    address payable party; // sender 
    address payable counterparty; // receiver 
    uint[] timelock; // One per subcontract 
    bytes32[][][] conditions; // One per (subcontract, path)
    uint height = 0; // pointer to current subcontract 
    
    function claim(uint i, uint j, bytes [] memory secrets) { 
        require(msg.sender == counterparty); 
        require (i == height); 
        require (secrets.length == conditions[i][j].length);
        for (uint k = 0; k < secrets.length; k++)
            require (H(secrets[k]) == conditions[i][j][k]); 
        counterparty.transfer(address(this).balance);}
    
    function disableSubcontract (uint i) {
        require(msg.sender == party);
        require (i == height); 
        require(block.number >= timelock[i]); 
        height ++; }
    
    function refund () { 
        require(msg.sender == party);
        require(height == timelock.length -1); 
        require (block.number >= timelock[height]); 
        party.transfer(address(this).balance);}}
    \end{lstlisting}
    \caption{Pseudocode for \CTLC{} smart contract.\label{fig:ctlc-sc}}
    \end{figure}

\paragraph{CTLC Implementation Using Adaptor Signatures}
Next, we describe a \CTLC{} implementation based on adaptor signatures~\cite{AumayrEEFHMMR21}, compatible with 
\tams that do not offer complex smart contract support, 
but instead offer a functionality restricted to 
(i) authorization of transactions based on a digital signature; and (ii) timelock functionality.
 
Interestingly, since (i) and (ii) are provided by virtually every existing blockchain (including Bitcoin), they all already support an adaptor signature-based \CTLC{} implementation, thereby bringing practical benefits such as reducing transaction fees since verifying a digital signature or checking a timelock is cheaper than executing a full-fledged smart contract. 

An adaptor signature enables the creation of a digital signature on a transaction conditioned on the knowledge of a cryptographic secret. 
Imagine that users $A$ and $B$ have a shared account $\textit{vk}_{AB}$ (i.e., a shared public key) with $\alpha$ coins. 
Then, $A$ and $B$ can jointly \emph{pre-sign} a transaction $m$ spending from $\textit{vk}_{AB}$ with respect to a condition $Y$. 
Afterward, each user on their own can \emph{adapt} the pre-signature into a valid signature using the secret $s$ corresponding to $Y$. 

An adaptor signature-based \CTLC{} is shown in~\cref{fig:ctlc-adaptor}. 
Rounded boxes denote transactions, squared boxes within represent \fundstext. 
The diamond represents choices for transferring \fundstext, denoted by arrows. 
The text over an arrow denotes the required conditions to take this choice, namely (i) $\textit{vk}_{AB}, \vec{Y}_i$ requires a signature from both $A$ and $B$ and the secrets for conditions $\vec{Y}_i$; (ii) $t_i$ requires that time  $t_i$ has passed. 
Intuitively, each transaction corresponds to a subcontract where (i) encodes claiming it and (ii) encodes disabling it.\footnote{For simplicity, we show a single claim option per subcontract. If needed (e.g., when there are duplicated \edgestext on the same level in a \gtree), additional claim operations can be encoded by adding spending conditions of the type ($\textit{vk}_{AB}, \vec{Y}_i$) to the corresponding transaction.}  

More concretely, to setup such a \CTLC{}, for each transaction $\textit{tx}_i$ both parties (i) create a pre-signature $\hat{\sigma}_i$ with respect to conditions $\vec{Y}_i$ to spend the asset to $B$; and (ii) sign the transfer of the asset to $\textit{tx}_{i+1}$. 
One technical subtlety here is that adaptor signatures support pre-signatures with respect to a single condition $Y$ whereas \CTLC{}s require several conditions $\vec{Y}_i$ per pre-signature. Fortunately, as shown in~\cite{TairiMS23}, one can securely merge $\vec{Y}_i$ conditions into $Y^*$ and several secrets $\vec{s}_i$ into $s^*$ such that $s^*$ is a valid secret for $Y^*$ iff secrets $\vec{s}_i$ are the valid ones for the  conditions $\vec{Y}_i$.

After all pre-signatures and signatures are created and verified by both parties, the \CTLC{} is set up. 
To execute, $B$ can \emph{adapt} the pre-signature on $\textit{tx}_1$ into a valid signature with secrets $\vec{s}_1$, thereby getting the \fundstext and finalizing the \CTLC{} contract. 
Otherwise, after $t_1$ expires, $A$ can use the signature on $\textit{tx}_1$ to transfer the \fundstext to $\textit{tx}_2$, effectively enabling the next subcontract.  
This process is repeated until (1) at any step $i$, $B$ \emph{adapts} the pre-signature on $\textit{tx}_i$ using the secrets $\vec{s}_i$; 
or (2) $A$ gets the \fundstext back after $t_n$ expires. 

\begin{figure}[tb]
    \centering
    \includegraphics[width=\columnwidth]{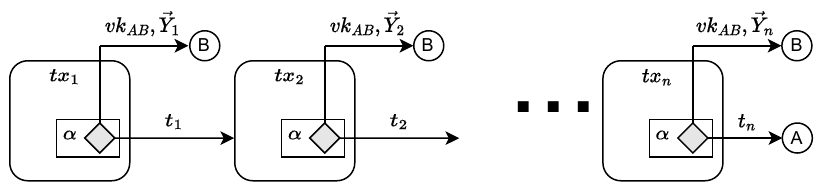}
    \caption{CTLC implementation in UTXO cryptocurrencies. 
    \label{fig:ctlc-adaptor}}
\end{figure}

\subsection{\CTLC{} Implementation in Payment Channels}

A payment channel (PC) allows two users to securely perform an arbitrary amount of instantaneous transactions between each other while including only two transactions on the blockchain: one for opening the PC and one to close it. 
To open a PC, two users publish a transaction that locks their \fundstext into a shared account (e.g., $\textit{vk}_{AB}$).  
Afterward, both users have the guarantee that \fundstext in $\textit{vk}_{AB}$ can only be transferred through (possibly many) transactions that they jointly sign and invalidate. 
Such set of transactions is called a PC state. 
The PC lifetime ends when one of the users publishes the last state on the blockchain. 
Moreover, if one of the users publishes a previously invalidated state, the counterparty can get all PC \fundstext. 

Next we describe how to implement \CTLC{} in a PC. 
Assume that users $A$ and $B$ have already opened a PC with $\alpha$ coins in it.  
Then, it is possible to lift the implementation of a \CTLC{} based on adaptor signatures (c.f.~\cref{fig:ctlc-adaptor}) to the PC setting. 
For that, both $A$ and $B$ compute a new PC state containing all transactions shown in~\cref{fig:ctlc-adaptor} along with a set of pre-signatures and signatures for each transaction.  
At that point, both users have the guarantee that they can execute the \CTLC{} in the PC. 
For instance, assume that timeout $t_1$ is expired, then $A$ can either (i) agree with $B$ to create another PC state that considers only transactions from $\textit{tx}_2$ to $\textit{tx}_n$; or (ii) publish $\textit{tx}_2$ in the blockchain, effectively closing the PC at the current state.  
In the latter case, the execution of the remaining \CTLC{} stays as described earlier. 

\section{Applications}
\label{sec:applications}

\begin{table}
    \caption{Summary of existing blockchain protocols for different applications and \tamsfull (\tams). Dotted arcs denote blockchain transfers, others are transfers on a PC. Abbreviations: Turing Complete (TC), Digital Signatures (DS), Trusted Hardware (TH). \label{table:summary-protocols}}
    \centering
    {\small
    \begin{tabular*}{\columnwidth}{c | p{2.8cm} | c}
        \hline
        \textbf{Applications} & \textbf{\Tams} & \textbf{\atgsacronym}\\
        \hline
        \multirow{ 2}{*}{$\begin{array}{c}
             \text{Multi-hop}  \\
             \text{Payments} 
        \end{array}$} & TC~\cite{state-channel-networks}, HTLC~\cite{lightning-network,fulgor} & \multirow{ 2}{*}{\scalebox{0.75}{%
            $\xymatrix @R=1pc @C=1.5pc{\circ \ar[r] & \circ \ar[r] & \circ \ar[r] &\circ}$
            }
            }\\
            &DS~\cite{multi-hop-locks}, TH~\cite{teechain} &\\
        \hline
        \multirow{ 2}{*}{Rebalancing} & TC~\cite{revive,hide-and-seek,cycle,shaduf}, & \multirow{ 2}{*}{\scalebox{0.75}{%
            $\xymatrix @R=0.3pc @C=1.5pc{\circ \ar[rd] & \circ \ar[l] &\circ \ar[l] &\circ \ar[l] \\
                                        & \circ \ar[r] & \circ \ar[ru] & }$
            }
            }\\
        & DS~\cite{amcu,zeta-rebalancing} &\\
        \hline
        \multirow{ 2}{*}{Loop-in} & \multirow{ 2}{*}{HTLC~\cite{loopin}} & \multirow{ 2}{*}{\scalebox{0.85}{%
            $\xymatrix @R=0.3pc @C=1.5pc{\circ \ar@{.>}[rd] &\circ \ar[l] &\circ \ar[l] \\
                                         & \circ \ar@{.>}[ru] & }$
            }
            }\\
            && \\
        \hline
        Atomic Multi-Path &HTLC~\cite{atomic-multipath-payments,split-payments, boomerang}, & \multirow{ 2}{*}{\scalebox{0.75}{%
            $\xymatrix @R=0.3pc @C=1.5pc{\circ \ar[r] \ar[rd] & \circ \ar[r] & \circ \ar[r] &\circ \\
                                        & \circ \ar[r] & \circ \ar[ru] & }$
            }
            }\\
        Payments & DS~\cite{cryptomaze} &\\
        \hline
        \multirow{ 2}{*}{Crowdfunding} & TC~\cite{pathshuffle}, & \multirow{ 2}{*}{\scalebox{0.85}{%
            $\xymatrix @R=0.3pc @C=1.5pc{\circ \ar[r] &\circ \ar[r] &\circ \\
                                         \circ \ar[r] & \circ \ar[ru] & }$
            }}\\
            & DS~\cite{amcu-crowdfunding}& \\
        \hline
        \multirow{ 2}{*}{$\begin{array}{c}
             \text{2-party}  \\
             \text{Atomic Swaps} 
        \end{array}$} & TC~\cite{p2dex}, HTLC~\cite{bitcoin-wiki-htlc},  & 
            \multirow{ 2}{*}{\scalebox{0.85}{%
            $\xymatrix@C=3.5pc{\circ \ar@/^0.5pc/@{.>}[r] & \circ \ar@/^0.5pc/@{.>}[l]}$
            }
            }\\
                             
            &DS~\cite{thyagarajan2022universal}, TH~\cite{BentovJ0BDJ19} &\\ 
        \hline 
        \multirow{ 3}{*}{$\begin{array}{c}
             \text{n-party}  \\
             \text{Atomic Swaps} 
        \end{array}$} & \multirow{ 3}{*}{TC~\cite{atomic-swaps,imoto2023atomic}} &  \multirow{ 2}{*}{\scalebox{0.85}{%
        \hspace{-13pt}$\xymatrix @R=1pc @C=1.5pc{\textit{n=3:} &\circ \ar@/^0.25pc/@{.>}[rd] \ar@/_0.25pc/@{.>}[ld] &\\
                             \circ \ar@/^0.25pc/@{.>}[rr] \ar@/_0.25pc/@{.>}[ru] & & \circ \ar@/^0.25pc/@{.>}[ll] \ar@/^0.25pc/@{.>}[lu]}$}}\\
        & & \\
        && \\
        \hline
    \end{tabular*}
    }
\end{table}

In this section, we overview real-world applications that can be captured in terms of \atgacronym specifications, and that, hence, can be generated using our framework. 
\cref{table:summary-protocols} summarizes existing protocols from the literature for such applications.
These protocols are application-specific and tailored to the capabilities of the underlying \tams (e.g., the support of Turing-complete smart contract languages), and consequently feature custom security and correctness notions and proofs. 

However, the security and correctness notions of these protocols coincide with the \atgacronym-based ones for the \atgsacronym depicted in the right-most column of~\cref{table:summary-protocols}.
Consequently, using our framework, we can generate protocols that provide the same guarantees by design, and that can be adapted to the desired use case by choosing an adequate \CTLC{} instantiation as discussed in~\Cref{sec:ctlc-impl}.
By instantiating with a fitting \CTLC{}, we do not only match the custom protocols' security, but we also arrive at protocols whose \tam-interactions coincide with the ones of the protocols in~\cref{table:summary-protocols}.
In particular, this means that our protocols do not incur any substantial overhead (e.g., in terms of additional transactions or execution steps to be conducted on a blockchain) w.r.t. the custom protocols. 

To illustrate how the protocols from~\cref{table:summary-protocols} can be specified in terms of \atgsacronym, we shortly overview their design goals: 

\paragraph{Multi-hop Payments in Payment-Channel Networks~\cite{multi-hop-locks}} 
As described in~\cref{sec:intro}, in payment channel networks (PCNs), multi-hop payments are realized by atomically executing (off-chain) transfers over bilateral payment channels (PCs) along a linear payment graph. 

\paragraph{Rebalancing in PCN~\cite{zeta-rebalancing}} 
Since a bilateral off-chain transfer in a PC redistributes the ownership of the PC \fundstext, the number of off-chain transfers is limited by the fraction of the coins owned by the transfer sender
(also called the sender’s balance in the PC). 
This limitation can be addressed with rebalancing protocols.
These protocols restore the payment capabilities of potential senders (S) on a PC pc by increasing their balance on pc through a cyclic (zero-sum) atomic multi-hop payment where S acts as a receiver (on pc)~\cite{zeta-rebalancing}. This behavior can be captured by a cyclic \atgacronym.

\paragraph{Loop-in~\cite{loopin}} 
Depending on the distribution of funds within the PCN, rebalancing may not always be possible. 
In such cases, so-called loop-in protocols~\cite{loopin} allow for the atomic integration of on-chain transactions into off-chain payment channels to provide the required funds. 

\paragraph{Atomic Multi-path Payments~\cite{atomic-multipath-payments}} 
Atomic Multi-path Payment protocols~\cite{atomic-multipath-payments} allow for the atomic execution of multiple linear multi-hop payments between the same sender and receiver. This is of practical relevance if there exists no single payment path with sufficient capacity to carry out a multi-hop payment. 
In this case the payment can be split across several paths, which need to be executed atomically.

\paragraph{Crowdfunding~\cite{amcu-crowdfunding}} 
Combining atomic multi-path payments from different senders enables crowdfunding applications where multiple users jointly pay a single user.
Atomicity here denotes that either all users contribute the pre-agreed amount to the crowdfunded user or 
all the funders get their coins back. 

\paragraph{Atomic Swaps~\cite{htlc-bip,atomic-swaps}} 
Two-party atomic swaps realize the atomic exchange of \fundstext that users hold on different blockchains. 
A generalization of this application is an $n$-party atomic swap~\cite{atomic-swaps}, where $n$ parties holding \fundstext at different blockchains perform a multilateral exchange. 

\paragraph{Limitations}
Applications that demand properties outside the scope of \atgacronym specifications cannot be generated using our framework. 
This applies e.g., to applications 
(i) whose requirements are not sufficiently characterized by the security and correctness guarantees discussed in~\Cref{cha:trees,sec:CTLCs};
(ii) whose security and functionality goals could only be captured with an \atgacronym that is not \insemi.

An example that falls in both categories is coin-mixing services~\cite{castejon2025mixbuy}.
Coin-mixing services enable multiple concurrent payments (involving different senders and receivers) that are routed through a single intermediary, a so-called mixer. 
The goal of these services is to hide from any third parties (including the mixer) which sender paid which receiver, a property usually called unlinkability.

It would be possible to represent such coin-mixing services as a graph-like structure to express that such services should ensure the atomic execution of all concurrent payments. 
However, satisfying this security and correctness notion alone falls short of ensuring the service's crucial unlinkability property.
Indeed, the construction presented in this paper excludes any unlinkability notion by design since the full \atgacronym structure needs to be known by all participants in order to construct the \gtree and set up \ctlcs accordingly.

Further, a coin-mixing service with multiple receivers would need to be described by an \atgacronym that is not \insemi, since there is no single node reachable from all receivers.

For applications that only suffer from restriction (ii), this restriction can often be lifted by adding cyclic payments of arbitrarily small amounts to the application to make their \atgacronym specification \insemi. 
We discuss this general strategy in more detail in the following paragraph.

\paragraph{Composing Applications} 
\begin{figure}[t]
    \centering
    \includegraphics[width=\columnwidth]{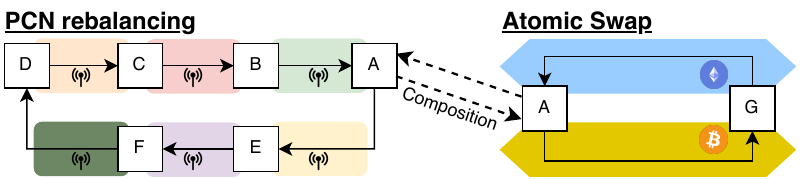}
    \caption{Composition of PCN rebalancing and atomic swap.  
    \label{fig:composition}}
\end{figure}
The applicability of our framework can be extended to applications that require more than a single \atgacronym by adding cyclic transfers that do not incur any financial harm on users.
Consider the scenario where two applications expressed as independent in-semiconnected graphs should be executed atomically. 
For example, in~\cref{fig:composition}, assume that user $A$ is interested in atomically executing a PCN rebalancing and an atomic swap where they are involved. 
The union of the corresponding \atgsacronym would not be in-semiconnected and hence out of scope of our framework. 
This issue can be solved by adding a cycle payment between the nodes of $A$ in these two graphs, see \cref{fig:composition}.
By adding such a payment, $A$ does not incur any loss as long as the \fundstext that are sent and received are the same. 
Moreover, the amount can be arbitrarily small (e.g., the smallest amount supported by the underlying \tam) since it is independent of the other applications. 

We can formally capture this observation as a general composition result: 
Assume that $\graphsymbol_1$ is an \atgacronym \insemi in node $A_1$ and $\graphsymbol_2$ is an \atgacronym \insemi in node $A_2$. 
Then the union of $\graphsymbol_1$  and $\graphsymbol_2$ is \insemi in $A_1$ and $A_2$ when adding the edges $(A_1, A_2)$ and $(A_2, A_1)$ to it.
A formal proof for this statement is given in 
\iffullversion
\cref{th:composingtwoinsemigraphs}. 
\else 
\cite{extended-version}.
\fi  

\section{Comparison with Related Work}
\label{sec:related-work}

Existing protocols across \hbeacronym can be grouped into: custom cryptographic protocols (as presented in~\cref{sec:applications}) and general protocols (yet for restricted class of applications). 
Since we already discussed the first group in~\cref{sec:applications}, 
we focus here in comparing with works in the second one.

Herlihy contributes in \cite{atomic-swaps} a general protocol for cross-currency swaps among several users, as can be represented by strongly connected graphs. 
The protocol relies on locking funds by transferring them into a stateful smart contract and making their release to the receiver subject to a complex unlocking procedure. 
This unlocking procedure requires the smart contract to store a copy of the whole graph. For unlocking, the receiver of the transfer needs to perform operations for all its paths in the graph to all nodes in a dedicated set of leaders, which (depending on the graph structure) may encompass all graph nodes. 

To improve upon the performance of~\cite{atomic-swaps}, Imoto et al.~\cite{imoto2023atomic} propose a protocol that also locks funds into a stateful smart contract as in~\cite{atomic-swaps} but make their release subject to an unlocking procedure more efficient than in~\cite{atomic-swaps}. 
For unlocking, the receiver needs to perform operations that scale with the number of users (instead of users and leaders as in~\cite{atomic-swaps}). 

Unfortunately, the protocols in~\cite{atomic-swaps,imoto2023atomic} are specific to blockchains with support for stateful smart contracts and restricted to applications that are represented as a strongly connected graph such as atomic swaps. Moreover, the security of these protocols is not analyzed in a realistic blockchain execution model. 
Instead, our approach can be implemented in virtually any blockchain (c.f.~\cref{sec:ctlc-impl}), supports applications specified as in-semiconnected \atgsacronym (c.f.~\cref{sec:applications}) and its security relies on a realistic blockchain model (c.f.~\cref{sec:CTLCs}). 

While being more general, our approach does not forfeit in performance, as shown in~\cref{table:asymptotic-results}. Here, as in~\cite{imoto2023atomic} we consider \emph{local time}: an upper bound on the computation cost to process the unlocking of an \arctext; and \emph{local space}: the total amount of bits that are stored in the blockchain per \arctext. 
Our protocol does not require storing the whole graph structure and the number of operations for unlocking funds is bounded by the number of subcontracts in a \ctlc{}, hence, scaling at most linearly with the number of graph nodes. 
We also confirm that these improvements are not only theoretical by showing that even for an unoptimized Ethereum-based implementation of \ctlc{}s, our protocol results in at least comparable gas cost for a small, concrete atomic swap (c.f. Appendix~\ref{sec:implementation}).

The authors of~\cite{atomic-swaps} acknowledge it as a limitation of their protocol that it relies on a set of leaders, since a single-leader protocol could be realized from simpler smart contract functionality and improve performance. The development of such a protocol is left as an open research problem, which we solve with this work. 
We defer a more detailed description and comparison with \cite{atomic-swaps,imoto2023atomic} to Appendix~\ref{sec:extended-related-work}.

\begin{table}[tb]
\caption{Comparison of asymptotic complexity for~\cite{atomic-swaps,imoto2023atomic} and ours. Here $\mathcal{L}$ denotes the set of leaders. \label{table:asymptotic-results}}
    \centering 
    {\small 
\begin{tabular}{|c | c | c | c}
    \hline
    & Local Time &  Local Space\\
    \hline
    Herlihy~\cite{atomic-swaps} & $O(\nodesymbol \cdot \mathcal{L})$ &  $O(\arcsymbol)$ \\
    \hline
    Imoto et al.~\cite{imoto2023atomic} & $O(\nodesymbol)$ & $O(\nodesymbol)$\\
    \hline
    Ours & $O(\nodesymbol)$   & $O(\nodesymbol)$\\
    \hline
\end{tabular}%
    }
\end{table}

\section{Conclusion}
\label{sec:conclusion}
 
We present a framework for secure-by-design protocols for \hbeacronym. 
The framework encompasses (i) the provably correct and secure translation from \atgacronym  specifications into \gtrees, an intermediate layer representing protocols as interactive \fundtext redistribution games among users in different \tams;
and (ii) a generic protocol that realizes \gtrees from a simple smart contract building block that we call \ctlc.

\section*{Acknowledgments}

We would like to thank the reviewers for their helpful feedback. 
This work has been supported by the Heinz Nixdorf Foundation through a Heinz Nixdorf Research Group (HN-RG) and funded by the Deutsche Forschungsgemeinschaft (DFG, German Research Foundation) under Germany’s Excellence Strategy—EXC 2092 CASA—390781972.
Further, this work has been partially supported by the ESPADA project (grant PID2022-142290OB-I00), MCIN/AEI/10.13039/501100011033/ FEDER, UE; 
and by the PRODIGY project (grant ED2021-132464B-I00), funded by MCIN/AEI/10.13039/501100011033/ and the European Union NextGenerationEU/ PRTR.

\bibliographystyle{IEEEtran}
\bibliography{bib}

\begin{appendices}

\section{Performance Evaluation and Comparison}
\label{sec:implementation}

In this section, we aim to compare the concrete cost of our protocol with those in~\cite{atomic-swaps,imoto2023atomic} for a small, concrete atomic swap. Since the smart contracts required in~\cite{atomic-swaps,imoto2023atomic} are only known to be realizable in Ethereum-like blockchains, we compare them with \actlc, our implementation of \CTLC{} using Ethereum as the \tam (c.f. ~\cref{sec:ctlc-impl}). The source code is made available at~\cite{CTLC-implementation}. 
To develop and test these three smart contracts, using as input the \arctext $B \rightarrow C$ in the running example of~\cref{fig:three-party-tree}, we have used the following toolset, testbed, and methodology. 

\paragraph{Toolset} 
Instead of directly interacting with the main currently deployed Ethereum blockchain, Ganache~\cite{ganache} allows one to spawn a local instance of an Ethereum blockchain where one can freely configure, e.g., (i) when new blocks are created; (ii) the creation of new and pre-funded Ethereum addresses; and beyond. This permits the creation of a controlled, safe environment to test smart contracts without harming the main Ethereum blockchain and their users (e.g., by clogging the main network with sample contract calls). The Truffle framework~\cite{truffle}, on the other hand, provides an API that eases the interaction with such a Ganache-based blockchain. Using Truffle, one can trigger (i) user-based actions (e.g., deploy a smart contract or call a function of a previously deployed smart contract); and (ii) miner-based actions (e.g., mine a block). Moreover, every time a user-based action is executed, Truffle reports a summary of its cost (e.g., transaction size or gas consumption). A thereby created environment is a standard mechanism to test smart contracts developed and maintained by the Ethereum community.

\paragraph{Testbed}
We spawned a Ganache-based blockchain with one pre-funded account per user. This models the (possibly many) assets each user wants to transfer as the ATG specifies. Moreover, we configured the blockchain parameters (e.g., gas limits) identical to those of the main Ethereum blockchain. To ease our experiments, we set each block to contain a single transaction so that we can easily compute and test timelock transactions based on blockchain length. With this setup, we extracted measurement data as follows. Given a smart contract (e.g., ethCTLC, the Ethereum-based instance of our CTLC described in~\cref{sec:ctlc-impl}), we first deploy it using Truffle and record the reported cost in terms of transaction size and gas consumption. After the smart contract is included in the blockchain, we use Truffle to interact with the smart contract by means of user-based actions (e.g., call the claim and refund operations) and record the reported cost. 

\paragraph{Methodology}
To compare the performance of ethCTLC with the smart contracts in \cite{atomic-swaps,imoto2023atomic}, we followed the same evaluation methodology for the three contracts. For each such contract, we created a script that uses Truffle and Ganache to (i) spawn a fresh blockchain; (ii) trigger in order the user actions to execute the contract calls required for claiming an edge in the best case. We repeated this process to claim an edge in the worst case and refund an edge (both in the best and worst case). In all cases, we consider the management of the same edge for the three contracts.

Our evaluation results are included in~\cref{table:eval-results}.
We show the gas cost to deploy the contract, claim an \arctext and refund an \arctext. 
We observe that the gas costs reflect the asymptotic performance of the compared protocols. For instance: the contract in~\cite{atomic-swaps} needs to store a complete copy of the graph and consequently has the highest deployment cost. Similarly, the contract in~\cite{atomic-swaps} also has the worst gas cost for the claim operation, reflecting the worst asymptotic performance of this approach. 
Perhaps more interestingly, we additionally observe that in the best case, both~\cite{atomic-swaps} and~\cite{imoto2023atomic} have better gas cost for refund than \actlc. 
This gap comes from the different approaches to handling the refund of an \arctext. In~\cite{atomic-swaps} and~\cite{imoto2023atomic}, the refund is implemented as a \emph{timeout}, that is, \fundstext can be claimed until a certain time $t$,  after which refunding the \arctext is the only option allowed.  
In \actlc, the refund is implemented as a \emph{timelock}, that is, the enabling of the $i$-th subcontract is not allowed until time $t_i$ and doing so effectively disables subcontract $i-1$. Consequently, the refund of the \arctext requires to have previously disabled all subcontracts. 
In a nutshell, using timeouts makes it possible to implement refund as a constant operation, while timelocks require a number of operations linear in the number of subcontracts. 
While we could have adopted the refund approach based on timeouts, we decided otherwise because timeouts are only realizable in cryptocurrencies with expressive smart contracts. 
By using timelocks instead, \CTLC{}s can be realized in cryptocurrencies with restricted smart contracts (e.g., Bitcoin) and support many more applications (c.f.~\cref{sec:applications}). 

\begin{table}[tb]
    {\resizebox{\columnwidth}{!}{%
    \begin{tabular}{c |  c | c | c | c | c }

         & Deploy & \multicolumn{2}{c}{Claim} & \multicolumn{2}{c}{Refund}\\
         \hline
         & & Best & Worst & Best & Worst\\
         \hline
    Herlihy~\cite{atomic-swaps} & $2080742$ & \multicolumn{2}{c|}{$385027$} & $34809$ & $315129$\\
    \hline
    Imoto et al.~\cite{imoto2023atomic} & $1278108$ & $57433$ & $92213$ & \multicolumn{2}{c}{$34417$}\\
    \hline
    Ours & $865120$ & $45478$ & $96427$ & \multicolumn{2}{c}{$80666$}\\
    \hline

    \end{tabular}%
    }}
    \caption{Comparison of gas cost for~\cite{atomic-swaps,imoto2023atomic} and ours. \label{table:eval-results}}
\end{table}

\section{Extended Related Work}
\label{sec:extended-related-work}

In this section, we describe and compare in detail with the works in~\cite{atomic-swaps,imoto2023atomic}, where the authors present a framework to design protocols for the concrete application of atomic swaps. 

    \paragraph{Detailed Comparison with~\cite{atomic-swaps}}
Herlihy presents in~\cite{atomic-swaps} a protocol for realizing strongly connected \swapgraphs. 
The protocol relies on a set of leaders forming a minimal feedback vertex set, meaning that removing the leaders from the graph results in an acyclic graph, and thus there is a unique path from each \nodetext to all leaders.
The \fundstext for each \arctext are getting locked in a complex (Ethereum-style) smart contract which enforces the protocol execution along the paths to the leaders, each of which holds an individual secret.
The swap execution starts when leaders partially unlock ingoing \arcstext in the first round by providing their secret and a signature on it. 
In the next round, every user whose outgoing \arctext was partially unlocked does the same for their ingoing \arcstext by providing the learned secret and their signature on the signature obtained in the previous round (as proof for the path to the leader through which the secret was obtained).
This procedure continues until all \arcstext have been reached via all their paths to all leaders.
Consequently, for claiming a fund, the corresponding \arctext must have been fully unlocked meaning that for all paths to all leaders, the corresponding leader secret and a nested signature on that secret of all users on the path to the leader must have been provided in the adequate round. 
A partial lock is implemented using a primitive called \emph{hash key} which checks that before a predefined timeout, the preimage for some hash value and a nested signature on that preimage for a specific user path is provided. 
If in any round, the expected partial unlocking of an \arctext did not happen, the \fundstext of that \arctext can be refunded to the owner. 

The work proves that this protocol is uniform, meaning that 
1) if all parties follow the protocol, then all \arcstext are executed (corresponding to our \gtree correctness result); 
2) a party following the protocol can never end up \emph{underwater}, meaning that at least one outgoing \arctext is triggered, whereas at least one ingoing \arctext is not triggered (roughly corresponding to our \gtree security result). 
They further prove an impossibility result showing that no such uniform swap protocol can exist for graphs that are not strongly connected. 
At first, this seems contradictory to our result (that proves the security of unfolding for \insemi graphs).
However, we opt for a slightly relaxed security notion: 
Instead of requiring for a user to be underwater to end up with at least one outgoing \arctext triggered and at least one ingoing \arctext missing, we only require that if an outgoing \arctext has been triggered also \emph{all} ingoing \arcstext  (which could potentially be none) must have been triggered. 
This means in particular that if a graph contains a user that functions as a sender only, the security notion still applies in a meaningful way.
We show in~\cref{sec:applications} that such scenarios, indeed, find applicability in practice.
For strongly connected graphs, our security notion coincides with the one from~\cite{atomic-swaps}.

A main advantage of our protocol is that it operates using a single leader. This advantage is two-fold. 
First, having a single leader enhances the liveness of the protocol: Leaders are in the position to block the execution (similar to how user $A$ can force the protocol to timeout by not pulling their ingoing \edgetext in the example presented in~\cref{sec:fromgraphstotreessubsection}) without encountering financial harm. 
Second, in~\cite{atomic-swaps}, it is observed that in single-leader scenarios no digital signatures and hashkeys are needed but only timeouts, and finding such a protocol for the general case is posed as an open challenge. 
Our work solves this challenge and, with that, brings multiple practical benefits (as demonstrated in~\cref{sec:ctlc-impl}):
1) The logic of \CTLC{}s, our smart contract for realizing \edgestext, is much simpler (since no hash keys need to be checked), resulting in improved on-chain performance, even when implementing \CTLC{}s on Ethereum.
2) The on-chain computation cost for executing a swap is asymptotically lower:
In the protocol~\cite{atomic-swaps}, in the worst case, users need to perform as many transactions unlocking the contract for an \arctext before claiming it as there are paths to all leaders in $\graphsymbol$. 
In our protocol, the worst case occurs if all \CTLC{} \edgestext time out, requiring disabling as many subcontracts as there are paths to the (single) leader. \linebreak
3) Our protocol has improved honest on-chain execution cost: If all protocol participants behave honestly, they can claim the assets for their \arcstext with a single blockchain transaction. In contrast, even in the honest case, the protocol from~\cite{atomic-swaps} requires unlocking all hash keys;
4) The smart contract based on hash keys is only known to be realizable in blockchains with expressive smart contract languages, while 
\CTLC{}s can be realized in blockchains with limited or even no scripting capabilities, enabling the implementation of our protocol among all existing cryptocurrencies and even on layer-two solutions, further reducing the on-chain cost.  

\paragraph{Detailed Comparison with~\cite{imoto2023atomic}} 
Imoto et al.~\cite{imoto2023atomic} build upon the work in~\cite{atomic-swaps} to propose a protocol that improves its performance drawbacks.  
As in~\cite{atomic-swaps}, the \fundstext for each \arctext in the graph are locked in a complex (Ethereum-style) smart contract. 
Different to~\cite{atomic-swaps}, the \fundstext are guarded by one secret per user. 
After one such contract per \arctext has been set up, the proposed protocol proceeds in rounds.  
In the $i$-th round, the funds at one \arctext can be released providing $i$-many signatures on the secrets of all users. 
In the first round, any user can claim the \fundstext on one \arctext with only their own signature on all secrets. But if they do so, those users whose funds where taken away have learnt a signature (in addition to their own) and can claim their ingoing \arcstext in the next round. 
Such protocol can take as many rounds as the diameter of the graph, after which \arcstext are timed out.
When compared to~\cite{atomic-swaps}, this work achieves the same security guarantees.  
However, the security analysis is in the same model of~\cite{atomic-swaps}, thereby sharing the same limitations. 
 Instead, as we detail in~\cref{sec:related-work}, this work improves upon~\cite{atomic-swaps} in storage requirements and computation cost. 
When compared to our work,~\cite{imoto2023atomic} and ours share similar storage requirements and computation cost, an improvement over~\cite{atomic-swaps}. 
Yet, \cite{imoto2023atomic} inherits several disadvantages from~\cite{atomic-swaps}, namely 1) it only considers strongly connected graphs and is restricted to the specific application of atomic cross-chain swaps; and 2) the required smart contract is only known to be realizable in cryptocurrencies with complex (Ethereum-style) smart contracts.

\iffullversion
\section{\CTLC{} Implementation in Bitcoin}
\label{sec:CTLC-in-BitML}

\newcommand{\bitmlcontractname}[1]{\textcolor{orange}{#1}}
\newcommand{\bitmlkeyword}[1]{\ensuremath{\mathsf{#1}}}
\newcommand{\bitmlparticipant}[1]{\ensuremath{\textcolor{red}{#1}}}
\begin{figure}[tb]
\begin{equation*}
\scalebox{0.86}{$
    \begin{aligned}
    &\bitmlcontractname{CTLC}(\bitmlparticipant{A}, \bitmlparticipant{B}) ::= \bitmlcontractname{Pay1} + \bitmlkeyword{after}~ t_2 : \tau. ( \bitmlcontractname{Pay2} + 
    \bitmlkeyword{after}~ t_3 : \bitmlkeyword{withdraw}~ \bitmlparticipant{A} )\\
    &\bitmlcontractname{Pay1} ::= \bitmlkeyword{reveal}~ \edgesecret{(B,A)}{\text{\edgeBAone}} \land \edgesecret{(A,B)}{\text{\edgeABtwo}}~ \bitmlkeyword{then}~\bitmlkeyword{withdraw }~ \bitmlparticipant{B}\\
    &\bitmlcontractname{Pay2} ::= \bitmlkeyword{reveal}~ \edgesecret{(C,A)}{\text{\edgeCAone}} \land \edgesecret{(B,C)}{\text{\edgeBCtwo}} \land \edgesecret{(A, B)}{\text{\edgeABthree}}~
    \bitmlkeyword{then}~\bitmlkeyword{withdraw }~ \bitmlparticipant{B}
    \end{aligned}
    $}
\end{equation*}
\caption{BitML contract realizing the \CTLC{} $\ctlcvar_{(A,B)}$ from Eq. \ref{eq:CTLC-running-example}.\label{fig:bitml-ctlc}}
\end{figure}
To show the implementation of \CTLC{}s in Bitcoin, we leverage BitML~\cite{BitML}. 
BitML is a domain-specific language for specifying contracts that describe transfers of Bitcoins among a set of users without relying on a trusted intermediary. 
A compiler is then provided to translate BitML contracts into Bitcoin transactions. 
Participants can execute the contract by appending these transactions to the Bitcoin blockchain according to compiled strategies for each user. 
Therefore, to implement \CTLC{}s in Bitcoin, we can express \CTLC{}s in BitML and use the existing compiler to extract the corresponding Bitcoin transactions. \cref{fig:bitml-ctlc} shows the \CTLC{} $\ctlcvar_{(A,B)}$ from \cref{eq:CTLC-running-example} implemented as a BitML contract. 
On the top level, it consists of two mutually exclusive execution choices (separated by the operator $+$). 
The first choice ($\bitmlcontractname{Pay1}$) intuitively corresponds to subcontract \subctlcABtwo{}: 
It enables $\userB$ to withdraw (claim) the contract funds when providing the secrets  $\edgesecret{(B,A)}{\text{\edgeBAone}}$ and $\edgesecret{(A,B)}{\text{\edgeABtwo}}$.
The subcontract's timeout $t_2$ is implemented by enabling the alternative execution choice $\tau. ( \bitmlcontractname{Pay2} + 
    \bitmlkeyword{after}~ t_3 : \bitmlkeyword{withdraw}~ \bitmlparticipant{A})$
starting from time $t_2$ (as indicated by BitML's $\bitmlkeyword{after}$ keyword). 
Taking this choice corresponds to timing out \subctlcABtwo{}.  Doing so has no further prerequisites (as indicated by the special action $\tau$) and results in the remaining contract 
$\bitmlcontractname{Pay2} + 
\bitmlkeyword{after}~ t_3 : \bitmlkeyword{withdraw}~ \bitmlparticipant{A}$.
Here, again $\bitmlcontractname{Pay2}$ corresponds to \subctlcABthree{} with its timeout $t_3$ being implemented by enabling party $\userA$ to withdraw the contract funds starting from $t_3$. 
\section{Graph Theory}
\label{sec:Graph Theory}

This section will briefly introduce the needed definitions and properties of directed graphs for this work. Further details and a rigorous treatment of directed graphs can be found in \cite{graphtheory}.  

\begin{definition}
    \label{def:digraph}
    A \emph{directed graph} (or \emph{digraph}) $\graphsymbol := (\nodesymbol, \arcsymbol)$ consists of a non-empty finite set $\nodesymbol$ of \emph{vertices} or \emph{nodes} and a finite set $\arcsymbol$ of ordered pairs of distinct vertices called \emph{arcs}. 
    For an arc $\arc=(A,B) \in \arcsymbol$ we call $A$ its \emph{sender}, $B$ its $\emph{receiver}$ and define the functions
    \begin{align*}
        \funsender(\arc) & := pr_1(\arc) = A,\\
        \funreceiver(\arc) & := pr_2(\arc) = B,
    \end{align*}
    where $pr_i(.)$ is the projection onto the $i$-th component of a tuple. 
    Two nodes $A, B \in \nodesymbol$ are \emph{connected} if there exists an arc $(A,B) \in \arcsymbol$ or $(B,A) \in \arcsymbol$, which will be visualized with $A \rightarrow B$ or $B \rightarrow A$ respectively. 
    We require the two vertices forming an arc to be distinct, which removes loops from one node back to itself from the set of arcs \cite[p.3] {graphtheory}.
\end{definition}

\begin{definition}\label{def:walk}
    A \emph{walk} in $\graphsymbol=(\nodesymbol,\arcsymbol)$ is a finite tuple of arcs
    $\walk := [\arc_{k-1}, \arc_{k-2},  ...,  \arc_1, \arc_0]$
    for some $\vert \walk \vert = k \in \N$ and 
    \begin{align*}
        &\forall 0 \leq j \leq k-1 : \, \arc_j \in \arcsymbol \text{ as well as } \\
        &\text{if } k \geq 2: \, \forall 0 \leq i \leq k-2 : \, \funsender(\arc_{i+1}) = \funreceiver(\arc_{i}).
    \end{align*} 
    We say that $\walk$ is a \emph{walk from} $\funsender(\arc_{k-1})$ \emph{to} $\funreceiver(\arc_{0})$ or a \emph{$(\funsender(\arc_{k-1}),\funreceiver(\arc_{0}))$-walk} \cite[p.7]{graphtheory}. 
    We use $\concatvec{\walk}{\walk'}$ to denote the concatenation of two walks $\walk$ and $\walk'$. 
\end{definition}

\begin{definition}\label{def:strongly connected}
    A digraph $\graphsymbol=(\nodesymbol,\arcsymbol)$ is \emph{strongly connected} if 
    \[
    \forall A,B \in \nodesymbol \, : \exists \, (A,B)\text{-walk} \, \land \, \exists \, (B,A)\text{-walk} \, \text{\cite[p.7]{graphtheory}}.
    \]
\end{definition}

\begin{definition} \label{def:neighbourhood}
    Let $A \in \nodesymbol$ be an arbitrarily chosen node from $\graphsymbol$, then we define the \emph{extended out-neighbourhood} of $A$ as
    \begin{align*}
        N^+_\graphsymbol(A) := \{ B \in \nodesymbol \setminus \{ A \} \mid \exists \, (A,B)\text{-walk} \}, 
    \end{align*}
    and the \emph{extended in-neighbourhood} of $A$ as
    \begin{align*}
        N^-_\graphsymbol(A) := \{ C \in \nodesymbol \setminus \{ A \} \mid \exists \, (C,A)\text{-walk} \}.
    \end{align*}
    The union of these sets is often referred to as the reachable set of a node, e.g., in \cite[p.16]{graphtheory}. 
\end{definition}

\begin{definition}\label{def:insemiconnected}
    Let $\graphsymbol=(\nodesymbol,\arcsymbol)$ be a digraph and $A \in \nodesymbol$. Then we call $\graphsymbol$ \emph{in-semiconnected w.r.t. $A$} if and only if 
     \begin{equation*}
              N^-_\graphsymbol(A) \cup \{ A \} = \nodesymbol.
     \end{equation*}
\end{definition}

\begin{example}
\label{ex:connected}
    The digraph 
    \begin{equation*}
        \xymatrix@C=3pc{
        A \ar[r] & B \ar[r] & C, }
    \end{equation*}
    is in-semiconnected w.r.t. $C$ but not strongly connected. For that a walk from $C$ to $B$ and from $B$ to $A$ has to be added. By including these arcs we obtain
    \begin{equation*}
        \xymatrix@C=3pc{
        A \ar@/^0.5pc/[r] & B \ar@/^0.5pc/[l] \ar@/^0.5pc/[r] & C, \ar@/^0.5pc/[l] }
    \end{equation*}
    which is now strongly connected. 
\end{example}

\begin{corollary} \label{cor:all_leaders}
It holds for every digraph $\graphsymbol = (\nodesymbol, \arcsymbol)$:
\begin{align*} 
    & \graphsymbol\text{ is strongly connected } \\
    & \Rightarrow
    \graphsymbol\text{ is in-semiconnected w.r.t. any $A\in \nodesymbol$}
\end{align*}
\end{corollary}

\begin{proof}
    Let $\graphsymbol=(\nodesymbol,\arcsymbol)$ be a strongly connected digraph and $A \in \nodesymbol$ chosen arbitrarily. By Definition \ref{def:strongly connected} we have
    \begin{equation*}
        \forall B \in \nodesymbol \setminus \{ A \} \, \exists (B,A)\text{-walk.}
    \end{equation*}
    This implies $N^-_\graphsymbol(A) = \nodesymbol \backslash \{ A \}$. 
\end{proof}

\begin{theorem} \label{th:composingtwoinsemigraphs}
    Let $\graphsymbol_1 = (\nodesymbol_1, \arcsymbol_1)$ be in-semiconnected w.r.t. $A \in \nodesymbol_1$ and $\graphsymbol_2 = (\nodesymbol_2, \arcsymbol_2)$ in-semiconnected w.r.t. $B \in \nodesymbol_2$. We define $\graphsymbol_3 := (\nodesymbol_3, \arcsymbol_3)$ with
    \begin{align*}
        \nodesymbol_3 &:= \nodesymbol_1 \cup \nodesymbol_2 , \\
        \arcsymbol_3 &:= \begin{cases}
            \arcsymbol_1 \cup \arcsymbol_2 \cup \{ (A,B), (B,A)\} & \text{ ,if } A \neq B \\
            \arcsymbol_1 \cup \arcsymbol_2 & \text{ ,if } A = B .
    \end{cases}
    \end{align*}
    Then $\graphsymbol_3$ is in-semiconnected w.r.t. $A$ and in-semiconnected w.r.t. $B$. 
\end{theorem}

\begin{proof}
    By Definition \ref{def:insemiconnected} we have to show that
    \begin{equation*}
              N^-_{\graphsymbol_3}(A) \cup \{ A \} = \nodesymbol_3
     \end{equation*}
     holds. Unfolding this statement using Definition \ref{def:neighbourhood} results in:
     \begin{equation*}
         \forall C \in \nodesymbol_3 \backslash \{ A \} \, \exists (C,A)\text{-walk}
     \end{equation*}
     Firstly, we assume $A=B$. Let $C \in \nodesymbol_3 \backslash \{ A \}$ be arbitrary. If $C \in \nodesymbol_1$ there is an $(C,A)\text{-walk}$ consisting of arcs in $\arcsymbol_1$. If $C \in \nodesymbol_2$, the same is true for $\arcsymbol_2$.
     Secondly, we assume $A \neq B$. We show that $\graphsymbol_3$ is in-semiconnected w.r.t. $A$. For all $C \in \nodesymbol_1 \backslash \{ A \}$ there again is a $(C,A)\text{-walk}$ consisting of arcs from $\arcsymbol_1$. For all $C \in \nodesymbol_2 \backslash \{ B \}$, by definition, a $(C,B)\text{-walk}$ $\walk$ exists with arcs from $\arcsymbol_2$. Additionally, $(B,A) \in \arcsymbol_3$ and so we get a $(C,A)\text{-walk}$ $\concatvec{\walk}{\unitvec{(B,A)}}$. 
     And finally, for $B$ there is a $(B,A)\text{-walk}$, more specifically $\unitvec{(B,A)}$. 
     Analogously we can show that $\graphsymbol_3$ is in-semiconnected w.r.t. $B$. 
\end{proof}
\section{Graph to Tree Conversion}
\label{sec:GraphtoTree}

As in the previous Section, we notate walks consisting of arcs in $\graphsymbol = (\nodesymbol, \arcsymbol)$ with 
$\walk = [a_{\vert \walk \vert - 1}, a_{\vert \walk \vert - 2}, ..., a_0] \in \arcsymbol^{\vert \walk \vert}.$

We will represent trees by their paths from the leaves to the root of the tree. 

\paragraph{Tree Unfolding} 
We first formally define the unfolding of a graph $\graphsymbol$ into a \textit{game tree}. 
This unfolding is given as follows:

\begin{definition}[Tree Unfolding] \label{def:treeunfoldingappendix}
    For a given digraph \linebreak $\graphsymbol = (\nodesymbol, \arcsymbol)$ and leader $A \in \nodesymbol$ we define the unfolding
    \begin{align} \small
        &\fununfold(\graphsymbol, A) := \{ \walk \mid \funpath(\graphsymbol, A, \walk) \} \text{ with } \label{eq:unfold}
    \end{align}
    \begin{align}
        & \funpath (\graphsymbol, A, \walk) := \exists X \in \nodesymbol : a_{0} = (X,A) \nonumber\\
        &\land \forall j \in \{\vert \walk \vert - 2,..., 0\} \, \exists X,Y,Z \in \nodesymbol : \\
        & \hspace{15pt} a_{j+1} = (Z,Y), a_{j} = (Y,X) \nonumber\\
        & \land \forall i,j \in \{\vert \walk \vert - 1,..., 0\} \forall X,Y,Z, \nonumber\\
        & \hspace{15pt} a_i = (X,Y), a_j = (Z,Y) \, \Rightarrow X=Z, i=j \nonumber\\
        & \land \forall X,Y, a_{\vert \walk \vert - 1} = (X,Y) \Rightarrow \exists j \in \{\vert \walk \vert - 2,..., 0\}, Z : \nonumber\\
        & \hspace{15pt} (Z,X) = a_j \lor \nexists Z: (Z,X) \in \arcsymbol. \nonumber
    \end{align}
    The condition $\funpath$ checks that a path is a walk in $\graphsymbol$ ending in node $A$ (first two conditions), in particular, whenever an edge $a_{j+1}$ ends in a node $Y$ the next edge $a_{j}$ should start in $Y$ again.
    The third condition says that no arc can be revisited. This means whenever there is a node $Y$ with two different edges in the same path, both ending in $Y$, they should already be the same. The fourth and last condition ensures that the path always ends with a repeated node or because there is no successor. 
    Thus, $\fununfold(\graphsymbol, A) =: \untree$ is a set of walks. To uniquely identify the edges in $\untree$, we index them with their partial walk, starting with the edge itself and ending with the leader, the root of the tree. For example $\edge$ is indexed with $\walk = [ (X,Y), ..., (Z,A) ]$ for a leader $A$ and some party $Z$. Note that $\walk$ is a partial path starting with $(X,Y)$ and following one of the paths in $\untree$ upwards to the leader. 

    In the following we will use $\concatvec{\walk_1}{\walk_2}$ to denote the concatenation of two paths $\walk_1$ and $\walk_2$. 
    In addition, we will use $\walk_2 \suffix \walk_1$ to denote that $\walk_1$ is a suffix of $\walk_2$ (so that $\exists \walk:~ \walk_2 = \concatvec{\walk}{\walk_1}$).

    Using this, we define what it means for an edge to be an element in the tree:
    \begin{align*}
        \edge \intree \untree : \Leftrightarrow 
        \exists \walk_1 \in \untree, \walk_2:
        \walk_1 = \concatvec{\walk_2}{\walk}
    \end{align*}
    In other words, there is a full walk $\walk_1$ in $\untree$ such that $\walk$ is part of it. 
    For edges $\edge, (X',Y')_{\walk'} \intree \untree$, we define what it means to be on the same path:
    \begin{align*}
        &\funonpath{\edge}{(X',Y')_{\walk'}} \\
        &: \Leftrightarrow \exists \walk'' \in \untree :  \walk'' \suffix \walk ~\land~ \walk'' \suffix \walk'
    \end{align*}
    We define the depth of $\edge$ by its tree level or equivalently the length of $\walk$:
    \begin{align*}
        \fundepth(\edge) := \vert \walk \vert
    \end{align*}
    The set of edges in the tree that are on the same path to the root is then given as:
    \begin{align*} 
    & \funonPathtoRoot(\untree, \e) := \\
    &\{ \e' \intree \untree \mid \funonpath{\e}{\e'} \land \fundepth(\e') \leq \fundepth(\e) \}
    \end{align*}
    This implies $\e \in \funonPathtoRoot(\untree, \e)$. 
\end{definition}
The rest of this section assumes a digraph $\fulldigraph$, which has been unfolded into a tree $\untree$. With $J$, we denote all levels of the tree and index it with $j \in J$. The parties in a given tree level $j$ are notated as the set $N^j$. 
\begin{definition} \label{def:gameTree}
    We define a \textit{game tree} as a finite set of walks
    \begin{align*}
        \untree :&= \{ \walk_1, \walk_2, ..., \walk_n \}, \\
        \forall 1 \leq i \leq n : \, \walk_i &= [a_{\vert \walk_i \vert - 1}, a_{\vert \walk_i \vert - 2}, ..., a_0]
    \end{align*}
    were each element of a $\walk_i$ is defined as a tuple $a_j = (X,Y)_{\walk'}$ where $X,Y$, with $X \neq Y$, come from an underlying set of nodes $\nodesymbol$ and $\walk' = [a_{j}, a_{j-1}, ..., a_0]$ with $\walk_i \suffix \walk'$. As all $\walk_i$ are walks, Def. \ref{def:walk} holds. For $\untree$ to be a game tree, the following two conditions need to hold:
    \begin{align*}
        (i) & ~\exists ! \, Y \in \nodesymbol \, \forall \walk_i \in \untree \, \exists X \in \nodesymbol : \\
        & ~a_0 = (X,Y)_{\unitvec{(X,Y)}} \text{ (Leader)} \\
        (ii) & ~\forall \walk =[a_{\vert \walk \vert - 1}, a_{\vert \walk \vert - 2}, ..., a_0], \\
        & \hspace{9pt} \walk' = [a'_{\vert \walk' \vert - 1}, a'_{\vert \walk' \vert - 2}, ..., a'_0] \in \untree : \\
        & \hspace{9pt} \exists i,j \in \N : a_i = a'_j \Rightarrow [a_i , ..., a_0] = [a'_j, ..., a'_0] \\
        & \hspace{9pt} \text{with } \linebreak \walk \suffix [a_i , ..., a_0], \walk' \suffix [a'_j, ..., a'_0] \text{ (Crossings)}
    \end{align*}
\end{definition} 

\paragraph{Outcome sets}
Based on a game tree $\untree$, we next define the set $ \outcomeset{\untree}$ of an honest user $B$, a participant in $\untree$.
Intuitively, the outcome set contains all possible tree executions that an honest user may observe when eagerly pulling all their ingoing tree edges whenever possible. 

\begin{definition}[Outcome Set]\label{def:outcomes}
    Let $\untree$ be a game tree. 
    Assume the following definitions: 
    \begin{align}
        \partialwalks (\untree)&:= \{ \walk' ~|~ \exists \walk \in \untree : \walk \suffix \walk' \}  \label{def:partialwalks} \\ 
        \partialtrees (\untree) &:= \{ \untree_p ~|~ \untree_p \subseteq \partialwalks(\untree) \} \label{def:partialtrees} \\ 
        \partialtreeoutcomes (\untree) &:= \{ \{ \edge \intree \untree_p \} ~|~  \untree_p \in \partialtrees(\untree) \} \label{def:partialtreeoutcomes}
    \end{align}

    Further, assume the following predicates: 

    \begin{align}
        \predNoDup{\untree}{B}{\outcome} &: \Leftrightarrow \edge \in \outcome \land (X,Y)_{\walk'} \in \outcome
        \label{eq:no-dup} \\
            & ~\land~ (X = B \lor Y = B) \Rightarrow \walk = \walk' \notag
    \end{align}
    \begin{align}
        & \predHonestRoot{\untree}{B}{\outcome} :\Leftrightarrow
          \label{eq:honest-root} \\
            & (X, B)_{\walk} \intree \untree ~\land~ \fundepth((X,B)_{\walk}) = 1 \Rightarrow (X,B)_{\walk} \in \outcome \notag
    \end{align}
    \begin{align}
        &\predEagerPull{\untree}{B}{\outcome} :\Leftrightarrow
         \label{eq:eager} \\
            & (X, B)_{\walk_1} \intree \untree ~\land~ (B,Y)_{\walk_2} \in \outcome \land \walk_1 = \concatvec{\unitvec{(X, B)}}{\walk_2} \notag \\
            & \qquad \Rightarrow \exists \walk_3: (X, B)_{\walk_3} \in \outcome \notag \\
            & \qquad \qquad ~\land~ \fundepth((X, B)_{\walk_3}) \leq \fundepth((X, B)_{\walk_1} ) \notag
    \end{align}

    Then the outcome set $\outcomeset{\untree}$ of user $B$ in $\untree$ is given as 
    \begin{align}
        \outcomeset{\untree} &:= \{ 
            \outcome \in \partialtreeoutcomes (\untree)  \mid 
              \hspace{15pt} \predNoDup{\untree}{B}{\outcome}  \\
            & \hspace{15pt} \hspace{62pt} ~\land~\predHonestRoot{\untree}{B}{\outcome} \notag \\
            & \hspace{15pt} \hspace{62pt} ~\land~\predEagerPull{\untree}{B}{\outcome} \} \notag
    \end{align}

\end{definition}

    The definition incrementally constructs $\outcomeset{\untree}$ by first defining the set $\partialtreeoutcomes(\untree)$ that contains a set of all its edges for each partial tree of $\untree$ (\Cref{def:partialtreeoutcomes}). 
    Partial trees (defined in \Cref{def:partialtrees}) are given by arbitrary subsets of partial walks (defined in~\Cref{def:partialwalks}) of the tree $\untree$.
    Finally, $\outcomeset{\untree}$ restricts $\partialtreeoutcomes(\untree)$ further to only those outcome sets that 
    do not contain any duplicate edges involving honest user $B$ (described by $\predNoDupP$) and that are compliant with an honest strategy of user $B$ (so satisfying the predicates $\predHonestRootP$ and $\predEagerPullP$): 
    More precisely,~\Cref{eq:honest-root} requires that if $B$ is the root, so has an ingoing edge $(X, B)_{\walk} \in \untree$, then this edge is also included in an outcome set $\outcome$. 
    This corresponds to $B$ always pulling all ingoing edges when being the root user.
    Next,~\Cref{eq:eager} requires that if $B$ has an ingoing edge $(X, B)_{\walk_1} \in \untree$ whose outgoing edge $(B,Y)_{\walk_2}$ is included in the outcome set $\outcome$ then also $(X, B)_{\walk_1}$ (or a duplicate thereof on the same or higher tree level) must be contained in $\outcome$.
    This corresponds to $B$ eagerly pulling all ingoing edges whenever possible.

To show the security of the tree-unfolding, we show that each outcome set $\outcome \in \outcomeset{\untree}$ constitutes a good outcome for user $B$, meaning that user $B$ does not end up \emph{underwater}. 
A user is considered underwater if the swap triggers some outgoing arcs of $\graphsymbol$ for user $B$ but not all their ingoing arcs of $\graphsymbol$.
We capture this notion formally with the following theorem:

\begin{theorem}[Security of tree unfolding]
    \label{thm:security-graph}
    Let $\graphsymbol = (\nodesymbol, \arcsymbol)$ be a digraph that is in-semiconnected w.r.t. $A \in \nodesymbol$ and \linebreak 
     $\untree = \fununfold(\graphsymbol, A)$ the tree unfolding of that graph and $B \in \nodesymbol$ be a node representing a user.
    Then, it holds that
    \begin{align}
        &\forall \outcome \in \outcomeset{\untree} : \nonumber \\
        &\forall \walk, \walk' : \bigl( \edge \in \outcome \land (X,Y)_{\walk'} \in \outcome \land B \in \{ X, Y\} \notag \\
        & \qquad \quad \Rightarrow \walk = \walk' \land (X,Y) \in \graphsymbol \bigr) \label{eq:underwater1}\\
        & \land \forall (B,Y)_{\walk} \in \outcome : \, \bigl( (X,B) \in \graphsymbol \Rightarrow \exists \walk': (X,B)_{\walk'} \in \outcome \bigr) \label{eq:underwater2}
    \end{align}
\end{theorem}

\begin{proof}
The proof will be carried out by first showing \eqref{eq:underwater1} and then showing \eqref{eq:underwater2}.

Let $\outcome \in \outcomeset{\untree}$
then by definition of $\outcomesetP$ we have that 
for all $\edge \in \outcome$ that $\edge \intree \untree$.
Since $\untree = \fununfold(\graphsymbol, A)$ only contains walks consisting of arcs from $\graphsymbol$. Also, all elements from $\outcome$ are indexed edges from $\graphsymbol$.

    Further, let $\edge, (X, Y)_{\walk'} \in \outcome$ and $B \in \{ X, Y\}$. Assume towards contradiction that $\walk \neq \walk'$.
    By~\Cref{def:outcomes}, it holds that $\predNoDup{\untree}{B}{\outcome}$ and consequently, by the definition of $\predNoDupP$ also $\walk = \walk'$ immediately giving a contradiction.
    
    Let $\outcome \in \outcomeset{\untree}$ and $(B,Y)_{\walk} \in \outcome, (X,B) \in \graphsymbol$. 
    We make a case distinction on whether there exists some edge $(Z,B) \in \walk$.
    \begin{itemize}
        \item Assume that $\walk = \concatvec{\walk_1}{\concatvec{\unitvec{(Z,B)}}{\walk_2}}$. 
        Since $(B,Y)_{\walk} \in \outcome$, by definition of $\fununfold$ there exists $\walk_t \in \untree$ such that $\walk_t \suffix \walk$.
        In particular, this means that (same as $\walk_t$) also $\walk$ is a walk in $\graphsymbol$ not visiting any arc twice.
        Consequently, $\walk_2$ does not contain any other edge $(Z', B)$. 
        And so, also $\walk' = \concatvec{\unitvec{(X,B)}}{\walk_2}$ is a walk in $\graphsymbol$ not visiting any node twice. 
        Consequently, by definition of $\fununfold$, there must be some $\walk_t' \in \untree$ such that $\walk_t' \suffix \walk'$ and so also $(X,B)_{\walk'} \intree \untree$. 
        We do another case distinction on $\walk_2$: 
        \begin{itemize}
            \item If $\walk_2$ is empty, then $\fundepth(\walk') = 1$ and hence by \linebreak $\predHonestRoot{\untree}{B}{\outcome}$ it follows from $(X,B)_{\walk'} \intree \untree$ that also $(X,B)_{\walk'} \in \omega$.  
            \item If $\walk_2 = \concatvec{\unitvec{(B, U)_{\walk_2}}}{\walk_3}$ then also $(B,U)_{\walk_2} \in \outcome$ (because $(B,Y)_{\walk} \in \outcome$ by definition of $\partialtreeoutcomes$ implies that also all $(V,W) \in \walk$ are included in $\outcome$). 
            By definition of \linebreak $\predEagerPull{\untree}{B}{\outcome}$, it follows from $(X,B)_{\walk'} \intree \untree$ and \linebreak $(B,U)_{\walk_2} \in \outcome$
            that there must be some $(X,B)_{\walk_4} \in \outcome$, which concludes the case. 
        \end{itemize}
        \item Assume that there is no $(Z,B) \in \walk$. 
        Since $(B,Y)_{\walk} \in \outcome$, by definition of $\fununfold$ there exists $\walk_t \in \untree$ such that $\walk_t \suffix \walk$.
        In particular, this means that (as $\walk_t$) also $\walk$ is a walk in $\graphsymbol$ not visiting any node twice.
        Then also $\walk' = \concatvec{\unitvec{(X,B)}}{\walk}$ is a walk in $\graphsymbol$ not visiting any node twice. 
        Consequently, by definition, there must be some $\walk_t' \in \untree$ such that $\walk_t' \suffix \walk'$ and so also $(X,B)_{\walk'} \intree \untree$. 
        By definition of $\predEagerPull{\untree}{B}{\outcome}$, it follows from this and  $(B,Y)_{\walk} \in \outcome$
        that there must be some $(X,B)_{\walk_3} \in \outcome$, which concludes the case. 
    \end{itemize}

\end{proof}

\begin{theorem}[Correctness of Tree Unfolding]\label{theorem:correctness-graph-tree-appendix}
    Let $\graphsymbol = (\nodesymbol, \arcsymbol)$ be a digraph that is in-semiconnected in $A \in \nodesymbol$ and 
    $\untree = \fununfold(\graphsymbol, A)$ the tree unfolding of that graph.
    If all parties $B_i \in \nodesymbol$ are honest, a representative of every edge in $\untree$ will be executed, which means
    \begin{align*}
        \mathcal{H} := \bigcap_{B_i \in \nodesymbol} \outcomesetBi{\untree} &\cong \untree \textit{ i.e.} \\
        \mathcal{H} / \sim  &= \{ \e \intree \untree \} / \sim  
    \end{align*}
    with $\edge \sim (X,Y)_{\walk'}$ for all $\walk$, $\walk'$. 
\end{theorem}
With $\mathcal{H} / \sim $ and $ \{ \e \intree \untree \} / \sim$ the creation of equivalence classes according to '$\sim$' is meant. 
\begin{proof}

    Let $ \outcome \in \mathcal{H}$. Then for all $B_i \in \nodesymbol$, we have $\outcome \in \outcomesetBi{\untree}$. 

    To prove the claim, we show that for all $\edge \in \outcome$ it holds that $(X, Y) \in \graphsymbol$ and that for all $(X, Y) \in \graphsymbol$ there exists $\walk$ such that $\edge \in \outcome$.
    The first claim follows directly from \cref{thm:security-graph}.
    To show the second claim, assume that $(X, Y) \in \graphsymbol$.
    We proceed by case distinction:
    \begin{itemize}
        \item If $Y = A$ then by the definition of $\fununfold$ it holds that \linebreak $(X,A)_{\unitvec{(X,A)}} \intree \untree$.
        Since $\outcome \in \outcomesetP(A, \untree)$, we also know that $\predHonestRoot{\untree}{A}{\outcome}$.
        Consequently, it immediately follows from the definition of $\predHonestRootP$ that 
        $$(X,A)_{\unitvec{(X,A)}} \in \omega .$$
        \item If $Y \neq A$ then since $\graphsymbol$ is in-semiconnected, there must be walks from $Y$ to the leader $A$. These walks can be ranked based on their lengths. With this, we end up with
        (potentially multiple) shortest walks $\walk_s$ (with $\length{\walk_s} > 1$) from $Y$ to $A$. Hence, they visit every node only once. Because there is the possibility of $Y$ appearing multiple times at the same tree level but on different walks, there can be multiple shortest walks with the same length.
        By definition of $\fununfold$ it holds that for all these walks $\walk_s$ we have that $(X,Y)_{\concatvec{\unitvec{(X,Y)}}{\walk_s}} \intree \untree$.
        We assume toward contradiction that there is no $\walk_s$ such that $(X,Y)_{\concatvec{\unitvec{(X,Y)}}{\walk_s}} \in \outcome$.
        
        Then either for all $\walk_s$ it holds that for all $(V,W)_{\concatvec{\unitvec{(V,W)}}{\walk_2}}$ with 
        \[
            \walk_s = \concatvec{\walk_1}{\concatvec{\unitvec{(V,W)}}{\walk_2}}
        \]
        for some $\walk_1, \walk_2$ that $(V,W)_{\concatvec{\unitvec{(V,W)}}{\walk_2}} \not \in \outcome$ 
        or there is a walk $\walk_j$ such that $(V,W)_{\concatvec{\unitvec{(V,W), (W, Z)}}{\walk_2}} \not \in \outcome$ with 
        \[
            \walk'_{s} = \concatvec{\walk_1}{\concatvec{\unitvec{(V,W), (W, Z)}}{\walk_2}}
        \]
        for some 
        $\walk_1, \walk_2$ and $(W,Z)_{\concatvec{\unitvec{(W, Z)}}{\walk_2}} \in \outcome$ 
        and for all $\walk_s$ it holds that for all 
        $(V',W')_{\concatvec{\unitvec{(V',W')}}{\walk_2'}}$ with 
        \[
            \walk_s = \concatvec{\walk_1'}{\concatvec{\unitvec{(V',W')}}{\walk_2'}}
        \]
        for some $\walk_1', \walk_2'$ and $\length{\walk_1'} \leq \length{\walk_1}$ that $(V',W')_{\concatvec{\unitvec{(V',W')}}{\walk_2'}} \not \in \outcome$.
        So either all elements of shortest walks $\walk_s$ are not contained in $\outcome$ or there must be a highest level where the shortest walk $\walk'_{s}$ contains an edge $(V,W)$ not in $\outcome$ followed by an edge $(W, Z)$ in $\outcome$. 
        The first case immediately leads to a contradiction because it implies that also $(V,W)_{\unitvec{(V,W)}} \not \in \outcome$ with $\walk_s =  \concatvec{\walk_1}{\unitvec{(V,W)}}$. However, since $\outcome \in \outcomesetP(W, \untree)$, we also know that $\predHonestRoot{\untree}{W}{\outcome}$.
        Consequently, it immediately follows from the definition of $\predHonestRootP$ that $(V,W)_{\unitvec{(V,W)}} \in \omega$.
        
        We, hence, assume the existence of $\walk'_s$ as described above. 
        Since $\outcome \in \outcomesetP(W, \untree)$, we also know that \linebreak $\predEagerPull{\untree}{W}{\outcome}$ and hence that there must be a $\walk_3$ such that $(V,W)_{\walk_3} \in \omega$ and 
        \[
            \length{\walk_3} \leq \length{\concatvec{\unitvec{(V,W), (W,Z)}}{\walk_2}}. 
        \]
        If $\length{\walk_3} < \length{\concatvec{\unitvec{(V,W), (W,Z)}}{\walk_2}}$ then $\concatvec{\walk_1}{\walk_3}$ would be a shorter walk from $Y$ to $A$ in $\graphsymbol$ than $\walk'_s$, which contradicts that $\walk'_s$ is a shortest path. \newline
        If $\length{\walk_3} = \length{\concatvec{\unitvec{(V,W), (W,Z)}}{\walk_2}}$ then $(V,W)_{\walk_3} \in \outcome$ with \linebreak $\walk_s = \concatvec{\walk_1'}{\walk_3}$ since 
        \begin{align*}
            \length{\walk'_s} &= \length{\concatvec{\walk_1}{\concatvec{\unitvec{(V,W), (W, Z)}}{\walk_2}}} \\
                         &= \length{\walk_1} + \length{\walk_3} = \length{\walk_s} = \length{\walk_1'} + \length{\walk_3}.
        \end{align*}
        So there is a shortest walk $\walk_s$ where an edge $(V',W')$ already occurs in $\outcome$ on a higher level than in $\walk'_s$, which contradicts the assumption that $\walk'_s$ is a shortest walk with the highest level where such edge occurs in $\outcome$. 
    \end{itemize}

\end{proof}

\section{Scalability of the Graph to Tree Conversion}
\label{sec:TtreeSize}

Albeit the \gtree 's size can grow exponentially with the number of users (exact bounds are given in the following), the CTLC-based protocol imposes tolerable on-chain cost per honest user because:
\begin{itemize}
    \item The number of executed edges per user is capped by the number of their arcs in the \atgacronym (and independent of the \gtree size).
    \item Applications relying on payment channels as instance of \tams, execute edges off-chain. 
\end{itemize}

For specifying an upper bound for the number of edges in a \gtree depending on the number of users in the underlying \atgacronym, we assume the largest \atgacronym possible:
Assume there are $N \in \N$ many users, and we have an \atgacronym where every user is directly connected to every other user with an arc. The unfolding process (\cref{def:treeunfoldingappendix}) for deducing a \gtree from an \atgacronym stops whenever a node is revisited and ensures that no walk in the \gtree features two edges linked to the same arc. 
The number of edges in this \gtree is now calculated using combinatorics. Let $A$ be the leader of the \gtree. Then, every walk features a minimum of $1$ and a maximum of $N-1$ other users than $A$. It ends once one of these users or $A$ appears a second time. Hence, the length of a walk, measured in edges, is given as $i + 1$, where $i$ is the number of its users other than $A$. The number of possible walks for a given $i$ is expressed by multiplying the ordered choice of users with the number of options for the last party in the walk: 
$$
	\frac{ (N-1)! }{ (N-1-i)! } ~i
$$
One of these walks has $(i+1)$ many edges. Summing over all $i$ leads us to the number of edges in the \gtree:
$$
	\sum_{i=1}^{N-1} \frac{ (N-1)! }{ (N-1-i)! } ~i (i+1)
$$
Since we considered the largest \atgacronym possible featuring $N$-many users, this is the upper bound for edges in the unfolded \gtree. 
Similarly, the maximum number of nodes is always smaller or equal than
$$
	\sum_{i=1}^{N-1} \frac{ (N-1)! }{ (N-1-i)! } ~2 i 
$$
as there can at most be $2$ instances of a user per walk. The maximum depth of a node in a \gtree is given by the longest possible walk, which has length $N$.
The maximum fan-out is $N-1$.

\section{\CTLC{} Semantics}
\label{sec:inferenceRules}
An environment is given as 
\[ 
    \environment := [ S_{com}, S_{rev}, \batchset, C_{adv}, C_{aut}, C_{en}, C_{cla}, F_{av}, F_{res}, t ]
\]
\begin{align*}
    \environmentCreatedSecrets &:= \{ \text{committed secrets} \} \\
    \environmentRevealedSecrets &:= \{ \text{revealed secrets} \} \\
    \environmentBatches &:= \{ \text{Set of advertised Batches} \} \\
    \environmentCTLCAdvertised & := \{ \text{advertised \CTLC{}s} \} \\
    \environmentCTLCEnabled &:= \{ \text{enabled \CTLC{}s} \} \\
    \environmentCTLCAuthorized &:= \{ \text{authorizations for \CTLC{}s} \} \\
    \environmentCTLCDecided &:= \{ \text{claimed \CTLC{}s} \} \\ 
    \environmentAvailableFunds &:= \{ \text{available funds} \} \\
    \environmentReservedFunds &:= \{ \text{reserved funds} \} \\
    \environmenttime &:= t
\end{align*}
To realize multiple environments at once, we are defining 
\begin{align*}
    \environmentvec := \fullenvironmentvec ,
\end{align*}
where every $\environment_{ch}$ is defined as $\environment$ above. 
The index $ch$ stands for "channel" and identifies one of the elements in $\environmentvec$.
In the main body of the paper, we use $\tammath$ and $\environment_{\tammath}$ instead of $ch$ and $\environment_{ch}$. This change of notation has no further meaning and is only for abbreviating the notation of the following formulas. Also, with "channel", we mean exactly the same object as previously referred to as "\tamfull". 
 
To identify who is a participant in a given environment, we define the assignment function $\conf$, which assigns users from $\environmentvec$ to every $\environmentch \in \environmentvec$ with
\begin{align*}
    \confuser{\environmentch} := \{ \text{users participating in } \environmentch\}. 
\end{align*}
This function $\conf$ is given together with $\environmentvec$ and the notation \linebreak $\conf \results \environmentvec$ means that $\conf$ is defined for all $\environmentch \in \environmentvec$. Additionally, $\conf$ cannot be altered and stays constant throughout all changes to $\environmentvec$. In the following, we will always assume that $\environmentvec$ is given together with $\conf$, and whenever we write $\environmentvec$, we mean $\conf \results \environmentvec$ implicitly. 
Further, we define
\begin{align*}
    \environmentvec.S_{com} &:= \bigcup_{1 \leq \channel \leq \vert \environmentvec \vert} \environmentchCreatedSecrets , \\
    \environmentvec.S_{rev} &:= \bigcup_{1 \leq \channel \leq \vert \environmentvec \vert} \environmentchRevealedSecrets , \\
    \environmentvec.t &:= \environment_1 .t = ... = \environment_{\vert \environmentvec \vert} .t, \\
    \environmentvec.\batchset &:= \environment_1 .\batchset = ... = \environment_{\vert \environmentvec \vert} .\batchset
\end{align*}
We also define the accumulation analogously for the other components of $\environmentch \in \environmentvec$. 
For changing only an element $\environmentch$ to $\environmentch'$ within $\environmentvec$ we use the update notation
\[
    \Update{\environmentvec}{\environmentch}{\environmentch'}.
\]
A \CTLC{} contract is given as a set of subcontracts \linebreak$\CTLCcontract := \{ \CTLCsubcontract_1, ..., \CTLCsubcontract_s \}$ with $\forall 1\leq i \leq s :$
\begin{align*} 
    \CTLCsubcontract_i := \bigl[X,Y,f^{\zeta}_X,\lambda_{i}, \funsecretset(\CTLCsubcontract_i) \bigr]
\end{align*}
where $X,Y \in \nodesymbol$, the set of participants in $\environmentvec$, and
\begin{align*}
    \funsecretset(\CTLCsubcontract_i) := \bigl\{ \{s_{\walk_{1}}^{id, 1},... ,s_{\walk_{n_i}}^{id, 1}\}, &\{s_{\walk_{1}}^{id, 2},... ,s_{\walk_{n_i}}^{id, 2}\}, \\
    ... , &\{s_{\walk_{1}}^{id, m},...,s_{\walk_{n_i}}^{id, m}\} \bigr\}
\end{align*}
is a set of secret sets. The used $x$ is a unique identifier of $\CTLCcontract$ where every subcontract $\CTLCsubcontract$ with the same identifier $x$ is an element of $\CTLCcontract$. 
The $f^{\zeta}_X$ is a token for a fund of user $X$ with a unique identifier $\zeta$.  
$\lambda_{i}$ is a timelock with $\lambda_{i} \in \R^+$ and if $i<s : \lambda_{i} \leq \lambda_{s}$. The identifier $id$ of the secrets is explained in the following paragraph.
For each of those $\CTLCsubcontract_i \in \CTLCcontract$ we define
\[
    \funposition(\CTLCsubcontract_i) := i.
\]
Advertised but not yet enabled contracts are notated with $\advCTLCcontract$. Once they are advertised, they are notated with $\CTLCcontract$. For $\CTLCsubcontract$, we don't use this notation as they are always considered as an element of a contract. The enabling and advertisement of contracts and sub-contracts is defined in the following semantics. 
A batch is given as a set of \CTLC{} contracts $\batch := \{ \advCTLCcontractid{x}_1, \advCTLCcontractid{x'}_2, ... \}$. 
With $\funhonest$, we denote the set of honest users in a given $\batch$.

The remaining element is a set, which, again, consists of sets containing secrets. Each secret $s_{\walk_{\iota}}^{id, j}$ is specific to an $id$, the identifier of $\batch$, and walk $\walk_{\iota}$. Here, in the standalone definition of $\CTLC$s, walks do not have a meaning yet, but this notation will be helpful in the following Appendix \ref{sec:honestStrategy} where we bring together $\CTLC$s with game trees. For now, view $\walk_{\iota}$ just as a differentiating index. 
Furthermore, each secret is owned by $\funowner(s_{\walk_{\iota}}^{id, j}) \in \confuser{\environmentch}$ for some $\environmentch \in \environmentvec$. Let $\Delta$ be a sufficient amount of time to execute an action on the specified $\channel$.

We define $s:= \vert \CTLCcontract \vert$. To access the individual components, we also define the following functions: 
\begin{align*}
    \funsender(\CTLCcontract) &= \funsender(\CTLCsubcontract_1)= ... = \funsender(\CTLCsubcontract_s) = X \\
    \funreceiver(\CTLCcontract) &= \funreceiver(\CTLCsubcontract_1) = ... = \funreceiver(\CTLCsubcontract_s) =Y \\
    \funusers(\CTLCcontract) &= \funusers(\CTLCsubcontract_1) = ... = \funusers(\CTLCsubcontract_s) \\
                            &= \{\funsender(\CTLCcontract), \funreceiver(\CTLCcontract)\} \\
    \funusers(\advCTLCcontract) &= \funusers(\CTLCcontract) \\
    \funusers(\batch) &= \bigcup_{\advCTLCcontract \in \batch} \funusers(\advCTLCcontract) \\
    \funfund(\CTLCcontract) &= \funfund(\CTLCsubcontract_1) = ... = \funfund(\CTLCsubcontract_s) = f^{\zeta}_X \\
    \timeoutsc(\CTLCsubcontract) &= t_0 + \lambda_{i} \Delta
\end{align*}
\begin{align*}
    \funsecret_1(\CTLCsubcontract_i) &= \{s_{\walk_{1}}^{id, 1},...,s_{\walk_{n_i}}^{id, 1}\} \\
    \funsecret_2(\CTLCsubcontract_i) &= \{s_{\walk_{1}}^{id, 2},...,s_{\walk_{n_i}}^{id, 2}\} \\
    ... & \\
    \funsecret_m(\CTLCsubcontract_i) &= \{s_{\walk_{1}}^{id, m},...,s_{\walk_{n_i}}^{id, m}\} \\
    \funsecret(\CTLCsubcontract_i) &= \{ \funsecret_1(\CTLCsubcontract_i) \} \cup \{ \funsecret_2(\CTLCsubcontract_i) \} \cup ...\\
\end{align*}

\begin{definition} \label{def:initialconfig}
    $\environmentvec := \fullenvironmentvec$ is called \textit{initial} if
    \begin{align*}
        \forall \environmentch \in \environmentvec : \environmentchCreatedSecrets &= \environmentchRevealedSecrets = \environmentchBatches = \environmentchCTLCAdvertised \\
        &= \environmentchCTLCEnabled = \environmentchCTLCAuthorized = \environmentchCTLCDecided \\
        &= \environmentchReservedFunds = \emptyset . 
    \end{align*}
    For all $\environmentch \in \environmentvec$ all elements/sets in $\environmentch$ are empty besides $\environmentchAvailableFunds$. In $\environmentchAvailableFunds$ there can be funds $f^{\zeta}_X$ for an $X \in \confuser{\environmentch}$. Each fund $f^{\zeta}_X \in \environmentvecAvailableFunds$ is unique with an unique identifier ${\zeta}$. For any contract used in the following semantics, a fund needs to be predefined in this initial environment because funds cannot be created afterward, and every contract needs a unique fund.
    The following will start with an $\textit{initial environment}$.
\end{definition}

To shorten the following rule, we define the well-formedness of a batch as a standalone function: \newcommand{\funwellformedbatch}{well-formed}
\begin{align}
    &\funwellformedbatch(\batch) :\Leftrightarrow \batch = \{ \advCTLCcontractid{x}_1, \advCTLCcontractid{x'}_2, ... \} \nonumber \\
    &\land \funhonest \neq \emptyset \nonumber\\ 
    &\land \forall \advCTLCcontract \in \batch : s, \lambda_1, ... , \lambda_s \in \N \land \advCTLCcontract := \{ \CTLCsubcontract_1, ..., \CTLCsubcontract_s \} \text{ with } \nonumber\\
    &\forall 1\leq i \leq s : \CTLCsubcontract_i := \bigl[X,Y,f^{\zeta}_X,\funsecret(\CTLCsubcontract_i), t_0 + \lambda_{i} \Delta \bigr], \nonumber\\
    &\text{and if } i<s : \lambda_i \leq \lambda_{i+1} . \label{ctlc:well-formedness}
\end{align}

\def\inferencescaling{0.9}
\begin{equation}\label{ctlc:advBatch} \small
\scalebox{\inferencescaling}{$
\inference[\text{[\advBatch]}]{ \batch \text{ with } \funwellformedbatch(\batch) \land \forall \advCTLCcontract \in \batch :
\\ \exists \channel : \funfund (\advCTLCcontract) \in \environmentchAvailableFunds 
\\ \land \forall A \in \funusers(\advCTLCcontract) : A \in \confuser{\environmentch} 
\\ \land \forall \CTLCsubcontract \in \advCTLCcontract : \funsecret(\CTLCsubcontract) \cap (\environmentCreatedSecrets \cup \environmentRevealedSecrets) = \emptyset ,
\\ \environmentvec' := [\environment_1', ..., \environment_{\vert \environmentvec \vert}'], \forall 1 \leq \channel \leq \vert \environmentvec \vert \text{ set}
\\ \fullenvironmentch ,
\\ \batchset' := \environmentchBatches \cup \batch
\\ \environmentch' := [ S_{\textit{com}}, S_{\textit{rev}}, \batchset' , C_{\textit{adv}}, C_{\textit{aut}}, C_{\textit{en}}, C_{\textit{cla}}, F_{\textit{av}}, F_{\textit{res}}, t ]}
{\environmentvec \overset{ \advBatch \, \batch }{\longrightarrow} \environmentvec' }
$}
\end{equation}
The $\advBatch$ rule checks that the proposed batch is well formed and that for every included contract the fund is available in a channel, which includes the users included in this contract. Additionally none of the secrets used in this batch should have been used before. 
Batches are global objects, and so they get copied to all environments. For talking about their secrets we define for $A \in \funusers(\batch)$
\begin{align*}
    S_A(\batch) &:= \{ s_{\walk}^{id} \mid \exists \CTLCsubcontract \in \advCTLCcontract \in \batch : ~ s_{\walk}^{id} \in \funsecret(\CTLCsubcontract) \}, \\
    S(\batch) &:= \bigcup_{A \in \funusers(\batch)} S_A(\batch).
\end{align*}

\begin{equation}\label{ctlc:commitBatch} \small
\scalebox{\inferencescaling}{$
\inference[\text{[\commitBatch]}]{ \batch \in \environmentchBatches, A \in \funusers(\batch) 
\\ S_A(\batch) \cap ( \environmentchCreatedSecrets \cup \environmentchRevealedSecrets) = \emptyset ,
\\ S_{com}' := S_{com} \cup S_A(\batch),
\\ \environmentvec' := [\environment_1', ..., \environment_{\vert \environmentvec \vert}'], \forall 1 \leq \channel \leq \vert \environmentvec \vert \text{ set}
\\ \fullenvironmentch ,
\\ \environmentch' := [ S_{\textit{com}}', S_{\textit{rev}}, \batchset, C_{\textit{adv}}, C_{\textit{aut}}, C_{\textit{en}}, C_{\textit{cla}}, F_{\textit{av}}, F_{\textit{res}}, t ]}
{\environmentvec \overset{A: \, \commitBatch \, \batch }{\longrightarrow} \environmentvec'}
$}
\end{equation}
With the $\commitBatch$ action a user commits to all secrets appearing in a given batch. It is a global action.

\begin{equation}\label{ctlc:advCTLC} \small
\scalebox{\inferencescaling}{$
\inference[\text{[\advCTLC]}]{
 \environmentch \in \environmentvec, \advCTLCcontract \in \batch \in \environmentchBatches, \advCTLCcontract \notin \environmentchCTLCAdvertised ,
\\ \forall \CTLCsubcontract \in \advCTLCcontract : \funsecret(\CTLCsubcontract) \subseteq \environmentchCreatedSecrets,
\\ \funfund(\advCTLCcontract) \in \environmentchAvailableFunds ,
\\ \funusers(\advCTLCcontract) \cap \funhonest \neq \emptyset , 
\\ \funusers(\advCTLCcontract) \subseteq \confuser{\environmentch},
\\ C_{adv}' := \environmentchCTLCAdvertised \cup \{ \advCTLCcontract \},
\\ \fullenvironmentch, 
\\ \environmentch' := [ S_{\textit{com}}, S_{\textit{rev}}, \batchset, C_{\textit{adv}}', C_{\textit{aut}}, C_{\textit{en}}, C_{\textit{cla}}, F_{\textit{av}}, F_{\textit{res}}, t ], }
{\environmentvec \overset{ \advCTLCch \, \advCTLCcontract}{\longrightarrow} \Update{\environmentvec}{\environmentch}{\environmentch'}}
$}
\end{equation}
\CTLC{}s are local objects, so advertising one is a local operation. 

\begin{equation}\label{ctlc:authCTLC} \small
\scalebox{\inferencescaling}{$
\inference[\text{[\authCTLC]}]{ \environmentch \in \environmentvec, \advCTLCcontract \in \environmentchCTLCAdvertised, (A,\advCTLCcontract) \notin \environmentchCTLCAuthorized , 
\\ \bigl( A = \funreceiver(\advCTLCcontract) 
\\ \lor (A = \funsender(\advCTLCcontract) \land (\funreceiver(\advCTLCcontract), \advCTLCcontract) \in \environmentchCTLCAuthorized) \bigr) , 
\\ \funfund(\advCTLCcontract) \in \environmentchAvailableFunds , 
\\ C'_{aut} = \environmentchCTLCAuthorized \cup \{ (A,\advCTLCcontract) \}, 
\\ \fullenvironmentch,
\\ \environmentch' := [ S_{\textit{com}}, S_{\textit{rev}}, \batchset, C_{\textit{adv}}, C'_{\textit{aut}}, C_{\textit{en}}, C_{\textit{cla}}, F_{\textit{av}}, F_{\textit{res}}, t ]}
{\environmentvec \overset{A: \, \authCTLC \, \advCTLCcontract }{\longrightarrow} \Update{\environmentvec}{\environmentch}{\environmentch'} }
$}
\end{equation}
Note that $\funsender(\advCTLCcontract)= \funowner(\funfund(\advCTLCcontract))$. For authorizing a \CTLC{}, all its secrets and the fund need to be available. To make sure that no contract can get executed without the consent of all included participants, both sender and receiver need to authorize them before it can proceed with enabling them in the next step. 

\begin{equation}\label{ctlc:enableCTLC} \small
\scalebox{\inferencescaling}{$
\inference[\text{[\enableCTLC]}]{ \environmentch \in \environmentvec, \advCTLCcontract \in \environmentchCTLCAdvertised, \CTLCcontract \notin \environmentchCTLCEnabled ,  \\
\CTLCcontract := \{ \CTLCsubcontract \}, \text{ for } \CTLCsubcontract \in \advCTLCcontract , s = \vert \advCTLCcontract \vert ,\\
\funposition(\CTLCsubcontract) = s, \\
Aut := \{ (\funsender(\advCTLCcontract), \advCTLCcontract),(\funreceiver(\advCTLCcontract), \advCTLCcontract) \},\\
Aut \subseteq \environmentchCTLCAuthorized , \, \funfund(\advCTLCcontract) \in \environmentchAvailableFunds , \\ 
A = \funsender(\CTLCcontract), C_{en}' := \environmentchCTLCEnabled \cup \{ \CTLCcontract \} ,\\
C_{aut}' := \environmentchCTLCAuthorized \backslash Aut, \, 
F_{av}' := \environmentchAvailableFunds \backslash \{ \funfund(\advCTLCcontract) \}, \\
F_{res}' := \environmentchReservedFunds \cup \{ \funfund(\advCTLCcontract) \} , \\
\fullenvironmentch \\ 
\environmentch' := [ S_{\textit{com}}, S_{\textit{rev}}, \batchset, C_{\textit{adv}}, C_{\textit{aut}}', C_{\textit{en}}', C_{\textit{cla}}, F_{\textit{av}}', F_{\textit{res}}', t ]}
{\environmentvec \overset{\enableCTLCch \, \CTLCcontract}{\longrightarrow} \Update{\environmentvec}{\environmentch}{\environmentch'}}
$}
\end{equation}
In case authorizations have been given and the fund for a \CTLC{} is available it can be enabled, which is a local action happening in one channel. With this the last of its subcontracts is made available in $\environmentchCTLCEnabled$. When a \CTLC{} has been enabled the remaining subcontracts can be enabled one by one by the sender.

\begin{equation}\label{ctlc:enableSubC} \small
\scalebox{\inferencescaling}{$
\inference[\text{[\enableSubC]}]{ \environmentch \in \environmentvec, \advCTLCcontract \in \environmentchCTLCAdvertised, \CTLCcontract \in \environmentchCTLCEnabled,
\\ \CTLCsubcontract \in \advCTLCcontract \backslash \CTLCcontract, A = \funsender(\CTLCcontract) ,
\\ C_{en}' := \environmentchCTLCEnabled \backslash \{ \CTLCcontract \} \cup \{ \CTLCcontract \cup \{ \CTLCsubcontract \} \} ,
\\ \fullenvironmentch, 
\\ \environmentch' := [ S_{\textit{com}}, S_{\textit{rev}}, \batchset, C_{\textit{adv}}, C_{\textit{aut}}, C_{\textit{en}}', C_{\textit{cla}}, F_{\textit{av}}, F_{\textit{res}}, t ]}
{\environmentvec \overset{A: \, \enableSubC \, \CTLCsubcontract}{\longrightarrow} \Update{\environmentvec}{\environmentch}{\environmentch'} }
$}
\end{equation}
Since \CTLC{}s are local objects, subcontracts are also local. Now that contracts and their subcontracts can be enabled the next step towards executing them is revealing their secrets. 

\begin{equation}\label{ctlc:revealSecret} \small
\scalebox{\inferencescaling}{$
\inference[\text{[\revealSecret]}]{ \environmentch \in \environmentvec, s^{id}_{\walk} \in \environmentchCreatedSecrets \backslash \environmentchRevealedSecrets,
\\   A = \funowner(s^{id}_{\walk}) \in \confuser{\environmentch},
\\  S_{com}' := \environmentchCreatedSecrets \backslash \{ s^{id}_{\walk} \},
\\  S_{rev}' := \environmentchRevealedSecrets \cup \{ s^{id}_{\walk} \},
\\ \fullenvironmentch, 
\\ \environmentch' := [ S_{\textit{com}}', S_{\textit{rev}}', \batchset, C_{\textit{adv}}, C_{\textit{aut}}, C_{\textit{en}}, C_{\textit{cla}}, F_{\textit{av}}, F_{\textit{res}}, t ]}
{\environmentvec \overset{A: \, \revealSecretch \, s^{id}_{\walk}}{\longrightarrow} \Update{\environmentvec}{\environmentch}{\environmentch'} }
$}
\end{equation}
Where committing to secrets was a global operation, revealing them is local.

\begin{equation}\label{ctlc:shareSecret} \small
\scalebox{\inferencescaling}{$
\inference[\text{[\shareSecret]}]{ \environmentch, \environment_{ch'} \in \environmentvec, s^{id}_{\walk} \in  \environment_{ch'}.S_{rev} \backslash \environmentchRevealedSecrets, 
\\  A \in \confuser{\environmentch} \cap \confuser{\environment_{ch'}},
\\  S_{rev}' := \environmentchRevealedSecrets \cup \{ s^{id}_{\walk} \},
\\ \fullenvironmentch, 
\\ \environmentch' := [ S_{\textit{com}}, S_{\textit{rev}}', \batchset, C_{\textit{adv}}, C_{\textit{aut}}, C_{\textit{en}}, C_{\textit{cla}}, F_{\textit{av}}, F_{\textit{res}}, t ]}
{\environmentvec \overset{A: \, \shareSecretch s^{id}_{\walk}}{\longrightarrow} \Update{\environmentvec}{\environmentch}{\environmentch'} }
$}
\end{equation}
With 'shareSecret', a secret already known in $\environment_{ch'}$ gets shared with an environment $\environmentch$. As participants in $\environment_{ch'}$ do not unlearn a secret it also stays in $\environment_{ch'}.S_{rev}$. Any user can execute this operation if they are part of $\environmentch$ and $\environment_{ch'}$. 

\begin{equation}\label{ctlc:timeout} \small
\scalebox{\inferencescaling}{$
\inference[\text{[\timeoutsc]}]{ \advCTLCcontract \in \environmentchCTLCAdvertised, \CTLCcontract \in \environmentchCTLCEnabled ,
\\ \vert \advCTLCcontract \vert > 1,  \CTLCsubcontract \in \advCTLCcontract , 
\\ \nexists \, \dotCTLCsubcontract \in \advCTLCcontract : \funposition(\dotCTLCsubcontract) < \funposition(\CTLCsubcontract),
\\ \funtimeout(\CTLCsubcontract) \leq \environmentchtime,
\\ C_{en}' := \environmentchCTLCEnabled \backslash \{ \CTLCcontract \} \cup \{ \CTLCcontract \backslash \{ \CTLCsubcontract \} \},
\\ C_{adv}' := \environmentchCTLCAdvertised \backslash \{ \advCTLCcontract \} \cup \{ \advCTLCcontract \backslash \{ \CTLCsubcontract \} \},
\\ \fullenvironmentch,
\\ \environmentch' := [ S_{\textit{com}}, S_{\textit{rev}}, B, C_{\textit{adv}}', C_{\textit{aut}}, C_{\textit{en}}', C_{\textit{cla}}, F_{\textit{av}}, F_{\textit{res}}, t ] }
{\environmentvec \overset{\timeoutsc \, (\CTLCcontract, \CTLCsubcontract)}{\longrightarrow} \Update{\environmentvec}{\environmentch}{\environmentch'} }
$}
\end{equation}
With $\timeoutsc$, one of the subcontracts in a $\CTLC{}$ can get disabled after its timelock has run out. If only one subcontract is left, the whole contract can be refunded instead. These actions remove them from the advertised and enabled contracts and as $\advCTLC$, $\enableCTLC$, and $\enableSubC$ where local actions this also happens locally in one $\environmentch$.  

\begin{equation}\label{ctlc:refund} \small
\scalebox{\inferencescaling}{$
\inference[\text{[\refund ]}]{ \advCTLCcontract \in \environmentchCTLCAdvertised, \CTLCcontract \in \environmentchCTLCEnabled ,
 \vert \advCTLCcontract \vert = 1,  \CTLCsubcontract \in \advCTLCcontract ,
\\ \environmentchtime \geq \funtimeout(\CTLCsubcontract),
\\ C_{en}' := \environmentchCTLCEnabled \backslash \{ \CTLCcontract \},
 C_{adv}' := \environmentchCTLCAdvertised \backslash \{ \CTLCcontract \},
\\ F_{av}' := \environmentchAvailableFunds \cup \{\funfund(\CTLCcontract)\}, 
 F_{res}' := \environmentchReservedFunds \backslash \{\funfund(\CTLCcontract)\},
\\ \fullenvironmentch, 
\\ \environmentch' := [ S_{\textit{com}}, S_{\textit{rev}}, B, C_{\textit{adv}}', C_{\textit{aut}}, C_{\textit{en}}', C_{\textit{cla}}, F_{\textit{av}}', F_{\textit{res}}', t ] }
{\environmentvec \overset{ \refund \, \CTLCcontract}{\longrightarrow} \Update{\environmentvec}{\environmentch}{\environmentch'} }
$}
\end{equation}
If the timelock for a subcontract has not been reached, its fund is reserved and all its secrets are known locally it can get claimed. 

\begin{equation} \label{ctlc:DecideCo} \small
\scalebox{\inferencescaling}{$
\inference[\text{[\DecideCo ]}]{ \advCTLCcontract \in \environmentchCTLCAdvertised, \CTLCsubcontract \in \CTLCcontract \in \environmentchCTLCEnabled, 
\\ \exists i : \, \funsecret_i(\CTLCsubcontract_{\iota}) \subseteq \environmentchRevealedSecrets, \funfund(\CTLCcontract) \in \environmentchReservedFunds, 
\\ \nexists \, \dotCTLCsubcontract \in \advCTLCcontract : \funposition(\dotCTLCsubcontract) < \funposition(\CTLCsubcontract),
\\ C_{adv}' := \environmentchCTLCAdvertised \backslash \{ \CTLCcontract \},
\\ C_{en}' := \environmentchCTLCEnabled \backslash \{ \CTLCcontract \},
 C_{cla}' := \environmentchCTLCDecided \cup \{ \CTLCcontract \cap \{ \CTLCsubcontract \} \},
\\ \fullenvironmentch, 
\\ \environmentch' := [ S_{\textit{com}}, S_{\textit{rev}}, B, C_{\textit{adv}}', C_{\textit{aut}}, C_{\textit{en}}', C_{\textit{cla}}', F_{\textit{av}}, F_{\textit{res}}, t ] }
{\environmentvec \overset{ \DecideCo (\CTLCcontract ,\CTLCsubcontract,  \funsecret_i(\CTLCsubcontract))}{\longrightarrow} \Update{\environmentvec}{\environmentch}{\environmentch'}} 
$}
\end{equation}

Note that the $\DecideCo$ action needs all secrets of one of the secret sets $\funsecret_i(\CTLCsubcontract_{\iota})$ to be revealed in one $\environmentch$, predefined by the location of $\funfund(\CTLCcontract)$. This means that it is possible to $\DecideCo$ $\CTLCsubcontract_{\iota}$ based on the secrets from $\funsecret_1(\CTLCsubcontract_{\iota})$ \textbf{or} $\funsecret_2(\CTLCsubcontract_{\iota})$ \textbf{or} ... .
Also, note that $\CTLCsubcontract_{\iota}$ can only be claimed when being the \emph{top-level contract}, meaning that the contract below (and hence all contracts below) has been timed out before.
One subtlety here is that also non-enabled subcontracts must have timed out (as long as the main CTLC has been enabled). 
This is achieved by checking for the appearance in the set $\advCTLCcontract \in \environmentchCTLCAdvertised$ of advertised subcontracts (belonging to the enabled contract $\CTLCsubcontract_{\iota}$). 
The reason for this modeling is the static nature of CTLCs: Subcontracts are pre-defined spending options that can be dynamically enabled by the contract senders but whose execution follows a strict hierarchical order. 
To ensure that low-hierarchy spending options are made available as expected, it needs to be ensured that high-priority spending options cannot be enabled at a later point in time, messing with the execution order. 
For this reason, the spending whole spending option (independent of whether being enabled or not) can time out. 

Finally, a claimed contract can be withdrawn. Here, in the second line of the rule, the owner of the fund belonging to $\CTLCcontract$ is set to the receiver of this contract. Before this action, the owner of this fund was the sender of $\CTLCcontract$. 

\begin{equation}\label{ctlc:CoEx} \small
\scalebox{\inferencescaling}{$
\inference[\text{[\CoEx]}]{ \advCTLCcontract \in \environmentchCTLCAdvertised, \CTLCsubcontract \in \CTLCcontract \in \environmentchCTLCDecided, 
\\ \funowner(\funfund(\CTLCcontract)) := \funreceiver(\CTLCcontract), \funfund(\CTLCcontract) \in \environmentchReservedFunds ,
\\ C_{cla}' := \environmentchCTLCDecided \backslash \{ \CTLCcontract \},
\\ F_{av}' := \environmentchAvailableFunds \cup \{ \funfund(\CTLCcontract) \},
\\ F_{res}' := \environmentchReservedFunds \backslash \funfund(\CTLCcontract),
\\ \fullenvironmentch, 
\\ \environmentch' := [ S_{\textit{com}}, S_{\textit{rev}}, B, C_{\textit{adv}}, C_{\textit{aut}}, C_{\textit{en}}, C_{\textit{cla}}', F_{\textit{av}}', F_{\textit{res}}', t ]
}
{\environmentch \overset{ \CoEx (\CTLCcontract, \CTLCsubcontract ) }{\longrightarrow} \Update{\environmentvec}{\environmentch}{\environmentch'}} 
$}
\end{equation}

\begin{equation}\label{ctlc:elapsetime} \small
\scalebox{\inferencescaling}{$
    \inference[\text{[elapse $\delta$]}]{
    \environmentvec' := [\environment_1', ..., \environment_{\vert \environmentvec \vert}'], \forall 1 \leq \channel \leq \vert \environmentvec \vert \text{ set}
\\  t' = t + \delta ,
\\  \fullenvironmentch,
\\  \environmentch' := [ S_{\textit{com}}, S_{\textit{rev}}, \batchset, C_{\textit{adv}}, C_{\textit{aut}}, C_{\textit{en}}, C_{\textit{cla}}, F_{\textit{av}}, F_{\textit{res}}, t' ] 
    }
    {\environmentvec \overset{ elapse \, \delta }{\longrightarrow} \environmentvec'}
    $}
\end{equation}
Time can only be changed in all $\environmentch \in \environmentvec$ at once an thus we write $\environmenttime := \environment_1.t = ... = \environment_{\vert \environmentvec \vert}.t$. 

\section{Honest User Strategy}
\label{sec:honestStrategy}

To model execution on \tams, we will adopt a symbolic execution model similar to the one presented in~\cite{BitML}. 
We will call sequences $\environmentvec_0 \overset{ \action_0}{\longrightarrow} \environmentvec_1 \overset{ \action_1}{\longrightarrow} \cdots$ runs where $\action_i$ are transition labels and $\environmentvec_0$ is an initial environment. 
We will refer to these transition labels as \emph{moves} or \emph{actions} in the following. Given a run
\[
    \run:= \environmentvec_0 \overset{ \action_0}{\longrightarrow} \cdots \overset{ \action_{n-1}}{\longrightarrow} \environmentvec_n 
\]
of length $n \in \N$ we set $\environmentvec_n =: \funlastEnv(\run)$ as the last environment of $\run$. 

A participant strategy is a function $\userstrategy{\honestuser}$ taking as input a run $\run$ and outputting a set of transition labels, indicating the actions that the user wants to schedule.
 \begin{definition}[Participant strategies]
 \label{def:participantstrategy}
 
The strategy of an (honest) user $\honestuser$ is a function $\userstrategy{\honestuser}$, taking as input a run $\run$.
The output is a set of $\action$-moves such that the following conditions hold: 
\begin{enumerate}
    \item Participant strategies can only output actions that are valid with respect to the semantics: 
    $$\forall \action \in \userstrategy{\honestuser}(\run): \run \overset{ \action}{\longrightarrow} ;$$
    \item Users can only schedule restricted actions if they are the ones to whom the action is restricted: 
    $$\forall \action \in \userstrategy{\honestuser}(\run): \action = U: \action' \Rightarrow \honestuser = U ;$$
    \item Participant strategies must be persistent, meaning that a participant strategy needs to keep scheduling an action as long as it is valid:
    $$\forall \action \in \userstrategy{\honestuser}(\run): \run \overset{ \action'}{\longrightarrow} \run'  \overset{ \action}{\longrightarrow} ~\Rightarrow \action \in \userstrategy{\honestuser}(\run \overset{ \action'}{\longrightarrow} \run') .$$
\end{enumerate}
\end{definition}

To model the power of miners in the execution of honest user actions in a blockchain ecosystem, we define the adversary strategy to be a function that, given a run and the outputs of the honest user strategies, can produce the next action to extend the run. 
This models both, the attacker's capability to order honest user actions arbitrarily and the adversary's power to include own transactions based on the knowledge gathered from the scheduled actions of honest users. 
To give basic guarantees to the honest user, the attacker strategy is restricted to only be able to make time pass once all honest users either agree to do so or have no more actions scheduled.
This ensures that honest users can always meet deadlines and that the protocol execution cannot advance (in time) without their scheduled actions being taken into account.

\begin{definition}[Adversary strategy] \label{def:adversarystrategy}
Let $\honestusers =  \{ \honestuser_i, \dots, \honestuser_k\}$ be a set of honest users.
An adversary strategy $\attackerstrategy$ (for $\honestusers$) is a function taking as input a run $\run$ and a list $\vec{\useractions} = [\useractions_1, \dots \useractions_k]$ of sets of moves for each $\honestuser_i \in \honestusers$. 
The output is a single adversary action $\action$ such that the following conditions hold:
\begin{enumerate}
    \item The adversary strategy can only output actions that are valid with respect to the semantics: $$\forall \action \in \attackerstrategy(\run, \vec{\useractions}): \run \overset{ \action}{\longrightarrow} ;$$
    \item Restricted actions of honest users can only be chosen by the adversary strategy if scheduled by the corresponding honest user strategy: 
    \begin{align*}
    &\attackerstrategy(\run, \vec{\useractions}) = U :\action' ~\land~ U \in \honestusers \\
    &\Rightarrow \exists i:  \honestuser_i = U ~\land~ U :\action' \in \useractions_i ;
    \end{align*}
    \item The adversary can only output a time elapse action if all users agree to do so:
    \begin{align*}
    &\attackerstrategy(\run, \vec{\useractions}) = \actelapse{\delta} \\
    &\Rightarrow \forall \honestuser_i \in \honestusers: \useractions_i = \emptyset \\
    & \hspace{10pt} ~\lor~ \exists \delta_i: \actelapse{\delta_i} \in \useractions_{i} ~\land~ \delta_i \geq \delta \geq \epsilon > 0
    \end{align*}
\end{enumerate}
We assume an $\epsilon > 0$ as a constant minimum size for all $\actelapse \delta$ actions throughout the run. 
\end{definition}

Based on these notions, we can define when a run $\run$ is conformant with a given set of strategies:

\begin{definition}[Conformant runs]
    Let $\honestusers =  \{ \honestuser_1, \dots, \honestuser_k\}$ be a set of honest users
and $\honestuserstrategies = \{ \userstrategy{\honestuser_1}, \dots, \userstrategy{\honestuser{_k}} \}$ a corresponding set of user strategies.
Let $\attackerstrategy$ be an adversary strategy (for $\honestusers$).  
We say that a run $\run$ \emph{conforms to $(\honestuserstrategies, \attackerstrategy)$} (written $(\honestuserstrategies, \attackerstrategy) \conforms \run$) if one of the following holds
\begin{enumerate}
    \item $\run = \environmentvec_0$ is an initial environment 
    \item $\run = \run'  \overset{ \action}{\longrightarrow} \environmentvec$ where $(\honestuserstrategies, \attackerstrategy) \conforms \run'$ and \\
    $\action = \attackerstrategy(\run', [ \userstrategy{\honestuser_1}(\run'), \dots, \userstrategy{\honestuser{_k}}(\run') ])$
\end{enumerate}
\end{definition}

We will also write $\honestuserstrategies \conforms \run$ as shorthand for 
\[
    \exists \attackerstrategy:(\honestuserstrategies, \attackerstrategy) \conforms \run .
\]
Note that for the sake of simplicity, as opposed to~\cite{BitML}, we assume here strategies to be functions (instead of PPTIME algorithms). 
To achieve computational soundness results, one could require these functions to be PPTIME algorithms and give them access to a (user-specific) source of randomness.
It will be evident to see that the honest user strategies presented in this work run in PPTIME. 

Another major adaption with respect to the model from~\cite{BitML} is that we do not explicitly model the values of secrets. 
This is because our semantics does not model any computations on secret values which allows us to only refer to characterize the function of these secrets purely in terms of user access to these secrets. 
This simplification substantially simplifies the theoretical model because we are not required to explicitly strip secret values from runs and user actions in order to model that those are not accessible to users and the attacker.

\paragraph{Preparation Phase}
This section establishes the formal connection between trees and \CTLC{}s. We will denote $\environmentvec = \funlastEnv(\run)$ in the following. We define $\treeobj$ as a set of \emph{tree specifications} consisting of elements $\fulltreeobj$ where $id$ is a unique identifier for this tree, which will be used to establish the connection to a batch. 
In the following, we assume a game tree $\untree$ (Def. \ref{def:gameTree}) to be given.
As defined in $\eqref{eq:unfold}$, it can result from unfolding a digraph $\graphsymbol$ with $\untree:= \fununfold(\graphsymbol, A)$. With $t_0 \in \R^+$, the beginning of the Execution Phase is determined. All $\funtimeout$s will be set after $t_0$, and for the honest user, the setup of contracts will be done before $t_0$, and executions only happen after $t_0$. The last element, $\specalone$ is defined as a function on all $\e \intree \untree$:
\[
    \spec{\e}:= (\edgefund{\e}, \channeledge{\e}) = (f^{\zeta}_X, \channel) 
\]
Here, $\edgefund{\e}$ defines the fund as a token in $\environment_{\channeledge{\e}}$ for the subcontract that will resemble this edge. $X$ is the sender of $\e$ and hence owner of $f^{\zeta}_X$. The \textit{specification} $\spec{\e}$ defines the preconditions for the initial environment for it to be able to resemble this edge. 
\begin{definition} \label{def:specification}
    A specification $\specalone$ on $\untree$ is \textit{valid} if for all
    $\e, \e' \intree \untree$ it holds 
    \begin{align*}
        &\spec{\e} = \spec{\e'} : \Leftrightarrow \\
        &\funsender(\e) = \funsender(\e') \land \funreceiver(\e) = \funreceiver(\e')
    \end{align*}
\end{definition}
This condition ensures that edges in different locations in the tree resulting from one arc in $\graphsymbol$ through the unfolding process feature the same output of $\specalone$. We call these \textit{duplicated edges}, and in the following, the equality of their specifications will be used to identify them. In the following, $\specalone$ will be a valid specification. 
 
We assign a unique secret 
\begin{equation} 
    \funsecret((A,B)_{\walk}, id) =: s_{\walk}^{id} \label{def:edgesecret}
\end{equation}
with $\funowner(s_{\walk}^{id}) = B$ to any edge $(A,B)_{\walk} \in \untree$. 
For defining the \CTLC{} contracts resulting from a given $\treeobj$, we first look at a single \linebreak $\fulltreeobj \in \treeobj$ and group together duplicated edges on the same level:
\begin{align} \label{def:funedgegroup}
    \funedgegroup := \bigcup_{\e \intree \untree} \{ \e' \intree \untree & \mid \fundepth(\e) = \fundepth(\e') \\
    & \hspace{-10pt} \land \funsender(\e) = \funsender(\e') \nonumber\\
    & \hspace{-10pt} \land \funreceiver(\e) = \funreceiver(\e') \} \nonumber
\end{align}
For a fixed $\Omega \in \funedgegroup$ we look at edges $\e \in \Omega$, and for these, we set
\begin{align} \label{def:hsecMapping}
    h_{sec}(\e, id) := \bigcup_{\e' \in \funonPathtoRoot(\untree, \e)} \{ \funsecret(\e', id) \}. 
\end{align}
For all $\Omega \in \funedgegroup$, we define secret sets:
\begin{align*}
    \funsecretset(\Omega, id, \untree) &:= \bigcup_{\e \in \Omega} \bigl\{ h_{sec}(\e, id) \bigr\}
\end{align*}
Note that if there is no additional edge on the same tree level with the same sender and receiver and disjoint path to the root, only the secret set of $\e$ itself is included.
By construction, for any fixed $\Omega$ there exists $X,Y \in \nodesymbol$, $j \in \N$ and $f^{\zeta}_X$ with:
\begin{align*}
    \funsender(\e) &= X &, \forall \e \in \Omega \\
    \funreceiver(\e) &= Y &, \forall \e \in \Omega \\
    \fundepth(\e) &= j &, \forall \e \in \Omega \\
    \funfund(\e) &= f^{\zeta}_X &, \forall \e \in \Omega 
\end{align*}
This means that all edges from $\Omega$ are between the same parties, their $\specalone$ features the same fund, and they are located on the same tree level. Again, let $\Delta$ be a sufficient amount of time to execute an action on the specified $\channel$.  For every $\Omega$ we demand $f^{\zeta}_X$ to be a fresh, unused fund with a unique identifier $\zeta$.
Based on the $X,Y,j$ and $f^{\zeta}_X$ from above we map a set $\Omega$ to a subcontract $\CTLCsubcontract$ with $\ctlcID = \ctlcIDtree{id}{X}{Y}$ and 
\begin{align*}
    & H(\Omega, \fulltreeobj) := \CTLCsubcontract \\
    &= [X,Y,f^{\zeta}_X, t_0 + j \Delta, secret-set(\Omega, id, \untree)] . \nonumber
\end{align*} 
For all $\e \in \Omega$ we define
\begin{align} \label{def:hMapping}
    h(\e, id) := \CTLCsubcontract.
\end{align}
Since $\funedgegroup$ is a partition of $\untree$, this definition and also the one of $h_{sec}$ can be extended to all $\e \intree \untree$. Later in this section (see (\ref{eq:honprep})), we will see that an honest user only commits to batches that are built according to this construction. Hence, $h$ gives us a surjective mapping from all edges in $\untree$ to the sub-contracts in a batch. 
We will now define this batch: 
\begin{align}
    &\treetoCTLCBadge(\fulltreeobj, \mathcal{S}):= \batch = \label{eq:treetoCTLCBadge} \\
    & \bigcup_{\e \intree \untree} \{ h(\e', id) \mid \e' \intree \untree \, \land \, \spec{\e'} = \spec{\e} \} \nonumber
\end{align}
By definition, this groups sub-contracts of duplicated edges together and gives a batch of \CTLC{}s that represents \mbox{$\fulltreeobj \in \treeobj$}, hence the same $id$ is used for $\batch$. Further, we define the mapping
\begin{align}
    &\treesettoCTLCBadges (\treeobj, \mathcal{S}) := \label{eq:treesettoCTLCBadges}\\
    &\bigcup_{\fulltreeobj \in \treeobj} \treetoCTLCBadge(\fulltreeobj, \mathcal{S}), \nonumber
\end{align}
which outputs a set of batches, one for each tree. 

For later use, we also define the following objects. 
Given a $\fulltreeobj \in \treeobj$ and let $A, B \in \nodesymbol$ be fixed. Then the number of edges $(A,B)$ in $\untree$ is given as
\begin{align}
    E_{(A,B)}^{\untree} &:= \{ \edge \intree \untree \mid X=A \land Y=B \}, \label{def:edgeswithB} \\
    s_{(A,B)}^{\untree} &:= \vert E_{(A,B)}^{\untree} \vert . \nonumber
\end{align}

For advertising and committing to batches we require the initial environment to be liquid with respect to a set of tree specifications $\fulltreeobj$, defined as follows.

\begin{definition}[Liquid Environment] \label{def:liquidenv}
    Let $\treeobj$ be a set of tuples of the form $\fulltreeobj$.
We say that a environment $\environmentvec$ is \textit{liquid w.r.t.} $\treeobj$ if the following condition holds:
\begin{align*}
    \forall& \fulltreeobj \in \treeobj: 
    \forall \e \intree \untree: \forall (f^{\zeta}_X, \channel):  \\
    & \specalone(\e) = (f^{\zeta}_X, \channel) \Rightarrow
    f^{\zeta}_X \in \environmentchAvailableFunds
\end{align*}
\end{definition}

We also define when we consider a $\treeobj$ to be well-formed. 

\begin{definition}[Well-Formed Treeobject] \label{def:wellformedtree}
Let $\treeobj$ be a set of tuples of the form $\fulltreeobj$.
We say that $\treeobj$ is \textit{well-formed} if all $\specalone$ functions are valid, all $id$s are unique, and the $\specalone$ functions do not intersect, formalized with the following condition:
\begin{equation*}
\scalebox{0.92}{$
\begin{aligned}
    &specInter(\treeobj) :\Leftrightarrow \\
    &\forall \fulltreeobj, (id', \untree', t'_0, \specalone') \in \treeobj ~\forall \e \intree \untree ~\forall \e' \intree \untree :\\
    &(\funsender(\e), \funreceiver(\e), id) \neq (\funsender(\e'), \funreceiver(\e'), id') \\
    & \Rightarrow \forall (f^{\zeta}_X, \channel), (f'^{\zeta}_{X'}, \channel'):  \\
    &\specalone(\e) = (f^{\zeta}_X, \channel) ~\land ~\specalone(\e') = (f'^{\zeta}_{X'}, \channel') \Rightarrow X \neq X'
\end{aligned}
$}
\end{equation*}
\end{definition}

Firstly, for every $\fulltreeobj \in \treeobj$, a Batch needs to be advertised. Anybody can advertise batches. Thus, it is in the honest strategy to advertise all not yet advertised batches. In the initial environment, only funds are given.
Batches are only advertised at or before $t_0 - \fundepth(\untree) \Delta$ with
\begin{equation} \label{eq:depthoftreedefinition}
    \fundepth(\untree) := max \{ \fundepth(\e) \mid \e \intree \untree \},
\end{equation}
because this ensures that the setup process can go through as desired.
The environment also needs to be liquid w.r.t. the current $\treeobj$. 
Additionally, honest users only accept well-formed $\treeobj$. This ensures that the same fund cannot be used for two different contracts. 
After $\advBatch$, $B$ commits to the batch.
We denote $\environmentvec = \funlastEnv(\run)$. 
The following sub-strategy, where $B$ is the currently operating party, formalizes this process for a given $\fulltreeobj \in \treeobj$:
{\small \begin{align*}
    &newBatch(\treeobj, \run) \\
    &:= \begin{cases}
        2
        &\text{ ,if }\forall \environmentch \in \environmentvec : \batch \notin \environmentchBatches, \\
        &\treeobj \text{ well-formed}, ~\environmentvec \text{ liquid w.r.t. } \treeobj , \\
        &\exists \fulltreeobj \in \treeobj: \\
        &\batch = \treetoCTLCBadge(\fulltreeobj, \mathcal{S}) , \\
        &\environmentvectime \, \leq \, t_0 - \fundepth(\untree) \Delta \\
        1
        &\text{ ,if } \exists \environmentch \in \environmentvec : \batch \in \environmentchBatches, \\
        &\treeobj \text{ well-formed}, ~\environmentvec \text{ liquid w.r.t. } \treeobj , \\
        &S_B(\batch) \neq \emptyset, \environmentchtime < t_0, \\
        &S_B(\batch) \nsubseteq \environmentchCreatedSecrets, \, \exists \fulltreeobj \in \treeobj: \\
        &\batch = \treetoCTLCBadge(\fulltreeobj, \mathcal{S})\\ 
        0
        &\text{ ,else.}
    \end{cases} \nonumber
\end{align*}}%
{\small \begin{align} \label{eq:honprep}
    &\widehat{\Sigma}_{B}^{\treeobj} (\run) \\
    &:= \begin{cases}
        \Bigl\{ \advBatch \, \batch \Bigr\}
        &\text{ ,if }newBatch(\treeobj, \run)=2, \\
        \Bigl\{ B:\commitBatch \, \batch \Bigr\}
        &\text{ ,if }newBatch(\treeobj, \run)=1,\\
        \emptyset &\text{ ,if }newBatch(\treeobj, \run)=0.
    \end{cases} \nonumber
\end{align}}%

In the first case, a corresponding $\batch$ for $\fulltreeobj$ gets advertised. 
If the batch is well-formed, meaning it is aligned with the said mapping, $B$ also commits to it. 
If none of these options is available, $B$ schedules nothing.

\paragraph{Enabling Phase}
Given $\fulltreeobj \in \treeobj$ and $\e \intree \untree$. We first define 2 helper functions that evaluate whether all ingoing edges have been enabled (1) and, if so, if an entirely new \CTLC{} should be advertised, an available \CTLC{} should be authorized or enabled or just an additional sub-contract (2) enabled. 
The condition (1) is checked in the function \newcommand{\ingoing}{ingoing}
{\small \begin{align} \label{eq:honingoing}
    &\ingoing(\e, \run) \\
    &:= \begin{cases} 
            1 &\text{,if } \e \intree \untree, \exists X, \walk: \, \e = (X,B)_{\walk} \textbf{ or} \\ 
            &\forall \e' \intree \untree \text{ with } \funonpath{\e}{\e'}, \fundepth(\e') = \fundepth(\e) + 1 \\
            & \exists \widetilde{\CTLCcontract} \in \environment_{\channeledge{\e'}}.C_{en}: h(\e',id) \in \widetilde{\CTLCcontract}  \\
            &\land \environmentvec.t < t_0, \\
            &\exists \batch \hspace{-2pt} := \treetoCTLCBadge(\fulltreeobj, \mathcal{S}) \hspace{-1pt} \in \hspace{-1pt} \environmentvec.\batchset \hspace{-2pt}: \\
            & h(\e,id) \in \CTLCcontract \in \batch \land h(\e',id) \in \widetilde{\CTLCcontract} \in \batch \\
            & \land \nexists \channel : h(\e,id) \in \CTLCcontract \in \environmentchCTLCEnabled \\
            0 &, \text{else.} \\
        \end{cases} \nonumber
\end{align}}%

The function $\ingoing$ checks for a given edge $\e$ and current run $R$ whether all edges below it have been enabled for this appearance of $B$ in the tree $\untree$. Additionally, it is checked that all these edges correspond to a tree.

This information gets used in the mapping $\newC$, which checks for condition (2) from above. For this we define
\[
Aut(\advCTLCcontract) := \{ (\funsender(\advCTLCcontract), \advCTLCcontract),(\funreceiver(\advCTLCcontract, \advCTLCcontract)) \},
\]
{\small \begin{align} 
        &newC(\e, R) \label{eq:CTLC-enable-check}\\
        &:=
        \begin{cases}
            4 &\text{,if } \ingoing(\e, \run) = 1 \land \exists X,\walk : \e = (B,X)_{\walk} \intree \untree\\
            &\land \exists h(\e, id) \in \advCTLCcontract \in \environmentvecCTLCAdvertised \\
            & \land \CTLCcontract \in \environmentvecCTLCEnabled \land h(\e, id) \notin \CTLCcontract,\\
            3 &\text{,if } \ingoing(\e, \run) = 1 \land \exists X,\walk : \e = (B,X)_{\walk} \intree \untree\\
            &\land \exists h(\e, id) \in \advCTLCcontract \in \environmentvecCTLCAdvertised \\
            & \land \CTLCcontract \notin \environmentvecCTLCEnabled \land Aut(\advCTLCcontract) \subseteq \environmentvecCTLCAuthorized ,\\
            2 &\text{,if } \ingoing(\e, \run) = 1 \\
            &\land \exists h(\e, id) \in \advCTLCcontract \in \environmentvecCTLCAdvertised \\
            & \land \CTLCcontract \notin \environmentvecCTLCEnabled \land \funfund(\advCTLCcontract) \in \environmentvecAvailableFunds \\
            &\land 
            if \, B = \funsender(\advCTLCcontract) : \,(\funreceiver(\advCTLCcontract), \advCTLCcontract) \in \environmentvecCTLCAuthorized\\
            1 &\text{,if } \ingoing(\e, \run) = 1 \land \exists X,\walk : \e = (B,X)_{\walk} \intree \untree\\
            &\land \nexists h(\e, id) \in \advCTLCcontract \in \environmentvecCTLCAdvertised \\
            & \land \nexists h(\e, id) \in \CTLCcontract \in \environmentvecCTLCEnabled\\
            & \land \forall \CTLCsubcontract \in \advCTLCcontract : \funsecret(\CTLCsubcontract) \subseteq \environmentvecCreatedSecrets,\\
            0 &, \ingoing(\e, \run) = 0.  \\
        \end{cases} \nonumber
\end{align}}%

In the first case the \CTLC{} of the subcontract for the given edge $\e$ has been advertised and enabled, but $h(\e, id)$ has not been enabled yet. Therefore, $B$ will enable it in the following substrategy. 
In the second case, $\CTLCcontract$ has not been enabled but advertised and authorized by the sender and receiver. Thus, $B$ will enable it.
The third case includes 
$\CTLCcontract$ not to be enabled but advertised. Additionally no other contract $\hat{\dot{\textit{c}}}^{\textit{x}}$
with the same identifier $x$ should have been authorized by $B$ before. If this true $B$ will authorize it.
In the fourth case, the \CTLC{} has neither been enabled nor advertised, so $B$ will advertise it under given well-formedness conditions. 
In the fifth and last case, $\ingoing(\e, \run) = 0$, and no action of $B$ is required. 

Then, the substrategy for these actions is defined as
\begin{equation} \label{eq:enable-check}
\scalebox{0.9}{$
{\small \begin{aligned} 
    &\overline{\Sigma}_{B}^{e} (\run) \\
    &:= \begin{cases}
        \Bigl\{ B: \enableSubC \, h(\e, id) \Bigr\} &\hspace{-8pt} \text{ ,if } newC(\e, \run) = 4, \environmentvec.t < t_0 ,\\
        \Bigl\{\enableCTLCch \, \CTLCcontract \Bigr\} &\hspace{-8pt} \text{ ,if } newC(\e, \run) = 3, h(\e, id) \in \CTLCcontract , \\
        &\hspace{-8pt} \channel = \channeledge{\e},\environmentvec.t < t_0 ,\\
        \Bigl\{ B: \authCTLC \, \advCTLCcontract \Bigr\} &\hspace{-8pt} \text{ ,if } newC(\e, \run) = 2, h(\e, id) \in \advCTLCcontract , \\
        &\hspace{-8pt} \nexists \hat{\dot{\textit{c}}}^{\textit{x}} :(B: \authCTLC \, \hat{\dot{\textit{c}}}^{\textit{x}}) \in \actions{\run},  \\
        &\hspace{-8pt} \environmentvec.t < t_0 ,\\
        \Bigl\{ \advCTLCch \, \advCTLCcontract \Bigr\} &\hspace{-8pt} \text{ ,if } newC(\e, \run) = 1, h(\e, id) \in \advCTLCcontract , \\
        &\hspace{-8pt} \channel = \channeledge{\e},\environmentvec.t < t_0 ,\\
        \emptyset &\hspace{-8pt} \text{ ,if }  newC(\e, \run) = 0. 
    \end{cases}
\end{aligned}}%
$}
\end{equation}

\paragraph{Execution Phase}
Before discussing the decision and execution of contracts, we want to ensure that contracts run into timeout or refund when possible. As any party can timeout or refund any contract as soon as their timeout has been reached, we make it part of the honest user strategy. For this, let $\e \intree \untree$ be given and we define
{\small \begin{align} \label{eq:CTLC-timeout}  
    &\widetilde{\Sigma}_{B}^{\e} (\run) \\
    &:= \begin{cases}
        \Bigl\{ \timeoutsc (\CTLCcontract, h(\e, id)) \Bigr\} &\text{ ,if } 
         \exists \CTLCcontract \in \environment_{\channeledge{\e}}.C_{en}: \\
         & \hspace{5pt} \honestuser \in \funusers(\CTLCcontract) 
         \land \vert \CTLCcontract \vert > 1 \\
        & \land \CTLCsubcontract_j = h(\e, id) \in \CTLCcontract \\
        &\land \funtimeout(h(\e, id)) \leq \environmentvec.t \\
        &\land \CTLCsubcontract_{j-1} \notin \CTLCcontract \\
        \Bigl\{ \refund \, \CTLCcontract \Bigr\} &\text{ ,if } \exists \CTLCcontract \in \environment_{\channeledge{\e}}.C_{en}: \\
        & \hspace{5pt} \honestuser \in \funusers(\CTLCcontract) \land \vert \CTLCcontract \vert = 1 \\
        & \land \exists h(\e, id) \in \CTLCcontract: \\
        & \hspace{5pt} \funtimeout(h(\e, id)) \leq \environmentvec.t \\
        &\land \refund \, (\CTLCcontract) \notin \actions{\run} \\
        \emptyset &\text{ ,else.} 
    \end{cases} \nonumber
\end{align}}%
In the first case, the subcontract corresponding to the given edge has run into timeout, but there is still another sub-contract in the \CTLC{} with a larger timeout. In the second case, the current sub-contract is the last one in the \CTLC{} so the whole \CTLC{} gets refunded.

Given an edge $\e \intree \untree$. For the corresponding subcontract to be claimable 3 things need to be given: \newcommand{\funenabled}{enabled} \newcommand{\funsecretsAv}{secretsAv} \newcommand{\funisIngoing}{isIngoing} \newcommand{\funoutgoing}{outgoing}
\begin{itemize}
    \item It needs to be enabled, meaning 
    \begin{equation*}
    \scalebox{0.95}{$    
    \begin{aligned}
        \funenabled(\e, \environment) :\Leftrightarrow &\exists \CTLCcontract \in \environment_{\channeledge{\e}}.C_{en} : \bigl( h(\e, id) \in \CTLCcontract \\
        \land &\nexists \CTLCsubcontract \hspace{-2pt} \in \advCTLCcontract \in \environment_{\channeledge{\e}}.C_{adv} :  \\
        &\funtimeout(\CTLCsubcontract) < \funtimeout(h(\e, id)) \bigr).
    \end{aligned}
    $}
    \end{equation*}
    \item The current time should be before its timeout, meaning 
    \begin{align*}
        t_0 \leq \environmentvec.t \leq \funtimeout(h(\e, id)).
    \end{align*}
    \item Secrets of other people for this contract should be revealed in the environment of the contract already s.t. only the ones of $B$ are missing. 
    We denote this condition as 
    \begin{align*}
        \funsecretsAv(\e) :\Leftrightarrow & \exists i \, \forall s^{id}_{\walk} \in \funsecret_i(h(e, id)) \\
        & \text{with }\hspace{-2pt} s^{id}_{\walk} \hspace{-1pt} \notin \hspace{-1pt} \environmentchRevealedSecrets \hspace{-2pt} : \hspace{-2pt} s^{id}_{\walk} \hspace{-2pt} = \hspace{-2pt} \funsecret(\e, id) .
    \end{align*}
\end{itemize}
Additionally, by construction, we want $B$ only to execute ingoing edges, and this only in case an outgoing edge has been executed before, which is formalized with:
\begin{equation} \label{fun:isingoing23}
\scalebox{0.9}{$
\begin{aligned}
    \funisIngoing(\ein, \run) :\Leftrightarrow & \exists Y \in \nodes, \walk \in \untree : \ein = (Y,B)_{\walk} \\
    & \land \text{ if } \fundepth(\ein) > 1 : \\
    & \forall \eout \in \walk \text{ with } \fundepth(\eout) = \fundepth(\ein) - 1 \\
    & \exists \CTLCcontract : \DecideCo(\CTLCcontract, h(\eout, id),  h_{sec}(\eout, id)) \\
    & \hspace{20pt} \in \actions{\run} \\
    & \land \text{ if } \fundepth(\ein) = 1 : \\
    &\forall \ein' \in \untree \text{ with } \fundepth(\ein') = 1 \, : \\
    &\exists \CTLCcontract, \channel : h(\ein', id) \in \CTLCcontract \in \environmentchCTLCEnabled \\
    &\hspace{20pt} \land \, \channel = \channeledge{\ein'} \\
    & \hspace{20pt} \lor \DecideCo(\CTLCcontract, h(\ein', id),  h_{sec}(\ein', id)) \\
    & \hspace{30pt} \in \actions{\run} 
\end{aligned}
$}
\end{equation}
Note that $\eout$ could also be specified using $\funonpath{.}{.}$ or \linebreak 
$\funonPathtoRoot(.)$ but by construction of $\untree$ they would specify exactly $\walk$ again. 
The following helper function validates the conditions from above \newcommand{\chContract}{chContract}
{\small \begin{align} \label{eq:CTLC-chContract}
        &\chContract(\ein, \run) \\
        &:=
        \begin{cases}
            3 &, \text{if } \exists h(\ein, id) \in \CTLCcontract \in \environment_{\channeledge{\ein}}.C_{dec} \\
            2 &, \text{if } \funenabled(\ein) \land \funsecretsAv(\ein) \land \funisIngoing(\ein, \run) \\
                &\land \,  t_0 \leq \environmentvec.t < \funtimeout(h(\ein, id))  \\
            1 &, \text{if } \funenabled(\ein) \land \funisIngoing(\ein, \run) \\
                &\land \, \exists s^{id}_{\walk} \in \funsecret(h(\ein,id)) \backslash \environment_{\channeledge{\ein}}.S_{rev} \, \\
                &\, \, \, \, \exists \channel': s^{id}_{\walk} \in \environment_{\channel'}.S_{rev} \land B\in \confuser{\environment_{\channel'}} \\
            0 &, \text{else.}
        \end{cases} \nonumber
\end{align}}%
In the first case, there is a decided/claimed contract with a subcontract resembling a given edge. It is part of the honest user strategy to execute any contract that is decided upon.
In the second case, the contract the subcontract for $\ein$ belongs to has been enabled, the remaining secrets are owned by $B$, $\ein$ is an ingoing edge, all previous sub-contracts of the same \CTLC{} have timed out, and the timeout has not run out yet. In the third case, there is a secret $a_{\iota}^A$ in another environment $\environment_{\channel'}$ different from $\environmentch$ which has been revealed there but not in $\environmentch$ and belongs to the subcontract coming from $\ein$. Additionally, $B$ is part of both environments. 

To avoid $B$ revealing secrets for two duplicated edges on the same level we define a function that looks up whether $B$ scheduled such an action before. From the partition $\funedgegroup$ (see \ref{def:funedgegroup}) we notice for any given $\e \intree \untree$
\[
    \exists ! \Omega \in \funedgegroup : \e \in \Omega .
\]
Based on this specific $\Omega$, we define the function 
\begin{align*}
    \funnodupl (\e, \run) :\Leftrightarrow & \nexists \e' \intree \untree : \e' \neq \e \\
    & \land \funsender(\e') = \funsender(\e) \\
    & \land \funreceiver(\e') = \funreceiver(\e) \\
    & \land \funsecret(\e', id) \in \environmentvecRevealedSecrets . 
\end{align*}
This completes the set of functions and conditions needed for the definition of the next sub-strategy:
\begin{equation}\label{eq:CTLC-execution}
\scalebox{0.88}{$
{\small \begin{aligned}  
    &\Sigma_{B}^{\e}(\run) \\
    &:= \begin{cases}
        \Bigl\{ \CoEx (\CTLCcontract , \CTLCsubcontract) \Bigr\} & \hspace{-5pt} \text{,if } \chContract(\e, \run) = 3, \\
        & \, \CTLCsubcontract = h(\e, id) \in \CTLCcontract \\
        \Bigl\{ \DecideCo (\CTLCcontract,\CTLCsubcontract , sec) \Bigr\} & \hspace{-5pt} \text{,if } \chContract(\e, \run) = 2, h(\e, id) \in \CTLCcontract, \\
       & \, \exists i : \, \funsecret_i(h(e,id)) \subseteq \environment_{\channeledge{\e}}.S_{rev}, \\
       & \, sec := \funsecret_i(h(e,id)) \\
        \Bigl\{ B: \revealSecretch \, s^{id}_{\walk} \Bigr\} & \hspace{-5pt} \text{,if } \chContract(\e, \run) = 2, \\
        & h(\e, id) \in \CTLCcontract \in \environmentvecCTLCEnabled, \\
       & \channel := \channeledge{\e}, s^{id}_{\walk} = \funsecret(\e, id), \\
       & \, \funsecret_i(h(e,id)) \nsubseteq \environmentchRevealedSecrets, \\
       & \, s^{id}_{\walk} \in \funsecret_i(h(e,id)) \backslash \environmentchRevealedSecrets, \\
       & \funnodupl (\e, \run) \\
        \Bigl\{ B: \shareSecretch \, s^{id}_{\walk} \Bigr\} & \hspace{-5pt} \text{,if } \chContract(\e, \run) = 1, h(\e, id) \in \CTLCcontract, \\
        & \channel := \channeledge{\e},\\
        \emptyset & \hspace{-5pt} \text{,if } \chContract(\e, \run) = 0 
    \end{cases}
\end{aligned}}%
$}
\end{equation}
In the first case, the \CTLC{} has already met all conditions for execution and, therefore, has been decided by moving it to $\environmentCTLCDecided$. It can get executed as long as its timeout has not run out yet. 
In the second case, all conditions are met, and $B$ executes $\DecideCo$.
In the third case, one secret of $B$ needs to be revealed before the contract can be decided. With $\funsecret_i(h(e,id))$, we mean the same secret set that was found in $\funsecretsAv(\ein)$, i.e. the same $i$. If this is true for more than one $i$ it is chosen arbitrarily from the ones fulfilling the condition. 
The $\funnodupl (\e, \run)$ condition implies that for two duplicated edges on the same level, only for one of them, the revealing of its edges gets scheduled by $B$. With this it is also ensured that only one of them can get claimed. By Definition (\ref{eq:treesettoCTLCBadges}), this situation cannot occur for $\shareSecretch \, s^{id}_{\walk}$. 
In the fourth case, a secret needed for claiming this contract has been revealed in another environment $B$ is part of thus, $B$ shares the secret with the environment of the contract $h(\e, id) \in \CTLCcontract$. 
If none of these 4 cases apply, $B$ cannot act on this contract and schedules no action. 

\paragraph{Combining Sub-Strategies}
All previously defined sub-strategies are dependent on an edge from a specific tree $\untree$, except $ \widehat{\Sigma}_{B}^{\treeobj} (\run)$ which is only $\treeobj$-dependent. Hence we first unite over all \linebreak $\fulltreeobj \in \treeobj$ and then over all edges $\e$ from a given $\untree$: 
\begin{align}\label{def:temphoneststrategy}
    &\Sigma_B^{\treeobj, temp}(\run) \\
    &:= \widehat{\Sigma}_{B}^{\treeobj} (\run) \bigcup_{\fulltreeobj \in \treeobj} \hspace{3pt} \bigcup_{\e \in \untree} \overline{\Sigma}_{B}^{e} (\run) \cup \widetilde{\Sigma}_{B}^{\e} (\run) \cup \Sigma_{B}^{\e}(\run) \nonumber
\end{align}

If $\Sigma_B^{\treeobj, temp}(\run)$ outputs no action, we want to wait and elapse time which is defined with 
\begin{align}\label{def:timetimetime}
    \Sigma_B^{\treeobj, local}(\run) := \begin{cases}
        \Bigl\{ \elapse \, \delta \Bigr\} &\text{ ,if } \Sigma_B^{\treeobj, temp}(\run) = \emptyset \\
        \Sigma_B^{\treeobj, temp}(\run) &\text{ ,else.} 
    \end{cases}
\end{align}
Here $\delta$ is defined as 
\begin{align*}
    \delta_{t_0} &:= t_0 + j \Delta - \environmentvec.t \\
            &\text{with } j = min \{ j \in \Z \mid t_0 + j \Delta - \environmentvec.t > 0 \} \\
    \delta &:= min \{ \delta_{t_0} \mid \fulltreeobj \in \treeobj \}
\end{align*}
which is exactly the time until the next timestamp at which new sub-contracts potentially become available for being enabled or for execution. Note that allowing $j$ to come from the Integers $\Z$ implies that both the Enabling Phase as well as the Execution Phase are covered. 
Since the honest user should not change his mind, we want to stick to an output as long as it has not been executed and is still a valid extension of the current run.
To formalize this notion, we set
\begin{align}
        \actions{R} &:= \{ \action \mid \run = \run_1 \overset{\action}{\longrightarrow} \run_2 \} , \label{eq:actionsset}\\
        \sactions{R} &:= \{ \action \mid \exists \environmentvec : \run \overset{\action}{\longrightarrow} \environmentvec \}. \label{eq:sactionsset} 
\end{align}
By including this addition, we recursively define the \textit{honest user strategy for an honest user $B$}:
\begin{align} \label{def:honeststrategy}
    \Sigma_B^{\treeobj}(\run) := \begin{cases}
        \{ \action \} &\hspace{-7pt} \text{,if } \exists \run_1, \run_2 : \run = \run_1 \rightarrow \run_2 \\
        &\hspace{-7pt} \exists \action \in \Sigma_B^{\treeobj}(\run_1) :\\
        &\hspace{-7pt} \action \notin \actions{R} \\
        &\hspace{-7pt} \land \action \in \sactions{R}  \\
        \Sigma_B^{\treeobj, local}(\run) &\hspace{-7pt} \text{,else.} 
    \end{cases}
\end{align}
We call the first case the consistency condition of the honest user strategy.
\section{Progression of Time}
\label{sec:TimeProgression}

This section formalizes that in a run including an honest participant time always progresses after a finite number of actions and thus a \textit{final run}, which will be formally defined here, will always be reached at some point. 
\begin{definition}
Given a run $\run \hspace{-1pt} = \hspace{-1pt} R_1 \overset{\action}{\longrightarrow} \environmentvec \hspace{-1pt}$ we define 
$$\funlastEnv(\run) := \environmentvec .$$
\end{definition}
In the following, we assume a run $\run$ of arbitrary length to be given and a set of honest users $\honestusers =  \{ \honestuser_1, \dots, \honestuser_k\}$ with strategies $\honestuserstrategies = \{ \userstrategy{\honestuser_1}, \dots, \userstrategy{\honestuser{_k}} \}$ to participate. We recap from Definition \ref{def:adversarystrategy} that an adversary can only schedule an $\elapse \, \delta$ action if all users agree. 

\begin{theorem}\label{th:timeelapseaftern}
    Let $\run' := \run \overset{\action_0}{\longrightarrow} \environmentvec_1 \overset{\action_1}{\longrightarrow} \environmentvec_2 \overset{\action_2}{\longrightarrow} ...$
    be an extension of $\run$ with $\honestuserstrategies \conforms \run'$.
    Then there exists $\delta \in \R^{>0}$ s.t. 
    \[
        \exists n \in \N : \alpha_n = \elapse \, \delta \lor \vert \run' \vert - \vert \run \vert \leq n .
    \]
\end{theorem}

\begin{proof}
    As noted above, an adversary can only schedule a finite number of actions between honest user actions. 
    All $\CTLC$ actions, defined in Appendix \ref{sec:inferenceRules}, are relative to a batch $\batch$. 
    Let $B \in \honestusers$. 
    From the honest user strategy, especially (\ref{eq:honprep}) and (\ref{eq:CTLC-enable-check}), we know that all honest users $B$ only schedule actions regarding a $\batch$ with 
    \[
        \batch = \treetoCTLCBadge(\fulltreeobj, \mathcal{S})
    \]
    for some $\fulltreeobj \in \treeobj$ and secret set $\mathcal{S}$. We assume $\treeobj$ to be finite, i.e., the number of game trees executed simultaneously is finite. Also, we assume every tree only to include finitely many walks, which again are assumed to be of finite length, i.e. 
    \[
        \forall \walk \in \untree \, \exists m \in \N: \walk = [a_{m-1}, a_{m-2}, ..., a_0 ].  
    \]
    Therefore, the total number of edges in $\treeobj$, given as the size of
    \[
       \bigcup_{\fulltreeobj \in \treeobj} \{\e \intree \untree \},
    \]
    is finite. Every such edge gets mapped to a subcontract \linebreak $\CTLCsubcontract = h(e, id)$, and based on the honest user strategy, only finitely many actions are implied by this. 
    
    Assume towards contradiction $\vert \run' \vert - \vert \run \vert > n$ and 
    \[
        \nexists n \in \N : \alpha_n = \elapse \, \delta .
    \]
    We show that any honest user $B_i$ only schedules finitely many actions before nothing but $\elapse \, \delta_i$ will be scheduled, regardless of the actions of other users. Combining this with the fact that an adversary can only schedule a finite number of actions between honest user actions and the assumption that $\treeobj$ is finite, we arrive at the desired contradiction.
    We proceed with going through the sub-strategies of the honest user strategy for a given $\fulltreeobj \in \treeobj$ one by one. 

    Let an honest user $B$ be given. 
    For the preparation phase we have $\widehat{\Sigma}_{B}^{\treeobj} (\run)$ (see (\ref{eq:honprep})) which depends on $newBatch(\treeobj, \run)$. In case $newBatch(\treeobj, \run) = 2$ all preconditions for advertising a new batch are fulfilled. Since it is checked that a batch has not been advertised before, this can only be the case as often as the cardinality of $\treeobj$ allows. In case $newBatch(\treeobj, \run) = 1$, party $B$ will commit to a given Batch. With $S_B(\batch) \nsubseteq \environmentchCreatedSecrets$, it is checked that this only happens once per batch. In case $newBatch(\treeobj, \run) = 0$ an empty set will be outputted, which, in case all other sub-strategies also output $\emptyset$, will be turned into an $\elapse \, \delta_i$ action (see (\ref{def:honeststrategy})). 

    For the enabling phase we have $\overline{\Sigma}_{B}^{e} (\run)$. In case 
    \[
        \overline{\Sigma}_{B}^{e} (\run) = \Bigl\{ \advCTLC_{\channel} \, \advCTLCcontract \Bigr\}
    \]
    $newC(\e, \run) = 1$ holds, which includes the condition \linebreak $\nexists h(\e, id) \in \advCTLCcontract \in \environmentCTLCAdvertised$. Hence, the corresponding contract for the sub-contract belonging to $\e$ has not been advertised, which implies that this action can only happen once. 
    Assume $\Bigl\{ \advCTLC_{\channel} \, \advCTLCcontract \Bigr\}$ happens twice in a run. Then, it needs to be removed from $\environmentvecCTLCAdvertised$ in between. This can only be done with a \textit{refund} $\CTLCcontract$ action, which by construction of $h(\e, id)$ in \eqref{def:hMapping} can only happen after $t_0$. This leads to a contradiction as it is a condition of 
    \[
        \overline{\Sigma}_{B}^{e} (\run) = \Bigl\{ \advCTLC_{\channel} \, \advCTLCcontract \Bigr\}
    \]
    that $\environmentvec.t < t_0$. 
    Analogously, the same is true for
    \[
        \overline{\Sigma}_{B}^{e} (\run) = \Bigl\{ B: \authCTLC \, \CTLCcontract \Bigr\}
    \]
    as it is checked for 
    $(B, id) \notin \environmentvecCTLCAuthorized$. This also holds for 
    \[
        \overline{\Sigma}_{B}^{e} (\run) = \Bigl\{ \enableCTLC_{\channel} \, \CTLCcontract \Bigr\}
    \]
    because of the $\CTLCcontract \notin \environmentvecCTLCEnabled$ condition, and for
    \[
        \overline{\Sigma}_{B}^{e} (\run) = \Bigl\{ B : \enableSubC \, h(\e, id) \Bigr\}
    \]
    with the
    \[
        \exists h(\e, id) \in \advCTLCcontract \in \environmentvecCTLCAdvertised \land \CTLCcontract \in \environmentvecCTLCEnabled \land h(\e, id) \notin \CTLCcontract
    \]
    condition. 

    The honest user strategy for the execution phase consists of two sub-strategies, $\widetilde{\Sigma}_{B}^{\e} (\run)$ and $\Sigma_{B}^{\e}(\run)$. In $\widetilde{\Sigma}_{B}^{\e} (\run)$ the actions \textit{timeout} and \textit{refund} are handled, and for both of them it is part of the conditions that they have not been scheduled before. In $\Sigma_{B}^{\e}(\run)$ we have 
    \[
        \Sigma_{B}^{\e}(\run) = \Bigl\{B : \revealSecretch \, s^{id}_{\walk} \Bigr\}
    \]
    conditionally linked to a sub-contract for which this secret is used. As noted previously, sub-contracts can only be enabled once. Therefore, this action can be scheduled only once per secret. Analogously 
    \[
        \Sigma_{B}^{\e}(\run) = \Bigl\{ \shareSecretch \, s^{id}_{\walk} \Bigr\}
    \]
    cannot be done more than once per secret, as it is linked to a sub-contract and its channel. 
    For 
    \[
        \Sigma_{B}^{\e}(\run) = \Bigl\{ \DecideCo ( \CTLCcontract, \CTLCsubcontract , \funsecret_i(\CTLCsubcontract_{\iota})) \Bigr\}
    \]
    it is guaranteed that $t_0 \leq \environmentvectime$, and with $\funisIngoing(\e, \run)$ that $B$ is the receiver of the contract. Hence, the authorization of $B$ is needed for this contract to be enabled (see (\ref{ctlc:enableCTLC})). Part of the condition for $B$ to authorize is $\environmentvectime < t_0$ according to the honest user strategy \linebreak (see (\ref{eq:enable-check})). Therefore, $\CTLCcontract$ cannot be enabled again after it was claimed by $B$. 
    For
    \[
        \Sigma_{B}^{\e}(\run) = \Bigl\{ \CoEx (\CTLCcontract, \CTLCsubcontract) \Bigr\}
    \]
    it is guaranteed in
    $\chContract(\ein, \run)$ that 
    \[
    h(\e, id) \in \CTLCcontract \in \environment_{\channeledge{\e}}.C_{dec}
    \]
    holds. The action $\DecideCo$ is the only one that can move contracts into $\environment_{\channeledge{\e}}.C_{dec}$ and $\CoEx$ removes them from this set. Hence, $\CoEx$ cannot be executed more often than $\DecideCo$ in a given run. As noted previously, $\DecideCo$ can only happen once per contract, which is, therefore, also the case for $\CoEx$
    
    In a setting with multiple $B_i \in \honestusers$, the above reasoning applies to all $B_i$. Hence all $\Sigma_{B_i}^{\treeobj}$ output $\{ \elapse \, \delta_i \}$ after $n$-many steps for some $\delta_i$. By Definition (see ($\ref{def:adversarystrategy}$)), an $\actelapse \delta$ action with $\delta_i \geq \delta$ is then appended to the run by the adversary. 
\end{proof}

\begin{corollary}\label{cor:anytimelimit}
    Let $\run' := \run \overset{\action_0}{\longrightarrow} \environmentvec_1 \overset{\action_1}{\longrightarrow} \environmentvec_2 \overset{\action_2}{\longrightarrow} ...$
    be an extension of $\run$ with $\honestuserstrategies \conforms \run'$. Given an arbitrary $t' \in \R^{>0}$ with $t' > \environmentvectime$ then there exists an $n \in \N$ s.t. $\environmentvec_{n+1}.t > t'$ or $\vert \run' \vert - \vert \run \vert \leq n$.
    In other words, extending a run can reach any point in time. 
\end{corollary}

\begin{proof}
    The Definition $\ref{def:adversarystrategy}$ of the adversary strategy sets a minimum size $\epsilon > 0$ for $\delta$. W.l.o.g. assume $\delta = \epsilon$ for all $\actelapse \delta$ actions in $\run'$. Thus 
    \[
        \left\lceil \frac{t' - \environmentvectime}{\epsilon} \right\rceil
    \]
    many $\actelapse \delta$ actions are needed to reach $t'$. From Theorem \ref{th:timeelapseaftern} we know that after a finite number of steps an $\actelapse \delta$ gets appended to the run and thus time $t'$ is reached after finitely many steps. 
\end{proof}

\begin{definition}\label{def:final run}
    A run $\run$ with $\environmentvec = \funlastEnv(\run)$ is considered \textit{final} if
    \begin{align*}
        &\forall \fulltreeobj \in \treeobj : \forall \e \in \untree : \\
        &\funtimeout (h(\e, id)) < \environmentvectime. 
    \end{align*}
\end{definition}

\begin{corollary}
    For any run $\run$ there exists an $n \in \N$ s.t. for all extensions 
    $
        \run' := \run \overset{\action_0}{\longrightarrow} \environmentvec_1 \overset{\action_1}{\longrightarrow} \environmentvec_2 \overset{\action_2}{\longrightarrow} ... \, ,
    $
    both of finite and infinite length, it holds:
    \[
        \run' \text{ is final} \, \lor \, \vert \run' \vert - \vert \run \vert \leq n
    \]
    This means that every run will become final at some point. 
\end{corollary}

\begin{proof}
    As it is explained in the proof of Theorem \ref{th:timeelapseaftern} the total number of edges in $\treeobj$ is finite. Thus the maximum
    \[
        max \{ \funtimeout (h(\e, id)) \mid \e \in \untree, \fulltreeobj \in \treeobj \}
    \]
    exists. By setting this value as $t'$ in Corollary \ref{cor:anytimelimit} the statement follows from it.
\end{proof}
\section{\CTLC{} Properties}
\label{sec:CTLCproperties}

\subsection{Strategy-independent properties}
We first define properties that are independent of the concrete, honest user strategies applied.

\begin{lemma}\label{lemma:atleastoneitemincontracts}
    Let $\run$ be a \CTLC{} run with $\environmentvec =\funlastEnv(\run)$.
    It holds
    \begin{align*}
        \forall \environmentch \in \environmentvec 
        : \CTLCcontract \in \environmentchCTLCEnabled \Rightarrow  \, \vert \CTLCcontract \vert \geq 1 .
    \end{align*}
\end{lemma}
\begin{proof}
    This statement directly follows from the inference rules \textit{\enableCTLC } (\ref{ctlc:enableCTLC}) and \textit{\enableSubC} (\ref{ctlc:enableSubC}) by induction. 
\end{proof}

\begin{lemma}\label{lemma:claimedcontractsoneitem}
    Let $\run$ be a \CTLC{} run with $\environmentvec =\funlastEnv(\run)$.
    It holds
    \begin{align*}
        \forall \environmentch \in \environmentvec 
        : \CTLCcontract \in \environmentchCTLCDecided \Rightarrow  \, \exists ! \, \CTLCsubcontract : \CTLCsubcontract \in 
        \CTLCcontract.
    \end{align*}
\end{lemma}
\begin{proof}
    This statement directly follows from the inference rules \textit{\DecideCo } (\ref{ctlc:DecideCo}) and \textit{\CoEx} (\ref{ctlc:CoEx}) by induction. 
\end{proof}

\begin{lemma}\label{lemma:enabledsubsetofadvertised}
    Let $\run$ be a \CTLC{} run with $\environmentvec =\funlastEnv(\run)$.
    It holds
    \begin{align*}
        \forall \environmentch \in \environmentvec 
        : \CTLCcontract \in \environmentchCTLCEnabled \Rightarrow  \, \advCTLCcontract \in \environmentchCTLCAdvertised \, \land \, 
        \CTLCcontract \subseteq \advCTLCcontract .
    \end{align*}
\end{lemma}
\begin{proof}
    This statement directly follows from the inference rules \textit{\enableCTLC } (\ref{ctlc:enableCTLC}) and \textit{\enableSubC} (\ref{ctlc:enableSubC}) by induction. 
\end{proof}

\begin{lemma}\label{lemma:somelemmainG}
    Let $\honestusers =  \{ \honestuser_1, \dots, \honestuser_k\}$ be a set of honest users
    and $\honestuserstrategies = \{ \userstrategy{\honestuser_1}, \dots, \userstrategy{\honestuser{_k}} \}$ a corresponding set of user strategies.
    Let $\attackerstrategy$ be an adversary strategy (for $\honestusers$) and let $\run$ be a run such that $(\honestuserstrategies, \attackerstrategy) \conforms \run$. 
    Further, let $\environmentvec =\funlastEnv(\run)$.
    Then it holds
    \begin{align*}
        \forall \environmentch &\in \environmentvec : \CTLCcontract \in \environmentchCTLCEnabled \\ \Rightarrow &(\funsender(\advCTLCcontract): \authCTLC \, \advCTLCcontract) \in \actions{\run} \\
         \land &(\funreceiver(\advCTLCcontract): \authCTLC \, \advCTLCcontract) \in \actions{\run} .
    \end{align*}
    See \cref{eq:actionsset} for the definition of $\actions{\run}$. 
\end{lemma}
\begin{proof}
    This statement directly follows from the inference rules \textit{\enableCTLC } (\ref{ctlc:enableCTLC}) and \textit{\authCTLC} (\ref{ctlc:authCTLC}) by induction. 
\end{proof}

\begin{lemma}\label{lemma:anothersomelemmainG}
    Let $\honestusers =  \{ \honestuser_1, \dots, \honestuser_k\}$ be a set of honest users
    and $\honestuserstrategies = \{ \userstrategy{\honestuser_1}, \dots, \userstrategy{\honestuser{_k}} \}$ a corresponding set of user strategies.
    Let $\attackerstrategy$ be an adversary strategy (for $\honestusers$) and let $\run$ be a run such that $(\honestuserstrategies, \attackerstrategy) \conforms \run$. 
    Further, let $\environmentvec =\funlastEnv(\run)$.
    Then it holds
    \begin{align*}
        &\CTLCcontract \in \environmentvecCTLCDecided \\
        &\Rightarrow 
        \exists \CTLCsubcontract : \, 
        \DecideCo (\CTLCcontract ,\CTLCsubcontract,  \funsecret_i(\CTLCsubcontract)) \in \actions{\run} 
    \end{align*}
\end{lemma}
\begin{proof}
    This statement directly follows from the inference rule \textit{\DecideCo } (\ref{ctlc:DecideCo}) by induction. 
\end{proof}

\begin{lemma}\label{lemma:authorized-before-enabled}
    Let $\honestusers =  \{ \honestuser_1, \dots, \honestuser_k\}$ be a set of honest users
    and $\honestuserstrategies = \{ \userstrategy{\honestuser_1}, \dots, \userstrategy{\honestuser{_k}} \}$ a corresponding set of user strategies.
    Let $\attackerstrategy$ be an adversary strategy (for $\honestusers$) and let $\run$ be a run such that $(\honestuserstrategies, \attackerstrategy) \conforms \run$. 
    Further, let $\environmentvec =\funlastEnv(\run)$.
    Then it holds
    \begin{align*}
        &\CTLCcontract \in \environmentvecCTLCEnabled \\
        &\Rightarrow 
        \exists \dotCTLCcontract \hspace{-3pt} : \hspace{-2pt}
        \dotCTLCcontract \hspace{-1pt} \supseteq \hspace{-1pt} \CTLCcontract \funsender(\dotCTLCcontract) \hspace{-2pt}: \authCTLC (\dotCTLCcontract) \hspace{-2pt} \in \hspace{-2pt} \actions{\run} 
    \end{align*}
\end{lemma}
\begin{proof}
    This statement directly follows from the inference rule \textit{\DecideCo } (\ref{ctlc:DecideCo}) by induction. 
\end{proof}

\begin{lemma}\label{lemma:authorized-before-claim}
    Let $\honestusers =  \{ \honestuser_1, \dots, \honestuser_k\}$ be a set of honest users
    and $\honestuserstrategies = \{ \userstrategy{\honestuser_1}, \dots, \userstrategy{\honestuser{_k}} \}$ a corresponding set of user strategies.
    Let $\attackerstrategy$ be an adversary strategy (for $\honestusers$) and let $\run$ be a run such that $(\honestuserstrategies, \attackerstrategy) \conforms \run$. 
    Further, let $\environmentvec =\funlastEnv(\run)$.
    Then it holds
    \begin{align*}
        \CTLCcontract \in \environmentvecCTLCDecided & \Rightarrow \,
        \exists \dotCTLCcontract : \,
        \dotCTLCcontract \supseteq \CTLCcontract \\
        &\hspace{2pt} \land \funsender(\dotCTLCcontract) \hspace{-1pt} : \hspace{-1pt} \authCTLC (\dotCTLCcontract) \hspace{-1pt} \in \hspace{-1pt} \actions{\run}.  
    \end{align*}
\end{lemma}
\begin{proof}
    This statement directly follows from the inference rule \textit{\DecideCo } (\ref{ctlc:DecideCo}) by induction. 
\end{proof}

\begin{lemma}\label{lemma:secsetsmonotonic}
    Let $\honestusers =  \{ \honestuser_1, \dots, \honestuser_k\}$ be a set of honest users
    and $\honestuserstrategies = \{ \userstrategy{\honestuser_1}, \dots, \userstrategy{\honestuser{_k}} \}$ a corresponding set of user strategies.
    Let $\attackerstrategy$ be an adversary strategy (for $\honestusers$) and let $\run$ be a run such that $(\honestuserstrategies, \attackerstrategy) \conforms \run$. 
    Further, let $\environmentvec =\funlastEnv(\run)$.
    Then it holds for all $\environmentvec_0, \environmentvecprime$ with $\run = \environmentvec_0 \longrightarrow \ldots \longrightarrow \environmentvecprime \longrightarrow \ldots \longrightarrow \environmentvec$ :
    \begin{align*}
        \forall \channel: \environmentchprimeRevealedSecrets \subseteq \environmentchRevealedSecrets 
    \end{align*}
\end{lemma}
\begin{proof}
    This statement directly follows from the inference rules 
    \textit{\revealSecret } (\ref{ctlc:revealSecret}) and \textit{\shareSecret } (\ref{ctlc:shareSecret}) by induction. 
\end{proof}

\subsection{Properties implied by the honest user strategy}

We show general properties of runs that can be enforced by the honest user strategy $\Sigma_B^{\treeobj}$ defined in Appendix~\ref{sec:honestStrategy}.

\begin{lemma}\label{lemma:somelemmaaboutrevealsecret}
    Let $B$ be an honest user, $\treeobj$ be a set of tuples of the form $\fulltreeobj$, which is well-formed according to \cref{def:wellformedtree}, and let $\Bstrategy$ be the honest user strategy for $B$ executing $\treeobj$. Let $\Astrategy$ be an arbitrary adversarial strategy. Let $\run$ be a run with $\Bstrategy , \Astrategy \results \run$, starting from an initial environment.
    Further, let $\environmentvec =\funlastEnv(\run)$.
    Then it holds for all $\e \intree \untree $ with $\honestuser = \funreceiver(\e)$:
    \begin{align*}
        &s = \funsecret(\e, id) \in \environment_{\channeledge{\e}}.S_{rev} \\ &\Rightarrow \honestuser : \revealSecretch \, s \in \actions{\run} \\
        &\hspace{13pt} \text{ with } \channel = \channeledge{\e}
    \end{align*}
\end{lemma}
\begin{proof}
    This statement directly follows from the inference rules \textit{\revealSecret } (\ref{ctlc:revealSecret}) and the fact that 
    \[
        \funowner(\funsecret(\e, id)) = \funreceiver(\e)
    \] 
    from \cref{def:edgesecret}. 
\end{proof}

Additionally, we show that whenever a contract has been claimed, it will not be enabled anymore for a second time. 

\begin{lemma} \label{lemma:againanotherlemmainG}
    Let $B$ be an honest user, $\treeobj$ be a well-formed set of tuples of the form $\fulltreeobj$ and $\Bstrategy$ the honest user strategy for $B$ executing $\treeobj$. Let $\Astrategy$ be an arbitrary adversarial strategy. Let $\run$ be a run with $\Bstrategy , \Astrategy \results \run$, starting from an initial environment.
    Further, let $\environmentvec =\funlastEnv(\run)$.
    Then it holds for \hspace{-1pt} all $id$:
    \begin{align*}
        &\DecideCo (\CTLCcontract ,\CTLCsubcontract,  \funsecret_i(\CTLCsubcontract_{\iota})) \in \actions{\run} \\
        &\land \honestuser \in \funusers(\CTLCcontract)\\
        &\Rightarrow \nexists \, \dotCTLCcontract \in \environmentvecCTLCEnabled \land \nexists ( B, \dotCTLCcontract ) \in \environmentvecCTLCAuthorized
    \end{align*}
\end{lemma}

\begin{proof}
    By induction on the length of $\run$. 
    The only interesting cases to consider are those where contracts are added to $\environmentvecCTLCEnabled$ and $\environmentvecCTLCAuthorized$ and where a contract is claimed.
    So we look at the cases where the last action $\action$ in $\run = \run' \overset{\action}{\longrightarrow} \Gamma$ is
    $\action = \enableCTLC (\CTLCcontract)$, \linebreak $\action = A: \authCTLC(\CTLCcontract)$ or 
    $\action = \DecideCo(\CTLCcontract, \CTLCsubcontract,  \funsecret_i(\CTLCsubcontract)).$

    If $\action = \enableCTLC (\CTLCcontract)$ we would have $( B, \dotCTLCcontract ) \in \environmentvecprimeCTLCAuthorized$, by the \textit{\enableCTLC } rule (\ref{ctlc:enableCTLC}), which is ruled out by I.H..

    If $\action = A: \authCTLC(\CTLCcontract)$ then we distinguish the cases on whether $A = B$. 
    If $A = B$, then we know that $\action$ must have been scheduled by $\Sigma_B^{\treeobj}$ and hence by definition of $\Sigma_B^{\treeobj}$, we have that
    $B: \authCTLC ( \hat{\dot{\textit{c}}}^{\textit{x}} ) \not \in \actions{\run'}$. 
    Now assume towards contradiction that 
    $$\DecideCo (\dotCTLCcontract ,\dotCTLCsubcontract,  \funsecret_i(\dotCTLCsubcontract_{\iota})) \in \actions{\run'}.$$ 
    Then there must be some $\run^*$ with $\run' = \run^* \overset{\action'}{\longrightarrow} \run^{**}$ for some $\run^{**}$ and $\funlastEnv(\run^*) =: \environmentvec^{*}$ such that $\dotCTLCcontract \in \environmentvec^{*}.C_{\textit{cla}} .$
    But then by~\Cref{lemma:authorized-before-claim}, we get that $B: \authCTLC(c'^x) \in \actions{\run^*} \subseteq \actions{\run'}$ for some $c'^x \supseteq \dotCTLCcontract$ leading to a contradiction. 
    If $A \neq B$, it does not have an effect on the Lemma. 

    If $\action = \DecideCo(\CTLCcontract, \CTLCsubcontract,  \funsecret_i(\CTLCsubcontract))$
    we have $\CTLCcontract \in \environmentvecprimeCTLCEnabled$ and $\CTLCcontract \notin \environmentvecprimeCTLCEnabled$ by the \textit{\DecideCo} rule (\ref{ctlc:DecideCo}). 
    For $\CTLCcontract \in \environmentvecprimeCTLCEnabled$ there needs to be $B : \authCTLC(\CTLCcontract) \in \actions{\run'}$ by \cref{lemma:somelemmainG}. This given authorization was then removed when enabling $\CTLCcontract$ according to the \textit{\enableCTLC } rule (\ref{ctlc:enableCTLC}). According to the honest user strategy, see \cref{eq:enable-check}, $B$ will not authorize another contract with the same identifier, and hence 
    \[
        \nexists (B : \authCTLC(\dotCTLCcontract)) \in \actions{\run'}
    \]
    which implies $\nexists ( B, \dotCTLCcontract ) \in \environmentvecCTLCAuthorized$ and $\nexists \, \dotCTLCcontract \in \environmentvecCTLCEnabled$.
    Part of the preconditions of $\DecideCo$ (\ref{ctlc:DecideCo}) is $\CTLCcontract \in \environmentvecCTLCEnabled$, and part of its effects is that it removes $\CTLCcontract$ from $\environmentvecCTLCEnabled$. Thus, it needs to be argued that $\CTLCcontract$ cannot be enabled again.
    
    Assume towards contradiction that in $\environmentvec =\funlastEnv(\run)$ we have $\CTLCcontract \in \environmentvecCTLCEnabled$. By \cref{lemma:somelemmainG}, we have that both sender and receiver need to authorize $\CTLCcontract$ again.
    Since $\funusers(\CTLCcontract) \cap Hon \neq \emptyset$, the honest user strategy determines at least one of these authorizations. 
    Now, the honest user strategy, more specifically \eqref{eq:enable-check}, tells us that honest users will only authorize a contract once. 
\end{proof}

\begin{lemma}\label{lemma:only-top-lvl-disabled}
    Let $B$ be an honest user, $\treeobj$ a well-formed set of tuples of the form $\fulltreeobj$ and $\Bstrategy$ the honest user strategy for $B$ executing $\treeobj$. Let $\Astrategy$ be an arbitrary adversarial strategy. Let $\run$ be a run with $\Bstrategy , \Astrategy \results \run$, starting from an initial environment.
    Further, let $\environmentvec =\funlastEnv(\run)$.
    Then it holds
    \begin{align*}
        &\CTLCsubcontract \in \CTLCcontract \in \environmentvecCTLCEnabled \, \land \, \funposition(\CTLCsubcontract) = j \\
        \land \, &B \in \funusers(\CTLCcontract) \\
        \land \, &\exists \run', \run'' : \run = \run' \rightarrow \run'' \text{ with }\environmentvecprime =\funlastEnv(\run') : \\
        &\exists \dotCTLCcontract, \dotCTLCsubcontract \in \dotCTLCcontract \in \environmentvecprimeCTLCEnabled \text{ with } \funposition(\dotCTLCsubcontract) > j \\
        & \Rightarrow \dotCTLCsubcontract \in \CTLCcontract \in \environmentvecCTLCEnabled .
    \end{align*}
\end{lemma}
\begin{proof}
    This statement follows from the inference rule \textit{\timeoutsc } (\ref{ctlc:timeout}) together with the previous \cref{lemma:againanotherlemmainG} by induction. The honest user strategy implies that once a contract has been claimed, it cannot be enabled again. 
\end{proof}

\begin{lemma}\label{lemma:fundsAvailable}
    Let $B$ be an honest user, $\treeobj$ a well-formed set of tuples of the form $\fulltreeobj$ and $\Bstrategy$ the honest user strategy for $B$ executing $\treeobj$. Let $\Astrategy$ be an arbitrary adversarial strategy. Let $\run$ be a run with $\Bstrategy , \Astrategy \results \run$, starting from an initial environment $\environmentvec_0$ which is liquid w.r.t. $\treeobj$, see \cref{def:liquidenv}.
    Further, let $\environmentvec =\funlastEnv(\run)$.
    Then it holds  
    \begin{align*}
        &\forall \fulltreeobj \in \treeobj ~\forall \e \intree \untree : B = \funsender(\e) \\
        &\land ~\nexists \batch \in \environmentvec.\batchset : h(\e,id) \in \CTLCcontract \in \batch \\
        &\Rightarrow \edgefund{\e} \in \environment_{\channeledge{\e}}.F_{av} 
    \end{align*}
\end{lemma}
\begin{proof}
    Recalling from
    \cref{def:liquidenv}
    we know that in $\environmentvec_0$ we have
    \[
        \forall \fulltreeobj \hspace{-1pt} \in \hspace{-1pt} \treeobj ~\forall \e \intree \untree \hspace{-2pt} : \hspace{-2pt} \edgefund{\e} \in \environment_{0, \channeledge{\e}}.F_{av} .
    \]
    Only the \textit{\enableCTLC } (\ref{ctlc:enableCTLC}) action removes funds from $\environmentvecAvailableFunds$. This action relies on $\advCTLCcontract \in \environmentvecCTLCAdvertised$, which can only be added by the \linebreak \textit{\advCTLC } (\ref{ctlc:advCTLC}) action, and the authorization of $B$. 
    Hence, if $B = \funowner(\edgefund{\e})$ does not give its authorization, it cannot be removed from $\environment_{\channeledge{\e}}.F_{av} $. 
    The action \textit{\advCTLC } relies on advertised batches $\batch \in \environmentvec.\batchset$ and hence, if a corresponding batch has not been advertised, which cannot be removed from $\environmentvec.\batchset$, the following steps also have not happened yet.. Hence the fund is still available.
\end{proof}

\subsection{Protocol Security}
\label{sec:protocolSecurity}
In the following, we will use $\e \in \untree$ instead of $\e \intree \untree$ for brevity of notation. 

\begin{definition}
    Given a $\fulltreeobj \in \treeobj$ and a valid specification $\specalone$ on $\untree$ (\ref{def:specification}), we define the set of all tree edges that involve $B$ as
    \[
        \hatpartialEx := \{ \edge \intree \untree \mid X = B \lor Y = B \}.    
    \]
    Subsets $\partialEx \subseteq \hatpartialEx$ are called \textit{consistent} by definition if and only if 
    \begin{align*}
        &\forall \e \in \partialEx : \\
        &\bigl( \forall \e' \in \funonPathtoRoot (\untree, \e) \cap \hatpartialEx : \e' \in \partialEx\\
        &\land \nexists \e'' \in \partialEx \backslash \{ \e \} : \spec{\e''} = \spec{\e} \bigr).
    \end{align*}
    In other words, for all edges in $\partialEx$, the edges appearing on its path to the root that are also in $\hatpartialEx$ are in $\varrho_B$. Therefore, $\varrho_B$ is closed regarding walks upwards in the tree. 
    Also, there are no duplicates in $\partialEx$. 
    We define 
    \[
        \partialEx \widehat{\subseteq} \untree : \Leftrightarrow \partialEx \subseteq 
        \{ \e \mid \e \intree \untree \}.
    \]
    We use $\partialEx \subseteq \untree$ instead of $\partialEx \widehat{\subseteq} \untree$ to simplify the notation.
\end{definition}

\begin{theorem}\label{th:mainSecurity}
    Let $B$ be an honest user, $\treeobj$ be a set of tuples of the form $\fulltreeobj$ (with $\untree$ being a tree, $id\in \N$ a unique id for the tree, $t_0 \in \R^{>0}$ and $spec(\e)$ the specification for $\e$), which is well-formed (\cref{def:wellformedtree}), and $\Bstrategy$ the honest user strategy for $B$ executing $\treeobj$. Let $\Astrategy$ be an arbitrary adversarial strategy. Then for all runs $\run$ with $\Bstrategy , \Astrategy \results \run$, starting from an initial environment, and for all $\fulltreeobj \in \treeobj$ there exists a consistent $\partialEx \subseteq\untree$ s.t. $\run \IDrelation \partialExFamily$ and $Inv(\run, \partialEx, \fulltreeobj)$ holds.
    
    Where 
    $
        \partialExFamily := \{ \partialEx \mid \fulltreeobj \in \treeobj \}
    $
    and 
    \begin{align}
        &\run \IDrelation \partialExFamily :\Leftrightarrow \nonumber \\
        & \forall (\CTLCsubcontract, \funsecret_i(\CTLCsubcontract)) \in \comcontracts{R} : \nonumber\\
        & \exists \partialEx, \e \in \partialEx : \CTLCsubcontract = h(\e, id) \nonumber \\
        & \hspace{20pt} \land \funsecret_i(\CTLCsubcontract) = h_{sec}(\e, id) \label{eq:IDrelation1}\\
        \land \, & \forall \partialEx, \e \in \partialEx : \label{eq:IDrelation2}\\
        & \hspace{20pt} (h(e,id), h_{sec}(\e, id)) \in \comcontracts{R} \nonumber 
    \end{align}
    and by recalling from Appendix \ref{sec:honestStrategy} we set
    \begin{align}
        \actions{\run} &\,= \{ \action \mid \run = \run_1 \overset{\action}{\longrightarrow} \run_2 \} \, \text{ (see (\ref{eq:actionsset}))} \nonumber 
    \end{align}
    \begin{align}
        \sactions{\run} &\,= \{ \action \mid \exists \environmentvec : \run \overset{\action}{\longrightarrow} \environmentvec \}  \, \, \, \, \text{ (see (\ref{eq:sactionsset}))} \nonumber
    \end{align}
    \begin{align}
        \comactions{\run} &:= \actions{\run} \cup \sactions{\run} \nonumber
    \end{align}
    \begin{align}
        \excontracts{\run} &:= \label{eq:excontracts} \\
        & \hspace{-25pt} \{ (\CTLCsubcontract, \funsecret_i(\CTLCsubcontract)) \mid \action \in \actions{\run} \nonumber\\
        & \hspace{-25pt} \land B = \funsender(\CTLCsubcontract) \nonumber\\
        & \hspace{-25pt} \land \action = \DecideCo(\CTLCcontract, \CTLCsubcontract,  \funsecret_i(\CTLCsubcontract)) \} \nonumber
    \end{align}
    \begin{align}
        \schecontracts{\run} &:= \label{eq:schecontracts} \\
        & \hspace{-25pt} \{ (\CTLCsubcontract, \funsecret_i(\CTLCsubcontract)) \mid \action \in \comactions{\run} \nonumber \\
        & \hspace{-25pt} \land B = \funreceiver(\CTLCsubcontract) \nonumber\\
        & \hspace{-25pt} \land \action = \DecideCo(\CTLCcontract,\CTLCsubcontract, \funsecret_i(\CTLCsubcontract)) \} \nonumber 
    \end{align}
    \begin{align}
        \comcontracts{R} &:= \nonumber \\ 
        & \hspace{-25pt} \excontracts{\run} \cup \schecontracts{\run} \nonumber
    \end{align}
    and it holds $Inv(\run, \partialEx, \fulltreeobj)$ defined as
    \begin{align*}
        Inv&(\run, \partialEx, \fulltreeobj) \\
        &:= \bigwedge_{S \in InvNames} Inv_S(\run, \partialEx, \fulltreeobj)
    \end{align*}
    consisting of the invariants defined below. 
    We define ten invariants
    \begin{align*}
        &InvNames := \\
        &\{ \text{in-secrets, secrets, in-schedule, levels, liveness,} \\
        & \text{init-liveness, deposits, setup, tree, auth} \}
    \end{align*} 
    The first invariant ensures that whenever the secret assigned to an ingoing edge has been revealed, the edge will get executed. 
    \begin{align*} \small
        &Inv_{\text{in-secrets}}(\run, \partialEx, \fulltreeobj) :\Leftrightarrow \\
        &\forall Y, \ein, \environmentvec \text{ with } \environmentvec = \funlastEnv(\run) \, : \\
        &\ein = (Y,B)_{\walk} \in \untree, \funsecret(\ein, id) \in \environmentvecRevealedSecrets \\
        &\Longrightarrow \ein \in \partialEx
    \end{align*}%

    It also should be assured that if an edge is executed, all secrets of the edges on the path to the root have been revealed.
    \begin{align*} \small
        &Inv_{\text{secrets}}(\run, \partialEx, \fulltreeobj) :\Leftrightarrow \\
        &\forall \environmentvec \text{ with } \environmentvec = \funlastEnv(\run) \, : \\
        &\bigl( \e \in \partialEx \\
        &\Longrightarrow \forall \e' \in \funonPathtoRoot(\untree,\e):\, \\
        & \hspace{22pt} \funsecret(\e', id) \in \environment_{\channeledge{\e}}.S_{rev} \bigr)
    \end{align*}%
    Additionally, the secret of an ingoing edge should only be published if the outgoing edge, if existent, had been claimed before. 
    \begin{align*} \small
        &Inv_{\text{in-schedule}}(\run, \partialEx, \fulltreeobj) :\Leftrightarrow \\
        &\forall X, Y, \ein, \environmentvec \text{ with } \environmentvec = \funlastEnv(\run) \, :\\
        &\bigl( \ein = (Y,B)_{\walk} \in \untree, \fundepth(\ein) > 1, \\
        &\funsecret(\ein, id) \in \environmentvecRevealedSecrets \\
        &\Longrightarrow \forall \eout = (B,X)_{\walk'} \in \untree \text{ with } \walk = \concatvec{\unitvec{(Y,B)_{\walk}}}{\walk'} :\\
        & \hspace{15pt} \DecideCo( \CTLCcontract,  h( \eout, id ), h_{sec}( \eout, id ) ) \in \funactions(\run) \bigr)
    \end{align*}%
    The \textit{levels} invariant ensures that once the timeout for a $h(\e, id)$ has run out, there is no channel in which the subcontract is enabled, and its fund is reserved. 
    \begin{align*} \small
        &Inv_{\text{levels}}(\run, \partialEx, \fulltreeobj) :\Leftrightarrow \\
        &\forall \e \in \untree \text{ with } B = \funsender(\e), \environmentvec \text{ with } \environmentvec = \funlastEnv(\run) \, :\\
        &\Bigl( \environmentvec.t > t_0 + \fundepth(\e) \Delta \\
        &\Longrightarrow \nexists \channel, \CTLCcontract \in \environmentchCTLCEnabled : h(\e, id) \in \CTLCcontract \\
        &\hspace{20pt} \land \funfund(\CTLCcontract) \in \environmentchReservedFunds \Bigr)
    \end{align*}%
    For the liveness of the protocol, we need to ensure that if an outgoing edge has been claimed and there is an ingoing one, its timelock has not run out yet, or a representative higher-up has been decided before. 
    \begin{align*} \small
        &Inv_{\text{liveness}}(\run, \partialEx, \fulltreeobj) :\Leftrightarrow \\
        &\forall Y, \eout, \environmentvec \text{ with } \environmentvec = \funlastEnv(\run) \, : \\
        & \Bigl( \eout = (B,X)_{\walk} \in \untree, j = \fundepth(\eout), \\
        & \DecideCo(\CTLCcontract, h(\eout, id), h_{sec}(\eout, id)) \in \funactions(\run) \\
        & \Longrightarrow \forall \ein = (Y,B)_{\walk'} \in \untree \text{ with } \walk' = \concatvec{\unitvec{(Y,B)}}{\walk} :\\
        & \bigl( \exists j' \leq j+1, \walk'' : (Y,B)_{\walk''} \in \partialEx \\
        & \land \fundepth((Y,B)_{\walk''}) = j' \bigr) \\
        &\lor \bigl( \exists \dot{c}^{x'} \in \environment_{\channeledge{\ein}}.C_{en} : h(\ein, id) \in \dot{c}^{x'} \\
        & \land \environmentvectime < t_0 + (j+1)\Delta \bigr) \Bigr)
    \end{align*}%
    Note that $\walk'' = \walk'$ is possible. 
    Similarly, it should be ensured that if $\e \in \partialEx$, it either has been claimed already or there is still enough time to do so before a respective $\timeoutsc$ action can happen.
    \begin{align*} \small
        &Inv_{\text{init-liveness}}(\run, \partialEx, \fulltreeobj) :\Leftrightarrow \\
        &\forall \e \in \untree, \environmentvec \text{ with } \environmentvec = \funlastEnv(\run) \, : \e \in \partialEx\\
        & \Rightarrow \exists \CTLCcontract : \DecideCo(\CTLCcontract, h(\e, id), h_{sec}(\e, id)) \in \funactions(\run) \\
        & \hspace{13pt} \lor \environmentvectime < \timeout(h(\e, id))
    \end{align*}%
    The corresponding funds should be available to execute enabled and claimed contracts. This is ensured in the following invariant. 
    \begin{align*} \small
        &Inv_{\text{deposits}}(\run, \partialEx, \fulltreeobj) :\Leftrightarrow \\
        &\forall \environmentvec \text{ with } \environmentvec = \funlastEnv(\run), \environmentch \in \environmentvec \, : \\
        &\bigl( \exists \CTLCcontract \in \environmentchCTLCEnabled \cup \environmentchCTLCDecided \Longrightarrow \funfund(\CTLCcontract) \in \environmentchReservedFunds \bigr)
    \end{align*}%
    The next invariant says that if a contract for an outgoing edge is enabled or authorized by both the sender and receiver, then the contract for the ingoing edge is enabled, or another representative at the same pr smaller tree level has already been executed. This is because if a contract was enabled for an outgoing edge, we need to show that the ingoing edge is either enabled or has already been executed. 
    \begin{align*} \small
        &Inv_{\text{setup}}(\run, \partialEx, \fulltreeobj) :\Leftrightarrow \\
        &\forall X, Y, \eout = (B,X)_{\walk} \in \untree, \environmentvec \text{ with } \environmentvec = \funlastEnv(\run) \, : \\
        &\bigl( h(\eout, id) \in \CTLCcontract \in \environment_{\channeledge{\eout}}.C_{\textit{en}} \\
        & \lor (B, \advCTLCcontract) \in \environment_{\channeledge{\eout}}.C_{\textit{aut}} \, , h(\eout, id) \in \advCTLCcontract \\
        & \Longrightarrow \forall \ein = (Y,B)_{\walk'} \in \untree \text{ with } \walk' = \concatvec{\unitvec{(Y,B)_{\walk}}}{\walk} :\\
        &(\exists (\CTLCcontract)' \in \environment_{\channeledge{\ein}}.C_{en} :  h(\ein, id) \in (\CTLCcontract)' \, \\
        &\lor \, \exists j' \leq \fundepth(\ein) : (Y,B)_{\walk''} \in \partialEx  \\
        &\hspace{10pt} \land \fundepth((Y,B)_{\walk''}) = j')
        \bigr)
    \end{align*}%
    The following invariant ensures all subcontracts, including the honest user $B$, are linked to at least one edge in a tree.
    \begin{align*} \small
        &Inv_{\text{tree}}(\run, \partialEx, \fulltreeobj) :\Leftrightarrow \\
        &\forall \environmentvec = \funlastEnv(\run) \, : \\
        &( \, (B, \advCTLCcontract) \in \environmentvecCTLCAuthorized \\
        &\Longrightarrow \forall \CTLCsubcontract \in \advCTLCcontract \, \exists \e \in \untree : \CTLCsubcontract = h(\e, id) ) \\
        &\land ( \CTLCcontract \in \environmentvecCTLCEnabled \cup \environmentvecCTLCDecided \land B \in \funusers(\CTLCcontract)\\
        &\Longrightarrow \forall \CTLCsubcontract \in \CTLCcontract \, \exists \e \in \untree : \CTLCsubcontract = h(\e, id) )
    \end{align*}%
    Finally, it is ensured that all authorized contracts also get enabled afterwards. 
    \begin{align*} \small
        &Inv_{\text{auth}}(\run, \partialEx, \fulltreeobj) :\Leftrightarrow \\
        &\forall \environmentch \in \environmentvec = \funlastEnv(\run), (B, \advCTLCcontract) \hspace{-2pt} \in \hspace{-2pt} \environmentchCTLCAuthorized, s := \size{\advCTLCcontract} \hspace{-2pt} \in \hspace{-2pt} \mathbb{N} \hspace{-2pt}: \\
        &(B = \funsender(\advCTLCcontract), \\
        &\CTLCcontract = \{ \CTLCsubcontract ~|~ \CTLCsubcontract \in \advCTLCcontract ~\land~ \funposition(\CTLCsubcontract) = s  \} \\
        & \Longrightarrow 
        \enableCTLC \, \CTLCcontract \in \Sigma_B^{\treeobj}(\run) 
    \end{align*}%
\end{theorem}

\begin{remark}
    Note that for the honest user $B$, only the claimed outgoing contracts are relevant, not the scheduled ones. Hence, only these are included in $\excontracts{R}$. For the ingoing ones, though, both already claimed and possibly claimed contracts are relevant, which is why they are both included in $\schecontracts{R}$. 
\end{remark}

\begin{proof}%
    Let $B, \treeobj, \Bstrategy, \Astrategy$ as stated in the Theorem and let $\run$ be a run with $\Bstrategy , \Astrategy \results \run$ and $\vert \run \vert = n$. 

    We prove by induction on $n$:
    \begin{align*}
        \forall &\fulltreeobj \in \treeobj \, \exists \partialEx \subseteq\untree \,\text{consistent}: \\ 
        &\run \IDrelation \partialExFamily \land Inv(\run, \partialEx, \fulltreeobj)
    \end{align*}

    \noindent \underline{\textbf{Case $n = 0$:}}
    
    \noindent
    By definition of $\run$ and $\vert \run \vert = 0$ we have $\funactions(R) = \emptyset$ and that $\environment = \funlastEnv(\run)$ is an initial environment, see \cref{def:initialconfig}. 
    We show that $a) \, \run \IDrelation \partialExFamily$ and $b) \, Inv(\run, \partialEx, \fulltreeobj)$. 
    \newline
    
    $a)$ Since $\funactions(R) = \emptyset$ by definition also $\schecontracts{R} = \emptyset$ and so it is sufficient to show that 
    $\excontracts{R} = \emptyset$ to show the statement (since trivially $\run \IDrelation \emptyset$ in this case). By the definition of $\excontracts{R}$ it is left to show 
    \begin{align*}
        \forall \action \in \sactions{R}:\, \action \neq \DecideCo(\CTLCcontract, \CTLCsubcontract, \funsecret_i(\CTLCsubcontract_{\iota})). 
    \end{align*}
    The set $\sactions{R}$ contains all $\action$ s.t. $\exists \environmentvec ' : \environmentvec \overset{\action}{\longrightarrow} \environmentvec '$. By the definition of the \CTLC \, semantics, the action $\DecideCo$ is impossible in an initial environment since $\environmentvecCTLCEnabled = \environmentvecCTLCDecided = \emptyset$.  
    \newline

    $b)$ To show that $Inv(\run, \partialEx, \fulltreeobj )$ holds, we go through the individual invariants. Let $\fulltreeobj \in \treeobj$ be given then:
    \begin{itemize}
        \item $Inv_{\text{in-secrets}}(\run, \emptyset, \fulltreeobj)$ trivially holds since \newline $\environmentvecRevealedSecrets = \emptyset$. 
        \item $Inv_{\text{secrets}}(\run, \emptyset, \fulltreeobj)$ holds since there is no $e \in \partialEx = \emptyset$ and therefore $\funonPathtoRoot(T, \e)$ is also empty.
        \item $Inv_{\text{in-schedule}}(\run, \emptyset, \fulltreeobj)$ holds because there are no revealed secrets yet, thus $\environmentvecRevealedSecrets = \emptyset$. Hence an $\ein \in \untree$ with $\funsecret(\ein) \in \environmentvecRevealedSecrets$ does not exist. 
        \item $Inv_{\text{levels}}(\run, \emptyset, \fulltreeobj)$ holds because no \CTLC s have been enabled yet. Thus, $\environmentvecCTLCEnabled = \emptyset$ and the conclusion of the invariant trivially holds. 
        \item $Inv_{\text{liveness}}(\run, \emptyset, \fulltreeobj)$ holds since $\actions{\run}= \emptyset$ and therefore no $\DecideCo$ action has happened yet. 
        \item $Inv_{\text{init-liveness}}(\run, \emptyset, \fulltreeobj)$ has no items to apply to since $\partialEx = \emptyset$, thus it holds.
        \item $Inv_{\text{deposits}}(\run, \emptyset, \fulltreeobj)$ is true because \newline $\environmentvecCTLCEnabled = \environmentvecCTLCDecided = \emptyset$.
        \item $Inv_{\text{setup}}(\run, \emptyset, \fulltreeobj)$ also holds because $$\environmentvecCTLCEnabled = \environmentvecCTLCAuthorized = \emptyset.$$
        \item $Inv_{\text{tree}}(\run, \partialEx, \fulltreeobj)$ is true because $$\environmentvecCTLCAuthorized = \environmentvecCTLCEnabled = \environmentvecCTLCDecided = \emptyset .$$
        \item $Inv_{\text{auth}}(\run, \partialEx, \fulltreeobj)$ holds for the same reason.
    \end{itemize}
    
    \noindent \underline{\textbf{Case $n > 0$:}} 
    
    \noindent With $\vert \run \vert > 0$ we know that it exists a run $\run'$ with $\run = \run' \overset{\action}{\longrightarrow} \Gamma$. By induction hypothesis (I.H.), we know there is also a $\partialExFamilyprime$ such that $\run' \IDrelation \partialExFamilyprime$ and $Inv(\run', \partialEx', \fulltreeobj)$ holds for all \linebreak $\fulltreeobj \in \treeobj$. Let $\fulltreeobj \in \treeobj$ be given. We show that there exists $\partialExFamily$ s.t. $\run \IDrelation \partialExFamily$ and \linebreak $Inv(\run, \partialEx, \fulltreeobj)$ holds. We proceed by case distinction on $\alpha$.

    \noindent \textbf{Case $\alpha = \advBatch$:}
    
    \noindent We show that 
    \begin{enumerate}[a\hspace{1pt}]
        \item \hspace{-5pt}) $\run \IDrelation \partialExFamilyprime$,
        \item \hspace{-5pt}) $Inv(\run, \partialEx', \fulltreeobj)$ hold.
    \end{enumerate}
    \noindent a) According to the $\advBatch$ rule only the batch set $\batchset$ changes between $\environmentvec'$ and $\environmentvec$, more specifically it changes between all $\environmentch$ as it is a global action. Therefore $\excontracts{\run} = \excontracts{\run'}$, $\schecontracts{\run} = \schecontracts{\run'}$ for all $id$ and thus 
    $$\run' \IDrelation \partialExFamilyprime \Rightarrow \run \IDrelation \partialExFamilyprime. $$
    \noindent b) All components in the environment remain unchanged except $\batchset$. Additionally all invariants within $Inv(\run, \partialEx', \fulltreeobj)$ only refer to actions in $\funactions(\run)$ that already existed in $\funactions(\run')$. Thus $Inv(\run, \partialEx', \fulltreeobj)$ is implied directly from \linebreak $Inv(\run', \partialEx', \fulltreeobj)$, i.e. the I.H..
    
    \noindent \textbf{Case $\alpha = A: \commitBatch$:}
    
    \noindent Similar to the previous case $\excontracts{\run} = \excontracts{\run'}$, $\schecontracts{\run} = \schecontracts{\run'}$ remain unchanged for all $id$ and so $\run \IDrelation \partialExFamilyprime$. Also analogously, $\commitBatch$ is a global action, and $\environmentvec$ remains unchanged except $\environmentCreatedSecrets$, which is not accessed by the invariants, and they refer only to actions in $\funactions(\run)$ that already existed in $\funactions(\run')$.
    
    \noindent \textbf{Case $\alpha = \advCTLC$:}
    
    \noindent The same reasoning applies here, except that $\advCTLC$ is a local action that only affects one $\environmentch \in \environmentvec$ where $\environmentchCTLCAdvertised$ has been changed. 
    Since none of the invariants accesses this component, they follow immediately from I.H..
    
    \noindent \textbf{Case $\alpha = A : \authCTLC \, \advCTLCcontract$:}
    
    \noindent Similar to the previous cases we argue that $\run \IDrelation \partialExFamilyprime$ holds. Out of the nine invariants from which $Inv(\run, \partialEx', \fulltreeobj)$ is constructed, only $Inv_{\text{tree}}(\run, \partialEx', \fulltreeobj)$ and $Inv_{\text{setup}}(\run, \partialEx', \fulltreeobj)$ access this component for a $\environmentch \in \environmentvec$ and $id$ with $\advCTLCcontract \in \batch$. 
    
    To show $Inv_{\text{tree}}(\run, \partialEx', \fulltreeobj)$, we only need to consider the case $A = B$ because the invariant only considers authorizations by the honest user $B$.
    In this case, for $(B, \advCTLCcontract) \in \environmentchCTLCAuthorized \backslash \environmentchprimeCTLCAuthorized$ we have to show
    \[
        \forall \CTLCsubcontract \in \advCTLCcontract \, \exists \e \in \untree : \CTLCsubcontract = h(\e, id) .
    \]
    Since this authorization is given by honest user $B$, it cannot have been scheduled by the adversary and is hence determined by the honest user strategy (of $B$). 
    From \eqref{eq:enable-check} we see that $B$ will only execute this action if $\CTLCcontract$ is part of a batch determined by
    \[
        \treetoCTLCBadge(\fulltreeobj, \mathcal{S}). 
    \]
    Hence, \eqref{eq:treetoCTLCBadge} gives us the desired property. 

    The second affected invariant $Inv_{\text{setup}}(\run, \partialEx', \fulltreeobj)$ also follows from the honest user strategy, see \eqref{eq:honingoing}.
    This is because, again, we only need to consider the case $A = B$ and according to the honest user strategy of $B$, it only schedules an authorization action for a CTLC representing an outgoing edge if contracts corresponding to all ingoing edges have been enabled (this is specified in \eqref{eq:honingoing}). Accordingly, the left disjunct of the invariant's conclusion must hold since $\environmentvecCTLCEnabled = \environmentvecprimeCTLCEnabled$.

    For $Inv_{\text{auth}}(\run, \partialEx', \fulltreeobj)$ we argue that since the authorization $(B, \advCTLCcontract) \in \environmentvecCTLCAuthorized$ is given by $B$, the action to do so is determined by the honest user strategy. Based on \eqref{eq:enable-check}, we know that the action 
    $
         B: \authCTLC \, \advCTLCcontract
    $
    will only be scheduled before $t_0$. For this $newC(\e, R)$ (\ref{eq:CTLC-enable-check}) ensures that all conditions for also enabling this contract are met.
    
    \noindent \textbf{Case $\action = \enableCTLC ~ \CTLCcontract$: } 
    
    \noindent
    We first show that $\run \IDrelation \partialExFamilyprime$. 
    This trivially follows from the inductive hypothesis if we can show that for all $id$
    \begin{align*}    
    \excontracts{\run'} &= \excontracts{\run} \text{ and } \\
    \schecontracts{\run'} &= \schecontracts{\run} . 
    \end{align*}
    For $id$ with $\CTLCcontract \notin \batch$, this follows immediately from I.H. so we consider the specific $id$ with $\CTLCcontract \in \batch$.
    The first claim follows directly from the definition of $\excontracts{\run}$.
    To show 
    $$\schecontracts{\run'} = \schecontracts{\run}$$
    we assume towards contradiction that there exists 
    \begin{equation*}
    \scalebox{0.93}{$
    (\CTLCsubcontract, \funsecret_i(\CTLCsubcontract)) \in \schecontracts{R} \backslash \schecontracts{R'}. (*)
    $}
    \end{equation*}
    Then there is $\environmentvec^*$ such that  $\run \overset{ \action'}{\longrightarrow} \environmentvec^*$ for 
    $$\action' = \DecideCo (\CTLCcontract, \CTLCsubcontract, \funsecret_i(\CTLCsubcontract_{\iota}))$$
    for some $\CTLCcontract$ with $\funsender(\CTLCsubcontract) = B$. 
    If $\action' = \DecideCo (\CTLCcontract, \CTLCsubcontract, \funsecret_i(\CTLCsubcontract_{\iota}))$ then by~\Cref{ctlc:DecideCo}
    we know 
    that $\CTLCsubcontract \in \CTLCcontract \in \environmentchCTLCEnabled$ 
    and there is a $i$ such that $\funsecret_i(\CTLCsubcontract_{\iota}) \subseteq \environmentchprimeRevealedSecrets$. 
    From 
    $$Inv_{\text{tree}}(\run', \partialEx', \fulltreeobj) ,$$ we know that $\exists \e \in \untree : \CTLCsubcontract = h(\e, id)$.
    From the way that $h$ is defined (\Cref{def:hMapping}) we can conclude that $\funsender(\e) = B$.
    And that there is some $\e'$ such that $\funsender(\e) = \funsender(\e')$, $\funreceiver(\e) = \funreceiver(\e')$, $\fundepth(\e) = \fundepth(\e')$, 
    $h(\e', id) = h(e, id)$ and 
    $$\funsecret_i(\CTLCsubcontract) = h_{sec}(\e', id) .$$ 
    Consequently, $\funsecret(\e', id) \in \environmentchprimeRevealedSecrets$.
    From the inductive hypothesis, we know that 
    $$Inv_{\text{in-secrets}}(\run', \partialEx', \fulltreeobj)$$
     holds
    and consequently also $\e' \in \partialEx'$.
    Since the inductive hypothesis also gives that $\run' \IDrelation \partialExFamilyprime$, we get 
    \begin{equation*}
    \scalebox{0.93}{$ 
    (\CTLCsubcontract = h(\e, id) = h(\e', id), \funsecret_i(\CTLCsubcontract))  \in \schecontracts{\run'} ,
    $}
    \end{equation*}
    contradicting (*). 

    To show the invariants, we first analyze the environment changes induced by the $\enableCTLC$ action.
    
    From the inference rule of the $\enableCTLC$ action, we know that there is $\environmentch' \in \environmentvec'$ and $\environmentch \in \environmentvec$ with
    \begin{align*}
        \environmentchCTLCEnabled \supsetneq \environmentchprimeCTLCEnabled, \,
        \environmentchAvailableFunds \subsetneq \environmentchprimeAvailableFunds, \, \environmentchReservedFunds \supsetneq \environmentchprimeReservedFunds.
    \end{align*}
    Affected by these changes are the invariants 
    \begin{itemize}
        \item $Inv_{\text{levels}}(\run, \partialEx', \fulltreeobj)$,
        \item $Inv_{\text{liveness}}(\run, \partialEx', \fulltreeobj)$,
        \item $Inv_{\text{deposits}}(\run, \partialEx', \fulltreeobj)$, 
        \item $Inv_{\text{setup}}(\run, \partialEx', \fulltreeobj)$,
        \item $Inv_{\text{tree}}(\run, \partialEx', \fulltreeobj)$, and
        \item $Inv_{\text{auth}}(\run, \partialEx, \fulltreeobj)$.
    \end{itemize}
    By the semantics of $\enableCTLC$ there is a new $\CTLCcontract \in \environmentchCTLCEnabled \backslash \environmentchprimeCTLCEnabled$ for some channel $\channel$. 
    
    For $Inv_{\text{levels}}(\run, \partialEx', \fulltreeobj)$ we show towards contradiction that if $B = \funsender(\e)$ and $\funfund(\CTLCcontract) \in \environmentchReservedFunds$ for this $\CTLCcontract \in \environmentchCTLCEnabled$ with $h(\e, id) \in \CTLCcontract$ then
    $
        \environmentvec.t < t_0 + \fundepth(\e) \Delta
    $
    holds. 
    
    Note that from $h(\e, id) \in \CTLCcontract$ and $B = \funsender(\e)$, we also know that $\funsender(\CTLCcontract) = B$.
    The preconditions for executing $\enableCTLC$ in this case require that for some $\advCTLCcontract$ it holds that $(B, \advCTLCcontract) \in \environmentchCTLCAuthorized$ and for $s = \size{\advCTLCcontract}$ it holds that $\CTLCcontract = \{ \CTLCsubcontract_i ~|~ \CTLCsubcontract_i \in \advCTLCcontract ~\land~ i = s  \}$. 
    Consequently, we know that $h(\e, id) \hspace{-1pt} = \hspace{-1pt} \CTLCsubcontract_i$ and from $Inv_{\text{auth}}(\run', \partialEx', \fulltreeobj)$ we can conclude that $\enableCTLC ~ \CTLCcontract \in \Sigma_B^{\treeobj}(\run')$.
    By definition of $\Sigma_B$, this only is the case if $\environmentvecprimetime < t_0$. 
    So, the claim follows since $\environmentvecprimetime = \environmentvectime$.

    $Inv_{\text{liveness}}(\run, \partialEx', \fulltreeobj)$ follows from I.H. since $\partialEx$ has not been changed between $\run'$ and $\run$ and $\environmentchCTLCEnabled \supsetneq \environmentchprimeCTLCEnabled$ holds. 
    
    For the $Inv_{\text{deposits}}(\run, \partialEx', \fulltreeobj)$ invariant there is now exactly one new $\CTLCcontract \in \environmentchCTLCEnabled \cup \environmentchCTLCDecided$ which also includes exactly one $\CTLCsubcontract$. 
    By the preconditions of the $\enableCTLC$ inference rule \linebreak $\funfund(\CTLCcontract) \in \environmentchprimeReservedFunds = \environmentchReservedFunds$. 
    
    For $Inv_{\text{setup}}(\run, \partialEx', \fulltreeobj)$ there is now exactly one more $\CTLCcontract \in \environmentvecCTLCEnabled$ to consider.
    From the inference rule for $\enableCTLC$ (\Cref{ctlc:enableCTLC}), we can conclude that $(B, \advCTLCcontract) \in \environment'_{\channeledge{\eout}}.C_{\textit{aut}}$ for a $\advCTLCcontract \supseteq \CTLCcontract$. 
    Consequently, if $h(\eout, id) \in \CTLCcontract$ with $\eout \intree \untree$ as the invariant's precondition desires it, then also $h(\eout, id) \in \advCTLCcontract$. 
    We can hence use $Inv_{\text{setup}}(\run', \partialEx', \fulltreeobj)$ (given by the inductive hypothesis) to immediately show the invariant's conclusion (using the fact that $\environment'_{\channel}.C_{\textit{en}} \subseteq \environment'_{\channel}.C_{\textit{en}}$ for all channels $\channel$).

    For $Inv_{\text{tree}}(\run, \partialEx', \fulltreeobj)$ we look at $\environmentvecprime = \funlastEnv(\run')$ and $\environmentvec = \funlastEnv(\run)$. From the inference rule of \textit{\enableCTLC} (\ref{ctlc:enableCTLC}) we know $(B,\advCTLCcontract) \in \environmentvecprimeCTLCAuthorized$ in case $\action = \enableCTLC \, \CTLCcontract$. Therefore this invariant follows from $Inv_{\text{tree}}(\run', \partialEx', \fulltreeobj)$. 

    For $Inv_{\text{auth}}(\run, \partialEx', \fulltreeobj)$ we look at the case of $\CTLCcontract$ where 
    \[
        \enableCTLC \, \CTLCcontract \in \Sigma_B^{\treeobj}(\run') \, \land \, \enableCTLC \, \CTLCcontract \notin \Sigma_B^{\treeobj}(\run). 
    \]
    Here the $\enableCTLC$ rule (\ref{ctlc:enableCTLC}) implies $(B, \advCTLCcontract) \notin \environmentvecCTLCAuthorized$, so this invariant still holds. For the other cases, we still have $\enableCTLC$ in the honest user strategy since the set of enabled \CTLC{}s has only gotten larger, see 
    $newC(\e, R)$ (\ref{eq:CTLC-enable-check}).
    
    \noindent \textbf{Case $\alpha = A: \enableSubC ~ \CTLCsubcontract$:}
    
    \noindent Here $\run \IDrelation \partialExFamilyprime$ holds with analogous reasoning as in the previous case.
     Again, this is because enabling a new (sub)contract representing an ingoing edge $\e$ of \honestuser{} does not allow for claiming such a contract immediately. Instead, this option is only available after $\honestuser$ revealed the corresponding edge secret $\funsecret(\e, id)$, which honest user $\honestuser$ will only do after enabling (this connection is made formal by $Inv_{\text{in-secrets}}$).

    Also, similar to the previous case, all components of the environment remain untouched except $\environmentchprimeCTLCEnabled' \neq \environmentchCTLCEnabled$, which implies that only 
    \begin{itemize}
        \item $Inv_{\text{levels}}(\run, \partialEx', \fulltreeobj)$,
        \item $Inv_{\text{liveness}}(\run, \partialEx', \fulltreeobj)$,
        \item $Inv_{\text{deposits}}(\run, \partialEx', \fulltreeobj)$, 
        \item $Inv_{\text{setup}}(\run, \partialEx', \fulltreeobj)$,
        \item $Inv_{\text{tree}}(\run, \partialEx', \fulltreeobj)$, and
        \item $Inv_{\text{auth}}(\run, \partialEx', \fulltreeobj)$.
    \end{itemize}
    can be affected for the specific $id$ with $\CTLCsubcontract \in \advCTLCcontract \in \batch$. 

    For $Inv_{\text{levels}}$, we show towards contradiction that if $B = \funsender(\e)$ and $\funfund(\CTLCcontract) \in \environmentchReservedFunds$ for this $\CTLCcontract \in \environmentchCTLCEnabled$ with $h(\e, id) \in \CTLCcontract$ then
    $
        \environmentvec.t < t_0 + \fundepth(\e) \Delta
    $
    holds. 

    For all $h(\e, id) \in \CTLCcontract \in \environmentchprimeCTLCEnabled$, this immediately follows from the inductive hypothesis. 
    So we only need to consider the case where $h(\e, id) = \CTLCsubcontract$.
    Since the preconditions for executing $\enableSubC$ in this case require that 
    it holds that $\CTLCsubcontract \in \CTLCcontract$ and $\funsender(\CTLCcontract) = A$
    we know that $A = \funsender(\CTLCcontract) = \funsender(\e) = B$.  
    Consequently, we know that $B: \enableSubC ~ \CTLCsubcontract \in \Sigma_B^{\treeobj}(\run')$ (since only the honest strategy of $\honestuser$ can schedule such actions). 
    By definition of $\Sigma_B$, this only is the case if $\environmentvecprimetime < t_0$. 
    So, the claim follows since $\environmentvecprimetime = \environmentvectime$. 
    
    $Inv_{\text{liveness}}(\run, \partialEx', \fulltreeobj)$ follows from the fact that there are more subcontracts in $\environmentchCTLCEnabled $ than in $\environmentchprimeCTLCEnabled$ and that $\partialEx'$ remained unchanged.

    For $Inv_{\text{deposits}}(\run, \partialEx', \fulltreeobj)$ we look at its preconditions: The $\enableSubC$ action adds a sub-contract to an already enabled \CTLC{} and thus
    \begin{align*}
        & \exists \CTLCcontract \in \environmentchprimeCTLCEnabled \cup \environmentchprimeCTLCDecided : \exists \CTLCsubcontract \in \CTLCcontract , \e \in \untree : \CTLCsubcontract = h(\e, id) \\
        & \Leftrightarrow \\
        & \exists \CTLCcontract \in \environmentchCTLCEnabled \cup \environmentchCTLCDecided : \exists \CTLCsubcontract \in \CTLCcontract , \e \in \untree : \CTLCsubcontract = h(\e, id). \\
    \end{align*}
    Additionally, $\environmentchprimeReservedFunds = \environmentchReservedFunds$ and so $$Inv_{\text{deposits}}(\run, \partialEx', \fulltreeobj)$$ is implied by I.H..

    For $Inv_{\text{setup}}(\run, \partialEx', \fulltreeobj)$ there is now exactly one more $\CTLCsubcontract \in \CTLCcontract \in \environmentvecCTLCEnabled$ to consider. 
    Consequently, we need to consider that $\CTLCsubcontract = h(\eout, id) \in \CTLCcontract$ with $\eout = (B, X)_{\walk} \intree \untree$ and show that
    the invariant's conclusion holds for this $\CTLCsubcontract$.
    We show that indeed for all $\ein = (Y,B)_{\walk'} \in \untree$ with $\walk' = \concatvec{\unitvec{(Y,B)_{\walk'}}}{\walk}$ 
    there is  $(\CTLCcontract)' \in \environment_{\channeledge{\ein}}.C_{en}$ such that  $h(\ein, id) \in (\CTLCcontract)'$ (left disjunct of the conclusion).
    Since the preconditions for executing $\enableSubC$ in this case require that 
    it holds that $\CTLCsubcontract \in \CTLCcontract$ and $\funsender(\CTLCcontract) = A$
    we know that $A = \funsender(\CTLCcontract) = \funsender(\eout) = B$. 
    Consequently, we know that $B: \enableSubC ~ \CTLCsubcontract \in \Sigma_B^{\treeobj}(\run')$ (since only the honest strategy of $\honestuser$ can schedule such actions). 
    By definition of $\Sigma_B$, this is only the case if $\ingoing(\eout, \run') = 1$. 
    This implies that there is some $(\CTLCcontract)' \in \environment'_{\channeledge{\ein}}.C_{en}$ such that $h(\ein, id) \in (\CTLCcontract)'$.
    Since $\environmentchCTLCEnabled$ contains strictly more subcontracts than $\environmentchprimeCTLCEnabled$ it is also ensured that there is some 
    $(\CTLCcontract)^* \in \environment_{\channeledge{\ein}}.C_{en}$ such that $h(\ein, id) \in (\CTLCcontract)^*$, concluding the case.

    $Inv_{\text{tree}}(\run, \partialEx', \fulltreeobj)$ holds since if for the newly enabled subcontract we have 
    $$\CTLCcontract \in \environmentvecCTLCEnabled \land B \in \funusers(\CTLCcontract)$$
    the $\authCTLC$ rule (\ref{ctlc:authCTLC}) implies $(B, \advCTLCcontract) \in \environmentvecprimeCTLCAuthorized$. 
    For this \linebreak $Inv_{\text{tree}}(\run', \partialEx', \fulltreeobj)$ applies for all $\CTLCsubcontract \in \advCTLCcontract$. 
    By \linebreak \cref{lemma:enabledsubsetofadvertised} we have that $\CTLCcontract \subseteq \advCTLCcontract$ always holds and so \linebreak $Inv_{\text{tree}}(\run, \partialEx', \fulltreeobj)$ is implied. 
    
    $Inv_{\text{auth}}(\run, \partialEx', \fulltreeobj)$ holds with analogous reasoning as in the previous case $\action = \enableCTLC ~ \CTLCcontract$. 

    \noindent \textbf{Case $\alpha = A: \, \revealSecretch \, s_{\walk}^{id}$:}
    
    \noindent For the specific $id$ of $s_{\walk}^{id}$ we define 
    \begin{equation} 
        \partialEx := \partialEx' \cup \{ \e \in \hatpartialEx \mid \exists \environmentvec^* : \run \overset{\action'}{\longrightarrow} \environmentvec^* \} \label{eq:somehelperequationhere}
    \end{equation}
    with 
    \begin{align*}
    &\alpha' = \DecideCo (\CTLCcontract , h(\e, id), h_{sec}(\e, id)), \\
    &s_{\walk}^{id} \in h_{sec}(\e, id), \funreceiver(\e) = \honestuser
    \end{align*}
    (for the definition of $h_{sec}$ see (\ref{def:hsecMapping}))
    and 
    \begin{align*}
        \partialExFamily := ( \partialExFamilyprime \backslash \{ \partialEx' \} ) \cup \{ \partialEx \}
    \end{align*}
    (only replacing this specific $\partialEx'$ for the $id$ of $s_{\walk}^{id}$)
    and show that
    \begin{enumerate}[a\hspace{1pt}]
        \item \hspace{-5pt}) $\run \IDrelation \partialExFamily$,
        \item \hspace{-5pt}) $\partialEx$ is consistent for all $id$ and
        \item \hspace{-5pt}) $Inv(\run, \partialEx, \fulltreeobj)$ holds.
    \end{enumerate}

    \noindent a) To show $\run \IDrelation \partialExFamily$ we start by stating that 
    \begin{align}
        \schecontracts{\run'} &\subseteq \schecontracts{\run}, \label{eq:helperstuff} \\
        \excontracts{\run'} &= \excontracts{\run}
    \end{align}
    holds. Revealing $s_{\walk}^{id}$ can allow for a new $\action'$ in $\sactions{\run}$. Having less possible \textit{claim} actions for $\partialEx$ is not possible since revealing a secret cannot remove the possibility of claiming a contract. This action moves a secret from a $\environmentchCreatedSecrets$ to $\environmentchRevealedSecrets$. 
    We now show the two parts (\ref{eq:IDrelation1}) and (\ref{eq:IDrelation2}) of $\run \IDrelation \partialExFamily$ individually.
    Nothing changes for all $id$ with $\partialEx = \partialEx'$, so we look at the one specific $id$ with $\partialEx \neq \partialEx'$. 

    (\ref{eq:IDrelation1}) Let $(\CTLCsubcontract, \funsecret_i(\CTLCsubcontract)) \in \schecontracts{\run}$ then by definition of $\schecontracts{\run}$ either $(\CTLCsubcontract, \funsecret_i(\CTLCsubcontract)) \in \schecontracts{\run'}$ or the most recent $\revealSecret$ action has made such an $\action'$ available with $\CTLCsubcontract = h(\e, id)$ and $\funsecret_i(\CTLCsubcontract) = h_{sec}(\e, id)$. 
    Further, since $\funreceiver(\e) = \honestuser$ and $\CTLCsubcontract = h(\e, id)$, we know that $\funreceiver(\CTLCsubcontract) = B$.
    Therefore $Inv_{\text{tree}}(\run', \partialEx, \fulltreeobj)$ applies and we get that there is some $\e$ such that $h(\e, id) = \CTLCsubcontract$.
    So $\CTLCsubcontract$ with $\funsecret_i(\CTLCsubcontract) = h_{sec}(\e, id)$ is the only possible element in 
    $$\schecontracts{\run} \backslash \schecontracts{\run'}.$$
    
    If $(\CTLCsubcontract, \funsecret_i(\CTLCsubcontract)) \in \schecontracts{\run'}$ then by I.H. also 
    $$\exists e \in \partialEx': h(\e, id) = \CTLCsubcontract \land h_{sec}(\e, id) = \funsecret_i(\CTLCsubcontract)$$ 
    and thus $\e \in \partialEx$. 
    If $\CTLCsubcontract \in \schecontracts{\run} \backslash \schecontracts{\run'}$ the above reasoning applies and gives us 
    \[
    \exists \e \in \partialEx : \, h(\e, id) = \CTLCsubcontract \land h_{sec}(\e, id) = \funsecret_i(\CTLCsubcontract).
    \]
    (\ref{eq:IDrelation2}) Let $\e \in \partialEx$. Then either $\e \in \partialEx'$ or $\e \not \in \partialEx'$ and 
    \begin{equation} \label{eq:proof1}
        \e \in \{ \e \in \hatpartialEx \mid \exists \environmentvec^* : \run \overset{\action'}{\longrightarrow} \environmentvec^* \} .
    \end{equation}
    If $\e \in \partialEx'$ then by I.H. also 
    \begin{equation*}
    \scalebox{0.88}{$
    \begin{aligned}
        (h(\e, id), h_{sec}(\e, id)) \in \schecontracts{\run'} ~\cup~\excontracts{\run'} \\
        \subseteq \schecontracts{\run} ~\cup~ \excontracts{\run}.
    \end{aligned}
    $}
    \end{equation*}
    In the case of \eqref{eq:proof1} we have 
    $$(h(\e, id), h_{sec}(\e, id)) \in \schecontracts{\run}$$ 
    by definition of $\schecontracts{\run}$ and \cref{eq:somehelperequationhere} given that $\funreceiver(\e) = \honestuser$.
    \newline
    
    \noindent b) Again, nothing changes for all $id$ with $\partialEx = \partialEx'$, so we look at the one specific $id$ with $\partialEx \neq \partialEx'$. To show that $\partialEx$ is consistent we assume towards contradiction that $\partialEx$ is not consistent meaning that there exists some $\e \in \partialEx$  and a predecessor $\e^{*} \in \funonPathtoRoot(\untree, \e) \cap \hatpartialEx$ such that $\e^{*} \notin \partialEx$ or
    \[
        \exists \e^{**} \in \partialEx : \spec{\e^{**}} = \spec{\e}.
    \]
    For the first case, we distinguish whether $e^*$ is an ingoing or an outgoing edge. 

    Let $e^* = (Y, B)_{\walk}$ be an ingoing edge. By assumption $e^* \notin \partialEx$ and by $Inv_{\text{in-secrets}}$ we know that $\funsecret(e^*) \nsubseteq \environmentvecRevealedSecrets$ holds. This contradicts that \mbox{$e \in \partialEx \land e \notin \partialEx'$} since for that it would need to hold that it exists a $\environmentvec'$ and $\CTLCsubcontract= h(\e, id)$ s.t. 
    $$\environmentvec \overset{\DecideCo(\CTLCcontract, h(\e, id), h_{sec}(\e, id))}{\longrightarrow} \environmentvec'$$
    is possible, implying $\funsecret(\e) \subseteq \environment_{\channeledge{\e}}.S_{rev}$. However, if $\e^*$ is a predecessor of $\e$ we have $ h_{sec}(\e^*, id) \subseteq  h_{sec}(\e, id)$, see (\ref{def:hsecMapping}). 

    Let $\e^* = (B, X)_{\walk}$ be an outgoing edge. Then there must be a corresponding ingoing edge $\e' = (Y,B)_{\walk'}$ with $\walk' = \concatvec{[(Y,B)]}{\walk}$. Then either $\e' \notin \partialEx$ and the previous reasoning applies or $\e' \in \partialEx$. In that case, we know from $Inv_{\text{secrets}}$ that $\funsecret(\e', id) \in \environment_{\channeledge{\e'}}.S_{\textit{rev}}$ and so from $Inv_{\text{in-schedule}}$ that \linebreak $\DecideCo( \CTLCcontract, h(e^*, id), h_{sec}(e^*, id)) \in \actions{\run}$ for $h(e^*, id) \in \CTLCcontract$. However, this implies 
    $$(h(e^*, id), h_{sec}(e^*, id)) \in \excontracts{R}$$
    and because of $\run \IDrelation \partialExFamily$ from point a), also $e^* \in \partialEx$. 
    This is a contradiction. 

    For the second case, we first show that for $\e \in \partialEx \backslash \partialEx'$, we have that $s_{\walk}^{id} = \funsecret(\e, id)$.
    If $s_{\walk}^{id} \neq \funsecret(\e, id)$ then $\funsecret(\e, id) \in \environmentchprimeRevealedSecrets$ 
    because for \linebreak $\run \overset{\action'}{\longrightarrow} \environmentvec^*$, it needs to hold that 
    for all $s \in h_{sec}(\e, id)$, $s \in \environmentchRevealedSecrets$, and so in particular, $\funsecret(\e, id) \in \environmentchRevealedSecrets$ (by \Cref{ctlc:DecideCo}). 
    However, since $\environmentchRevealedSecrets = \environmentchprimeRevealedSecrets \cup \{ s_{\walk}^{id} \} $ (by \Cref{ctlc:revealSecret}), if $s_{\walk}^{id} \neq \funsecret(\e, id)$ then it must hold that \linebreak 
    $\funsecret(\e, id) \in \environmentchprimeRevealedSecrets$. 
    In this case, however, by \linebreak $Inv_{\textit{in-secrets}}(\run', \partialEx', (id, \untree, t_0, \specalone))$, we also know that $\e \in \partialEx'$.
    By construction, we know that 
    $$\funowner(\funsecret(\e, id)) = \funreceiver(\e) = \honestuser$$
    and so by \Cref{ctlc:revealSecret}, we can conclude that $A = \honestuser$.
    Consequently, we know that $B: \, \revealSecretch \, s_{\walk}^{id} \in \Sigma_B^{\treeobj}(\run')$ (since only the honest strategy of $\honestuser$ can schedule such actions). 
    By definition of $\Sigma_B$, this only is the case if $\funnodupl (\e, \run')$ holds.
    If now there would be an edge $\e^{**} \in \partialEx : \spec{\e^{**}} = \spec{\e}$ 
    then $\funsender(e) = \funsender(e^{**})$ and $\funreceiver(\e) = \funreceiver(e^{**})$ which would immediately contradict $\funnodupl (\e, \run)$.
    
    \noindent c) To show that $Inv(\run, \partialEx, \fulltreeobj )$ holds, we go through the invariants individually. In the following $\environmentvec$ will always notate $\environmentvec = \funlastEnv(\run)$.
    Again, nothing changes for all $id$ with $\partialEx = \partialEx'$, so we look at the one specific $id$ with $\partialEx \neq \partialEx'$.

    \noindent \underline{$Inv_{\text{in-secrets}}(\run, \partialEx, \fulltreeobj)$}
    Let $\ein = (Z,B)_{\walk} \in \untree$ for some $Z$ and \mbox{$\funsecret(\ein, id) \in \environmentvecRevealedSecrets$}. 
    By the $\revealSecret$ rule we know that $\environmentchRevealedSecrets = \environmentchprimeRevealedSecrets \cup \{ s^{id}_{\walk} \}$. If $\channel \neq \channeledge{\ein}$ this invariant remains unaffected so let $\channel = \channeledge{\ein}$.  
    Either $\funsecret(\ein, id) \in \environmentchprimeRevealedSecrets$ or $s^{id}_{\walk} = \funsecret(\ein, id)$.
    In the first case, we know by the I.H. that $\ein \in \partialEx'$ and thus $\ein \in \partialEx$.
    
    In the second case, we know $A=B$ since $B = \funreceiver(\ein)$ and \eqref{def:edgesecret} from the honest user strategy. 
    Consequentially $\revealSecretch \, s^{id}_{\walk} \in \Bstrategy(\run')$. By the definition of the honest user strategy, this is only the case if $\chContract(\ein, \run') = 2$, see \linebreak \eqref{eq:CTLC-execution}. 
    We show that this implies all the preconditions for executing $\run \overset{\action'}{\longrightarrow} \environmentvec^*$ according to~\Cref{ctlc:DecideCo}. 
    $\chContract(\ein, \run') = 2$ ensures that the subcontract is enabled (there is some $\CTLCcontract \in \environmentchCTLCEnabled$ such that $h(\e, id)\in \CTLCcontract$) and that the secrets of all edges on the path to the root (but $h(\e, id)$) are available in $\environmentchprimeRevealedSecrets$, and so all relevant secrets are available $\environmentchRevealedSecrets$.
    From this, we can also conclude by \cref{lemma:enabledsubsetofadvertised} that a corresponding contract has been advertised, and (using $Inv_\textit{deposits}$) that the corresponding fund is reserved.
    We are hence left to show that $h(\e, id)$ is the top-level contract in $\CTLCcontract$.
    $\chContract(\ein, \run') = 2$ also gives us that there is no $\CTLCsubcontract{}' \in \CTLCcontract$ with a timeout smaller than the one of $h(\e, id)$. 
    By construction, this means that $h(\e, id)$ is the top-level contract in $\CTLCcontract$. 
    Correspondingly $\run \overset{\action'}{\longrightarrow} \environmentvec^*$ holds by \Cref{ctlc:DecideCo}.

    \noindent \underline{$Inv_{\text{secrets}}(\run, \partialEx, \fulltreeobj)$}
    Let $\e, \e'$ such that \linebreak $\e \in \partialEx$ and $\e' \in \funonPathtoRoot(\untree, \e)$. We need to show that 
    $$\funsecret(\e', id) \in \environment_{\channeledge{\e}}.S_{rev}.$$
    
    If $\e \in \partialEx'$ this holds by I.H., because once a secret is in \linebreak $\environment_{\channeledge{\e}}.S_{rev}$, it cannot be removed from there. 

    If $e \in \partialEx \backslash \partialEx'$ then we know by (\ref{eq:somehelperequationhere}) that 
    $\exists \environmentvec^* : \run \overset{\action'}{\longrightarrow} \environmentvec^*$
    for $\alpha' = \DecideCo (\CTLCcontract , h(\e, id), h_{sec}(\e, id)), s_{\walk}^{id} \in h_{sec}(\e, id)$. 
    Consequently, by the inference rule for claim~\Cref{ctlc:DecideCo}, we know that
    $$h_{sec}(\e, id) \subseteq \environment_{\channeledge{\e}}.S_{rev} .$$
     By definition of $h_{sec}$(\ref{def:hsecMapping}) we have
    \[
        \forall e' \in \funonPathtoRoot(\untree, \e): \funsecret(e', id) \in h_{sec}(\e, id)
    \]
    and so $\funsecret(e', id) \in \environment_{\channeledge{\e}}.S_{rev}$. 
    
    \noindent \underline{$Inv_{\text{in-schedule}}(\run, \partialEx, \fulltreeobj)$}
    Let $Z, \ein$ s.t. $\ein = (Z,B)_{\walk} \in \untree, \fundepth(\ein) > 1,  \funsecret(\ein) \in \environmentvecRevealedSecrets$ and let $\eout, X$ s.t. $\eout \in \untree, \eout = (B,X)_{\walk'}$, $\walk = \concatvec{\unitvec{\ein}}{\walk'}$. We need to show that 
    \begin{equation} \label{eq:somehleperthinkshereandthere}
        \DecideCo( \CTLCcontract,  h( \eout, id ), h_{sec}( \eout, id ) ) \in \funactions(\run).
    \end{equation}
    If $\funsecret(\ein, id) \in \environmentvecprimeRevealedSecrets$, then the claim immediately follows from I.H. since $\funactions(\run') \subseteq\funactions(\run)$.
    If $\funsecret(\ein, id) \notin \environmentvecprimeRevealedSecrets$ then $\funsecret(\ein, id) = s_{\walk}^{id}$ with $\funowner(s_{\walk}^{id}) = B$. 
    Therefore, this action is determined by the honest user strategy and hence \eqref{eq:CTLC-execution} and \eqref{eq:CTLC-chContract} imply \eqref{eq:somehleperthinkshereandthere}.

    \noindent \underline{$Inv_{\text{levels}}(\run, \partialEx, \fulltreeobj)$}
    Let $\e \in \untree, B = \funsender(\e),  \linebreak \environmentvectime > t_0 + \fundepth(\e) \Delta.$
    Since $\environmentvectime = \environmentvecprimetime$, $\environmentvecCTLCEnabled = \environmentvecprimeCTLCEnabled$ and $\environmentvecReservedFunds = \environmentvecprimeReservedFunds$ we get  
    $Inv_{\text{levels}}(\run, \partialEx, \fulltreeobj)$
    directly from \linebreak $Inv_{\text{levels}}(\run', \partialEx', \fulltreeobj)$. 
    
    \noindent \underline{$Inv_{\text{liveness}}(\run, \partialEx, \fulltreeobj)$}
    Let $X, \eout = (B,X)_{\walk} \in \untree$ and $j := \fundepth(\eout)$, $\DecideCo(\CTLCcontract, h(\eout, id), h_{sec}(\eout, id) \in \actions{\run}$.
    We need to show that for every $\ein (Z,B)_{\walk'} \in \untree$ with \newline $\walk' = \concatvec{\unitvec{(Y,B)_{\walk}}}{\walk}$:
    \begin{align}
        &\exists j' \leq j+1, \walk'' : (Y,B)_{\walk''} \in \partialEx \land \fundepth((Y,B)_{\walk''}) = j'  \label{eq:helpinliveness1} \\
        & \text{or} \nonumber \\
        &\exists \CTLCcontract \in \environment_{\channeledge{\ein}}.C_{en} : \label{eq:helpinliveness2}\\
        & \qquad h(\ein, id) \in \CTLCcontract \land \environmentvectime < t_0 + (j+1)\Delta \nonumber
    \end{align}
    By I.H. we know that \eqref{eq:helpinliveness1} or \eqref{eq:helpinliveness2} holds for $\run'$ with $\environmentvecprime$ and $\partialEx'$. In case $\eqref{eq:helpinliveness1}$ holds for $\run'$ it also holds for $\run$ and $\partialEx$ since
    \begin{align*}
        &\DecideCo{(\CTLCcontract, h(\eout, id), \funsecret_i(\CTLCsubcontract_{\iota}))} \in \actions{\run} \\
        &\Rightarrow \DecideCo{(\CTLCcontract, h(\eout, id), \funsecret_i(\CTLCsubcontract_{\iota}))} \in \actions{\run'}
    \end{align*}
    and $\partialEx' \subseteq \partialEx$.
    In case \eqref{eq:helpinliveness2} holds for $\run'$ and $\partialEx'$ it also holds for $\run$ and $\partialEx$ because $\environmentvecCTLCEnabled = \environmentvecprimeCTLCEnabled$, and $\environmentvectime = \environmentvecprimetime$. 

    \noindent \underline{$Inv_{\text{init-liveness}}(\run, \partialEx, \fulltreeobj)$} Assume $\e \in \partialEx \backslash \partialEx'$, as for all others the invariant is implied directly by I.H.. Since $\action'$ from (\ref{eq:somehelperequationhere}) is not yet in $\actions{\run}$ 
    \[
        \nexists \CTLCcontract : \DecideCo(\CTLCcontract, h(\e, id), h_{sec}(\e, id)) \in \funactions(\run) .
    \]
    Since $\funreceiver(\e) = B $ we also have $\funowner(s^{id}_{\walk}) = B$ by the $h_{sec}$ function (\ref{def:hsecMapping}). Therefore $\action \in \Sigma_B^{\treeobj}(\run')$ which implies 
    $$\environmentvectime = \environmentvecprimetime < \timeoutsc(h(\e, id)),$$
    see (\ref{eq:CTLC-execution}).
    
    \noindent \underline{$Inv_{\text{deposits}}(\run, \partialEx, \fulltreeobj)$}
    For all $\environmentch \in \environmentvec$ we have 
    \[
        \environmentchCTLCEnabled \cup \environmentchCTLCDecided = \environmentch'.C_{\textit{en}} \cup \environmentch'.C_{\textit{cla}} .
    \]
    Furthermore $\untree$ stays constant throughout and for all $\environmentch \in \environmentvec$ it holds $\environmentchReservedFunds = \environmentch'.F_{\textit{res}}$, thus $Inv_{\text{deposits}}(\run, \partialEx, \fulltreeobj)$ follows directly from I.H..
    
    \noindent \underline{$Inv_{\text{setup}}(\run, \partialEx, \fulltreeobj)$} \linebreak
    Let $X, Y, \eout = (B,X)_{\walk} \in \untree, \CTLCcontract, h(\eout, id) \in \CTLCcontract \in \environment_{\channeledge{\eout}}.C_{\textit{en}}$ or $(B, \advCTLCcontract) \in \environment_{\channeledge{\eout}}.C_{\textit{aut}} \, , h(\eout, id) \in \advCTLCcontract$. We need to show that $\forall \ein = (Y,B)_{\walk'} \in \untree \text{ with } \walk' = \concatvec{\unitvec{(Y,B)_{\walk}}}{\walk}:$
    \begin{align*} 
        &\exists (\CTLCcontract)' \in \environment_{\channeledge{\ein}}.C_{en} :  h(\ein, id) \in (\CTLCcontract)' \\ 
        \lor \, &\exists j' \leq \fundepth(\ein) : (Y,B)_{\walk''} \in \partialEx \\
        & \qquad \land \fundepth((Y,B)_{\walk''}) = j' .
    \end{align*}
    By I.H., we know this holds for $\run'$ with $\environmentvecprime$ and $\partialEx'$. Additionally 
    \begin{align*}
        \environment_{\channeledge{\ein}}.C_{en} &= \environment'_{\channeledge{\ein}}.C_{en}, \\
        \environment_{\channeledge{\eout}}.C_{\textit{aut}} &= \environment'_{\channeledge{\eout}}.C_{\textit{aut}}, 
    \end{align*}
     and $\partialEx' \subseteq \partialEx$, thus it also holds for $\run$ with $\environmentvec$ and $\partialEx$. 

    \noindent $Inv_{\text{tree}}(\run, \partialEx, \fulltreeobj)$ remains unaffected since 
    \begin{align*}
        \environmentvecCTLCAuthorized &= \environmentvecprimeCTLCAuthorized, \, \environmentvecCTLCEnabled = \environmentvecprimeCTLCEnabled, \\
        \environmentvecCTLCDecided &= \environmentvecprimeCTLCDecided. 
    \end{align*}

    \noindent \underline{$Inv_{\text{auth}}(\run, \partialEx, \fulltreeobj)$}
    remains unaffected since the honest user strategy does not rely on revealed secrets for its decision to enable a contract, see
    $newC(\e, R)$ (\ref{eq:CTLC-enable-check}).

    \noindent \textbf{Case $\action = A: \shareSecretch s^{id}_{\walk}$:}

    \noindent According to the inference rule of $\shareSecretch$ only $\, \environmentchprime$ is effected by this action, all other elements of $\environmentvecprime$ remain untouched, especially $\environment'_{\channel'}$ with $\channel' \neq \channel$. Furthermore, $\environmentchprimeRevealedSecrets \subseteq \environmentchRevealedSecrets$. 
    For showing $\run \IDrelation \partialExFamilyprime$ we first notice that 
    \begin{align*}
        \excontracts{R'} &= \excontracts{R}, \\
        \schecontracts{R'} & \subseteq \schecontracts{R}
    \end{align*}
    for the $id$ of $s^{id}_{\walk}$. In the following, we look at this specific $id$ only, since nothing changes for the others. 
    Assume towards contradiction that there is a 
    $$(\CTLCsubcontract, \funsecret_i(\CTLCsubcontract)) \in \schecontracts{R} \backslash \schecontracts{R'}.$$ 
    Then, there needs to be a 
    \begin{equation*}
    \scalebox{0.93}{$
    \action' = \DecideCo(\CTLCcontract,\CTLCsubcontract, \funsecret_i(\CTLCsubcontract)) \in \sactions{\run} \backslash \sactions{\run'} ,
    $}
    \end{equation*}
    so $\exists \environmentvec^* : \run \overset{\action'}{\longrightarrow} \environmentvec^*$. 
    By the inference rule of $\DecideCo$ (\Cref{ctlc:DecideCo}), this implies that $\CTLCsubcontract \in \CTLCcontract \in \environment_{\channel'}.S_{\textit{en}}$ and $\funsecret_i(\CTLCsubcontract) \subseteq \environment_{\channel'}.S_{\textit{rev}}$.
    Since $B = \funreceiver(\CTLCsubcontract)$ the invariant $Inv_{\text{tree}}(\run', \partialEx', \fulltreeobj)$ applies and hence there is a $\e \in \untree$ with $\CTLCsubcontract = h(\e, id)$ and hence by construction of $h$ there is an edge $\e' \in \untree$ such that $\CTLCsubcontract = h(\e', id)$ and $h_{sec}(\e', id) = \funsecret_i(\CTLCsubcontract)$ and $B = \funreceiver(\e')$. 
    Consequently, $\funsecret(\e', id) \in \environment_{\channel'}.S_{\textit{rev}}$.
    But then also $\funsecret(\e', id) \in \environment'_{\channel^*}.S_{\textit{rev}}$ for some $\channel^*$
    because either $\funsecret(\e', id) = s^{id}_{\walk}$ and then by \Cref{ctlc:shareSecret} $\funsecret(\e', id) \in \environment'_{\channel}.S_{\textit{rev}}$ or $\funsecret(\e', id) \neq s^{id}_{\walk}$ and then $\funsecret(\e', id) \in \environment'_{\channel'}.S_{\textit{rev}}$ 
    (since by \Cref{ctlc:shareSecret} $ \environment_{\channel'}.S_{\textit{rev}} =  \environment'_{\channel'}.S_{\textit{rev}} ~\cup~ \{ s^{id}_{\walk} \}$). 
    But then $Inv_{\text{in-secrets}}(\run', \partialEx', \fulltreeobj)$ implies $\e' \in \partialEx'$ and hence by $\run' \IDrelation \partialEx'$ we know that \linebreak $(\CTLCsubcontract, \funsecret_i(\CTLCsubcontract)) \in \schecontracts{\run'}$ (contradicting our initial assumption).

    We now go through the invariants individually. 
    Since $\environmentvecRevealedSecrets = \environmentvecprimeRevealedSecrets$ (so the set of revealed secrets over all channels stays constant) the invariant $Inv_{\text{in-secrets}}(\run, \partialEx', \fulltreeobj) $ is not affected. $Inv_{\text{secrets}}(\run, \partialEx', \fulltreeobj)$ is also not affected since \linebreak $\shareSecretch$ does not remove secrets from channels. For \linebreak $Inv_{\text{in-schedule}}(\run, \partialEx', \fulltreeobj)$ the same argument applies as for $Inv_{\text{in-secrets}}(\run, \partialEx', \fulltreeobj).$ The remaining invariants do not depend on any $\environmentchRevealedSecrets$ and so are implied directly by I.H..

    $Inv_{\text{auth}}(\run', \partialEx', \fulltreeobj)$
    remains unaffected since the honest user strategy does not rely on revealed secrets for its decision to enable a contract, see
    $newC(\e, R)$ (\ref{eq:CTLC-enable-check}).
    
    \noindent \textbf{Case $\action = \timeoutsc (\dotCTLCcontract, \dotCTLCsubcontract)$:} 

    \noindent We show that $\run  \IDrelation \partialExFamilyprime$. 
    Since no $\DecideCo$ or $\CoEx$ action is executed it holds immediately that
     $$
        \excontracts{\run} = \excontracts{\run'}
     $$
    for all $id$.
    We additionally show for all $id$ that 
    $$
        \schecontracts{\run} = \schecontracts{\run'}.
    $$
    For all $id$ with $\widehat{\dotCTLCcontract} \notin \batch$ it holds trivially, for the one $id$ with $\widehat{\dotCTLCcontract} \in \batch$
    we assume towards contradiction that 
    $$ \schecontracts{\run} \neq \schecontracts{\run'}$$
     and consider the cases 
    $$(\CTLCsubcontract, \funsecret_i(\CTLCsubcontract)) \in  \schecontracts{\run} \backslash  \schecontracts{\run'}$$ 
    and 
    $$(\CTLCsubcontract, \funsecret_i(\CTLCsubcontract)) \in  \schecontracts{\run'} \backslash  \schecontracts{\run} .$$

    If $(\CTLCsubcontract, \funsecret_i(\CTLCsubcontract)) \in  \schecontracts{\run} \backslash  \schecontracts{\run'}$ then (since no new $\DecideCo$ or $\CoEx$ action is executed)
    this means that $\run \overset{\action'}{\longrightarrow} \environmentvec^*$ for some $\environmentvec^*$ and $\action' = \DecideCo(\CTLCcontract, \CTLCsubcontract, \funsecret_i(\CTLCsubcontract))$ and $\funreceiver(\CTLCsubcontract) = \honestuser$.
    Since by the inference rule for $\DecideCo$ (\Cref{ctlc:DecideCo}) we know that $\CTLCsubcontract \in \CTLCcontract \in \environmentchCTLCEnabled$ 
    for some channel $\channel$ we can conclude using $Inv_\textit{tree}$ that there is some $\e \in \untree$ such that $\CTLCsubcontract = h(\e, id)$ and hence by construction also that there is some $\e' \in \untree$ with $h(\e', id) = \CTLCsubcontract$ and $h_{sec}(\e', id) = \funsecret_i(\CTLCsubcontract)$ and $\funreceiver(\CTLCsubcontract) = \funreceiver(\e) = \funreceiver(\e') = \honestuser$.
    If $\run \overset{\action'}{\longrightarrow} \environmentvec^*$ then by the inference rule for $\DecideCo$ (\Cref{ctlc:DecideCo}), we know that 
    $\funsecret_i(\CTLCsubcontract) \subseteq \environmentvecRevealedSecrets = \environmentvecprimeRevealedSecrets$ and so also $\funsecret(\e', id) \in \environmentvecRevealedSecrets$ since by construction $\funsecret(\e', id) \in h_{sec}(\e', id) = \funsecret_i(\CTLCsubcontract)$. 
    This allows us to conclude using $Inv_\textit{in-secrets}$ (for $\run'$ and $\partialEx'$) that $\e' \in \partialEx'$ and so by $\run'  \IDrelation \partialExFamilyprime$ that 
    $$(h(\e', id), h_{sec}(\e', id)) \in \schecontracts{\run'},$$
    contradicting the original assumption. 

    If $(\CTLCsubcontract, \funsecret_i(\CTLCsubcontract)) \in  \schecontracts{\run'} \backslash  \schecontracts{\run}$ then this means that $\run' \overset{\action'}{\longrightarrow} \environmentvec^*$ for some $\environmentvec^*$ and 
    $$\action' = \DecideCo(\CTLCcontract, \CTLCsubcontract, \funsecret_i(\CTLCsubcontract))$$
    and $\funreceiver(\CTLCsubcontract) = \honestuser$. 
    Since by the inference rule for $\DecideCo$ (\Cref{ctlc:DecideCo}) we know that $\CTLCsubcontract \in \CTLCcontract \in \environmentchprimeCTLCEnabled$ 
    for some channel $\channel$ we can conclude using $Inv_\textit{tree}$ that there is some $\e \in \untree$ such that $\CTLCsubcontract = h(\e, id)$ and hence by construction also that there is some $\e' \in \untree$ with $h(\e', id) = \CTLCsubcontract$ and $h_{sec}(\e', id) = \funsecret_i(\CTLCsubcontract)$ and $\funreceiver(\CTLCsubcontract) = \funreceiver(\e) = \funreceiver(\e') = \honestuser$.
    From $\run'  \IDrelation \partialExFamilyprime$ we hence also know that $\e' \in \partialEx'$ and consequently using \linebreak
    $Inv_{\text{init-liveness}}(\run', \partialEx', \fulltreeobj)$
    that either 
    $$\DecideCo(\CTLCcontract, h(\e', id), h_{sec}(\e', id)) \in \funactions(\run')$$ 
    or $\environmentvecprimetime < t_0 + \fundepth(\e') \Delta $.
    The first case gives us that 
    $$(h(\e', id), h_{sec}(\e', id)) \in \schecontracts{\run} ,$$ contradicting the original assumption.
    The second case immediately contradicts the precondition of \Cref{ctlc:timeout} 
    since 
    $$\funtimeout(h(\e, id)) = t_0 + \fundepth(\e') \Delta .$$

    To show $Inv(\run, \partialEx', \fulltreeobj)$, we look at the 4 invariants that are potentially affected, which are $Inv_{\text{levels}}$, $Inv_{\text{liveness}}$, \linebreak $Inv_{\text{deposits}}$ and $Inv_{\text{setup}}$ and $Inv_{\text{setup}}$. 
    The changes from \linebreak $\action = \timeoutsc (\dotCTLCcontract, \dotCTLCsubcontract)$ do not influence the other invariants.

    We first show $Inv_{\text{levels}}$:
    Assume that $\e \in \untree$ with $B = \funsender(\e)$ and $\environmentvec.t > t_0 + \fundepth(\e) \Delta $.  
    Since $\environmentvec.t = \environmentvec'.t$ we get from the inductive hypothesis that 
    $$ (*) \, \nexists \channel, \CTLCcontract \in \environmentchprimeCTLCEnabled : h(\e, id) \in \CTLCcontract \land \funfund(\CTLCcontract) \in \environmentchprimeReservedFunds .$$

    Since the inference rule for $\timeoutsc$ (\Cref{ctlc:timeout}) only removes $\dotCTLCsubcontract$ from $\dotCTLCcontract \in \environmentvecprimeCTLCEnabled$, if there would be a $\CTLCcontract \in \environmentchCTLCEnabled$ with $h(\e, id) \in \CTLCcontract$ and  $\funfund(\CTLCcontract) \in \environmentchReservedFunds$ then there would also be some $\CTLCcontract* \in \environmentchprimeCTLCEnabled$ with $h(\e, id) \in \CTLCcontract*$ and also  $\funfund(\CTLCcontract) \in \environmentchprimeReservedFunds$ (since $\environmentchReservedFunds =\environmentchprimeReservedFunds$ for all channels $\channel$). 
    This immediately contradicts $(*)$.
    
    To show $Inv_{\text{liveness}}$, let $\eout = (B,X)_{\walk} \in \untree, j = \fundepth(\eout)$ and $\ein = (Y,B)_{\walk'} \in \untree \text{ with } \walk' = \concatvec{\unitvec{(Y,B)_{\walk}}}{\walk}$ be given as stated in \linebreak $Inv_{\text{liveness}}(\run, \partialEx', \fulltreeobj)$ which are the same as in
    \linebreak $Inv_{\text{liveness}}(\run', \partialEx', \fulltreeobj)$. By I.H. we either have
    \begin{align}
        &\exists j' \leq j+1 : (Y,B)_{\walk''} \in \partialEx' \land \fundepth((Y,B)_{\walk''}) = j' \text{ \textbf{or}} \label{eq:timehelper2}\\
        &\exists \CTLCcontract \in \environmentvecprimeCTLCEnabled : h(\ein, id) \in \CTLCcontract \land \environmentvecprimetime < t_0 + (j+1)\Delta . \label{eq:timehelper3}
    \end{align}
    In the case of \eqref{eq:timehelper2}, the conclusion trivially holds. 
    In the case of \eqref{eq:timehelper3}, we make a case distinction on $\dotCTLCsubcontract = h (\ein, id)$. 
    If $\dotCTLCsubcontract = h (\ein, id)$ we immediately arrive at a contradiction because the precondition of \Cref{ctlc:timeout} requires that $\funtimeout(\dotCTLCsubcontract) \leq \environmentvecprimetime$ and by construction $\funtimeout(\dotCTLCsubcontract) = \funtimeout(h(\ein, id)) = t_0 + \fundepth(\ein) \Delta$ and $\fundepth(\ein) = j + 1$. 
    If $\dotCTLCsubcontract \neq h (\ein, id)$ then we know from  $\CTLCcontract \in \environmentvecprimeCTLCEnabled : h(\ein, id) \in \CTLCcontract$ that also there is some $\CTLCcontract* \in \environmentvecCTLCEnabled : h(\ein, id) \in \CTLCcontract*$
    since $\environmentvecCTLCEnabled$ coincides with $\environmentvecCTLCEnabled$ with the only exception of $\dotCTLCsubcontract$ being removed from one channel $\channel$. With $\environmentvectime = \environmentvecprimetime$ this shows the conclusion. 

    For showing $Inv_{\text{deposits}}(\run, \partialEx', \fulltreeobj)$ we notice that
    $\environmentvecReservedFunds = \environmentvecprimeReservedFunds$ and that 
    by the inference rule of $\timeoutsc$ (\Cref{ctlc:timeout}) it holds that if $\CTLCcontract \in \environmentchprimeCTLCEnabled$
    then also for some  $\CTLCcontract* \in \environmentchprimeCTLCEnabled$ with $\funfund(\CTLCcontract*) = \funfund(\CTLCcontract)$ 
    (because the rule removes at most one subcontract from $\CTLCcontract$, which leaves the contract's funds unchanged).
    Consequently, the invariant follows directly from the inductive hypothesis. 

    To show $Inv_{\textit{setup}}$, assume $\eout = (B,X)_{\walk} \in \untree$.
    We distinguish the cases $h(\eout, id) \in \CTLCcontract \in \environment_{\channeledge{\eout}}.C_{\textit{en}}$
    and 
    \linebreak
    $(B, \advCTLCcontract) \in \environment_{\channeledge{\eout}}.C_{\textit{aut}}$ such that  $h(\eout, id) \in \advCTLCcontract$. 

    Assume that $h(\eout, id) \in \CTLCcontract \in \environment_{\channeledge{\eout}}.C_{\textit{en}}$ (*).
    Then we know that there is also some $\CTLCcontract* \in \environment'_{\channeledge{\eout}}.C_{\textit{en}}$ such that $h(\eout, id) \in \CTLCcontract*$ (since the $\timeoutsc$ rule atmost removes $\dotCTLCsubcontract$ from $\dotCTLCcontract$). 
    Consequently, from $Inv_{\text{setup}}(\run', \partialEx', \fulltreeobj)$ we get that 
    \begin{align}
        & \exists (\CTLCcontract)' \in \environment'_{\channeledge{\ein}}.C_{en} :  h(\ein, id) \in (\CTLCcontract)' \label{eq:setuphelper10}\\
        & \lor \hspace{-1pt} \exists j' \hspace{-1pt} \leq \hspace{-1pt} \fundepth(\ein) \hspace{-1pt} : \hspace{-1pt} (Y,B)_{\walk''} \hspace{-2pt} \in \partialEx' \label{eq:setuphelper11} \\
        & \qquad \land \fundepth((Y,B)_{\walk''}) = j' \nonumber
    \end{align}

    In the case of \eqref{eq:setuphelper10} we could have $h(\ein, id) = \dotCTLCsubcontract$ and thus $\nexists (\CTLCcontract)' \in \environment_{\channeledge{\ein}}.C_{en} :  h(\ein, id) \in (\CTLCcontract)'$. If 
    $ \dotCTLCsubcontract =  h(\ein, id)$ 
    by the precondition of $\timeoutsc$ (\ref{ctlc:timeout}) and the definition of $h$ we have
    \begin{align*}
        \environmentvectime \hspace{-1pt} &\geq \hspace{-1pt} \timeout(h(\ein, id)) \hspace{-1pt} \\
        &> \hspace{-1pt} \timeout(h(\eout, id)) \hspace{-1pt} = \hspace{-1pt} t_0 + \fundepth(\eout) \Delta
    \end{align*}
    So we get from $Inv_{\textit{levels}}$ that 
    \[
        \nexists \CTLCcontract \in \environment_{\channeledge{\eout}}.C_{\textit{en}}: \, h(\eout, id) \in \CTLCcontract
    \]
    which contradicts our assumption (*).
    In the case of \eqref{eq:setuphelper11} the conclusion trivially holds. 

    Next, assume that $(B, \advCTLCcontract) \in \environment_{\channeledge{\eout}}.C_{\textit{aut}}$ such that \linebreak $h(\eout, id) \in \advCTLCcontract$.
    In this case we also know that \linebreak $(B, \advCTLCcontract) \in \environment'_{\channeledge{\eout}}.C_{\textit{aut}}$ (since the $\timeoutsc$ rule does not change authorizations).
    
    Consequently, from $Inv_{\text{setup}}(\run', \partialEx', \fulltreeobj)$ we again get that 
    \begin{align}
        & \exists (\CTLCcontract)' \in \environment'_{\channeledge{\ein}}.C_{en} :  h(\ein, id) \in (\CTLCcontract)' \text{ \textbf{or}}\label{eq:setuphelper12}\\
        &\exists j' \leq \fundepth(\ein) : (Y,B)_{\walk''} \in \partialEx'  \land \fundepth((Y,B)_{\walk''}) = j' \label{eq:setuphelper13}
    \end{align}

    Again, the claim immediately follows for \Cref{eq:setuphelper13}.
    For the case of \Cref{eq:setuphelper12} we only need to consider the case $h(\ein, id) = \dotCTLCsubcontract$ where have
    by the precondition of $\timeoutsc$ (\ref{ctlc:timeout}) and the definition of $h$:
    \begin{align*}
        \environmentvectime &\geq \timeout(h(\ein, id)) > \timeout(h(\eout, id)) 
        \\ &= t_0 + \fundepth(\eout) \Delta > t_0
    \end{align*}
    Therefore, $Inv_{\text{auth}}(\run, \partialEx', \fulltreeobj)$ contradicts the assumption which is proven in the next paragraph independently. This is because
    \begin{equation*}
    \scalebox{0.95}{$
        \environmentvectime \geq t_0 \Rightarrow \nexists \advCTLCcontract \in \environmentvecCTLCAdvertised:
        (B, \advCTLCcontract) \in \environmentvecCTLCAuthorized  \land B = \funsender(\advCTLCcontract)
        $}
    \end{equation*}
    is a direct implication of this invariant. In the invariant we implied $\environmentvectime < t_0$, which is the negation of $\environmentvectime \geq t_0$, the precondition here. Therefore, implying that such an $(B, \advCTLCcontract) \in \environmentvecCTLCAuthorized$ does not exist, which is the negation of the invariants precondition, follows directly. 
    
    $Inv_{\text{tree}}(\run, \partialEx', \fulltreeobj)$ only applies to less subcontracts as $Inv_{\text{tree}}(\run', \partialEx', \fulltreeobj)$ and so it is implied directly by I.H..
    For $Inv_{\text{auth}}(\run, \partialEx', \fulltreeobj)$ we look at $\environmentvectime = \environmentvecprimetime$. 
    For 
    $$(B, \widehat{\dot{\textit{c}}^x}) \in \environmentvecCTLCAuthorized = \environmentvecprimeCTLCAuthorized$$
    a substrategy (\ref{eq:enable-check}) of the honest user strategy implies that $B$ schedules 
    \begin{equation} \label{thingsthingsandevenmorethings}
         \enableCTLC \, \dotCTLCcontract
    \end{equation}
    right after $B: \authCTLC \,  \widehat{\dot{\textit{c}}^x}$. The only action removing authorizations is $\enableCTLC$ (\ref{ctlc:enableCTLC}) itself. Based on the same substrategy we know that $B: \authCTLC \, \widehat{\dot{\textit{c}}^x}$ only gets scheduled if $\environmentvectime < t_0$. 
    By \cref{def:adversarystrategy}, we know that no time elapses as long as $B$ does not agree, and so $\environmentvectime < t_0$ still holds. By \cref{def:hMapping} we have $\timeout(\dotCTLCsubcontract) > t_0$ and so $\environmentvectime > t_0$, which is a contradiction to the given action $\action$. 
    
    \noindent \textbf{Case $\action = \refund \, \dotCTLCcontract$:}

    To show that $\run  \IDrelation \partialExFamilyprime$ the same reasoning can be applied as in the $\timeoutsc$ case. 

    From the $\refund$ rule, we also know
    \begin{align} \label{somehelperrealitsations}
        &\environmentvecCTLCEnabled \subseteq \environmentvecprimeCTLCEnabled \land \environmentvecCTLCAdvertised \subseteq \environmentvecprimeCTLCAdvertised \\
        \land &\environmentvecReservedFunds \subseteq \environmentvecprimeReservedFunds \land \environmentvecAvailableFunds \supseteq \environmentvecprimeAvailableFunds . \nonumber
    \end{align}
    Similar to the previous case $\action = \timeoutsc (\dotCTLCcontract, \dotCTLCsubcontract)$ out of \linebreak
    $Inv(\run, \partialEx', \fulltreeobj)$ only 5 invariants are affected for the $id$ with $\widehat{\dotCTLCcontract} \in \batch$: 
    \begin{align*}
        &Inv_{\text{levels}}, Inv_{\text{liveness}}, Inv_{\text{deposits}}, Inv_{\text{setup}}, Inv_{\text{tree}} \text{, and } Inv_{\text{auth}}
    \end{align*}
    For proving $Inv_{\text{levels}}(\run, \partialEx', \fulltreeobj)$ we note that no time has passed, $\environmentvectime = \environmentvecprimetime$, and the set of enabled contracts as well as reserved funds have only gotten smaller, see (\ref{somehelperrealitsations}). Therefore, it follows immediately from I.H..

    The proof for $Inv_{\text{liveness}}(\run, \partialEx', \fulltreeobj)$ works analogously to its proof in the $\alpha = \timeoutsc(\dotCTLCcontract, \dotCTLCsubcontract)$ case.

     For showing $Inv_{\text{deposits}}(\run, \partialEx', \fulltreeobj)$ we look at the $\refund$ rule (\ref{ctlc:refund}) and notice that exactly $\funfund(\CTLCcontract)$ gets removed from $\environmentchprimeReservedFunds$ alongside $\CTLCcontract$ from $\environmentchprimeCTLCEnabled$. Thus \linebreak $Inv_{\text{deposits}}(\run, \partialEx', \fulltreeobj)$ applies to one less contract than $Inv_{\text{deposits}}(\run', \partialEx', \fulltreeobj)$. The fund for no other contract can be missing in $\environmentchReservedFunds$ since every contract has a unique fund according to \cref{def:initialconfig}. 

    Again, the proof for $Inv_{\text{setup}}(\run, \partialEx', \fulltreeobj)$ works analogously to its proof in the $\alpha = \timeoutsc(\dotCTLCcontract, \dotCTLCsubcontract)$ case considering that $\dotCTLCcontract = \{ \dotCTLCsubcontract \}$ (which is given since the preconditions of the $\refund$ rule require $\dotCTLCcontract$ to be a singleton set).

    $Inv_{\text{tree}}(\run, \partialEx', \fulltreeobj)$ is implied by I.H. as $\environmentchCTLCEnabled$  only got smaller or stayed equal compared to $\environmentchprimeCTLCEnabled$. Their implication in this invariant is unaffected.

    $Inv_{\text{auth}}(\run, \partialEx', \fulltreeobj)$ is implied by analogous reasoning as in the previous case $\action = \timeoutsc (\dotCTLCcontract, \dotCTLCsubcontract)$. 
    
    \noindent \textbf{Case $\action = \DecideCo (\dotCTLCcontract ,\dotCTLCsubcontract,  \funsecret_i(\dotCTLCsubcontract))$: } 

    \noindent 
    For the specific $id$ with $\widehat{\dotCTLCcontract} \in \batch$ we define 
    \begin{align*}
        \partialEx := \partialEx' \cup \{ \e \in \untree \mid & \e = (B,X)_{\walk} \land \dotCTLCsubcontract = h(\e, id) \\
        &\land \funsecret_i(\dotCTLCsubcontract) = h_{sec}(\e, id) \}
    \end{align*}
    and 
    \begin{align*}
        \partialExFamily := ( \partialExFamilyprime \backslash \{ \partialEx' \} ) \cup \{ \partialEx \}
    \end{align*}
    (only replacing this specific $\partialEx'$ for the $id$ of $\widehat{\dotCTLCcontract} \in \batch$)
    and show that
    \begin{enumerate}[a\hspace{1pt}]
        \item \hspace{-5pt}) $\run \IDrelation \partialExFamily$,
        \item \hspace{-5pt}) $\partialEx$ is consistent and
        \item \hspace{-5pt}) $Inv(\run, \partialEx, \fulltreeobj)$ holds.
    \end{enumerate}

    \noindent a) For showing that $\run \IDrelation \partialExFamily$, we first show that 
    $$ \schecontracts{\run} = \schecontracts{\run'}.$$
    Assume towards contradiction that (for the $id$ of $\batch$ specifically, for the others it holds by definition)
    $$\schecontracts{\run} \neq \schecontracts{\run'}.$$
    We consider the cases 
    \begin{align*}
    &(\CTLCsubcontract, \funsecret_i(\dotCTLCsubcontract)) \hspace{-1pt} \in \hspace{-1pt} \schecontracts{\run} \backslash \schecontracts{\run'} \\
    &\text{and} \\
    &(\CTLCsubcontract, \funsecret_i(\dotCTLCsubcontract)) \hspace{-1pt} \in \hspace{-1pt} \schecontracts{\run'} \backslash \schecontracts{\run}
    \end{align*}
    individually. 
    Assume that there is some 
    $$(\CTLCsubcontract, \funsecret_i(\dotCTLCsubcontract)) \hspace{-1pt} \in \hspace{-1pt} \schecontracts{\run} \backslash \schecontracts{\run'}.$$
    This means that $\run \overset{\action'}{\longrightarrow} \environmentvec^*$ for some $\environmentvec^*$ and 
    $$\action' = \DecideCo(\CTLCcontract, \CTLCsubcontract, \funsecret_i(\CTLCsubcontract))$$
     and $\funreceiver(\CTLCsubcontract) = \honestuser$.
    Since by the inference rule for $\DecideCo$ (\Cref{ctlc:DecideCo}) we know that $\CTLCsubcontract \in \CTLCcontract \in \environmentchCTLCEnabled$ 
    for some channel $\channel$ we can conclude using $Inv_\textit{tree}$ that there is some $\e \in \untree$ such that $\CTLCsubcontract = h(\e, id)$ and hence by construction also that there is some $\e' \in \untree$ with $h(\e', id) = \CTLCsubcontract$ and $h_{sec}(\e', id) = \funsecret_i(\CTLCsubcontract)$ and $\funreceiver(\CTLCsubcontract) = \funreceiver(\e) = \funreceiver(\e') = \honestuser$.
    If $\run \overset{\action'}{\longrightarrow} \environmentvec^*$ then by the inference rule for $\DecideCo$ (\Cref{ctlc:DecideCo}), we know that 
    $\funsecret_i(\CTLCsubcontract) \subseteq \environmentvecRevealedSecrets = \environmentvecprimeRevealedSecrets$ and so also $\funsecret(\e', id) \in \environmentvecRevealedSecrets$ since by construction $\funsecret(\e', id) \in h_{sec}(\e', id) = \funsecret_i(\CTLCsubcontract)$. 
    This allows us to conclude using $Inv_\textit{in-secrets}$ (for $\run'$ and $\partialEx'$) that $\e' \in \partialEx'$ and so by $\run'  \IDrelation \partialEx'$ that $\CTLCsubcontract = h(\e', id) \in \schecontracts{\run'}$, contradicting the original assumption. 

    If $(\CTLCsubcontract, \funsecret_i(\dotCTLCsubcontract)) \in  \schecontracts{\run'} \backslash  \schecontracts{\run}$ then this means that $\run' \overset{\action'}{\longrightarrow} \environmentvec^*$ for some $\environmentvec^*$ and 
    $$\action' = \DecideCo(\CTLCcontract, \CTLCsubcontract, \funsecret_i(\CTLCsubcontract))$$ and $\funreceiver(\CTLCsubcontract) = \honestuser$. 
    We further know that there is no $\environmentvec^\dagger$ such that $\run' \overset{\action'}{\longrightarrow} \environmentvec^\dagger$
    and also $\CTLCsubcontract \neq \dotCTLCsubcontract$ (since otherwise by definition $\CTLCsubcontract \in  \schecontracts{\run}$).
    By the inference rule of $\DecideCo$ (\Cref{ctlc:DecideCo}), we know that all environment components influencing the rules precondition but $\environmentchCTLCEnabled$ and $\environmentCTLCAdvertised$ stay unaffected.
    Further, for $\CTLCsubcontract \not \in \dotCTLCcontract$ we know that $\run' \overset{\action'}{\longrightarrow}$ since
    if $\CTLCcontract \neq \dotCTLCcontract$ then 
    $\CTLCcontract \in \environmentchprimeCTLCEnabled$ implies 
    $\CTLCcontract \in \environmentchCTLCEnabled$ and 
    $\CTLCcontract \in \environmentchprimeCTLCAdvertised$ implies 
    $\CTLCcontract \in \environmentchCTLCAdvertised$ for some channel $\channel$.
    We, hence, are left to consider the case that $\CTLCsubcontract \in \dotCTLCcontract$.
    By the inference rule for $\DecideCo$ (\Cref{ctlc:DecideCo}) we know that $\CTLCsubcontract \in \CTLCcontract \in \environmentchprimeCTLCEnabled$,
    so we can conclude using $Inv_\textit{tree}$ that there is some $\e \in \untree$ such that $\CTLCsubcontract = h(\e, id)$ and $\funreceiver(\CTLCsubcontract) = \funreceiver(\e)$ and $\funsender(\CTLCsubcontract) = \funsender(\e)$.
    Similarly, we can conclude that there is some $\e' \in \untree$ such that $\dotCTLCsubcontract = h(\e', id)$ and $\funreceiver(\dotCTLCsubcontract) = \funreceiver(\e')$ and $\funsender(\dotCTLCsubcontract) = \funsender(\e')$.
    Since all subcontracts in the same CTLC have the same sender and receiver, we can also conclude that $\funreceiver(\e) = \funreceiver(\e')$, $\funsender(\e) = \funsender(\e')$. 
    And consequently from \Cref{def:specification} that $\spec{\e} = \spec{\e'}$. 
    Since by definition, $\dotCTLCsubcontract \in \schecontracts{\run'}$ and by assumption $\CTLCsubcontract \in \schecontracts{\run'}$, from $\run' \IDrelation \partialEx'$, we know that $\e, \e' \in \partialEx'$. 
    However, since $\partialEx'$ is consistent (by I.H.) this is contradicting given that $\spec{\e} = \spec{\e'}$.

    For $\excontracts{\run}$ though we have
    \begin{align*}
        &\excontracts{\run} \\
        &= \excontracts{\run'} \cup \{ (\dotCTLCsubcontract, \funsecret_i(\dotCTLCsubcontract)) \}
    \end{align*}
    if $B = \funsender(\dotCTLCsubcontract)$, otherwise  
    $$\excontracts{\run} = \excontracts{\run'} .$$ 
    So if $\CTLCsubcontract \in \excontracts{\run}$ then either $\CTLCsubcontract \in \excontracts{\run'}$ and hence (by $\run' \IDrelation \partialEx'$) also $\e \in \partialEx' \subseteq \partialEx$ for some $\e \in \untree$ with $\CTLCsubcontract = h(\e, id)$.
    Or $\CTLCsubcontract  = \dotCTLCsubcontract$ with $B = \funsender(\CTLCsubcontract)$.   

    Since by the inference rule of the $\DecideCo$ action (\Cref{ctlc:DecideCo}) we know that $\dotCTLCsubcontract \in \dotCTLCcontract \in \environmentchprimeCTLCEnabled$ 
    for some channel $\channel$, \linebreak $Inv_{\text{tree}}(\run', \partialEx', \fulltreeobj)$ implies
    \[
        \exists e \in \untree : \e = (B,X)_{\walk} \land \dotCTLCsubcontract = h(\e, id) . 
    \]
     By construction of $h$ there then is also some $\e' \in \untree$ with $h(\e', id) = \dotCTLCsubcontract$ and $h_{sec}(\e', id) = \funsecret_i(\dotCTLCsubcontract)$ and $\funsender(\CTLCsubcontract) = \funsender(\e) = \funsender(\e') = \honestuser$.
     So by definition of $\partialEx$ also $\e' \in \partialEx$. 

     Correspondingly, if $\e \in \partialEx$ then either $\e \in \partialEx'$ and (by $\run' \IDrelation \partialEx'$) $h(e, id) \in \excontracts{\run'} \subseteq \excontracts{\run}$. 
     Or $h(\e, id) = \dotCTLCsubcontract$ and hence $h(\e, id) \in \excontracts{\run}$ by construction.

    \noindent b) To show that $\partialEx$ is consistent we assume towards contradiction that $\partialEx$ is not consistent meaning that there exists a predecessor $\e^{*} \in \funonPathtoRoot(\untree, \e) \cap \hatpartialEx$ of the one $\e \in \partialEx \backslash \partialEx'$ such that \linebreak (1) $\e^{*} \notin \partialEx$ or (2) 
    $
        \exists \e^{**} \in \partialEx' : \spec{\e^{**}} = \spec{\e}.
    $
    The situation for (1) is similar to the one we have covered for $$\alpha = A: \, \revealSecretch \, s_{\walk}^{id}.$$ Therefore, the argumentation is analogous.

    For (2), assume towards contradiction that there is $\e^{**} \in \partialEx'$ such that $\spec{\e^{**}} = \spec{\e}$.
    Since $\spec{\e^{**}} = \spec{\e}$ we know from~\Cref{def:specification} also that $\funsender(\e) = \funsender(\e^{**})$ and \linebreak $\funreceiver(\e) = \funreceiver(\e^{**})$. 
    Consequently, by construction, 
    $\e^{**}$ and $\e$ are part of the exact same \CTLC{}, see \cref{eq:treetoCTLCBadge}. This $\CTLCcontract$ has a unique identifier, and according to \cref{lemma:againanotherlemmainG}, cannot be enabled again after it has been claimed. Therefore, 
    $$\action = \DecideCo (\dotCTLCcontract ,h(\e, id),  h_{sec}(\e, id))$$
    is not possible, which is a contradiction. 
    
    \noindent c) Firstly, we note that for all $\environmentch \in \environmentvec$ all sets are the same as they are in $\environmentchprime \in \environmentvecprime$ except for one $\environmentch$, where we have 
    \begin{align*}
    & \environmentchCTLCAdvertised \subsetneq \environmentchprimeCTLCAdvertised ~\land~ \environmentchCTLCEnabled \subsetneq \environmentchprimeCTLCEnabled \\
    ~\land~ & \environmentchCTLCDecided \supsetneq \environmentchprimeCTLCDecided .
    \end{align*}

    $Inv(\run, \partialEx', \fulltreeobj)$ is shown by going through the invariants individually.
    
\noindent \underline{$Inv_{\text{in-secrets}}(\run, \partialEx', \fulltreeobj)$} immediately follows from $Inv_{\text{in-secrets}}(\run', \partialEx', \fulltreeobj)$ since $\partialEx' \subseteq \partialEx$ and $\environmentchRevealedSecrets = \environmentchprimeRevealedSecrets$.

\noindent \underline{$Inv_{\text{secrets}}(\run, \partialEx', \fulltreeobj)$} is implied for all $\e \in \partialEx'$ by I.H. and for $\e \in \partialEx \backslash \partialEx'$ we have $h_{sec}(\e, id) \subseteq \environment_{\channeledge{\e}}.S_{rev}$ by definition of $\partialEx$ for this case. The invariant is then implied by the definition of $h_{sec}$ in \cref{def:hsecMapping}. 

\noindent \underline{$Inv_{\text{in-schedule}}(\run, \partialEx', \fulltreeobj)$} is implied directly from I.H. as $\environmentvecRevealedSecrets = \environmentvecprimeRevealedSecrets$ and $\actions{\run'} \subsetneq \actions{\run}$.

\noindent \underline{$Inv_{\text{levels}}(\run, \partialEx', \fulltreeobj)$} directly follows from the inductive hypothesis, because the $\DecideCo$ rule (\Cref{ctlc:DecideCo}) only removes contracts from $\environmentvecCTLCEnabled$.

\noindent \underline{$Inv_{\text{liveness}}(\run, \partialEx, \fulltreeobj)$} Is proven by first setting 
$\eout = (B,X)_{\walk} \in \untree$  and $j = \fundepth(\eout)$ and $\ein = (Y,B)_{\walk'}$ with $\walk' = \concatvec{[(\ein)]}{\walk}$.
We distinguish the cases $h(\eout, id) \neq \dotCTLCsubcontract$ and $h(\eout, id) = \dotCTLCsubcontract$. 

        First consider $h(\eout, id) \neq \dotCTLCsubcontract$.
        We know from 
        $$Inv_{\text{liveness}}(\run', \partialEx', \fulltreeobj)$$ that either 
        $$
        \exists j' \leq j+1, \walk'' : (Y,B)_{\walk''} \in \partialEx' \land \fundepth((Y,B)_{\walk''}) = j'
        $$
        or 
        \begin{align*}
        &\exists \CTLCcontract \in \environment'_{\channeledge{\ein}}.C_{en} : \\
        &h(\ein, id) \in \CTLCcontract \land \environmentvecprimetime < t_0 + (j+1)\Delta.
        \end{align*}
        In the first case, the claim immediately follows since $\partialEx' \subseteq \partialEx$ and $\fundepth(\ein) = 1 + \fundepth(\eout)$.
        In the second case, we distinguish whether $h(\ein, id) \in \dotCTLCcontract$. 
        If $h(\ein, id) \not \in \dotCTLCcontract$ then we know that $\dotCTLCcontract \neq (\CTLCcontract)'$ and hence 
        $h(\ein, id) \in (\CTLCcontract)'$ and $\environmentvectime = \environmentvecprimetime < t_0 + (j+1)\Delta$.
        If  $h(\ein, id) \in \dotCTLCcontract$ 
        then we know that $\funsender(\dotCTLCsubcontract) = \funsender(h(\ein, id)) = \funsender(\ein) = Y$ and
        $$\funreceiver(\dotCTLCsubcontract) = \funreceiver(h(\ein, id)) = \funreceiver(\ein) = Y$$ (since they are subcontracts of the same CTLC).
        Further, by \linebreak $Inv_\textit{tree}(\run', \partialEx', \fulltreeobj)$ we know that there is some $\e \in \untree$ such that $h(\e, id) = \dotCTLCsubcontract$ (since by \cref{ctlc:DecideCo} we know that $\dotCTLCsubcontract \in \dotCTLCcontract \in \environment'_{\channeledge{\ein}}.C_{en}$).
        Consequently, $\e = (Y, B)_{\walk''}$ for some $\walk''$ (since $\funsender(\e) = \funsender(\dotCTLCsubcontract) = Y$ and $\funreceiver(\e) = \funreceiver(\dotCTLCsubcontract) = B$). 
        By definition of $\schecontracts{\cdot}$, \linebreak $(\dotCTLCsubcontract, h_{sec}(\e, id)) \in \schecontracts{\run'}$ and so by $\run' \IDrelation \partialEx'$, also $\e \in \partialEx' \subseteq \partialEx$. 
        So, we are left to show that $\fundepth(\e) \leq j+1$. 
        Assume towards contradiction that $\fundepth(\e) > j+1$.
        By construction, we know that $\funtimeout(\dotCTLCsubcontract) = t_0 + \fundepth(\e) \Delta$, so in this case $\funtimeout(\dotCTLCsubcontract)> t_0 + (j+1) \Delta$. 
        However, $\funtimeout(h(\ein, id)) = t_0 + \fundepth(\ein) \Delta = t_0 + j \Delta$, so 
        $\funtimeout(\dotCTLCsubcontract) >  \funtimeout(h(\ein, id))$. 
        But then, by construction 
        $\funposition(\dotCTLCsubcontract) > \funposition(h(\ein, id))$. 
        Since $h(\ein, id), \dotCTLCsubcontract \in \dotCTLCcontract \in \environment'_{\channeledge{\ein}}.C_{en}$ we know by 
        \cref{lemma:enabledsubsetofadvertised}
        that $h(\ein, id), \dotCTLCsubcontract \in \advCTLCcontract \in \environment'_{\channeledge{\ein}}.C_{adv}$.
        By the inference rule of $\DecideCo$ we know that if 
        $\funposition(\dotCTLCsubcontract) > 1$ then there is no $(\CTLCsubcontract)^* \in \advCTLCcontract$ with $\funposition((\CTLCsubcontract)^* ) < \funposition(\dotCTLCsubcontract)$. 
        This is contradicted by $h(\ein, id) \in  \advCTLCcontract$.
        
        Next, we consider the case $h(\eout, id) = \dotCTLCsubcontract$. 
        Then we know from $Inv_{\text{setup}}(\run', \partialEx', \fulltreeobj)$ that either 
        \begin{align*}
        &\exists (\CTLCcontract)' \in \environment'_{\channeledge{\ein}}.C_{en} :  h(\ein, id) \in (\CTLCcontract)' \text{ \textbf{or}} \\
        &\exists j' \hspace{-1pt} \leq \fundepth(\ein) \hspace{-1pt} : \hspace{-1pt} (Y,B)_{\walk''} \hspace{-1pt} \in \partialEx'  \land \fundepth((Y,B)_{\walk''}) \hspace{-1pt} = j'. 
        \end{align*}
        In the second case, the claim immediately follows since $\partialEx' \subseteq \partialEx$ and $\fundepth(\ein) = 1 + \fundepth(\eout)$.
        In the first case, we need to first show that 
        $\exists (\CTLCcontract)' \in \environment_{\channeledge{\ein}}.C_{en} :  h(\ein, id) \in (\CTLCcontract)'$.
        This trivially holds since $\environment_{\channeledge{\ein}}.C_{en}$ is the same as $\environment'_{\channeledge{\ein}}.C_{en}$ with the only exception that $\dotCTLCcontract$ got removed. 
        However, $\dotCTLCcontract$ could not have been $(\CTLCcontract)'$ (the contract containing $h(\ein, id)$) since 
        \begin{align*}
        \funsender(h(\ein, id)) &= \funsender (\ein) \\
        &\neq \funreceiver(\ein) = \honestuser = \funsender(\dotCTLCcontract)
        \end{align*}
        but if it would hold that $h(\ein, id) \in \dotCTLCcontract$ then we would need to have that $\funsender(h(\ein, id)) = \funsender(\dotCTLCcontract)$.
        We are hence, left to show that 
        $\environmentvectime < t_0 + (j+1) \Delta$. 
        Assume towards contradiction that 
        $\environmentvectime \geq t_0 + (j+1) \Delta$
        then also $\environmentvecprimetime = \environmentvectime > t_0 + j \Delta$
        and  because $j = \fundepth(\eout)$ consequently by 
        $Inv_{\textit{levels}}(\run', \partialEx', \fulltreeobj)$
        we get that 
        $$\nexists \channel, \CTLCcontract \in \environmentchprimeCTLCEnabled : h(\eout, id) \in \CTLCcontract \land \funfund(\CTLCcontract) \in \environmentchprimeReservedFunds$$
        However, since $\dotCTLCsubcontract = h(\eout, id)$ by inference rule of $\DecideCo$ (\Cref{ctlc:DecideCo}) we know that $h(\eout, id) \in \dotCTLCcontract \in \environmentchprimeCTLCEnabled$ for some channel $\channel$. By $Inv_{\text{deposits}}(\run', \partialEx', \fulltreeobj)$ we also have $\funfund(\dotCTLCcontract) \in \environmentchprimeReservedFunds$ which leads to a contradiction.
        
\noindent \underline{$Inv_{\text{deposits}}(\run, \partialEx', \fulltreeobj)$} follows directly from I.H. as \linebreak $\environmentchCTLCEnabled \cup \environmentchCTLCDecided $ only shrinks while $\environmentchReservedFunds$ remains unchanged.    

\noindent \underline{$Inv_{\text{setup}}(\run, \partialEx', \fulltreeobj)$} is shown by setting \newline $\eout = (B,X)_{\walk} \in \untree$  and $j = \fundepth(\eout)$ and $\ein = (Y,B)_{\walk'}$ with $\walk' = \concatvec{[(\ein)]}{\walk}$.
        Further, let either be
        \begin{equation}
            h(\eout, id) \in \CTLCcontract \in \environment_{\channeledge{\eout}}.C_{\textit{en}} 
            \label{eqn:claim-setup-one}
        \end{equation}
        or
        \begin{equation}
            (B, \advCTLCcontract) \in \environment_{\channeledge{\eout}}.C_{\textit{aut}} \, , h(\eout, id) \in \advCTLCcontract
            \label{eqn:claim-setup-two}
        \end{equation}
        If \Cref{eqn:claim-setup-one} holds then also 
        $$h(\eout, id) \in \CTLCcontract \in \environment'_{\channeledge{\eout}}.C_{\textit{en}} $$ 
        (since $\DecideCo$ only removes CTLCs from $ \environment'_{\channeledge{\eout}}.C_{\textit{en}}$)
        and similarly if \Cref{eqn:claim-setup-two} holds also 
        $$(B, \advCTLCcontract) \in \environment'_{\channeledge{\eout}}.C_{\textit{aut}} \, , h(\eout, id) \in \advCTLCcontract.
        $$
        Hence, in both cases we know by $Inv_{\text{setup}}(\run', \partialEx', \fulltreeobj)$ that either 
        \begin{equation}
            \exists (\CTLCcontract)' \in \environment'_{\channeledge{\ein}}.C_{en} :  h(\ein, id) \in (\CTLCcontract)'
            \label{eqn:claim-setup-three} 
        \end{equation}
        or 
        \begin{equation}
            \exists j' \leq \fundepth(\ein) : (Y,B)_{\walk''} \in \partialEx'  \land \fundepth((Y,B)_{\walk''}) = j'.
            \label{eqn:claim-setup-four}
        \end{equation}
        If \Cref{eqn:claim-setup-four} holds then the claim follows immediately since $\partialEx' \subseteq \partialEx$. 
        If \Cref{eqn:claim-setup-three} holds we distinguish whether $h(\ein, id) \in \dotCTLCcontract$. 
        If $h(\ein, id) \not \in \dotCTLCcontract$ then we know that $\dotCTLCcontract \neq (\CTLCcontract)'$ and hence 
        $h(\ein, id) \in (\CTLCcontract)' \in \environment_{\channeledge{\ein}}.C_{en}$.
        If $h(\ein, id) \in \dotCTLCcontract$ 
        then we know that $\funsender(\dotCTLCsubcontract) = \funsender(h(\ein, id)) = \funsender(\ein) = Y$ and
        $\funreceiver(\dotCTLCsubcontract) = \funreceiver(h(\ein, id)) = \funreceiver(\ein) = Y$ (since they are subcontracts of the same CTLC).
        Further, by \linebreak $Inv_\textit{tree}(\run', \partialEx', \fulltreeobj)$ we know that there is some $\e \in \untree$ such that $h(\e, id) = \dotCTLCsubcontract$ and $\funsecret_i(\dotCTLCsubcontract) = h_{sec}(\e, id)$ (since by \cref{ctlc:DecideCo} we know that $\dotCTLCsubcontract \in \dotCTLCcontract \in \environment'_{\channeledge{\ein}}.C_{en}$).
        Consequently, $\e = (Y, B)_{\walk''}$ for some $\walk''$ (since $\funsender(\e) = \funsender(\dotCTLCsubcontract) = Y$ and $\funreceiver(\e) = \funreceiver(\dotCTLCsubcontract) = B$). 
        By definition of \linebreak $\schecontracts{\cdot}$, $(\dotCTLCsubcontract, h_{sec}(\e, id)) \in \schecontracts{\run'}$ and so by $\run' \IDrelation \partialEx'$, also $\e \in \partialEx' \subseteq \partialEx$. 
        So, we are left to show that $\fundepth(\e) \leq j+1$. 
        Assume towards contradiction that $\fundepth(\e) > j+1$.
        By construction, we know that $\funtimeout(\dotCTLCsubcontract) = t_0 + \fundepth(\e) \Delta$, so in this case $\funtimeout(\dotCTLCsubcontract) >  t_0 + (j+1) \Delta$. 
        However, $\funtimeout(h(\ein, id)) = t_0 + \fundepth(\ein) \Delta = t_0 + j \Delta$, so 
        $\funtimeout(\dotCTLCsubcontract) >  \funtimeout(h(\ein, id))$. 
        But then, by construction 
        $\funposition(\dotCTLCsubcontract) > \funposition(h(\ein, id))$. 
        Since $h(\ein, id), \dotCTLCsubcontract \in \dotCTLCcontract \in \environment'_{\channeledge{\ein}}.C_{en}$ we know by 
        \cref{lemma:enabledsubsetofadvertised}
        that $h(\ein, id), \dotCTLCsubcontract \in \advCTLCcontract \in \environment'_{\channeledge{\ein}}.C_{adv}$.
        By the inference rule of $\DecideCo$ we know that if 
        $\funposition(\dotCTLCsubcontract) > 1$ then there is no $(\CTLCsubcontract)^* \in \advCTLCcontract$ with $\funposition((\CTLCsubcontract)^* ) < \funposition(\dotCTLCsubcontract)$. 
        This is contradicted by $h(\ein, id) \in  \advCTLCcontract$.
        
\noindent \underline{$Inv_{\text{tree}}(\run, \partialEx', \fulltreeobj)$} is not affected 
as the rule only removes elements from $\environmentvecprimeCTLCEnabled$ and only adds an element to set $\environmentvecCTLCDecided$, which previously was in $\environmentvecprimeCTLCEnabled$ and hence is guaranteed to satisfy the condition. 

\noindent \underline{$Inv_{\text{init-liveness}}(\run, \partialEx', \fulltreeobj)$} holds since we only need to consider the case of the newly added $\e \in \partialEx \, \backslash \, \partialEx'$.
        In this case by construction of $\partialEx$, we know that 
        \[
            \DecideCo(\dotCTLCcontract, h(\e, id), h_{sec}(e, id)) \in \actions{\run}. 
        \]

\noindent \underline{$Inv_{\text{auth}}(\run, \partialEx', \fulltreeobj)$}
is implied by \newline $Inv_{\text{init-liveness}}(\run, \partialEx', \fulltreeobj)$ combined with 
the fact that 
$$\exists j' \leq \fundepth(\ein) : (Y,B)_{\walk''} \in \partialEx  \land \fundepth((Y,B)_{\walk''}) = j')$$
from this invariant cannot happen since this implies that $B$ revealed $\funsecret((Y,B)_{\walk''})$, see $Inv_{\text{tree}}$. This only happens if $t_0 < \environmentvectime$ according to \cref{eq:CTLC-execution} which would be a contradiction based on the reasoning we applied in \cref{thingsthingsandevenmorethings}. 

    \noindent \textbf{Case $\action = \CoEx ( \CTLCcontract, \CTLCsubcontract )$ :} 
    
    \noindent We show that 
    \begin{enumerate}[a\hspace{1pt}]
        \item \hspace{-5pt}) $\run \IDrelation \partialExFamilyprime$ and
        \item \hspace{-5pt}) $Inv(\run, \partialEx', \fulltreeobj)$ hold.
    \end{enumerate}
    As $\partialExFamilyprime$ stays the same for $\run'$ and $\run$ we know by I.H. that all $\partialEx'$ are consistent. 

    \noindent a) The relation $\run \IDrelation \partialExFamilyprime$ also follows from I.H. given that we can show for all $id$
    \begin{align*}
        \excontracts{\run} &= \excontracts{\run'} \text{ and} \\
        \schecontracts{\run} &= \schecontracts{\run'} .
    \end{align*}
    We look at the $id$ with $\advCTLCcontract \in \batch$ specifically since this holds for all others by I.H..
    By definition we have $\excontracts{\run'} \subseteq \excontracts{\run}$. We show that also $\excontracts{\run} \subseteq $ \linebreak $\excontracts{\run'}$ holds.

For $\environmentvecprime = \funlastEnv(\run')$ all $\environmentchprime$ remain untouched except for one specific $\channel$. Let this $\channel$ be fixed then the following statements are true:
    \begin{align*}
         \environmentchRevealedSecrets = \environmentchprimeRevealedSecrets, \,
         \environmentchCTLCEnabled = \environmentchprimeCTLCEnabled, \,
         \environmentchCTLCAdvertised &= \environmentchprimeCTLCAdvertised, \\
         \environmentchCTLCDecided \subsetneq \environmentchprimeCTLCDecided, \,
         \environmentchReservedFunds \subsetneq \environmentchprimeReservedFunds, \,
         \environmentchtime &= \environmentchprimetime
    \end{align*}
    Since 
    \begin{align*}
    &\excontracts{\run} \\
    & = \excontracts{\run'} \cup \{ (\CTLCsubcontract, \funsecret_i(\CTLCsubcontract)) \}
    \end{align*}
    we only need to show that $(\CTLCsubcontract, \funsecret_i(\CTLCsubcontract)) \in \excontracts{\run'}$. 
    From the $\CoEx$ rule (\ref{ctlc:CoEx}) combined with \cref{lemma:claimedcontractsoneitem} we know that 
    $$ \CTLCcontract = \{ \CTLCsubcontract \} \in \environmentvecprimeCTLCDecided $$ 
    and thus,
    $
        \DecideCo (\CTLCcontract ,\CTLCsubcontract,  \funsecret_i(\CTLCsubcontract)) \in \actions{\run'} 
    $
    is implied. 
    Consequently, $\CTLCsubcontract \in \excontracts{\run'}$. 

    We next show that $\schecontracts{\run} = \schecontracts{\run'}$. 
    To this end, we first show that 
    $\schecontracts{\run'} \subseteq \schecontracts{\run}$.
    Assume towards contradiction that there is some $$(\dotCTLCsubcontract, \funsecret_i(\dotCTLCsubcontract)) \in \schecontracts{\run'} \backslash \schecontracts{\run}.$$
    Then this can only be the case if $\run' \overset{\action'}{\longrightarrow} \environmentvec^*$ for some $\environmentvec^*$ and $\action' = \DecideCo(\dotCTLCcontract, \dotCTLCsubcontract, \funsecret_i(\dotCTLCsubcontract))$ 
    but there is no $\environmentvec^*$ such that $\run \overset{\action'}{\longrightarrow} \environmentvec^*$. 
    By the definition of the $\CoEx$ rule (\Cref{ctlc:CoEx}) this could only be the case if $\funfund(\dotCTLCcontract) = \funfund(\CTLCcontract)$ (so the action is not possible anymore because the required funds got removed from $\environmentvecprimeAvailableFunds$).
    This is ruled out by the uniqueness of funds from \cref{def:initialconfig}.
    
    We next show that  
    $\schecontracts{\run} \subseteq \schecontracts{\run'}$.
    Assume towards contradiction that there is some $$(\dotCTLCsubcontract, \funsecret_i(\dotCTLCsubcontract)) \in \schecontracts{\run} \backslash \schecontracts{\run'}.$$ 
    This could only be the case if $\run \overset{\action'}{\longrightarrow} \environmentvec^*$ for some $\environmentvec^*$ and 
    $$\action' = \DecideCo(\dotCTLCcontract, \dotCTLCsubcontract, \funsecret_i(\dotCTLCsubcontract))$$
    but there is no $\environmentvec^*$ such that $\run' \overset{\action'}{\longrightarrow} \environmentvec^*$. 
    However, this cannot be the case since the $\CoEx$ rule only removes elements from $\environmentvecprimeCTLCDecided$ and $\environmentvecprimeAvailableFunds$.
    Note in particular 
    $$(\CTLCsubcontract, \funsecret_i(\CTLCsubcontract)) \in \schecontracts{\run'}$$
     by \cref{lemma:anothersomelemmainG}. 

    \noindent b) From the aforementioned changes in $\environmentvec$ compared to $\environmentvecprime$ only $Inv_{levels}$, $Inv_{deposits}$ and $Inv_{tree}$ are affected.
    \begin{itemize}
        \item $Inv_{\text{levels}}(\run, \partialEx', \fulltreeobj)$ holds since the preconditions have not changed and $\environmentchReservedFunds \subseteq \environmentchprimeReservedFunds$.
        \item $Inv_{\text{deposits}}(\run, \partialEx', \fulltreeobj)$ is implied by the $\CoEx$ rule (\ref{ctlc:CoEx}). Let $\CTLCcontract \in \environmentchprimeCTLCDecided \backslash \environmentchCTLCDecided$ then
        \[
            \environmentchprimeReservedFunds \backslash \environmentchReservedFunds = \{ \funfund(\CTLCcontract) \} .
        \]
        Therefore, it still holds for all other contracts based on \linebreak $Inv_{\text{deposits}}(\run', \partialEx', \fulltreeobj)$ and the uniqueness of funds (\cref{def:initialconfig}).
        \item $Inv_{\text{tree}}(\run, \partialEx', \fulltreeobj)$ is also implied by I.H. because only $\environmentvecCTLCDecided$ got smaller. With $\environmentvecCTLCDecided \subseteq \environmentvecprimeCTLCDecided$ this invariant applies to fewer contracts now and is unchanged for the remaining. Therefore, it still holds. 
    \end{itemize}
    
    \noindent \textbf{Case $\action = \elapse \, \delta$:} 
    
    \noindent We show that 
    \begin{enumerate}[a\hspace{1pt}]
        \item \hspace{-5pt}) $\run \IDrelation \partialExFamilyprime$,
        \item \hspace{-5pt}) $Inv(\run, \partialEx', \fulltreeobj)$ hold.
    \end{enumerate}
    As $\partialExFamilyprime$ stays the same for $\run'$ and $\run$ we know by I.H. that all $\partialEx'$ are consistent. 
    
    \noindent a) For all $id$ the equality 
    \[
        \excontracts{\run} = \excontracts{\run'}.
    \]
    directly follows from its definition since $\action$ is not a $\DecideCo$ action. To show 
    \[
        \schecontracts{\run} = \schecontracts{\run'}.
    \]
    for all $id$
    we look at the rule for $\DecideCo$ (\ref{ctlc:DecideCo}).
    Here, we see that it does not depend on $\environmentvectime$.
    Hence, since none of the actions relevant for $\schecontracts{\run}$ are time-dependent, this set is not influenced by $\action = \elapse \, \delta$. 
    
    \noindent b) The only invariants affected by $\environmentvectime > \environmentvecprimetime$
    are
    $Inv_{\text{levels}}$, $Inv_{\text{liveness}}$, $Inv_{\text{init-liveness}}$, and $Inv_{\text{auth}}$. The following arguments apply to all $\partialEx' \in \partialExFamilyprime$. 
    
    \noindent \underline{$Inv_{\text{levels}}(\run, \partialEx', \fulltreeobj)$}: 
    By Definition \ref{def:adversarystrategy} of adversarial strategies, we know that an $\elapse \, {\delta}$ action can only be executed if all honest users proposed an action $\elapse \, {\delta_i}$ such that $\delta_i \geq \delta$. 
    So, in particular, we know that $\Bstrategy(\run') = \elapse \, \delta_B $ with $\delta_B \geq \delta$ from the honest user strategy (\ref{def:honeststrategy}). 
    Now let $e \in \untree$ with $B = \funsender(\e)$ and 
    \begin{equation} \label{eq:starstarstar}
    \environmentvectime > t_0 + \fundepth(\e) \Delta 
    \end{equation}
    We show that $\nexists \channel, \CTLCcontract \in \environmentchCTLCEnabled : h(\e, id) \in \CTLCcontract$ with $\funfund(\CTLCcontract) \in \environmentchReservedFunds$.
We proceed by case distinction on $\environmentvecprimetime > t_0 + \fundepth(\e) \Delta$.

\noindent \underline{$\environmentvecprimetime > t_0 + \fundepth(\e) \Delta$ :} In this case, the claim immediately follows from 
$Inv_{\text{levels}}(\run', \partialEx', \fulltreeobj)$, nothing has changed.

\noindent \underline{$\environmentvecprimetime < t_0 + \fundepth(\e) \Delta$ :} 
In this case, by definition of $\Bstrategy$, we know that 
$\delta_B = t_0 + j \Delta - \environmentvecprimetime$
for the minimal $j$ such that $t_0 + j \Delta - \environmentvecprimetime > 0$ (since the honest user strategy makes time progress only to the next time step).
If $j \leq \fundepth(\e)$, then $\environmentvectime \leq t_0 + \fundepth(\e) \Delta$ since 
$\environmentvectime 
= \environmentvecprimetime + \delta 
\leq \environmentvecprimetime + \delta_B 
= \environmentvecprimetime + t_0 + j \Delta - \environmentvecprimetime 
\leq t_0 + \fundepth(\e) \Delta$.
With this, we have a contradiction to (\ref{eq:starstarstar}).
If $j > \fundepth(\e)$ then 
$t_0 + \fundepth(\e) \Delta - \environmentvecprimetime \leq 0$
and so $\environmentvecprimetime \geq t_0 + \fundepth(\e) \Delta$ immediately contradicting the assumption of the case. 

\noindent \underline{$\environmentvecprimetime = t_0 + \fundepth(\e) \Delta$ :} 
In this case, we show that $\Bstrategy$ would schedule an action different from $\elapse~ \delta$ if 
$$\exists \channel, \CTLCcontract \in \environmentchCTLCEnabled : h(\e, id) \in \CTLCcontract$$
with $\funfund(\CTLCcontract) \in \environmentchReservedFunds$ holds.
Assume towards contradiction that 
$\exists \channel, \CTLCcontract \in \environmentchCTLCEnabled : h(\e, id) \in \CTLCcontract$ with $\funfund(\CTLCcontract) \in \environmentchReservedFunds$. 
By definition of $h(\e, id)$ (\ref{def:hMapping}) we have 
\[
    \timeout(h(\e, id)) = t_0 + \fundepth(\e) \Delta .
\]
Consequently it holds that $\timeout(h(\e, id)) \leq \environmentvecprimetime$.
Based on \cref{lemma:atleastoneitemincontracts} we distinguish the cases whether $\size{\CTLCcontract} = 1$ or $\size{\CTLCcontract} > 1$.

If $\size{\CTLCcontract} = 1$ by definition of $\widetilde{\Sigma}_{B}^{\e}$ (\Cref{eq:CTLC-timeout}), we would have that $\widetilde{\Sigma}_{B}^{\e} (\run') \ni \refund \, \CTLCcontract$
and so $\Bstrategy(\run') \not \ni \elapse(\delta_B)$.

If $\size{\CTLCcontract} > 1$, 
we show that also there is no $\dotCTLCsubcontract \in \CTLCcontract$ with 
$$\funposition(\dotCTLCsubcontract) < \funposition(h(\e, id)) .$$
If there would be such a contract, we would know by $Inv_\textit{tree}$ that there is some $\e' \in \untree$ such that $\dotCTLCsubcontract = h(\e', id)$. 
Further, we would know that 
$$\funsender(\e') = \funsender(\dotCTLCsubcontract) = \funsender(\e) = B .$$ 
By well-formedness of CTLCs (\cref{ctlc:well-formedness}) we would further know that also 
$\timeout(\dotCTLCsubcontract) <  \timeout(h(\e, id))$
and so consequently that $\fundepth(\e') < \fundepth(\e)$. 
So also $\environmentvecprimetime > t_0 + \fundepth(\e') \Delta$
Applying $Inv_{\text{levels}}(\run', \partialEx', \fulltreeobj)$ would immediately contradict that $h(\e', id) = \dotCTLCsubcontract \in \CTLCcontract \in \environmentchprimeCTLCEnabled = \environmentchCTLCEnabled $.
Consequently, by the definition of $\widetilde{\Sigma}_{B}^{\e'}$ (\Cref{eq:CTLC-timeout}), we would have that $\widetilde{\Sigma}_{B}^{\e'} (\run') \ni \timeoutsc(\CTLCcontract, h(\e', id))$
and so $\Bstrategy(\run') \not \ni \elapse(\delta_B)$.

\noindent \underline{$Inv_{\text{liveness}}(\run, \partialEx', \fulltreeobj)$}: 
    
\noindent Let 
    \begin{align*}
        &\eout = (B,X)_{\walk} \in \untree, j = \fundepth(\eout), \\
        &\DecideCo(\CTLCcontract, h(\eout, id), h_{sec}(\eout, id)) \in \funactions(\run) \text{ and } \\
        &\ein = (Y,B)_{\walk'} \in \untree \text{ with } \walk' = \concatvec{\unitvec{(Y,B)_{\walk}}}{\walk}.
    \end{align*}
According to the preconditions of $Inv_{\text{liveness}}(\run, \partialEx', \fulltreeobj)$. 
We show that either
    \begin{align*}
        (i) \, &\exists j' \leq j+1, \walk'' : (Y,B)_{\walk''} \in \partialEx' \\
                &\land \fundepth((Y,B)_{\walk''}) = j' \text{ \textbf{or} }\\
        (ii)\, &\exists \CTLCcontract \in \environment_{\channeledge{\ein}}.C_{en} : h(\ein, id) \in \CTLCcontract \\
                &\land \environmentvectime < t_0 + (j+1)\Delta.
    \end{align*}
From $Inv_{\text{liveness}}(\run', \partialEx', \fulltreeobj)$ we know that either 
    \begin{align*}
    \text{(a) } &\exists j' \leq j+1, \walk'' : (Y,B)_{\walk''} \in \partialEx' \\
        &\land \fundepth((Y,B)_{\walk''}) = j' \text{ or }\\
    \text{(b) } &\exists \CTLCcontract \in \environment'_{\channeledge{\ein}}.C_{en} : h(\ein, id) \in \CTLCcontract \\
        &\land \environmentvecprimetime < t_0 + (j+1)\Delta.
    \end{align*}
In the case of (a), (i) is implied immediately as they are the same.  
In the case of (b), we also have 
\[
\forall \environmentch \in \environmentvec: \environmentchCTLCEnabled = \environmentchprimeCTLCEnabled ,
\]
from the \textit{elapse} $\delta$ rule (\ref{ctlc:elapsetime}). Thus (ii) holds whenever $$\environmentvectime < t_0 + (j+1)\Delta .$$
We now show that whenever $\environmentvectime \geq t_0 + (j+1)\Delta$ statement (i) is implied, which concludes the proof. 
For this, we argue that $\ein \in \partialEx'$ which fulfills (i) with $j' = j + 1$ and $\walk' = \walk''$ except if $\funnodupl (\ein, \run)$ is not true. The case in which $\funnodupl (\ein, \run)$ is not true will be dealt with individually at the end of the proof. 
    
To do so, due to $Inv_{\text{in-secrets}}(\run, \partialEx', \fulltreeobj)$, it is sufficient to show that $\funsecret(\ein) \in \environmentvecRevealedSecrets$. 

Assume towards contradiction that $\funsecret(\ein) \notin \environmentvecRevealedSecrets$. 
We will show that in this case $\Bstrategy(\run') \not \ni \elapse(\delta_B)$, which would contradict that $\action = \elapse(\delta)$.
From (b) we know that $$\exists \CTLCcontract \in \environment'_{\channeledge{\ein}}.C_{en} : h(\ein, id) \in \CTLCcontract .$$
For this $\CTLCcontract$ we continue by case distinction based on the existence of another sub-contract with lower timeout, in other words we distinguish the cases in which $h(\ein, id)$ is the current top-contract in $\CTLCcontract$ and in which it is not. 

Assume that there \underline{exists} a subcontract $\CTLCsubcontract \in \CTLCcontract$ with 
\begin{equation} \label{somehelperthingshereandthere}
    \timeout(\CTLCsubcontract) < \timeout(h(\ein, id)).
\end{equation}
From $Inv_{\text{tree}}(\run', \partialEx', \fulltreeobj)$ we know $\exists \e \in \untree$ with $\CTLCsubcontract = h(\e, id)$ and consequently $\fundepth(\e)=: j_{\e} < j + 1$
since we know by construction of $h$ (\ref{def:hMapping}) that $\funtimeout(\CTLCsubcontract) = t_0 + j_{\e} \Delta$.
Further, we have that $\environmentvectime = \environmentvecprimetime + \delta$ and $\delta \leq \Delta$ and $\environmentvectime \geq t_0 + (j+1)\Delta$.
So it holds that $\environment'.t = \environmenttime - \delta \geq  \environmenttime - \Delta \geq  t_0 + \Delta j \geq t_0 + j_{\e} \Delta$. 
Consequently, we would have that $\widetilde{\Sigma}_{B}^{\eout} (\run') \ni \timeoutsc(\CTLCcontract, h(\eout, id))$
and so $\Bstrategy(\run') \not \ni \elapse(\delta_B)$, see (\ref{eq:CTLC-timeout}).

Assume that there \underline{exists no} subcontract $\CTLCsubcontract \in \CTLCcontract$ fulfilling \eqref{somehelperthingshereandthere}.
Since $\DecideCo(h(\eout, id)) \in \funactions(\run)$, we know that $\DecideCo(h(\eout, id)) \in \funactions(\run')$ and hence also (since $\run' \IDrelation \partialEx'$) that $\eout \in \partialEx'$.
From $Inv_{\text{secrets}}(\run', \partialEx', \fulltreeobj)$, we know that for all 
$\e' \in \funonPathtoRoot(\untree, \eout)$ it holds
$$\funsecret(\e', id) \in \environment_{\channeledge{\eout}}.S_{rev} .$$ 
Therefore $\chContract(\ein, \run) = 2$ and if \linebreak $\funnodupl (\ein, \run)$ is true all preconditions for $\widetilde{\Sigma}_{B}^{\ein}(\run') = \revealSecret_{\channel}(\funsecret(\ein, id))$ with $\channel = \channeledge{\ein}$, see (\ref{eq:CTLC-execution}), are fulfilled. Hence $\Bstrategy(\run')$ outputs this action instead of $\elapse(\delta_B)$. 

If $\funnodupl (\ein, \run)$ is not true we have
\begin{align*}
    \exists \e'' \intree \untree : \e'' \neq \ein
    & \land \funsender(\e'') = \funsender(\ein) \\
    & \land \funreceiver(\e'') = \funreceiver(\ein) \\
    & \land \funsecret(\e'', id) \in \environmentvecRevealedSecrets . 
\end{align*}
From $Inv_{\text{in-schedule}}(\run', \partialEx', \fulltreeobj)$ we know then $\e'' \in \partialEx'$.
        Since $\run' \IDrelation \partialEx'$, we know that either $\alpha' \in \run'$ or $\run' \overset{\alpha'}{\longrightarrow} \environmentvec^*$ for some $\environmentvec^*$ for $\alpha' = \DecideCo(\dotCTLCcontract, h(\e'', id), h_{sec}(\e'', id))$. 
        If $\alpha' \in \run'$ then we know that there was some prefix $\run^*$ of $\run'$ such that \linebreak $\run^* \overset{\alpha'}{\longrightarrow} \environmentvec^\dagger$ for some $\environmentvec^\dagger$ and hence by the inference rule of $\DecideCo$ (\Cref{ctlc:DecideCo}), 
        also $h(\e'', id) \in \dotCTLCcontract$. Since $\CTLCcontract$ has the same identifier as $\dotCTLCcontract$, it cannot be enabled after $\alpha'$ due to \cref{lemma:againanotherlemmainG}. 
        So we are left to consider the case that $\run' \overset{\alpha'}{\longrightarrow} \environmentvec^*$. 
        In this case by the inference rule of $\DecideCo$ (\Cref{ctlc:DecideCo}), we know that all preconditions for the execution of $h(\e'', id)$ are satisfied.
        In addition, using $Inv_{in-schedule}$ for $\funsecret(\e'', id) \in \environmentvecRevealedSecrets = \environmentvecprimeRevealedSecrets$, we can conclude that the corresponding outgoing edges of $\e''$ have been claimed in $\run'$.
        We further can show that $\environmentvecprimetime < \timeout(h(\e'', id))$ because by $Inv_\textit{init-liveness}$ we know that either $\alpha' \in \run'$ (leading to a contradiction as shown above) or $\environmentvecprimetime < \timeout(h(\e'', id))$.
        Therefore, $\chContract(\e'', \run) = 2$ and all preconditions for $\widetilde{\Sigma}_{B}^{\e''}(\run') = \alpha'$, see (\ref{eq:CTLC-execution}), are fulfilled. Hence, $\Bstrategy(\run')$ outputs this $\alpha'$ instead of $\elapse(\delta_B)$. 

\noindent \underline{$Inv_{\text{init-liveness}}(\run, \partialEx', \fulltreeobj)$}:

\noindent Let $\e \in \partialEx'$. 
From the I.H. we know that either 
\begin{equation} \label{eq:elapse-init-liveness-one}
    \exists \CTLCcontract : \DecideCo(\CTLCcontract, h(\e, id), h_{sec}(\e, id)) \in \funactions(\run') \text{ or}
\end{equation}
\begin{equation} \label{eq:elapse-init-liveness-two}
    \environmentvecprimetime < \timeoutsc(h(\e, id)). 
\end{equation}
If \cref{eq:elapse-init-liveness-one} holds, the claim immediately follows because if 
$\DecideCo(\CTLCcontract, h(\e, id), h_{sec}(\e, id)) \in \funactions(\run')$
 also 
 $$\DecideCo(\CTLCcontract, h(\e, id), h_{sec}(\e, id)) \in \funactions(\run)$$ (since $\actions{\run'} \subseteq \actions{\run}$).
Assume that \cref{eq:elapse-init-liveness-two} holds. 
By I.H. we also know that $\run'  \IDrelation \partialEx'$. 
Consequently, 
\begin{equation*}
\scalebox{0.9}{$
(h(\e, id), h_{sec}(\e, id)) \in \excontracts{\run'} \cup \schecontracts{\run'} .
$}
\end{equation*}
We distinguish the cases 
$$(h(\e, id), h_{sec}(\e, id)) \in \excontracts{\run'} \text{ and}$$ 
$$(h(\e, id), h_{sec}(\e, id)) \in \schecontracts{\run'} .$$ 
If $(h(\e, id), h_{sec}(\e, id)) \in \excontracts{\run'}$ then by Definition (\ref{eq:excontracts})
$$\DecideCo(\CTLCcontract, h(\e, id), h_{sec}(\e, id)) \in \actions{\run'} .$$
Next, we consider the case that 
$$(h(\e, id), h_{sec}(\e, id)) \in \schecontracts{\run'} .$$ 
By the definition of $\schecontracts{\run'}$ we hence know that either 
$\DecideCo(\CTLCcontract, h(\e, id), h_{sec}(\e, id)) \in \funactions(\run')$ or 
$\run' \overset{\alpha'}{\longrightarrow} \environmentvec^*$ for some $\environmentvec^*$,
$\alpha' = \DecideCo(\CTLCcontract, h(\e, id), h_{sec}(\e, id))$ 
 and some $\CTLCcontract$.
For the first case, it again follows from $\actions{\run'} \subseteq \actions{\run}$. 

In the second case, we know that all preconditions for executing $\action'$ are met. 
We show that in this case, the honest user strategy would schedule $\alpha'$ instead of $\elapse~ \delta$. 
Using $Inv_{in-schedule}$ for $\funsecret(\e, id) \in h_{sec}(\e, id) \subseteq \environmentvecprimeRevealedSecrets$, we can conclude that the corresponding outgoing edges of $\e$ have been claimed in $\run'$.
Therefore, $\chContract(\e'', \run) = 2$ and all preconditions for $\widetilde{\Sigma}_{B}^{\e}(\run') = \alpha'$, see (\ref{eq:CTLC-execution}), are fulfilled. Hence, $\Bstrategy(\run')$ outputs this $\alpha'$ instead of $\elapse(\delta_B)$. 

\noindent \underline{$Inv_{\text{auth}}(\run, \partialEx', \fulltreeobj)$} is implied by I.H. combined with the reasoning we applied in \cref{thingsthingsandevenmorethings}. 

\end{proof}

\paragraph{Protocol Security}

\begin{theorem}[Protocol Security] \label{th:protocolsecurity}
Let $B$ be an honest user, $\treeobj$ be a set of tuples of the form $\fulltreeobj$, which is well-formed, and $\Bstrategy$ the honest user strategy for $B$ executing $\treeobj$. Let $\Astrategy$ be an arbitrary adversarial strategy. Then for all final runs $\run$ with $\Bstrategy , \Astrategy \results \run$, starting from an initial environment, with $\environmentvec := \funlastEnv(\run)$, and for all $\fulltreeobj \in \treeobj$ there exists $\widetilde{\outcome}_{\treeid} \in \outcomeset{\untree} \cup \{ \emptyset \}$ s.t. for their family 
$$
    \wideoutcomeFamily := \{ \widetilde{\outcome}_{\treeid} \mid \fulltreeobj \in \treeobj \}
$$
it holds
    \begin{align*}
         & \forall \treeid, \e \in \widetilde{\outcome}_{id}: \e \in \hatpartialEx \\
         & \Rightarrow \exists \CTLCcontract: \DecideCo(\CTLCcontract, h(\e, id), h_{sec}(\e, id)) \in \actions{\run}
    \end{align*}
and
    \begin{align*}
         &\forall \DecideCo (\CTLCcontract ,\CTLCsubcontract,  \funsecret_i(\CTLCsubcontract)) \in \actions{\run}: \\
         &B \in \funusers(\CTLCcontract) \Rightarrow \exists \treeid, \e \in  \widetilde{\outcome}_{id}: \CTLCsubcontract = h(\e, id) \\
         & \qquad \qquad \qquad \qquad \land \funsecret_i(\CTLCsubcontract) =  h_{sec}(\e, id).
    \end{align*}
\end{theorem}

\begin{proof}
    By \cref{th:mainSecurity}, for every $\fulltreeobj \in \treeobj$ it exists a consistent $\partialEx \subseteq \untree$ s.t. for their family $\partialExFamily$ we have \linebreak $\run \IDrelation \partialExFamily$. By \cref{eq:IDrelation1} we then have 
    \begin{align} \label{somestuffweneedsomewhereelse}
        \forall & \DecideCo (\CTLCcontract ,\CTLCsubcontract,  \funsecret_i(\CTLCsubcontract)) \in \actions{\run}: \\
        &\exists \partialEx, \e \in \partialEx : \nonumber \\
        &\CTLCsubcontract = h(\e, id) \land \funsecret_i(\CTLCsubcontract) = h_{sec}(\e, id). \nonumber
    \end{align}
    Since $\run$ is final, the $\partialEx$ stay constant for all extensions of $\run$. 
    Let $\fulltreeobj \in \treeobj$ be given. For this, we now show
    \begin{align*}
        \exists \widetilde{\outcome}_{id}\in \outcomeset{\untree} \cup \{ \emptyset \} : \hatpartialEx \cap \widetilde{\outcome}_{id}= \partialEx .
    \end{align*}
    For this, we construct 
    \begin{align*}
        \widetilde{\outcome}_{id}:= \{ \e \in \untree \mid \exists \e' \in \partialEx : \e \in \funonPathtoRoot(\untree, \e') \} ,
    \end{align*}
    which, since $\partialEx$ is consistent, fulfills $\hatpartialEx \cap \widetilde{\outcome}_{id}= \partialEx$. 
    Therefore, it is left to show that
    $
        \widetilde{\outcome}_{id}\in \outcomeset{\untree} \cup \{ \emptyset \}
    $
    holds. 
    In case $\partialEx = \emptyset$ we also get $\widetilde{\outcome}_{id}= \emptyset$ by the above definition. 
    
    In case $\partialEx \neq \emptyset$ it follows immediately $\widetilde{\outcome}_{id}\neq \emptyset$. 
    Hence, it is left to show that if $\widetilde{\outcome}_{id} \neq \emptyset$ we have $\widetilde{\outcome}_{id} \in \outcomeset{\untree}$.   
    We recall from \cref{def:outcomes}:
    \begin{align*}
        \outcomeset{\untree} &:= \{ 
            \outcome \in \partialtreeoutcomes (\untree)  \mid 
            & \predNoDup{\untree}{B}{\outcome}  \\
            && ~\land~\predHonestRoot{\untree}{B}{\outcome}  \\
            && ~\land~\predEagerPull{\untree}{B}{\outcome} \}
    \end{align*}
    $\predNoDup{\untree}{B}{\widetilde{\outcome}_{id}}$ follows from $\hatpartialEx \cap \widetilde{\outcome}_{id}= \partialEx$ and the fact that $\partialEx$ is consistent and hence does not contain any duplicate edges (so since all edges involving $B$ in $\widetilde{\outcome}_{id}$ are also contained in $\partialEx$ also $\widetilde{\outcome}_{id}$ cannot contain any duplicate edges involving $B$).

    For $\predHonestRoot{\untree}{B}{\widetilde{\outcome}_{id}}$, we look at the situation where $B$ sits in the root of $\untree$.
    Assume that there is an edge $\e = (X,B)_{[X,B]} \intree \untree$.
    We show that then also $\e \in \widetilde{\outcome}_{id}$.
    Since $\partialEx$ is consistent and not empty there also needs to be an $\e^* :=(Y,B)_{[(Y,B)]} \in \untree \cap \partialEx$ for some user $Y$. 
    The edges $\e$ and $\e^*$ could coincide but don't need to. 
    Assume towards contradiction that $\e \not \in \widetilde{\outcome}_{id}$. 

    Since $\e^* \in \partialEx$ by 
    $$
        Inv_{\text{secrets}}(\run, \partialEx, \fulltreeobj)
    $$
    we know that
    $$
        s= \funsecret(\e^*, id) \in \environment_{\channeledge{\e^*}}.S_{rev}.   
    $$
    Then \cref{lemma:somelemmaaboutrevealsecret} implies 
    $$(\honestuser : \revealSecretch \, s) \in \actions{\run}$$ 
    with $\channel = \channeledge{\e}$. 
    Therefore it exist $\run^{\dagger}$ and $\run^*$ with 
    \[
        \run = \run^{\dagger} \overset{\honestuser : \revealSecretch \, s}{\longrightarrow} \run^*
    \]
    and so $(\honestuser : \revealSecretch \, s) \in \Bstrategy(\run^{\dagger})$. 
    By definition of $\Bstrategy$, see \cref{eq:CTLC-execution}, we know that $\chContract(\e^*, \run^{\dagger}) = 2$ holds in this case.
    We show that $\chContract(\e, \run^{\dagger}) = 2$
    also holds.
     The condition $\funsecretsAv(\e)$ follows immediately from the fact that $h_{sec}(\e, id)$ has only one secret owned by $B$. 
    The condition $\funisIngoing(\e, \run^{\dagger})$ follows immediately from $\funisIngoing(\e^*, \run^{\dagger})$
    and since $$\funtimeout(h(\e, id)) = \funtimeout(h(\e^*, id))$$ the $\timeout$ condition is also implied. 
    Thus it is left to show that $\funenabled(\e)$ holds. Because 
    \[
        \fundepth(\e) = \fundepth(\e^*) = 1
    \]
    and it is dictated by the definition of $h$ (\ref{def:hMapping}) that no smaller $\timeout$ is possible, we can imply that there can be no subcontract with a smaller timelock than the one of $h(e, id)$ in the same $\CTLCcontract$.

    Further, from $\funisIngoing(\e^*, \run^{\dagger})$ we know that either $h(\e, id)$ must be enabled (in $h(\e, id) \in \CTLCcontract \in \funlastEnv(\run^\dagger).C_{en}$) or \linebreak $\DecideCo(\CTLCcontract, h(\e, id), h_{sec}(\e, id)) \in \actions{\run^\dagger}$. 
    In the second case, we immediately know from $\run \IDrelation \partialEx$ that $\e \in \partialEx$ and hence $\e \in \widetilde{\outcome}_{id}$, which would contradict the assumption. 
    If $h(\e, id) \in \CTLCcontract \in \funlastEnv(\run^\dagger).C_{en}$ is enabled then $\chContract(\e, \run^{\dagger}) = 2$ holds. 
    
    We next show that $\funnodupl (\e, \run^{\dagger})$. 
    For this, we assume towards contradiction that $\funnodupl (\e, \run^{\dagger})$ does not hold. This means 
    \begin{align*}
    \exists \e' \intree \untree : \e' \neq \e &\land \funsender(\e') = \funsender(\e) \\
    & \land \funreceiver(\e') = \funreceiver(\e) \\
    & \land s' := \funsecret(\e', id) \in \environmentvec^{\dagger}.S_{rev}  
    \end{align*}
    where $\environmentvec^{\dagger} = \funlastEnv(\run^{\dagger})$. 

    By \cref{lemma:somelemmaaboutrevealsecret} we then have
    \begin{align} \label{somehelpersmallstuffneededhereorthere}
        \run^{\dagger} = \dot{\run} \overset{\honestuser : \revealSecretch \, s'}{\longrightarrow}\widetilde{\run}
    \end{align}
    for some $\dot{\run}$ and $\widetilde{\run}$.
    From \eqref{fun:isingoing23} of the honest user strategy we know $\exists \CTLCcontract: \,  h(\e, id) \in \CTLCcontract \in \environmentvec^{\dagger}.C_{en}$. 
    Hence, if there exists some $\dotCTLCcontract$ (so a contract with the same id as $\CTLCcontract$) with $ h(\e', id) \in \dotCTLCcontract \in \funlastEnv(\dot{\run}).C_{en}$, which needs to be the case for the action in (\ref{somehelpersmallstuffneededhereorthere}), we also have
    \[
        h(\e, id) \in \dotCTLCcontract \in \funlastEnv(\dot{\run}).C_{en},
    \]
    by \cref{lemma:only-top-lvl-disabled} because 
    $$\funposition (h(\e, id)) > \funposition(h(\e', id))$$ 
    by construction of $h$.
    This contradicts
    $$(\honestuser : \revealSecretch \, s') \in \Bstrategy(\dot{\run})$$
    because $\Bstrategy$ would not schedule $\honestuser : \revealSecretch \, s'$ if $h(\e', id)$ is not the top-level contract in $\dotCTLCcontract$ and hence cannot be executed. 

    Consequently, $\funnodupl (\e, \run^{\dagger})$ holds. 
    And so that $\Bstrategy$ would schedule $\honestuser : \revealSecretch \, s^*$ on $\run^{\dagger}$ given that $s^*$ has not yet been revealed.

    More formally, for $\channel = \channeledge{\e}$ we have that either 
    \begin{align*}
        &(a) \, s^* = \funsecret(\e, id) \in \environmentch^{\dagger}.S_{rev}  \text{ \textbf{or} } \\
        &(b) \, (\honestuser : \revealSecretch \, s^*) \in \Bstrategy(\run^{\dagger}). 
    \end{align*}
    In case $(a)$ from $Inv_{\text{in-schedule}}$ we get immediately that $\e \in \partialEx$ (since also $s^*  \in \environmentch.S_{rev}$ because the secrets sets monotonically increase, see \cref{lemma:secsetsmonotonic}), which would contradict our assumption $\e \not \in \widetilde{\outcome}_{id}$.

    In case (b) we argue why during the run $\run^*$ an environment must be reached where $\revealSecretch \, s^*$ is finally executed and hence we can use $Inv_{\text{in-schedule}}$ again to show that $\e \in \partialEx$ (and thus arrive at a contradiction).
    To this end,  we show that there must be $\run'$ and $\run''$ with
    \[
        \run^* = \run' \overset{\elapse \, \delta}{\longrightarrow} \run''
    \]
    and $s^* \in \environmentchprimeRevealedSecrets$. 
    Consequently $$\funsecret(\e, id) \in \environmentch''.S_{rev} = \environmentchRevealedSecrets .$$ 
    Thus $\e \in \partialEx$.
    Clearly, such $\run'$ and $\run''$ must exist since $\environmentvec''.t > \environmentvec^\dagger.t$ (because $\environmentvec^\dagger.t \leq \timeout(h(\e, id))$) and elapse is the only rule to increase time.
    Assume towards contradiction that $s^* \notin \environmentchprimeRevealedSecrets$.
    Then 
    $$(\honestuser : \revealSecretch \, s^*) \in \Bstrategy(\run')$$
    by the persistence of the honest user strategy (\ref{def:honeststrategy}) and the fact that for any $\widetilde{\environmentvec}$ and any channel $\channel$ the set $\widetilde{\environment}_{\channel}.S_{rev}$ can only increase in size (if there would have been a point in $\run'$ where $\honestuser : \revealSecretch \, s^*$ was not possible anymore this would imply starting from this point, $s^*$ would reside in the $S_{rev}$ set, contradicting $s^* \notin \environmentchprimeRevealedSecrets$).

    However, in order to execute the $\elapse \, \delta$ action, we would need to have 
    $$\{ \elapse \, \delta' \} = \Bstrategy(\run')$$
    for some $\delta'$ (since $\Bstrategy$ only schedules an $\elapse$ action if no other action is scheduled).
    This gives us the final contradiction.

    For $\predEagerPull{\untree}{B}{\widetilde{\outcome}_{id}}$ let 
    \[
    (X, B)_{\walk_1} \intree \untree ~\land~ (B,Y)_{\walk_2} \in \widetilde{\outcome}_{id} ~\land~ \walk_1 = \concatvec{\unitvec{(X, B)}}{\walk_2}
    \]
    be given. We set $j := \fundepth((B,Y)_{\walk_2})$. 
    From $(B,Y)_{\walk_2} \in \widetilde{\outcome}_{id}$, we know that also $(B,Y)_{\walk_2} \in \partialEx$
    and hence invariant \\ $Inv_{\text{liveness}}(\run, \partialEx, \fulltreeobj)$ applies. 
    Thus either
    \begin{align*}
        &1) \, \exists j' \leq \vert \walk_1 \vert, \walk_3 : (X,B)_{\walk_3} \in \partialEx \land \vert \walk_3 \vert = j' \text{ \textbf{or} } \\
        &2) \, \exists \CTLCcontract \in \environment_{\channeledge{(X, B)_{\walk_1}}}.C_{en} : h((X, B)_{\walk_1}, id) \in \CTLCcontract \\
        & \hspace{6pt} \land \environmentvectime < t_0 + (j+1)\Delta
    \end{align*}
    In the first case, the claim holds immediately, since by construction $(X,B)_{\walk_3} \in \partialEx$ implies $(X,B)_{\walk_3} \in \widetilde{\outcome}_{id}$.
    In the second case,  we have a contradiction since we know that $\run$ is final and so 
    \[
        \environmentvectime > t_0 + (j+1)\Delta = \funtimeout((X, B)_{\walk_1}).
    \]

    With that, we showed that $\widetilde{\outcome}_{id} \in \outcomeset{\untree}$ for $\widetilde{\outcome}_{id} \neq \emptyset$. 
    To show the first implication of the final statement, we still need to show that if $\e \in \widetilde{\outcome}_{id} ~\cap~ \hatpartialEx$ then also 
    $\DecideCo(\CTLCcontract, h(\e, id), h_{sec}(\e, id)) \in \actions{\run}$.
    If $\e \in \widetilde{\outcome}_{id} ~\cap~ \hatpartialEx$ then $\e \in \partialEx$ (by construction) and hence from 
    $\textit{Inv}_\textit{init-liveness}$, we know that either \linebreak $\DecideCo(\CTLCcontract, h(\e, id), h_{sec}(\e, id)) \in \actions{\run}$
    or $\environmentvectime < \timeout(h(\e, id))$
    In the first case, the claim follows immediately. 
    In the second case, we immediately arrive at a contradiction to $\run$ being final.

    The second implication of the final statement immediately follows from~ (\ref{somestuffweneedsomewhereelse}) and $\partialEx \subseteq \widetilde{\outcome}_{id}$.
    This concludes the proof.
\end{proof}

\paragraph{End-to-end Security}
When $\untree = \fununfold(\graphsymbol, A)$ is given by the unfolding process of a digraph $\graphsymbol$, which is in-semiconnected w.r.t. $A \in \nodesymbol$, 
the execution of the protocol never leaves an honest party \textit{underwater} in $\graphsymbol$. We use the same notion of \textit{underwater} as in \cref{thm:security-graph}. 
To formally state our end-to-end security statement, we assume to be provided a set $\graphobj$ of graph specifications of the form $\fullgraphobj$
where $id$ denotes a unique identifier, $\graphsymbol$ the graph to be executed, $t_0$ the execution starting time and $\graphspecalone$ is a graph specification that maps arcs to a pair $(f^{\zeta}_X, \channel)$ of funds and channels. 

We will denote with $\specFromGraph$ the tree specification induced by the graph specification $\graphspecalone$, formally defined as 
\begin{align*}
    \specFromGraph(\e) := \graphspecalone(\funsender(\e), \funreceiver(\e))
\end{align*}

Note that $\specFromGraph$ is by construction a valid tree specification. 

We now formally state our end-to-end security statement: 

\begin{theorem}[End to End Security] \label{th:endtoendsecurityhereintheappendix}
    Let $B$ be an honest user, $\graphobj$ be a set of tuples of the form $\fullgraphobj$ with digraph $\graphsymbol$ in-semiconnected w.r.t. $A \in \nodesymbol$.
    Be $\treeobj$ well-formed, see \cref{def:wellformedtree}, and given as 
    \begin{align*}
        \treeobj := \{ & \fulltreeobj ~|~ \fullgraphobj \in \graphobj ~\land~ \\
        & \qquad \qquad \specalone = \specFromGraph ~\land~ \untree = \fununfold(\graphsymbol, A) \}
    \end{align*}
    Let $\Bstrategy$ be the honest user strategy for $B$ executing $\treeobj$. Let $\Astrategy$ be an arbitrary adversarial strategy. Then, for all final runs $\run$ with $\Bstrategy, \Astrategy \results \run$, starting from an initial environment it holds:
    \begin{align*}
    \forall &\DecideCo (\CTLCcontract ,\CTLCsubcontract,  \funsecret_i(\CTLCsubcontract)) \in \actions{\run}, X \in \nodesymbol: \\ 
        & \funsender(\CTLCcontract) = \honestuser \, \land \, \funreceiver(\CTLCcontract) = X \\
        & \Rightarrow \exists \fullgraphobj \in \graphobj: ~(B,X) \in \graphsymbol \\
        & \qquad \land x = \ctlcIDtree{id}{B}{X}\\
        & \qquad \land \forall (Y,B) \in \graphsymbol \, 
        \exists \dot{c}^{x'}, \dot{sc}^{x'}, j : \ctlcID' = \ctlcIDtree{id}{Y}{B} \\
        & \qquad  \land \DecideCo (\dot{c}^{x'} ,\dot{sc}^{x'},  \funsecret_j(\dot{sc}^{x'})) \in \actions{\run}
    \end{align*}    
    Intuitively, every claim action taking money away from $\honestuser$ corresponds to an outgoing arc in $\graphsymbol$ for which all ingoing arcs were claimed.
\end{theorem}

\begin{proof}
    Let $\DecideCo (\CTLCcontract ,\CTLCsubcontract,  \funsecret_i(\CTLCsubcontract)) \in \actions{\run}$
    with \linebreak $\funsender(\CTLCcontract) = \honestuser$ and $\funreceiver(\CTLCcontract) = X$, for some user $X$. 
    From~\cref{th:protocolsecurity} we have that for every $\fulltreeobj \in \treeobj$ there is some $ \widetilde{\outcome}_{id} \in \outcomeset{\untree} \cup \{ \emptyset \}$
    such that 
    \begin{align} \label{helper:prot-sec-1}
        \forall &\e \in \widetilde{\outcome}_{id}: \e \in \hatpartialEx \\
        & \qquad \Rightarrow \exists \CTLCcontract: \DecideCo(\CTLCcontract, h(\e, id), h_{sec}(\e, id)) \in \actions{\run} \nonumber
    \end{align}
    and 
    \begin{align}\label{helper:prot-sec-2}
        &\forall \DecideCo (\CTLCcontract ,\CTLCsubcontract,  \funsecret_i(\CTLCsubcontract)) \in \actions{\run} : B \in \funusers(\CTLCcontract)\\
        & \Rightarrow \hspace{-2pt} \exists \treeid, \e \hspace{-1pt} \in \widetilde{\outcome}_{id} \hspace{-1pt} : \CTLCsubcontract \hspace{-1pt} = \hspace{-1pt} h(\e, id) \hspace{-1pt} \land \hspace{-1pt} \funsecret_i(\CTLCsubcontract) \hspace{-1pt} = \hspace{-1pt}  h_{sec}(\e, id). \nonumber
    \end{align}
    So consequently, (\ref{helper:prot-sec-2}) immediately gives us for $\fulltreeobj \in \treeobj$ that there is some $\e \in  \widetilde{\outcome}_{id} $ which fulfills 
    \begin{align}
         \CTLCsubcontract = h(\e, id) \, \land \, \funsecret_i(\CTLCsubcontract) =  h_{sec}(\e, id) 
    \end{align}
    By definition of $h$, see \cref{def:hMapping}, $\e = (B,X)_{\walk}$ for some walk $\walk$ and $\ctlcID = \ctlcIDtree{id}{B}{X}$. Since $\e \intree \untree$ we have $(B,X) \in \graphsymbol$ for 
    $$\fullgraphobj \in \graphobj .$$ 

    Let now $(Y,B) \in \graphsymbol$.
    Then \Cref{thm:security-graph} immediately gives us that there is also some $\walk'$ such that $\e' = (Y,B)_{\walk'} \in \widetilde{\outcome}_{id}$. 
    From (\ref{helper:prot-sec-1}) we then also have that $\DecideCo(\dotCTLCcontractid{\ctlcID'}, h(\e', id), h_{sec}(\e', id)) \in \actions{\run}$ for some $\dotCTLCcontractid{\ctlcID'}$ and so by definition of $h$, we know that $\ctlcID' = \ctlcIDtree{id}{Y}{B}$ 
    and $h_{sec}(\e', id) = \funsecret_\iota(h(\e', id)) $ for some $\iota$.
    This concludes the proof. 

\end{proof}

\begin{remark} \label{remark:uniqueness}
    If we assume the preconditions from \cref{th:endtoendsecurityhereintheappendix} and 
    \[
        \action = \DecideCo (\CTLCcontract ,\CTLCsubcontract,  \funsecret_i(\CTLCsubcontract)) \in \actions{\run} 
    \]
    with $\funsender(\CTLCcontract) = \honestuser$ and $\funreceiver(\CTLCcontract) = X$ then there exist $\run'$, $\run''$ s.t.
    \[
        \run = \run' \overset{\action}{\longrightarrow} \run'' .
    \]
    By applying \cref{lemma:againanotherlemmainG} we get that there are no 
    $\dotCTLCcontract$, $\dotCTLCsubcontract$, $\iota$ for which
    \[
        \action' \in \actions{\run'} \lor \action' \in \actions{\run''}
    \]
    with $\action' = \DecideCo (\dotCTLCcontract ,\dotCTLCsubcontract,  \funsecret_{\iota}(\dotCTLCsubcontract))$ holds. 
    This means that a CTLC (with sender $\honestuser$) with the same identifier can never be claimed twice during a run. 
    By construction of $h$, contracts with different sender or receiver get assigned different identifiers. 
    Therefore, in \cref{th:endtoendsecurityhereintheappendix} two different \textit{claim} actions in $\actions{\run}$ can never be linked to the same arc in $\graphsymbol$.
    This rules out that there could be several claim actions that are mapped to the same arc $(B,X)$ in $\graphsymbol$. 
\end{remark}

\paragraph{Protocol Correctness}
Next, we will prove the correctness of the tree protocol. 
Intuitively, the protocol for the tree $\untree$ is correct if, given that all users of the tree $\untree$ are honest, the protocol computes edges corresponding to the 'optimal' tree that lies in the intersection of the outcome sets of all users.

To ensure protocol correctness, some additional prerequisites need to be satisfied (beyond the users being honest)
\begin{itemize}
    \item The protocol setup phase needs to be started in time. The time $t_0$ is part of the protocol specification but denotes the time starting from when the execution of the tree can start. However, to ensure that all tree edges are set up and enabled by $t_0$, the setup process needs to start before $\Delta * \fundepth(\untree)$ since setting up each level of the tree may take up to time $\Delta$. 
    \item The funds of users (according to $\specalone$) needs to be available in the initial environment.
\end{itemize}

We further define the user set of a tree set as follows: 
\begin{align*}
    \funusers(\treeobj) := \bigcup_{ \e \in \{ \e ~|~  \exists \fulltreeobj \in \treeobj: \e \intree \untree \}} \funusers(\e) 
\end{align*}

Intuitively, liquidity states that for all edges (without duplicates) of all trees (identified by $(\funsender(\e), \funreceiver(\e), id)$) there is a unique fund as specified in the spec $\specalone$.

To prove protocol correctness, we first show that the setup process is correct, meaning that if the protocol is started on time, all tree edges will be enabled by the time the protocol's starting time $t_0$ is reached. 

\begin{theorem}[Setup Correctness] \label{th:setupcorrectness}
    Let $\honestusers =  \{ \honestuser_1, \dots, \honestuser_k\}$ be a set of honest users. 
    Let $\treeobj$ be well-formed, see \cref{def:wellformedtree}, and $\funusers(\treeobj) \subseteq \honestusers$.
    Let $\honestuserstrategies^{\treeobj} = \{ \treestrategy{\honestuser_1}, \dots, \treestrategy{\honestuser{_k}} \}$ a set of honest user strategies executing $\treeobj$.
    Let $\Astrategy$ be an arbitrary adversarial strategy (for $\honestusers$). 
    Let $\run$ with $(\honestuserstrategies^{\treeobj}, \Astrategy) \results \run$ be a final run starting from a liquid initial environment $\environmentvec_0$ for $\treeobj$, see \cref{def:liquidenv}. 
    Further, let $\environmentvec =\funlastEnv(\run)$.
    We recall from \cref{eq:depthoftreedefinition}
    $$
        \fundepth(\untree) = \textit{max} \{ \fundepth(\e) \mid \e \intree \untree \} .
    $$
    Then if the run $\run$ started on time, meaning
    \[
        \textit{min} \{ t_0 - \fundepth(\untree)\Delta \mid \fulltreeobj \in \treeobj \} > \environmentvec_0.t
    \]
    we have for all $\fulltreeobj \in \treeobj$ and $m \in \N$ with
    \[
        \fundepth(\untree) \geq m \, \land \, t_0 > \environmentvectime \geq t_0 - (\fundepth(\untree) - m) \Delta
    \]
    that the following implication holds for all $\e \intree \untree$:
    \begin{align*}
        &\fundepth(\e) + m \geq \fundepth(\untree)\\
        &\Rightarrow \exists \CTLCcontract \in \environment_{\channeledge{\e}}.C_{en} : h(\e, id) \in \CTLCcontract
    \end{align*}
    
\end{theorem}

\begin{proof}
    Let $\honestusers$, $\treeobj$, $\honestuserstrategies$ and $\run$ as stated in the Theorem with $\vert \run \vert = n$. 
    The Theorem is proven by induction on $n$.

    \noindent \underline{\textbf{Case $n = 0$:}}

    \noindent 
    Let $\fulltreeobj \in \treeobj$ be given. By definition of $\run$ and $\vert \run \vert = 0$ we have
    $\environmentvec =\funlastEnv(\run) = \environmentvec_0$.
    Let 
    \begin{align}
        &t_0 > \environmentvectime = \environmentvec_0.t \geq t_0 - (\fundepth(\untree) - m) \Delta \label{helper11112}\\
        \land \, &t_0 - \fundepth(\untree)\Delta > \environmentvec_0.t \label{helper11114}
    \end{align}
    according to the preconditions. 
    Combining these inequalities yields:
    \begin{align*}
        t_0 \hspace{-1pt} - \hspace{-1pt} (\fundepth(\untree) \hspace{-1pt} - \hspace{-1pt} m) \Delta \hspace{-1pt} \overset{(\ref{helper11112})}{\leq} \hspace{-1pt} \environmentvec_0.t = \environmentvectime \hspace{-1pt} \overset{(\ref{helper11114})}{<} \hspace{-1pt} t_0 \hspace{-1pt} - \hspace{-1pt} \fundepth(\untree)\Delta
    \end{align*}
    Thus, $t_0 - (\fundepth(\untree) - m) \Delta = t_0 - \fundepth(\untree) \Delta + m \Delta < t_0 - \fundepth(\untree)\Delta$ has to hold for a potential $m \in \N$, which therefore cannot exist. 
    Hence, the claim trivially holds. 
    
    \noindent \underline{\textbf{Case $n > 0$:}}

    \noindent 
    With $\vert \run \vert > 0$ we know that it exists a run $\run'$ with 
    $$\run = \run' \overset{\action}{\longrightarrow} \Gamma .$$ 
    By induction hypothesis (I.H.), the Theorem holds for $\run'$ and $\environmentvecprime = \funlastEnv(\run')$.

    We proceed by case distinction on $\alpha$.
    Out of the components of $\environmentvec$, only $\environmentvectime$ and $\environmentvecCTLCEnabled$ are relevant for the Theorem. According to the \CTLC{} Semantics, see Appendix~\ref{sec:inferenceRules}, the only actions influencing these components are 
    \[
        \enableCTLC, \, \enableSubC, \, \timeoutsc, \, \refund, \,\DecideCo, \, \elapse \, \delta .
    \]
    Since all users are honest, all actions are determined by the honest user strategy. 
    By $\enableCTLC$ and $\enableSubC$ only additional contracts or subcontracts get added to $\environmentchCTLCEnabled$. Thus, the claim still follows from I.H..
    According to the honest user strategy, see \cref{eq:CTLC-timeout}, $\timeoutsc$, and $\refund$ are only executed if there is an $\e \intree \untree$ with 
    $$\timeout(h(\e, id)) \leq \environmentvectime .$$
    By definition of $h$, see \cref{def:hMapping}, we have 
    $$\timeout(h(\e, id)) = t_0 + \fundepth(\e) \Delta .$$
    As it is a precondition for this Theorem that $\environmentvectime < t_0$, the actions $\timeoutsc$, and $\refund$ are not performed by honest users yet. 
    For $\DecideCo$ it is also a precondition that $t_0 \leq \environmentvectime$ holds. Thus, honest users do not schedule this action yet. 
    The only relevant action remaining is
    $
        \action = \elapse \, \delta .
    $
    As $\action = \elapse \, \delta$ does only get executed if all parties, in this case, all parties from $\honestusers$, agree. Hence for all $B \in \honestusers$ we have $\Bstrategy(\run') = \{ \elapse \, \delta \}$. According to \cref{def:timetimetime}, this is only the case if no other action is possible based on the rules of the honest user strategy. 
    In particular, this includes $newC(\e, \run) = 0$ for all $\e \intree \untree$ and hence
    $\ingoing(\e, \run) = 0$, 
    see \cref{eq:honingoing}. 

    So let $\e \intree \untree$ be given. Then there exist users $B,X \in \honestusers$ such that $\e = (B,X)_{\walk}$, which is an outgoing edge for the honest user $B$. Hence for $\ingoing(\e, \run) = 0$ inside of the honest user strategy of $B$ we need to rule out that $\ingoing(\e, \run) = 1$, see \cref{eq:honingoing}. Since $\e \neq (Y,B)_{\walk}$, i.e. it is not an ingoing edge for $B$, there are 3 remaining conditions which are all concatenated with a logical 'and' condition. 
    For the whole statement to be false, and thus implying $\ingoing(\e, \run) = 0$, at least one of them needs to be false. These conditions are:
    \begin{align*}
        & \\
        &(a)\, \forall \e' \intree \untree \text{ with } \funonpath{\e}{\e'}, \fundepth(\e') = \fundepth(\e) + 1 \\
        & \exists \widetilde{\CTLCcontract} \in \environment_{\channeledge{\e'}}.C_{en}: h(\e',id) \in \widetilde{\CTLCcontract}  \\
        & \\
        &(b)\, \environmentvec.t < t_0, \\
        & \exists \batch := \treetoCTLCBadge(\fulltreeobj, \mathcal{S}) \in \environmentvec.\batchset : \\
        & h(\e,id) \in \CTLCcontract \in \batch \land h(\e',id) \in \widetilde{\CTLCcontract} \in \batch \\
        & \\
        &(c)\, \nexists \channel : h(\e,id) \in \CTLCcontract \in \environmentchCTLCEnabled \\
    \end{align*}
    We now argue that $(a)$ and $(b)$ cannot be false, and thus $(c)$ is. The negation of $(c)$ then gives us the desired property for the statement to hold.
    
    For $(a)$ we look at the
    case where $m \in \N$ is such that
    \begin{align}
        \environmentvecprimetime &< t_0 - (\fundepth(\untree) - m) \Delta \text{ but} \label{idfjghkdfjghkldjg}\\
        \environmentvectime = \environmentvecprimetime + \delta &\geq t_0 - (\fundepth(\untree) - m) \Delta \label{idfjghkdfjghkldjg2}
    \end{align}
    since for other $m' \in \N$
    the claim follows immediately from I.H.. For these $m'$, the preconditions have not changed. 
    Since $\delta \leq \Delta$, the difference between $\environmentvecprimetime + \delta$ and $\environmentvecprimetime$ is smaller than $1 \Delta$.
    Thus, there can only be one $m \in \N$ fulfilling both (\ref{idfjghkdfjghkldjg}) and (\ref{idfjghkdfjghkldjg2}) at the same time.  So, let $m \in \N$ be fixed in this way.
    
    Additionally, we only need to consider edges $\e$ with 
    \[
        \fundepth(\e) + m = \fundepth(\untree)
    \]
    since for an edge $\e'$ with
    \[
        \fundepth(\e') + m > \fundepth(\untree)
    \]
    we know that 
    \[
        \fundepth(\e') + m' \geq \fundepth(\untree)
    \]
    holds for some $m' < m$ and thus the I.H. applies. 
    Therefore, only if the given $\e \intree \untree$ and $m \in \N$ fulfill these properties, this case is not implied directly by I.H.
    
    So let $\e$ and $m$ be this way. 
    For (a) to be false, there would need to be an $\e'$ with $\fundepth(\e') = \fundepth(\e) + 1$ for which
    \[
        \nexists \widetilde{\CTLCcontract} \in \environment_{\channeledge{\e'}}.C_{en}: h(\e',id) \in \widetilde{\CTLCcontract}. 
    \]
    This is a contradiction to the I.H. since
    \[
        \fundepth(\e) + m \hspace{-1pt} = \hspace{-1pt} \fundepth(\untree) \hspace{-1pt} \Rightarrow \hspace{-1pt} \fundepth(\e') + (m-1) \hspace{-1pt} \geq \hspace{-1pt} \fundepth(\untree)
    \]
    and so the I.H. applies for $m-1$
    and hence, $h(\e',id)$ would be enabled. 
    Therefore, this cannot be the case. 

    For $(b)$ to be false, we either have $\environmentvec.t \geq t_0$, here the assumption would be contradicted, or 
    \begin{align*}
        &\nexists \batch := \treetoCTLCBadge(\fulltreeobj, \mathcal{S}) \in \environmentvec.\batchset : \\
        & h(\e,id) \in \CTLCcontract \in \batch \land h(\e',id) \in \widetilde{\CTLCcontract} \in \batch .
    \end{align*}
    This means that there is no batch representing all edges in $\untree$. In this case, we have $newBatch(\untree, \run') = 2$, since $\environmentvecBatches = \environmentvecprimeBatches$ and \cref{lemma:fundsAvailable} ensures that all necessary funds are still available. Therefore we have 
    \[ 
        \widehat{\Sigma}_{B}^{\untree} (\run') = \Bigl\{ \advBatch \, \batch \Bigr\},
    \]
    see \cref{eq:honprep}. By \cref{def:timetimetime}, this implies that the honest user strategy outputs this action instead of $\elapse \, \delta$:
    \[
        \advBatch \, \batch \, \in \, \Bstrategy(\run')
    \]
    And since all users need to agree in order to elapse time, we reach a contradiction. 

    Thus, the only possibility for $\ingoing(\e, \run) = 0$ to hold is that $(c)$ is false. This means
    \[
         \exists \channel : h(\e,id) \in \CTLCcontract \in \environmentchCTLCEnabled .
    \]
    Since all users are honest, \cref{eq:enable-check} dictates $\channel = \channeledge{\e}$, which concludes the proof.
\end{proof}

Using the correctness of the setup, we can prove protocol correctness. 
Intuitively, protocol correctness states that if all users are honest, and the protocol specification is consistent with the blockchain state, and execution is started in time, then the final run will reflect an execution corresponding to the ideal outcome $\outcome^{*}_{\treeid}$ that lies in the intersection of the outcome sets of all honest users.

\begin{theorem}[Protocol Correctness]
    \label{th:protocolcorrectness}
    Let $\honestusers =  \{ \honestuser_1, \dots, \honestuser_k\}$ be a set of honest users. 
    Let $\treeobj$ be a well-formed set of tuples of the form $\fulltreeobj$ and $\funusers(\treeobj) \subseteq \honestusers$.
    Let $\honestuserstrategies^{\treeobj} = \{ \treestrategy{\honestuser_1}, \dots, \treestrategy{\honestuser{_k}} \}$ a set of honest user strategies executing $\treeobj$.
    Let $\Astrategy$ be an arbitrary adversarial strategy (for $\honestusers$). 
    Let $\run$ with $(\honestuserstrategies^{\treeobj}, \Astrategy) \results \run$ be a final run starting from an initial configuration $\environmentvec_0$ that is liquid w.r.t. $\treeobj$. 
    Further, let 
    $
        \environmentvec_0.t < t_0 - \fundepth(\untree)\Delta ,
    $
    then for all $\fulltreeobj \in \treeobj$ it exists $\outcome^{*}_{\treeid} \in \bigcap_{\honestuser \in \honestusers} \outcomeset{\untree}$  s.t. for their family \linebreak
    $
        \{ \outcome^{*}_{\treeid} \}_{\treeobj} := \{ \outcome^{*}_{\treeid} \mid \fulltreeobj \in \treeobj \}
    $
    it holds
    \begin{align*}
      & \bigl ( \forall \treeid, \e \in \outcome^{*}_{\treeid} \\
      & ~\exists \CTLCcontract: \DecideCo(\CTLCcontract, h(\e, id), h_{sec}(\e, id)) \in \actions{\run} \bigr ) \\
     \land~ &  \bigl ( \forall \DecideCo (\CTLCcontract ,\CTLCsubcontract,  \funsecret_i(\CTLCsubcontract)) \in \actions{\run}: \\
     & ~\funusers(\CTLCcontract) \cap \honestusers \neq \emptyset \\
     & \qquad \Rightarrow \exists \treeid, \e \in \outcome^{*}_{\treeid}: \CTLCsubcontract = h(\e, id) \\
     & \qquad \qquad \hspace{3pt} \land \, \funsecret_i(\CTLCsubcontract) =  h_{sec}(\e, id)  \bigr ) .
    \end{align*}
\end{theorem}

\begin{proof}
    Let $\fulltreeobj \in \treeobj$ be given. We define
    \begin{equation*}
    \scalebox{0.95}{$
    \begin{aligned}
        &\outcome^{*}_{\treeid} := \\
        &\{ \e \intree \untree \mid \exists \CTLCcontract : \DecideCo(\CTLCcontract, h(\e, id), h_{sec}(\e, id)) \in \actions{\run} \}.
    \end{aligned}
    $}
    \end{equation*}
    For the statement, it is sufficient to show that 
    \begin{align*}
        \outcome^{*}_{\treeid} \in \bigcap_{\honestuser \in \honestusers} \outcomeset{\untree}. 
    \end{align*}
    To this end, we need to show the following properties:
    \begin{enumerate}[a)\hspace{1pt}]
        \item $\forall \e \in \outcome^{*}_{\treeid} \forall \e' \in \funonPathtoRoot(\untree, \e): \, \e' \in \outcome^{*}_{\treeid} $
        \item $\forall \honestuser \in \honestusers: \, \predNoDup{\untree}{B}{\outcome^{*}_{\treeid}}$
        \item $\forall \honestuser \in \honestusers: \, \predHonestRoot{\untree}{B}{\outcome^{*}_{\treeid}}$
        \item $\forall \honestuser \in \honestusers: \, \predEagerPull{\untree}{B}{\outcome^{*}_{\treeid}}$
    \end{enumerate} 
    
    a) 
    Assume towards contradiction 
    \[
        \exists \e \in \outcome^{*}_{\treeid} \exists \e' \in \funonPathtoRoot(\untree, \e): \, \e' \notin \outcome^{*}_{\treeid} .
    \]
    Then there needs to be an $\e^* = (X,Y)_{\walk^*} \in \outcome^{*}_{\treeid}$ and \linebreak $\e^{\dagger} = (Y,Z)_{\walk^{\dagger}} \in \funonPathtoRoot(\untree, \e)$ with $\e^{\dagger} \notin \outcome^{*}_{\treeid}$ and $\walk^* = \concatvec{\unitvec{(X,Y)}}{\walk^{\dagger}}$. 
    Since $\e^* \in \outcome^{*}_{\treeid}$ implies 
    \[
        \DecideCo(\CTLCcontract, h(\e^*, id), h_{sec}(\e^*, id)) \in \actions{\run}.
    \]
    we can use \cref{th:protocolsecurity} with $\e^* \in \hatpartialExY$ (and hence \linebreak $Y \in \funusers(h(\e^*, id)) = \funusers(\CTLCcontract)$) to conclude that there is some $\widetilde{\outcome}_Y \in \outcomesetvar{Y}{\untree}$ such that $\e^* \in \widetilde{\outcome}_Y$.

    By \cref{def:outcomes} of outcome sets, we then have \linebreak $\e^{\dagger} \in \widetilde{\outcome}_Y$ (since outcome sets are constructed from partial trees) which implies (by~\cref{th:protocolsecurity})
    \[
        \DecideCo(\CTLCcontract, h(\e^{\dagger}, id), h_{sec}(\e^{\dagger}, id)) \in \actions{\run}.
    \]
    Therefore $\e^{\dagger} \in \outcome^{*}_{\treeid}$ by construction of $\outcome^{*}_{\treeid}$. 

    b) 
    Let $\honestuser \in \honestusers$.
    Assume towards contradiction that 
    $\e = (X,Y)_{\walk} \in \outcome^{*}_{\treeid}$ and $\e' = (X,Y)_{\walk'} \in \outcome^{*}_{\treeid}$ and $B \in \{ X, Y \}$ and $\walk \neq \walk'$. 
    By the definition of $\outcome^{*}_{\treeid}$ we hence know that 
    $\DecideCo(\CTLCcontract, h(\e, id), h_{sec}(\e, id)) \in \actions{\run} $
    and 
    $\DecideCo(\CTLCcontract, h(\e', id), h_{sec}(\e', id)) \in \actions{\run}$.
    From \cref{th:protocolsecurity} we know that there is some $\widetilde{\outcome}_B \in \outcomesetvar{B}{\untree}$
    such that 
    \begin{align} \label{helper-sec-dup-1}
        &\forall \e \in \hatpartialEx: \e \in \widetilde{\outcome}_B\\
        &\qquad \Rightarrow \exists \CTLCcontract: \DecideCo(\CTLCcontract, h(\e, id), h_{sec}(\e, id)) \in \actions{\run} \nonumber
    \end{align}
    and 
    \begin{align} \label{helper-sec-dup-2}
        &\forall \DecideCo (\CTLCcontract ,\CTLCsubcontract,  \funsecret_i(\CTLCsubcontract)) \in \actions{\run}: B \in \funusers(\CTLCcontract) \\
        &\Rightarrow \exists \e \in  \widetilde{\outcome}_B: \CTLCsubcontract = h(\e, id) \, \land \, \funsecret_i(\CTLCsubcontract) =  h_{sec}(\e, id) \nonumber
    \end{align}

    Consequently, we can conclude with \cref{helper-sec-dup-1} that 
    also $\e \in \widetilde{\outcome}_B$ and $\e' \in \widetilde{\outcome}_B$ (by definition of $h$ we know that $\funusers(\e) = \funusers(h(\e, id))$).
    However, from $\widetilde{\outcome}_B \in \outcomesetvar{B}{\untree}$ we know that also 
    $\predNoDup{\untree}{B}{\widetilde{\outcome}_B}$ and so $\walk \neq \walk'$, giving a contradiction. 

    c)
    Assume towards contradiction 
    \[
        \exists \e = (X,B)_{\unitvec{(X,B)}} \in \untree \, \land \, \e \notin \outcome^{*}_{\treeid}.
    \]
    Hence $\DecideCo(\CTLCcontract, h(\e, id), h_{sec}(\e, id)) \notin \actions{\run}$.

    Let $\funlastEnv(\run) =: \environmentvec$. Since $\run$ is final, see \cref{def:final run}, we have
    $
        t_0 + \fundepth(\untree) \Delta < \environmentvectime .
    $
    By the precondition, we also have
    \[
        \environmentvec_0.t \leq t_0 - \fundepth(\untree)\Delta 
    \]
    for the initial environment $\environmentvec_0$. 

    Then we know that there must have be runs $\run^1$ and $\run^2$ such that  $\run = \run^1 \overset{\elapse \, \delta }{\longrightarrow} \run^2$
    for some $\delta$ and 
    for $\run^{1a} :=  \run^1 \overset{\elapse \, \delta }{\longrightarrow} \environmentvec^2$
    and $\environmentvec^1 = \funlastEnv(\run^1)$ it holds that 
    $\environmentvec^1.t < t_0$ and $\environmentvec^2.t \geq t_0$. 
    Since we know from the definition of the honest user strategy that $\delta \leq \Delta$, we also have 
    that $\environmentvec^1.t = \environmentvec^2.t - \delta \geq t_0 - \Delta$. 

    Consequently, we can apply \cref{th:setupcorrectness} using $\fundepth(\e) = 1$ and $m:= (\fundepth(\untree) - 1)$ (then $\fundepth(\untree) - m = 1$ and $\fundepth(\e) + m = \fundepth(\untree)$) which results in 
    \begin{equation} \label{someequationsomewhereattheend}
        \exists \CTLCcontract \in \environment^1_{\channeledge{\e}}.C_{en} : h(\e, id) \in \CTLCcontract.
    \end{equation}

    Then by definition of the $\elapse$ rule, we also know that 
    \begin{equation} \label{helper:e-enabled}
        \CTLCcontract \in \environment^2_{\channeledge{\e}}.C_{en} : h(\e, id) \in \CTLCcontract.
    \end{equation}

    We now show that then also $\action := X: \revealSecretch \, s_{\unitvec{(X,B)}}^{id} \, \in \treestrategy{X}(\run^{1a})$.
    By the definition of $\treestrategy{X}$, for this we need to show that 
    \begin{enumerate}
        \item $s_{\unitvec{(X,B)}}^{id} \not \in \environment^1_{\channeledge{\e}}.S_{rev}$
        \item  $\funnodupl (\e, \run^{1a})$
        \item  $\chContract(\e, \run^{1a}) = 2$.
    \end{enumerate}
    To show the first two conditions, it is sufficient to show that the secret of no edge $\e'$ with $\funsender(\e') = X$ is revealed yet.
    Assume towards contradiction that $\funsecret(\e', id) \in \environment^1_{\channeledge{\e}}.S_{rev}$.
    Then using ~\Cref{lemma:somelemmaaboutrevealsecret} it follows that
    $X: \revealSecretch \, \funsecret(\e', id) \in \run^1$ and so there must be a prefix $\run^{1b}$ of $\run^1$ such that 
    $$\run^{1b} \overset{X: \revealSecretch \, \funsecret(\e', id) }{\longrightarrow}$$ 
    is a prefix of $\run^1$. 
    Since the action is restricted by $X$, we know that then also 
    $X: \revealSecretch \, \funsecret(\e', id) \, \in \treestrategy{X}(\run^{1b})$. 
    However, this would imply that $\funlastEnv(\run^{1b}).t \geq t_0$ (according to the definition of $\treestrategy{X}$ secrets are only revealed after $t_0$).
    This immediately gives a contradiction, because from $\environmentvec^1.t < t_0$ we have that also $\funlastEnv(\run^{1b}).t \geq t_0$ since time only increases over a run.  

    To show the last condition, we are left to show that 
    \begin{itemize}
        \item $\funenabled(\e)$
        \item  $\funsecretsAv(\e)$
        \item $\funisIngoing(\e, \run^{1b})$
        \item $t_0 \leq \environmentvec^2.t < \funtimeout(h(\e, id))$ 
    \end{itemize}

    $\funenabled(\e)$ follows from~\cref{helper:e-enabled} and construction of $h$ since $\fundepth(\e) = 1$ and hence $h(\e, id)$ is the top-level contract in $\CTLCcontract$. 
    We get $\funsecretsAv(\e)$ since by construction $h_{sec}(\e, id)$ contains only the single secret $\funsecret(\e', id)$.
    Further, $\funisIngoing(\e, \run^{1b})$ holds since $\fundepth(\e) = 1$ and we can use 
    \cref{th:setupcorrectness} with $m:= (\fundepth(\untree) - 1)$ to show for all $\ein$ with $\fundepth(\ein) = 1$  that 
    \begin{equation}
        \exists \CTLCcontract \in \environment^1_{\channeledge{\ein}}.C_{en} : h(\ein, id) \in \CTLCcontract.
    \end{equation}
    and so also 
    \begin{equation}
        \exists \CTLCcontract \in \environment^2_{\channeledge{\ein}}.C_{en} : h(\ein, id) \in \CTLCcontract.
    \end{equation}

    Finally, we already know that 
    $t_0 \leq \environmentvec^2.t$.
    This leaves us with showing that  $\environmentvec^2.t < \funtimeout(h(\e, id))$. 
    Since by definition of $h$, we have that $\funtimeout(h(\e, id)) = t_0 + \Delta$ this follows immediately from 
    $\environmentvec^2.t = \environmentvec^1.t + \delta < t_0 + \delta < t_0 + \Delta$ (since $\delta < \Delta$ and $\environmentvec^1.t < t_0$). 

    With this, we know that $X: \revealSecretch \, s_{\unitvec{(X,B)}}^{id} \, \in \treestrategy{X}(\run^{1b})$

    Since we know that $t_0 + \fundepth(\untree) \Delta < \environmentvectime$
    and $\environmentvec^2.t < t_0 + \Delta$ and $\fundepth(\untree) \geq 1$, we know that there must be $\run^3$ and $\run^4$ such that 
    $\run^2 = \run^3 \overset{\elapse \delta'}{\longrightarrow} \run^4$. 
    
    From the persistence of $\treestrategy{X}$ and the monotonicity (\Cref{lemma:secsetsmonotonic}) of the revealed secrets, we can conclude that either 
    $$X: \revealSecretch \, s_{\unitvec{(X,B)}}^{id} \, \in \treestrategy{X}(\run \overset{\elapse \, \delta }{\longrightarrow} \run^3)$$ 
    or 
    $s_{\unitvec{(X,B)}}^{id} \in \funlastEnv(\run^3).S_\textit{rev}$
    
    The first case would immediately lead to a contradiction since then it could not be (by definition of $\treestrategy{X}$) that $\treestrategy{X}$ schedules an $\elapse$ action. 
    In the second case, we know (by \Cref{lemma:secsetsmonotonic}) that also 
    $s_{\unitvec{(X,B)}}^{id} \in \environmentvec.S_\textit{rev}$. 
    Using \Cref{th:mainSecurity}, we know that there exists some $\partialEx$ such that $\run \IDrelation \partialEx$ and by $Inv_\textit{insecrets}$ that $\e in \partialEx$.
    Further, $Inv_\textit{init-liveness}$ gives us with $t_0 + \funtimeout(h(\e, id) \leq t_0 + \fundepth(\untree) \Delta < \environmentvectime$
    that $\DecideCo(\CTLCcontract, h(\e, id), h_{sec}(\e, id)) \notin \actions{\run}$, contradicting the original assumption. 

    d) 
    Let $\e = (X, B)_{\walk_1} \intree \untree $ and $\e' = (B,Y)_{\walk_2} \in \outcome^{*}_{\treeid}$ with 
    $\walk_1 = \concatvec{\unitvec{(X, B)}}{\walk_2} .$

    By the definition of $\outcome^{*}_{\treeid}$ we hence know that 
    $$\DecideCo(\CTLCcontract, h(\e', id), h_{sec}(\e', id)) \in \actions{\run}.$$
    From \cref{th:protocolsecurity} we know that there is some $\widetilde{\outcome}_B \in \outcomesetvar{B}{\untree}$
    such that 
    \begin{align}  \label{helper-eagerpull-prot-sec-1}
        &\forall \e \in \hatpartialEx: \e \in \widetilde{\outcome}_B\\
        &\Rightarrow \exists \CTLCcontract: \DecideCo(\CTLCcontract, h(\e, id), h_{sec}(\e, id)) \in \actions{\run} \nonumber
    \end{align}
    and 
    \begin{align} \label{helper-eagerpull-prot-sec-2}
        &\forall \DecideCo (\CTLCcontract ,\CTLCsubcontract,  \funsecret_i(\CTLCsubcontract)) \in \actions{\run}: B \in \funusers(\CTLCcontract) \\
        & \Rightarrow \exists \e \in  \widetilde{\outcome}_B: \CTLCsubcontract = h(\e, id) \, \land \, \funsecret_i(\CTLCsubcontract) =  h_{sec}(\e, id) \nonumber
    \end{align}

    Consequently, we can conclude with \cref{helper-eagerpull-prot-sec-2} that also $\e' \in \widetilde{\outcome}_B$.
    Since $\widetilde{\outcome}_B \in \outcomesetvar{B}{\untree}$ we know in particular that $\predEagerPull{\untree}{B}{\widetilde{\outcome}_B}$ holds and consequently also that there exists some $\walk_3$ such that $\e'' = (X,B)_{\walk_3} \in \widetilde{\outcome}_B$ and $\fundepth(\e'') \leq \fundepth(\e)$. 
    With this, we can use~\cref{helper-eagerpull-prot-sec-1} to obtain that 
    $$\DecideCo(\dotCTLCcontract, h(\e'', id), h_{sec}(\e'', id)) \in \actions{\run}$$ 
    for some $\dotCTLCcontract$ and so by construction of $\outcome^{*}_{\treeid}$ also $\e'' \in \outcome^{*}_{\treeid}$ what concludes the case.  
    
\end{proof}

\paragraph{End-to-end Correctness}
Similar to our end-to-end security statement, we formulate an end-to-end correctness statement that ensures that if a tree protocol resulting from a graph is executed by honest users then the final executions of CTLC contracts involving honest users exactly correspond to the arcs in the graph. 

\begin{theorem}[End-to-end Correctness]
    \label{th:fullprotocolcorrectness}
    \hspace{10pt} \newline Let $\honestusers =  \{ \honestuser_1, \dots, \honestuser_k\}$ be a set of honest users. 
    Let $\graphobj$ be a set of tuples of the form $\fullgraphobj$ with digraph $\graphsymbol$ in-semiconnected w.r.t. $A \in \nodesymbol$
    and $\nodesymbol \subseteq \honestusers$. 
    Be $\treeobj$ well-formed and given as 
    \begin{align*}
        \treeobj := \{ &\fulltreeobj ~|~ \fullgraphobj \in \graphobj \\
        &  \qquad ~\land~ \specalone = \specFromGraph ~\land~ \untree = \fununfold(\graphsymbol, A) \}.
    \end{align*}
    Let $\honestuserstrategies^{\treeobj} = \{ \treestrategy{\honestuser_1}, \dots, \treestrategy{\honestuser{_k}} \}$ a set of honest user strategies executing $\treeobj$.
    Let $\Astrategy$ be an arbitrary adversarial strategy (for $\honestusers$). 
    Let $\run$ with $(\honestuserstrategies^{\treeobj}, \Astrategy) \results \run$ be a final run starting from an initial configuration $\environmentvec_0$ that is liquid w.r.t. $\treeobj$. 
    Further, let
    $
        \environmentvec_0.t < t_0 - \fundepth(\untree)\Delta 
    $
    for all $\fulltreeobj \in \treeobj$. 
    Then it holds 
    \begin{equation*}
    \scalebox{0.9}{$
    \begin{aligned}
        &\bigl( \forall \DecideCo (\CTLCcontract ,\CTLCsubcontract,  \funsecret_i(\CTLCsubcontract)) \in \actions{\run}, X, Y \in \nodesymbol: \\ 
        & \funsender(\CTLCcontract) = X \, \land \, \funreceiver(\CTLCcontract) = Y \, \land \, \{ X, Y \} \cap \honestusers \neq \emptyset \\
        & \Rightarrow \exists \fullgraphobj \in \graphobj : (X,Y) \in \graphsymbol \\
        & \qquad \land x = \ctlcIDtree{id}{X}{Y} \bigr ) \\
        & \land \bigl ( \forall \fullgraphobj \in \graphobj ~\forall (X,Y) \in \graphsymbol \, 
        \exists \dot{c}^{x'}, \dot{sc}^{x'}, j : \\
        &  \ctlcID' = \ctlcIDtree{id}{X}{Y} \hspace{-1pt} \land \hspace{-1pt} \DecideCo (\dot{c}^{x'} ,\dot{sc}^{x'}, \funsecret_j(\dot{sc}^{x'})) \hspace{-1pt} \in \hspace{-1pt} \actions{\run} \bigr )
    \end{aligned}
    $}
    \end{equation*}
\end{theorem}

\begin{proof}
From~\cref{th:protocolcorrectness}, we know that for all $\fulltreeobj \in \treeobj$ there exists some $\outcome^{*}_{\treeid} \in \bigcap_{\honestuser \in \honestusers} \outcomeset{\untree}$
such that 
\begin{align} \label{helper:final-prot-correct-1}
    \forall \e \in \outcome^{*}_{\treeid} \hspace{-2pt} : \hspace{-2pt} \exists \CTLCcontract: \DecideCo(\CTLCcontract, h(\e, id), h_{sec}(\e, id)) \in \actions{\run}
\end{align}
and 
\begin{align} \label{helper:final-prot-correct-2}
    &\forall \DecideCo (\CTLCcontract ,\CTLCsubcontract, \funsecret_i(\CTLCsubcontract)) \in \actions{\run}: \nonumber \\ 
    &\funusers(\CTLCcontract) \cap \honestusers \neq \emptyset \nonumber \\
    & \Rightarrow \exists \e \in \outcome^{*}_{\treeid}: \CTLCsubcontract = h(\e, id) \, \land \, \funsecret_i(\CTLCsubcontract) =  h_{sec}(\e, id)
\end{align}

To show the first conjunction of the statement assume that $\DecideCo (\CTLCcontract ,\CTLCsubcontract,  \funsecret_i(\CTLCsubcontract)) \in \actions{\run}$,  $X, Y \in \nodesymbol$, $\funsender(\CTLCcontract) = X$, $\funreceiver(\CTLCcontract) = Y$, and $\{ X, Y \} \cap \honestusers \neq \emptyset$. 
Using~\cref{helper:final-prot-correct-2}, we can conclude that there exists some $\e \in \outcome^{*}_{\treeid}$ such that 
$\CTLCsubcontract = h(\e, id)$ and $\funsecret_i(\CTLCsubcontract) =  h_{sec}(\e, id)$.
From \Cref{theorem:correctness-graph-tree-appendix}, we immediately get that $(X,Y) \in \graphsymbol$ and by definition of $h$ that $\ctlcID' = \ctlcIDtree{id}{X}{Y}$. 

To show the second conjunction, assume that $(X,Y) \in \graphsymbol$ for some $\fullgraphobj \in \graphobj$. 
Then from~\Cref{theorem:correctness-graph-tree-appendix}, we have that there must be some $\walk$ such that $\e = (X, Y)_{\walk} \in \outcome^{*}_{\treeid}$. 
Using~\Cref{helper:final-prot-correct-1}, we can conclude that there exists some $\CTLCcontract$ with \linebreak $\DecideCo(\CTLCcontract, h(\e, id), h_{sec}(\e, id)) \in \actions{\run}$
This closes the case because we know that by construction $h(\e, id) = \CTLCsubcontract$, $h_{sec}(\e, id) = \funsecret_j(\CTLCsubcontract)$ and 
$\ctlcID = \ctlcIDtree{id}{X}{Y}$ for some $\CTLCsubcontract$ and $j \in \mathbb{N}$. 
\end{proof}

\begin{remark}
The proposed protocol can easily be extended to allow for multiple arcs in the same direction between two users. This means we would have graphs like:
    \begin{equation*}
        \xymatrix@C=3pc{
        A \ar@/^0.5pc/[r] \ar@/_0.5pc/[r] & B }
    \end{equation*}
In the unfolding process, this implies multiple, identical edges on the same level. 
The only differentiating factor is then the underlying fund. Since their fund is different, they belong to two different $\CTLC{}$s. In the presented constructions and proofs, nothing would change systematically besides adding another index to differentiate the objects. 
\end{remark}

\else 
\fi 

\end{appendices}

\end{document}